\documentclass[12pt]{article}
\usepackage[T1]{fontenc}
\usepackage[utf8]{inputenc}
\usepackage{amssymb,amsmath}
\usepackage{amsfonts}
\usepackage{mathtools}
\usepackage{amsthm}
\usepackage{upgreek}
\usepackage{enumitem}
\usepackage{amssymb}

\usepackage{graphicx,psfrag}
\usepackage{amsfonts}
\usepackage{mathtools}

\usepackage{geometry}
\usepackage[colorlinks=true,linkcolor=blue, citecolor=blue, urlcolor = blue]{hyperref}
\usepackage{lmodern}
\usepackage{avant}
\usepackage[round,sort,authoryear]{natbib}
\usepackage{comment}
\usepackage{setspace}
\usepackage{mathptmx}

\counterwithin*{equation}{section}

\theoremstyle{definition}
\newtheorem{theorem}{Theorem}
\newtheorem{definition}{Definition}
\newtheorem{assumption}{Assumption}
\newtheorem{lemma}{Lemma}
\newtheorem{example}{Example}

\newtheorem{proposition}{Proposition}
\newtheorem{remark}{Remark}
\newtheorem{corollary}{Corollary}

\newcommand\norm[1]{\left\lVert#1\right\rVert}

\DeclarePairedDelimiter{\floor}{\lfloor}{\rfloor}

\usepackage{amsmath,amssymb}
\newcommand{\vertiii}[1]{{\left\vert\kern-0.25ex\left\vert\kern-0.25ex\left\vert #1 
    \right\vert\kern-0.25ex\right\vert\kern-0.25ex\right\vert}}

\newif\ifshow 
\showfalse 

\ifshow
  \includecomment{wrap}
\else
  \excludecomment{wrap} 
\fi

\newcommand{\nc}{\newcommand}
\nc{\mbb}{\mathbb}\nc{\bb}{\mathbb}
\nc{\mbf}{\mathbf}\nc{\mb}{\mathbf}
\nc{\mc}{\mathcal}
\nc{\msf}{\mathsf}\nc{\ms}{\mathsf}
\nc{\acc}{\ms{acc}}
\nc{\ack}{\ms{ack}}
\nc{\alp}{\alpha}\nc{\al}{\alpha}\nc{\gka}{\alpha}
\nc{\ap}{\ms{ap}}
\nc{\apd}{\ms{apd}}
\nc{\base}{\ms{base}}\nc{\ba}{\ms{base}}
\nc{\bet}{\beta}\nc{\gkb}{\beta}
\nc{\boucle}{\ms{loop}}\nc{\Loop}{\ms{loop}}\nc{\lo}{\ms{loop}}
\nc{\bu}{\bullet}
\nc*{\cc}{\raisebox{-3pt}{\scalebox{2}{$\cdot$}}}
\nc{\centre}{\ms{center}}\nc{\Center}{\ms{center}}\nc{\cen}{\ms{center}}\nc{\ce}{\ms{center}}
\nc{\ci}{\circ}
\nc{\code}{\ms{code}}\nc{\cod}{\ms{code}}\nc{\decode}{\ms{decode}}\nc{\encode}{\ms{encode}}
\nc{\de}{:\equiv}
\nc{\dr}{\right}\nc{\ga}{\left}
\nc{\ds}{\displaystyle}
\nc{\ep}{\varepsilon}
\nc{\eq}{\equiv}
\nc{\ev}{\ms{ev}}
\nc{\fib}{\ms{fib}}
\nc{\funext}{\ms{funext}}\nc{\fu}{\ms{funext}}
\nc{\gam}{\gamma}
\nc{\glue}{\ms{glue}}\nc{\gl}{\ms{glue}}
\nc{\happly}{\ms{happly}}\nc{\ha}{\ms{happly}}
\nc{\id}{\ms{id}}
\nc{\ima}{\ms{im}}
\nc{\inc}{\subseteq}
\nc{\ind}{\ms{ind}}
\nc{\inl}{\ms{inl}}
\nc{\inr}{\ms{inr}}
\nc{\isContr}{\ms{isContr}}\nc{\co}{\ms{isContr}}\nc{\iC}{\ms{isContr}}\nc{\ic}{\ms{isContr}}
\nc{\isequiv}{\ms{isequiv}}\nc{\iseq}{\ms{isequiv}}\nc{\ieq}{\ms{isequiv}}
\nc{\ishae}{\ms{ishae}}\nc{\ish}{\ms{ishae}}\nc{\ih}{\ms{ishae}}
\nc{\isProp}{\ms{isProp}}\nc{\prop}{\ms{isProp}}\nc{\iP}{\ms{isProp}}\nc{\ip}{\ms{isProp}}
\nc{\isSet}{\ms{isSet}}\nc{\isS}{\ms{isSet}}\nc{\iss}{\ms{isSet}}\nc{\iS}{\ms{isSet}}\nc{\is}{\ms{isSet}}
\nc{\lam}{\lambda}
\nc{\LEM}{\ms{LEM}}\nc{\lem}{\ms{LEM}}\nc{\LE}{\ms{LEM}}
\nc{\lv}{\lvert}\nc{\rv}{\rvert}\nc{\lV}{\lVert}\nc{\rV}{\rVert}
\nc{\Map}{\ms{Map}}
\nc{\merid}{\ms{merid}}\nc{\meri}{\ms{merid}}\nc{\mer}{\ms{merid}}\nc{\me}{\ms{merid}}
\nc{\N}{\bb N}
\nc{\na}{\ms{nat}}
\nc{\nn}{\noindent}
\nc{\one}{\mb1}
\nc{\oo}{\operatorname}
\nc{\pd}{\prod}
\nc{\ps}{\mc P}
\nc{\pa}{\ms{pair}^=}
\nc{\ph}{\varphi}
\nc{\ppmap}{\ms{ppmap}}
\nc{\pr}{\ms{pr}}
\nc{\Prop}{\ms{Prop}}
\nc{\qinv}{\ms{qinv}}\nc{\qin}{\ms{qinv}}\nc{\qi}{\ms{qinv}}
\nc{\rec}{\ms{rec}}
\nc{\refl}{\ms{refl}}
\nc{\seg}{\ms{seg}}
\nc{\Set}{\ms{Set}}
\nc{\sm}{\scriptstyle}
\nc{\sms}{\ms s}
\nc{\sq}{\square}
\nc{\suc}{\ms{succ}}\nc{\su}{\ms{succ}}
\nc{\tb}{\textbf}
\nc{\then}{\Rightarrow}
\nc{\tms}{\ms t}
\nc{\tx}{\text}
\nc{\transport}{\ms{transport}}\nc{\tr}{\ms{transport}}
\nc{\two}{\mb2}
\nc{\Type}{\text-\ms{Type}}\nc{\type}{\text-\ms{Type}}\nc{\ty}{\text-\ms{Type}}
\nc{\U}{\mc U}
\nc{\ua}{\ms{ua}}
\nc{\uniq}{\ms{uniq}}
\nc{\univalence}{\ms{univalence}}
\nc{\vide}{\varnothing}
\nc{\ws}{\ms{sup}}
\nc{\zero}{\mb0}

\title{\textbf{High Dimensional Time Series Regression Models:\\ Applications to Statistical Learning Methods}\footnote{Part of these lecture notes were prepared during my Ph.D. studies at the Department of Economics, University of Southampton. Moreover, these notes were motivated by discussions and presentations during the weekly sessions of the \textit{Time Series and Machine Learning} Reading Group at the School of Mathematical Sciences, Statistics Division, University of Southampton. The author is grateful to Zudi Lu and Chao Zheng for stimulating discussions. The author is also grateful for stimulating seminar talks and discussions at the \textit{Southampton Statistical Sciences Research Institute} seminar series. This draft was prepared during the academic year 2022-2023 at the Department of Economics, University of Exeter Business School.\\ 

Dr. Christis Katsouris is a Lecturer in Economics, University of Exeter Business School, Exeter EX4 4PU, United Kingdom. \textit{Email Address}: \textcolor{blue}{christiskatsouris@gmail.com} }}

\author{\textbf{Christis Katsouris}\\ Department of Economics, University of Southampton\\ University of Exeter Business School}

\date{\today}

\begin{document}

\maketitle

\begin{abstract}
\vspace*{-0.8 em}
These lecture notes provide an overview of existing methodologies and recent developments for estimation and inference with high dimensional time series regression models. First, we present main limit theory results for high dimensional dependent data which is relevant to covariance matrix structures as well as to dependent time series sequences. Second, we present main aspects of the asymptotic theory related to time series regression models with many covariates. Third, we discuss various applications of statistical learning methodologies for time series analysis purposes. 
\end{abstract}


\maketitle

\newpage 
   
\begin{small}
\begin{spacing}{0.9}
\tableofcontents
\end{spacing}
\end{small}

\newpage

\section{Introduction}

These lecture notes series present a unified framework for recent developments of deep learning methodologies in time series econometrics and statistics. Specifically, statistical machine learning methods and theory are now widespread in social sciences, economics\footnote{The idea of representing the dynamic interactions of economic agents similar to the structure of brain neural networks dates back to \cite{ashby1957introduction}. Furthermore the application of neural networks for modelling purposes dates back to the fundamental works of \cite{hebb1949organization}  and \cite{rosenblatt1958perceptron}.}, finance and mathematical modelling. Neural networks can be considered as a type of linear sieve estimation where the basis functions themselves are flexibly learned from the data by optimizing over many combinations of simple functions (see, \cite{farrell2021deep}). Deep reinforcement learning has gain attention in recent years especially when modelling dynamic interactions of economic agents. In particular deep neural nets contain many hidden layers of neurons between the input and output layers, which has been found to exhibit superior performance across a variety of contexts. Our goal is to consolidate the theoretical and practical underpinnings of machine learning methods when modelling time series data for estimation, inference and forecasting purposes.   

In order to express precise conditions for ergodicity, we turn to the Lyapounov exponent, a concept which is well-known by those studying stability of dynamical systems. Therefore, in the context of nonlinear time series, nonstability means explosive behaviour. Thus, the Lyapounov exponent, as we define it for the state space model $\left\{ X_t \right\}$ of a time series, is given by 
\begin{align}
\gamma = \underset{ n \to \infty }{ \mathsf{lim \ inf} } \ \underset{ \norm{x} \to \infty }{ \mathsf{lim \ sup} } \ \frac{1}{n}
\end{align}
which measures the drift of the process when the sample size is large. Define with 
\begin{align}
\ell_1 &= \left\{ ( x_n )_{n=1}^{ \infty } : \sum_{ n \geq 1 } | x_n | < \infty \right\}    
\\
\ell_2 &= \left\{ ( x_n )_{n=1}^{ \infty } : \sum_{ n \geq 1 }  x_n^2 < \infty \right\}    
\end{align}
The sets $\ell_1, \ell_2$ and $\ell_{ \infty }$ are all vector spaces and it holds that $\norm{ ( x_n ) }_1 = \sum_{ n \geq 1 } | x_n |$, $\norm{ ( x_n ) }_2 = \sum_{ n \geq 1 } \left(  \sum_{ n \geq 1 } x_n^2 \right)^{ 1 / 2 }$ define norms on $\ell_1$ and $\ell_2$ spaces respectively. For any $x \in \mathbb{R}^n$, $\norm{ x } = \sqrt{ \sum_{i=1}^n x_i^2 }$, $\norm{ x }_{ \ell_1 } = \sum_{ i=1 }^n | x_i |$ and $\norm{ x }_{ \infty } = \underset{ 1 \leq i \leq n }{ \mathsf{max} } | x_i |$, denote $\ell_2$, $\ell_1$ and $\ell_{ \infty }-$norms respectively. We also denote with $\norm{ x }_{ \ell_0 } = \sum_{ i=1 }^n \mathbf{1} \left\{ x_i \neq 0 \right\}$ which is simply the number of non-zero entries of $x$. For any $n \times n$ matrix $M$, $\norm{ M }_{  \ell_{ \infty } } =  \mathsf{max}_{ 1 \leq i \leq n } \sum_{ j = 1 }^n | M_{i,j} |$ denotes the induced $\ell_{ \infty }$ matrix norm. Detailed statistical theory for high dimensional statistics applications are presented in \cite{wainwright2019high} (see also \cite{vershynin2018high}). Let $\left\{ X_{t,n} \right\}_{ t = 1}^n$, $n \in \mathbb{N}$ be a random array and $\left\{ \mathcal{F}_t \right\}_{t = - \infty}^{ \infty }$ a filtration such that $X_{t,n}$ is $\mathcal{F}_t-$measurable for all $t$ and $n$. Based on the aforementioned notation we discuss important theoretical results from the statistics and econometrics literature. 

\newpage

\section{Limit theory for High Dimensional Dependent Data}

\subsection{Limit theory of  Covariance Matrices for Linear Processes}

\subsubsection{Covariance and Precision matrix estimation for high dimensional time series}

Following the framework of \cite{Chen2013covariance}, suppose we have $n$ temporally observed $p-$dimensional vectors $\left( \mathbf{z}_i \right)_{ i = 1}^n$ having mean zero and covariance matrix $\Sigma_i = \mathbb{E} \left( \mathbf{z}_i, \mathbf{z}_i \right)$ whose dimension is $p \times p$. Our goal is to estimate the covariance matrices $\Sigma_i$ and their inverses $\Omega_i = \Sigma_i^{-1}$ based on the data matrix $Z_{ p \times n} = \left( \mathbf{z}_1,..., \mathbf{z}_n \right)$. In the classical situation where $p$ is fixed, $n \to \infty$ and $\mathbf{z}_i$ are mean zero independent and identically distributed $\textit{i.i.d}$ random vectors, it is well known that the sample covariance matrix
\begin{align}
\hat{\Sigma}_n = \frac{1}{n} \sum_{ i = 1 }^n \mathbf{z}_i \mathbf{z}_i^{ \top }
\end{align}
is a consistent and well behaved estimator of $\Sigma$, and $\hat{\Omega}_n = \hat{\Sigma}_n^{-1}$ is a natural and good estimator of $\Omega$.

However, when the dimensionality $p$ grows with $n$, random matrix theory asserts that $\hat{\Sigma}_n$ is no longer a consistent estimate of $\Sigma$ in the sense that its eigenvalues do not converge to those of $\Sigma$, as the Marcenko-Pastur law states. Moreover, it is clear that $\hat{\Omega}_n$ is not defined when $\hat{\Sigma}_n$ is not invertible in the high-dimensional case with $p >> n$.  

\begin{itemize}

\item Let $T_u \left( \widehat{\Sigma}_u \right) = \mathbf{Q} \hat{ \Lambda } \mathbf{Q}^{\top} = \sum_{ j = 1}^p \hat{\lambda}_j \mathbf{q}_j \mathbf{q}_j^{\top}$ be its eigen-decomposition, where $\mathbf{Q}$ is an orthonormal matrix and $\hat{ \Lambda }$ is a diagonal matrix. For $v > 0$, consider
\begin{align}
\tilde{S}_v = \sum_{j=1}^p \left( \hat{\lambda}_j \vee v \right) \mathbf{q}_j \mathbf{q}_j^{\top},  
\end{align}
where $0 < v \leq \sqrt{p} \bar{\omega}$ and $\omega^2$ is the rate of convergence. 

\item Let $\mu_1,...,\mu_p$ be the diagonal elements of $\mathbf{Q}^{\top} \Sigma \mathbf{Q}$. Then, by Theorem 2.1 in \cite{Chen2013covariance}, it holds that $\sum_{j=1}^p \left( \hat{\lambda}_j - \mu_j \right)^2 \leq p^2 \bar{\omega}^2$, and consequently
\begin{align*}
\left| \tilde{S}_v - \Sigma \right|^2_F 
\leq 
2 \left| \tilde{S}_v - T_u \left( \hat{\Sigma} \right) \right|^2_F + 2 \left| T_u \left( \hat{\Sigma}_u \right) - \Sigma \right|_F^2 
&\leq 
2 \sum_{j=1}^p \left( \hat{\lambda}_j - \left( \hat{\lambda}_j  \vee v \right) \right)^2 + 2 \bar{ \omega }^2 p^2 
\\
&\leq 
2 \sum_{j=1}^p \left( 2 \hat{\lambda}_j^2 \mathbf{1} \left\{ \hat{\lambda}_j \leq 0 \right\} + 2 v^2 \right) + 2 \bar{ \omega }^2 p^2  .
\end{align*}

\end{itemize}
Furthermore, if $\hat{\lambda}_j \leq 0$, since $\mu_i \geq 0$, we have that $\left| \hat{\lambda}_j \right| \leq \left| \hat{\lambda}_j  - \mu_i \right|$. Then,  
\begin{align}
\left| \tilde{S}_v - \Sigma \right|_F^2 \leq 4 v^2 p + 6 \bar{\omega} p^2 \leq 10 \bar{\omega}^2 p^2. 
\end{align}

\newpage

Notice that the eigenvalues of $\tilde{S}_v$ are bounded below by $v$, and thus it is positive definite such that 
\begin{align*}
v = \left( p^{-1} \sum_{ j,k = 1}^p u^2 \times \mathbf{1} \left\{ \left| \hat{\sigma}_{jk} \right| \geq u \right\} \right)^{1 / 2}
\end{align*}
The same positive-definization procedure also applies to the spectral norm and its rate can be similarly preserved.

\paragraph{Sparsity Conditions}

We begin by presenting the commonly used sparsity condition defined in terms of the strong $\ell^q-$ball such that 
\begin{align}
\mathcal{G}_r \left( \mathcal{M} \right) = \left\{ \Sigma \ \bigg| \ \underset{ j \leq p }{ \text{max} } \sigma_{jj} \leq 1  ;  \underset{ 1 \leq k \leq p }{ \text{max} } \ \sum_{ j = 1 }^p \left| \sigma_{jj} \right|^r \leq \tilde{M} \right\} , 0 \leq r \leq 1. 
\end{align} 
In particular, when $r = 0$, the sparsity condition implies that $\underset{ 1 \leq k \leq p  }{ \text{max} } \sum_{ j = 1 }^p \mathbf{1} \left\{ \sigma_{jk} \neq 0 \right\} \leq \tilde{M}$.

\begin{example}
(Stationary linear process). An important special class is the vector linear process as defined
\begin{align}
\label{linear.process}
\mathbf{z}_i = \sum_{ j = 0 }^{ \infty } A_j \mathbf{e}_{i-j}
\end{align}
where $A_j$ are $p \times p$ matrices, and $\mathbf{e}_{i}$ are \textit{i.i.d} mean zero random vectors with finite covariance matrix $\Sigma_e = \mathbb{E} \left( \mathbf{e}_{i} \mathbf{e}_{i}^{\top} \right)$.  A linear process representation is a filter which ensures the dependence on the innovation sequences. Then, $\mathbf{z}_i$ exists almost surely with covariance matrix $\Sigma = \sum_{ j = 0 }^{ \infty } A_j \Sigma_e A_j^{\top}$ if the latter converges. Assume that the innovation vector $\mathbf{e}_{i} = \left( e_{1i},..., e_{pi} \right)^{\top}$, where $e_{ji}$ are \textit{i.i.d} with mean zero, variance 1 and $e_{ji} \in \mathcal{L}^{2q}, q > 2$, and the coefficient matrices $A_i = \left( a_{i,jk}  \right)_{ 1 \leq j, k \leq p }$ satisfy 
\begin{align}
\text{max}_{ j \leq p } \sum_{k=1}^p a_{i,jk}^2 = \mathcal{O} \left( i^{-2 - 2\gamma } \right), \gamma > 0
\end{align}
\end{example}
Take the AR(1) process, $\mathbf{z}_i = A \mathbf{z}_{i-1} + \mathbf{e}_i$, where $A$ is the real matrix with spectral norm $\rho(A) < 1$, it is of the form as in \eqref{linear.process} with $A_j = A^j$. and the functional dependence measure $\theta_{ i, 2q, j} = \mathcal{O} \left( \rho(A)^i \right)$.

\begin{example}
Consider the estimation framework for the graphical Lasso estimator with off-diagonal entries peranilzed by the 1-norm. In particular, for $\textit{i.i.d}$ $p-$dimensional vectors with polynomial moment condition, it can be shown that $p = \mathcal{O} \left( \left( \frac{n}{d^2} \right)^{\frac{q}{2 \tau} }  \right)$ for some $\tau > 2$, where $d$ is the maximum degree in the Gaussian graphical model, then 
\begin{align}
\frac{1}{ p^2 } \left| \hat{\Omega}_n - \Omega \right| = \mathcal{O}_p \left( \frac{s+p}{ p^2 } . \frac{ p^{2 \tau / q} }{n} \right),
\end{align}
where $s$ is the number of non-zero off-diagonal entries in $\Omega$.

\newpage

We now compare the results with the CLIME (constrained $L_1-$minimization for inverse matrix estimation) method, a non-Lasso type estimator which is estimated as below
\begin{align}
\text{minimize} \ \left| \Theta \right|_1 \ \text{subject to} \ \left| \hat{\Sigma}_n \Theta - I \right|_{ \infty } \leq \lambda_n. 
\end{align} 

Notice that we can also consider the slightly modified version of the graphical Lasso: Let $V = \text{diag} \left( \sigma_{11}^{1 / 2}, ....,    \sigma_{pp}^{1 / 2} \right)$ and $R$ be the correlation matrix with $\hat{V}$ and $\hat{R}$ be the sample counterparts. We estimate
\begin{align}
\Omega = V^{-1} K  V^{-1} \ \ \ \text{by } \ \ \ \hat{\Omega}_{\lambda} = \hat{V}^{-1} \hat{K}_{\lambda} \hat{V}^{-1} 
\end{align}
where
\begin{align}
\hat{K}_{ \lambda } = \underset{ \Phi > 0 }{ \text{arg min} } \big\{ \text{trace} \left( \Phi \hat{R} \right) - \text{log det} \left( \Psi \right) + \lambda \left| \Psi^{-} \right|_1  \big\}. 
\end{align}
\end{example}

Furthermore, the structure of the proposed risk matrix has different properties in contrast to the covariance matrix. For instance, the estimation procedure is not invariant to variable permutations. In contrast the covariance matrix itself has the property of being invariant to variable permutations. Therefore, this property allows the implementation of high dimensional techniques in inference problems. For instance, in the case of the covariance matrix \cite{ledoit2004well} proposed a way to compute an optimal linear combination of the sample covariance with the identity matrix, which also results in shrinkage of eigenvalues. Furthermore, shrinkage estimators are invariant to variable permutations, but they do not affect the eigenvectors of the covariance, only the eigenvalues, and it has been shown that the sample eigenvectors are also not consistent when $p$ is large. Therefore, developing also estimation methods which are invariant to variable permutations is important.        

\begin{definition}
For any sequence $\underline{X}$, the $\phi-$mixing coefficient $\phi_k$ is defined as follows:
\begin{align}
\phi_k ( \underline{X} ) = \mathsf{sup} \bigg\{ P( B | A ) - P(B) | : A \in \sigma_{\ell}, B \in \sigma_{\ell + k}^{ \prime }, \ell \geq 1 \bigg\}.   
\end{align}
\end{definition}

\begin{remark}
Notice that the sub-Gaussianity assumption can be still employed in a time series setting (e.g., see \cite{wong2020lasso}). To do this, we combine the Rini-mixing condition with the sub-Gaussianity assumption which ensures that there is enough weak dependence captured as well as that the vectors of the model have good properties for estimation and inference. 
\end{remark}    

\begin{lemma}
Let $Z_i$ be $\textit{i.i.d}$  $\mathcal{N} \left( \mathbf{0}, \Sigma_p  \right)$ and $\lambda_{ \text{max} } \left( \Sigma_p \right) \leq \bar{k} < \infty$. Then if $\Sigma_p = [ \sigma_{ab} ]$, 
\begin{align}
\mathbb{P} \left[ \left| \sum_{i=1}^n \left( Z_{ij} Z_{ik} - \sigma_{jk} \right) \right| \geq  n \nu \right] \leq c_1 \text{exp} \left( - c_2 n \nu^2 \right) \ \ \ \text{for} \ \ \ | \nu | \leq \delta
\end{align}
where $c_1, c_2$ and $\delta$ depend on $\bar{k}$ only. 
\end{lemma}

\newpage

\subsubsection{Main Results on Probability Bounds}

Following the framework of \cite{Chen2013covariance} we focus in how to establish the convergence theory for covariance matrix estimates, we shall use the functional dependence measure of Wu (2005). Recall that $Z_{ji} = g_j ( \mathcal{F}_i ), 1 \leq j \leq p$, where $g_j(.)$ is the $j$th coordinate projection of the $\mathbb{R}^p-$valued measurable function $\mathbf{g}$. For $w > 0$, the functional dependence measure of $Z_{ji}$ is defined by 
\begin{align}
\theta_{ i, w , j } 
= \norm{ Z_{ji} - Z_{ji}^{\prime} }_w = \left( \mathbb{E} \left| Z_{ji} - Z_{ji}^{\prime} \right|^w \right)^{1 / w},
\end{align}
where $Z_{ji}^{\prime} = g_j ( \mathcal{F}_i^{\prime} ), \mathcal{F}_i^{\prime}  = \left( ..., \mathbf{e}_{-1}, \mathbf{e}_{0}^{\prime},      \mathbf{e}_{1},..., \mathbf{e}_{i} \right)$ and $\mathbf{e}_{0}^{\prime}$ and $\mathbf{e}_{0}^{\prime}$, $\mathbf{e}_{\ell}$, $\ell \in \mathbb{Z}$, are \textit{i.i.d}. In other words, $Z_{ji}^{\prime}$ is a coupled version of $Z_{ji}$ with $\mathbf{e}_{0}$ in the latter replaced by an \textit{i.i.d} copy of $\mathbf{e}_{0}^{\prime}$.   

\paragraph{Proof of Theorem 2.1}
(\cite{Chen2013covariance}) We first assume that $\alpha > 1/2 - 1/q$. Notice that 
\begin{align*}
\mathbb{E} \left| T_u \left( \hat{ \Sigma }_u \right) - \Sigma \right|^2_F 
&= 
\sum_{ j,k = 1 }^p \mathbb{E} \left[ \hat{\sigma}_{jk} \mathbf{1} \left\{  \left| \hat{\sigma}_{jk} \right| \geq u  \right\}  - \sigma_{jk}   \right]^2
\leq 
2 \sum_{ j,k = 1 }^p  \mathbb{E} \left( W_{jk}^2 \right) + 2 B( u / 2),
\\
W_{jk} &= \hat{\sigma}_{jk} \mathbf{1} \left\{  \left| \hat{\sigma}_{jk} \right| \geq u  \right\} - \sigma_{jk} \mathbf{1} \left\{  \left| \sigma_{jk} \right| \geq u / 2  \right\}  
\\
B(u) &= \sum_{ j,k = 1 }^p  \sigma_{jk}^2 \mathbf{1} \left\{  \left| \sigma_{jk} \right| < u  \right\}.
\end{align*}
Since the functional dependence measure for the product process $\left(  Z_{ji} Z_{ki} \right)_i$ has a bounded inequality, under the decay condition $\Theta_{ m, 2q } \leq C m^{- \alpha }, \alpha > 1/2 - 1/q$. we have by Theorem 2(ii) in Wu (2013) that 
\begin{align}
\mathbb{P} \left( \left| \xi_{jk} \right| > v \right) \leq \frac{C_2 n}{ ( nv )^q } + C_3 e^{ - C_4 nv^2}
\end{align}
holds for all $v > 0$. Using integration by parts we obtain
\begin{align}
\mathbb{E} \left[ \xi_{jk}^2 \mathbf{1} \left\{ \left| \xi_{jk} \right| > v  \right\} \right] = v^2 \mathbb{P} \left(  \left| \xi_{jk} \right| > v \right) + \int_{ v^2 }^{ \infty } \mathbb{P} \left(  \left| \xi_{jk} \right| > \sqrt{w} \right) dw.  
\end{align}

\medskip

\begin{remark}
Note that sparsity assumptions in the literature of high dimensional covariance matrices are usually imposed with respect to the inverse of the covariance matrix. In other words, if the $(j,k)-$th component of $\Sigma^{-1}$ is zero, then the variables $Z_j$ and $Z_k$ are partially uncorrelated, given the other variables. More specifically, the current way of defining sparsity in high dimensional models based on an underline covariance structure is to consider the notion of partial correlation as a measure of conditional independence (or dependence) in graphical based models. On the other hand, \cite{katsouris2023statistical}, defines a novel tail dependency matrix in which case the absence of a link between two nodes on the network implies the conditional tail independence\footnote{Associate Professor Sebastian Engelke (University of Geneva) gave a seminar with title: "Causality for extreme values" at the S3RI Departmental Seminar Series at the University of Southampton, on the 9th of December 2021.} (see, also \cite{engelke2020graphical}).
\end{remark}

\newpage 

Moreover, some useful limit results are also presented in the paper of \cite{xiao2013asymptotic}. In particular, the sample covariance between columns $x_i$ and $x_j$ is defined as 
\begin{align}
\hat{\sigma}_{ij} = \frac{1}{n} \left( x_i - \bar{x}_i \right)^{\top} \left( x_i - \bar{x}_i \right). 
\end{align}  

In high-dimensional covariance inference, a fundamental problem is to establish an asymptotic distributional theory for the maximum deviation
\begin{align}
M_n = \underset{ 1 \leq i \leq j \leq m }{ \text{max} } \left|  \hat{\sigma}_{ij} - \sigma_{ij} \right|.
\end{align}

\paragraph{Proofs of (iii) and (iv)}

The first step is the truncation step. 

\begin{enumerate}

\item[\text{Step 1}.] (Truncation Step)
We truncate $X_{n,k,i}$ by 
\begin{align}
\tilde{X}_{n,k,i} = X_{n,k,i} I \left\{ \left| X_{n,k,i} \right| \leq n^{1 / 4} \big/ \text{log} (n) \right\}
\end{align}

Define $\tilde{M}_n$ similarly as $M_n$ with $X_{n,k,i}$ being replaced by its truncated version $\tilde{X}_{n,k,i}$,  
\begin{align}
\mathbb{P} \left( \tilde{M}_n \neq M_n \right) \leq nm \mathcal{M}_n (p) n^{ - p / 4} \left( \text{log} (n) \right)^p \leq C \mathcal{M}_n (p) n^{ - \delta / 4} \left( \text{log} (n) \right)^p = o(1).
\end{align}
Furthermore, since the asymptotics are not affected for notational simplicity, we still use $\tilde{X}_{n,k,i}$ to denote its centered version with mean zero. 

Define $\tilde{\sigma}_{n,i,j} = \mathbb{E} \left( \tilde{X}_{n,k,i} \tilde{X}_{n,k,j} \right)$ and $\tilde{ \tau }_{ n,i,j } = \text{Var} \left( \tilde{X}_{n,k,i} \tilde{X}_{n,k,j} \right)$. Furthermore, denote with 
\begin{align}
M_{n,1} &= \underset{ 1 \leq i < j \leq m }{ \text{max} } \frac{1}{ \sqrt{ \tilde{\tau}_{n,i,j}  } } \left| \frac{1}{n} \sum_{k=1}^n      \tilde{X}_{n,k,i} \tilde{X}_{n,k,j} - \tilde{\sigma}_{n,i,j}  \right|
\\
M_{n,2} &= \underset{ 1 \leq i < j \leq m }{ \text{max} } \frac{1}{ \sqrt{ \tilde{\tau}_{n,i,j}  } } \left| \frac{1}{n} \sum_{k=1}^n \tilde{X}_{n,k,i} \tilde{X}_{n,k,j} - \sigma_{n,i,j}  \right|
\end{align}
Simple calculations are given by 
\begin{align}
&\underset{ 1 \leq i \leq j \leq m }{ \text{max} } \left| \tilde{\sigma}_{n,i,j} - \sigma_{n,i,j}  \right| \leq Cn^{ - (p-2) / 4} \left( \text{log} (n) \right)^p, 
\\
&\underset{ \alpha, \beta \in \mathcal{I}_n }{ \text{max} } \left|   \text{Cov} \left( \tilde{X}_{n,\alpha} \tilde{X}_{n,\beta}  \right)  - \text{Cov} \left(  X_{n,\alpha} X_{n,\beta} \right) \right| \leq Cn^{ - (p-2) / 4} \left( \text{log} (n) \right)^p, 
\end{align}

\item[\text{Step 2}.] (Effect of estimated means) Set $\bar{X}_{n,i} = \frac{1}{n} \sum_{k=1}^n \tilde{X}_{ n,k,i}$. Define with 
\begin{align}
M_{n,3} = \underset{ 1 \leq i \leq j \leq m }{ \text{max} } \frac{1}{ \sqrt{ \tilde{\tau}_{n,i,j}  } } \left| \frac{1}{n} \sum_{k=1}^n \left( \tilde{X}_{n,k,i} - \bar{X}_{n,i} \right) \left( \tilde{X}_{n,k,j} - \bar{X}_{n,j} \right) - \sigma_{n,i,j}       \right|.
\end{align}

\newpage 

Furthermore, observe that 
\begin{align}
\left| M_{n,3} - M_{n,2} \right| \leq \underset{ 1 \leq i \leq j \leq m }{ \text{max} } \frac{ \left| \bar{X}_{n,i} \bar{X}_{n,j}   \right| }{ \sqrt{ \tilde{\tau}_{n,i,j}  } } \leq \underset{ 1 \leq i \leq m }{ \text{max} } \left| \bar{X}_{n,i} \right|^2 . \left(   \underset{ 1 \leq i \leq j \leq m }{ \text{min} } \tilde{\tau}_{n,i,j} \right)^{- 1/ 2}. 
\end{align}

Furthermore, using Bernestein's inequality we can show that 
\begin{align}
\underset{ 1 \leq i \leq m }{ \text{max} } \left| \bar{X}_{n,i}  \right| = \mathcal{O}_p \left( \sqrt{ \frac{\text{log}(n)}{n}  }  \right),
\end{align}
which in together with the previous result implies that 
\begin{align}
\left| M_{n,3} - M_{n,2} \right| = \mathcal{O}_p \left( \frac{\text{log}(n)}{n} \right).
\end{align}

\item[\text{Step 3}.] (Effect of estimated variances)

Denote by $\check{ \sigma }_{n,i,j}$ the estimate of $\tilde{ \sigma }_{n,i,j}$ 
\begin{align}
\check{ \sigma }_{n,i,j} = \frac{1}{n} \sum_{ k = 1 }^n \left( \tilde{X}_{n,k,i} - \bar{X}_{n,i} \right) \left( \tilde{X}_{n,k,j} - \bar{X}_{n,j} \right).
\end{align}
Thus, since in the definition of $\tilde{M}_n$, $\tilde{\tau}_{ n,i,j}$ is unknown, and is estimated by 
\begin{align}
\check{ \tau }_{ n,i,j }  = \frac{1}{n} \sum_{ k = 1 }^n \left[ \left( \tilde{X}_{n,k,i} - \bar{X}_{n,i} \right) \left( \tilde{X}_{n,k,j} - \bar{X}_{n,j} \right) - \check{ \sigma }_{n,i,j} \right]^2. 
\end{align}
Therefore, in order to show that the exponential limit holds for $\tilde{M}_n$, it suffices to verify that 
\begin{align}
\underset{ 1 \leq i \leq m }{ \text{max} }  \left| \check{\tau}_{n,i,j} - \tilde{\tau}_{n,i,j} \right| = o_p \left( 1 / \text{log} (n) \right).
\end{align}
Set 
\begin{align}
\check{\tau}_{n,i,j,1} 
&= 
\frac{1}{n} \sum_{k=1}^n \left[ \left( \tilde{X}_{n,k,i} - \bar{X}_{n,i} \right) \left( \tilde{X}_{n,k,j} - \bar{X}_{n,j} \right) - \tilde{ \sigma }_{n,i,j} \right]^2,
\\
\check{\tau}_{n,i,j,2} 
&= 
\frac{1}{n} \sum_{k=1}^n \left( \tilde{X}_{n,k,i} \tilde{X}_{n,k,j} - \tilde{ \sigma }_{n,i,j} \right).
\end{align}
Using the probability limit, we know that 
\begin{align}
\underset{ 1 \leq i \leq j \leq m }{ \text{max} } \left| \check{\tau}_{n,i,j,1} - \check{\tau}_{n,i,j} \right| = \mathcal{O}_p \left( \text{log}(n) / n \right). 
\end{align}

\newpage

Since
\begin{align}
\left( \tilde{X}_{n,k,i} \tilde{X}_{n,k,j} - \tilde{\sigma}_{n,i,j}  \right)^2 \leq 64 n / \left( \text{log}(n) \right)^4. 
\end{align}
We have that 
\begin{align*}
\underset{ 1 \leq i \leq j \leq m }{ \text{max} } \mathbb{P} \left( \left| \check{\tau}_{n,i,j,2} - \tilde{\tau}_{n,i,j} \right| \geq \left( \text{log}(n) \right)^{-2} \right) 
&\leq 
\left[ \frac{ Cn }{ n \left( \text{log}(n) \right)^{-2} . \left[ n \left( \text{log}(n) \right)^{-3} \right]^{ q \wedge 1 } } \right]^{\text{log}(n)}
\\
&\leq  \left[ \frac{ C \left( \text{log}(n) \right)^5 }{ n^{ q \wedge 1} }   \right]^{\text{log}(n)} 
\end{align*}

Then, it follows that 
\begin{align}
\underset{ 1 \leq i \leq j \leq m }{ \text{max} } \left| \check{\tau}_{n,i,j,2} - \tilde{\tau}_{n,i,j} \right| = \mathcal{O}_p \left[ \left( \text{log}(n) \right)^{-2} \right].
\end{align}
Thus, it remains to prove that 
\begin{align}
\underset{ 1 \leq i \leq j \leq m }{ \text{max} } \left| \check{\tau}_{n,i,j,1} - \check{\tau}_{n,i,j,2} \right|  
\end{align}

\end{enumerate}

\paragraph{A normal comparison principle}

Suppose that for each $n \geq 1$, $\left( X_{n,i} \right)_{ i \in \mathcal{I}_n }$ is a Gaussian random vector whose entries have mean zero and variance one, where $\mathcal{I}_n$ is an index set with cardinality $\left| \mathcal{I}_n \right| = s_n$. Let $\mathcal{\Sigma}_n \left( r_{n,i,j} \right)_{ i,j \in \mathcal{I}_n }$ be the covariance matrix of $\left( X_{n,i} \right)_{ i \in \mathcal{I}_n }$. Assume that $s_n \to \infty$ as $n \to \infty$. We impose either of the following two conditions.

\begin{enumerate}

\item[(B1)] For any sequence $( b_n )$ such that $b_n \to \infty$, $\gamma_n ( n, b_n ) = o \left( 1 / \text{log} (b_n) \right)$ and 
\begin{align}
\underset{ n \to \infty }{ \text{lim sup} } \ \gamma_n < 1. 
\end{align}

\item[(B2)] For any sequence $( b_n )$ such that $b_n \to \infty$, $\gamma_n ( n, b_n ) = o \left( 1 \right)$ and 
\begin{align}
\sum_{ i \neq j \in \mathcal{I}_n } r^2_{n,i,j} = O \left( s_n^{2 - \delta } \right) \ \ \text{for some} \ \ \delta > 0 \ \ \ \text{and} \ \  \underset{ n \to \infty }{ \text{lim sup} } \ \gamma_n < 1
\end{align}
where 
\begin{align}
\gamma( n, b_n ) := \underset{ i \in \mathcal{I}_n }{ \text{sup} } \ \underset{ \mathcal{A} \subset \mathcal{I}_n, | \mathcal{A} | = b_n }{ \text{sup} } \ \underset{ j \in \mathcal{I}_n }{ \text{inf} } \left| r_{n,i,j} \right|, \ \ \ \gamma_n := \underset{ i, j \in \mathcal{I}_n, i \neq j }{ \text{sup} } \left| r_{n,i,j} \right|. 
\end{align}

\end{enumerate}

\newpage

\subsubsection{Limiting Laws of Coherence of Random Matrices with Applications}

Furthermore, following the framework of \cite{cai2011limiting}, we present some useful lemmas below. 

\begin{lemma}
Let $h_i = \norm{ x_i - \bar{x}_i } / \sqrt{n}$ for each $i$. Then, 
\begin{align}
\norm{ n \Gamma_n - X_n^{\top} X_n } \leq \left( b_{n,1}^2 + 2 b_{n,1}   \right) W_n b_{n,3}^{-2} + n b_{n,3}^{-2} b_{n,4}^2,  
\end{align}
where 
\begin{align*}
b_{n,1} &= \underset{ 1 \leq i \leq p }{ \text{max} } | h_i - 1 |, \ \ \ \ W_n = \underset{ 1 \leq i < j \leq p }{ \text{max} } \left| x_i^{\top} x_j \right|, 
\\
b_{n,3} &= \underset{ 1 \leq i \leq p }{ \text{min} } h_i, \ \ \ b_{n,4} = \underset{ 1 \leq i \leq p }{ \text{max} } \left| \bar{x}_i \right|. 
\end{align*}
\end{lemma}

\begin{lemma}
Let $\xi, 1 \leq i \leq n$, be independent random variables with $\mathbb{E} \xi_i = 0$. Set 
\begin{align}
s_n^2 = \sum_{ i = 1}^n \mathbb{E} \xi_i^2 , \ \ \ \rho_n^2  \sum_{ i = 1}^n \mathbb{E} \left| \xi_i \right|^3, \ \ \ S_n = \sum_{ i = 1}^n \xi. 
\end{align}
Assume that $\underset{ 1 \leq i \leq n }{ \text{max} } \left| \xi_i \right| \leq c_n s_n$ for some $0 < c_n \leq 1$. Then, 
\begin{align}
\mathbb{P} \left( S_n > x s_n \right) = e^{\gamma( x / s_n ) } \left( 1 = \Phi(x) \right) \left( 1 + \theta_{n,x} (1 + ) s_n^{-3} \rho_n  \right)
\end{align}
for $0 < x \leq 1 / ( 18 c_n )$, where $\left| \gamma(x) \right| \leq 2 x^3 \rho_n$ and $\left| \theta_{n,x} \right| \leq 36$.  
\end{lemma}
Therefore, we have that 
\begin{align}
\mathbb{P} \left( W_n \leq a_n \right) \leq e^{ - \lambda_n } + b_{1,n} + b_{2,n}, 
\end{align}
Suppose we consider an $\mathbb{R}^d-$valued time series process $\left\{ \underline{X}_t, t \in \mathbb{Z} \right\}$ with $\underline{X}_t = \left( X_{1,t},...., X_{d,t} \right)^{\top}$, and we have data $\underline{X}_1,..., \underline{X}_n$ at hand to use for estimation and inference purposes.

\medskip

\begin{corollary}
The process $\left\{ \underline{X}_t, t \in \mathbb{Z} \right\}$ is assumed to be strictly stationary and its $(d \times d)$ autocovariance matrix $\mathbf{C}(h) = \left( C_{ij} (h) \right)_{i,j = 1,...,d}$ at lag $h \in \mathbb{Z}$ is    
\begin{align}
\mathbf{C}(h) = \mathbb{E} \left[ \left( \underline{X}_{t + h} - \underline{\mu} \right) \left( \underline{X}_{t + h} - \underline{\mu} \right)^{\top} \right],
\end{align}
where $\underline{\mu} = \mathbb{E} \left[ \underline{X}_{t} \right]$, and the sample autocovariance is denoted with $\hat{ \mathbf{C} }(h) = \left( \hat{C}_{ij}(h) \right)_{i,j = 1,...,d}$  at lag $| h | < n$ and can be defined accordingly.
\end{corollary}

\newpage

\subsubsection{Large deviations for quadratic forms}

Large deviations for quadratic forms of stationary processes have been extensively studied in the literature. In particular, we need an upper bound for the tail probability under less restrictive conditions. In this section, we follow the framework proposed by \cite{xiao2012covariance}. More specifically, we prove a result on probabilities of large deviations of quadratic forms of stationary processes, which take the form
\begin{align}
Q_T = \sum_{ 1 \leq s \leq t \leq T} a_{s,t} X_s X_t. 
\end{align}
The coefficients $a_{s,t} = a_{T,s,t}$ may depend on $T$. Throughout this section we assume that sup$_{s,t} \left| a_{s,t} \right| \leq 1$, and $a_{s,t} = 0$ when $| s - t | > B_T$, where $B_T \to \infty$, and $B_T = \mathcal{O} ( T^{\gamma} )$ for some $0 < \gamma < 1$. 

\medskip

\begin{remark}
Notice that large deviations for quadratic forms of stationary processes have been extensively examined in the literature, which include: (i) large deviations principle for Gaussian processes, (ii) functional large deviation principle, (iii) moderate deviations principle, (iv) moderate deviation principles for quadratic forms of Gaussian processes and (v) moderate deviations of periodograms of linear processes as well as Cramer-type moderate deviation for spectral density estimates of Gaussian processes.  
\end{remark}

\begin{theorem}
Assume that $X_t \in \mathcal{L}^p, p > 4$, $\mathbb{E} X_t = 0$, and $\Theta_p (m) = \mathcal{O} ( m^{-a} )$. Set $c_p = ( p + 4) e^{p / 4} \Theta_4^2$. For any $M > 1$, let $x_t = 2 c_p \sqrt{ TM B_T \text{log} B_T }$. Assume that $B_T \to \infty$ and $B_T = \mathcal{O} \left( T^{\gamma} \right)$ for some $0 < \gamma < 1$. Then, for any $\gamma < \beta < 1$, there exists a constant $C_{p, M, \beta } > 0$ such that the following holds
\begin{align*}
\mathbb{P} \big( \left| \mathbb{E}_0 Q_T \right| \geq x_T \big) \leq C_{p,M, \beta } x_T^{- p / 2} \left( logT \right) \big[ ( T B_T )^{p/ 4} T^{- \alpha \beta p / 2} + T B_T^{p/2 -1 - \alpha \beta p / 2} + T   \big]  +  C_{p, M, \beta} B_T^{-M}. 
\end{align*}
\end{theorem}

\begin{proof}
Without loss of generality, assume $B_T \leq T^{\gamma}$. For $\gamma < \beta < 1$, let $m_T = \floor{ T^{ \sqrt{\beta} } }$, $\tilde{X}_T = \mathcal{H}_{t - m_T } X_t$ and 
\begin{align}
Q_T = \sum_{ 1 \leq s \leq t \leq T} a_{s,t} X_s X_t. 
\end{align}

We have that 
\begin{align}
\mathbb{P} \left[ \left| \mathbb{E}_0 \left( Q_T - \tilde{Q}_T \right)  \right| \geq c_p M^{1 / 2} \sqrt{ T B_T \left( \text{log} B_T \right) } \right]
\end{align}

Split $[1, T]$ into blocks $\mathcal{B}_1,..., \mathcal{B}_{b_T}$ of size $2 m_T$, and define
\begin{align}
Q_T = \sum_{ t \in \mathcal{B}_k } \sum_{ 1 \leq s \leq t } a_{s,t} \tilde{X}_s \tilde{X}_t. 
\end{align}

\end{proof}

\newpage 

\subsubsection{Asymmetry Helps}

In this section, we present the main results obtained by \cite{chen2021asymmetry}. We being by considering the low-rank matrix completion of \cite{chen2021asymmetry}.

\paragraph{Low-rank matrix completion} Suppose that $\mathbf{M}$ is generated using random partial entries of $\mathbf{M}^{\star}$ as 
\begin{align}
M_{ij} = 
\begin{cases}
\frac{1}{p} M_{ij}^{\star} & \text{with probability} \ p,  
\\
0  & \text{else},
\end{cases}
\end{align}
where $p$ denotes the fraction of the entries of $\mathbf{M}^{\star}$ being revealed. We have that $\mathbf{H} = \mathbf{M} - \mathbf{M}^{\star}$ is zero-mean and obeys $\left| H_{ij}  \right| \leq \frac{ \mu }{ np } := B$ and Var$\left( H_{ij} \right) \leq \frac{ \mu^2 }{ p n^2 }$. For instance, in the case where $p \approx \frac{  \mu^2 \text{log} (n) }{ n }$, then invoking Corollaries 1-3 of \cite{chen2021asymmetry} yields
\begin{align}
\frac{ \left| \lambda - \lambda^{\star} \right| }{  \left| \lambda^{\star} \right|  } \leq \frac{1}{ \sqrt{n }} \sqrt{ \frac{ \mu^3 \text{log}(n) }{ pn }   }, 
\end{align} 
which gives that 
\begin{align}
\text{min} \big\{ \norm{ \mathbf{u} - \mathbf{u} }_{ \infty },  \norm{ \mathbf{u} + \mathbf{u} }_{ \infty }   \big\} \leq  \frac{1}{ \sqrt{n }} \sqrt{ \frac{ \mu^3 \text{log}(n) }{ pn }   }
\end{align}
with high probability, where $\mathbf{a} \in \mathbf{R}^n$ is any fixed unit vector. Empirically, eigen-decomposition outperforms SVD in estimating both the leading eigenvalue and eigenvector of $\mathbf{M}^{\star}$

\paragraph{Why asymmetry helps?} According to\cite{chen2021asymmetry}, if we consider their Theorem 3, focusing on the case with $\lambda^{\star} = 1$ for simplicity, then the key ingredient is the Neumann trick stated in Theorem 2. Specifically, in the rank-1 case we can expand
\begin{align}
\mathbf{u} = \frac{1}{ \lambda } \left( \mathbf{u}^{\star \top}     \mathbf{u} \right) \sum_{s=0}^{ \infty } \frac{1}{ \lambda^s } \mathbf{H}^s \mathbf{u}^{\star}. 
\end{align}
We obtain that 
\begin{align}
\left| \mathbf{a}^{\top} \left( \mathbf{u} - \frac{ \mathbf{u}^{\star \top} \mathbf{u} }{ \lambda } \mathbf{u}^{\star}  \right)  \right|
= 
\left| \frac{ \mathbf{u}^{\star \top} \mathbf{u} }{ \lambda }   \sum_{s=1}^{\infty} \frac{ \mathbf{a}^{\top} \mathbf{H}^s \mathbf{u}^{ \star } }{ \lambda^s } \right| \leq \sum_{s=1}^{\infty } \left| \frac{ \mathbf{a}^{\top} \mathbf{H}^s \mathbf{u}^{ \star } }{ \lambda^s } \right|
\end{align}
where the last inequality holds since (i) $\left| \mathbf{u}^{ \star \top}  \mathbf{u} \right| \leq 1$, and (ii) $\lambda$ is real-valued and obeys $\lambda \approx 1$ if $ \norm{ \mathbf{H}  } << 1$. As a result, the perturbation can be well controlled as long as $\left| \mathbf{a}^{\top} \mathbf{H}^s \mathbf{u}^{\star}  \right|$ is small for every $s \leq 1$.

\newpage

\begin{itemize}
\item (Asymmetric case) When $\mathbf{H}$ is composed of independent zero-mean entries each with variance $\sigma_n^2$, we obtain that 
\begin{align}
\mathbb{E} \left[ \mathbf{a}^{\top} \mathbf{H}^2 \mathbf{u}^{\star}    \right]  = \mathbf{a}^{\top}  \mathbb{E} \left[ \mathbf{H}^2  \right] \mathbf{u}^{\star}   = \mathbf{a}^{\top} \left( \sigma^2 \mathbf{I} \right) \mathbf{u}^{\star} = \sigma^2 \mathbf{a}^{\top} \mathbf{u}^{\star}
\end{align}

\item (Symmetric case) When $\mathbf{H}$ is symmetric and its upper triangular part consists of independent zero-mean entries with variance $\sigma_n^2$, it holds that  
\begin{align}
\mathbb{E} \left[ \mathbf{a}^{\top} \mathbf{H}^2 \mathbf{u}^{\star}    \right]  = \mathbf{a}^{\top}  \mathbb{E} \left[ \mathbf{H}^2  \right] \mathbf{u}^{\star}   = \mathbf{a}^{\top} \left( n \sigma^2 \mathbf{I} \right) \mathbf{u}^{\star} 
= 
n\sigma^2 \mathbf{a}^{\top} \mathbf{u}^{\star}. 
\end{align}
\end{itemize}

\paragraph{Pertrubation Analysis for the rank-r case}

\paragraph{Eigenvalue perturbation for the rank-r case} The eigenvalue perturbation analysis can be extended to accommodate the case where $\mathbf{M}^{\star}$ is symmetric and rank$-r$. As before, we assume that the $r$ nonzero eigenvalues of  $\mathbf{M}^{\star}$ obey that $| \lambda_1 | \geq ... \geq | \lambda_r |$ 

\begin{theorem}[\cite{chen2021asymmetry}]
\label{rank-r}
(Perturbation of linear forms of eigenvectors (rank-r)) Consider a rank$-r$ symmetric matrix $\mathbf{M}^{\star} \in \mathbb{R}^{n \times n}$ with incoherence parameter $\mu$. Suppose that 
\begin{align}
\frac{ \sigma \sqrt{ n \text{ log} (n) } , B \text{ log} (n)}{ \lambda^{\star}_{ \text{max} } } \leq \frac{ c_1 }{ \kappa }
\end{align}
\end{theorem}
for some sufficiently small constants $c_1 > 0$. Then for any fixed unit vector $\mathbf{a} \in \mathbb{R}^n$ and any $1 \leq \ell \leq r$, with probability at least $1 - \mathcal{O} \left( n^{-10} \right)$ one has that 
\begin{align}
\left| \mathbf{a}^{\top} \left( \mathbf{u}_{\ell} - \sum_{ j = 1}^{r} \frac{ \lambda_j^{\star } \mathbf{u}_j^{ \star \top } \mathbf{u}_{ \ell } }{ \lambda_j } \mathbf{u}_j^{ \star } \right) \right| 
&\leq 
\text{max} \left\{ \sigma \sqrt{n \text{log} (n) } , B \text{log}(n)  \right\} \frac{ \kappa }{ \left|  \lambda_{\ell } \right| } \sqrt{ \frac{  \mu r }{ n }  }
\\
& \leq \frac{ \text{max} \left\{ \sigma \sqrt{n \text{log} (n) } , B \text{log}(n)  \right\}  }{ \lambda_{ \text{max} }^{\star } }  \kappa^2  \sqrt{ \frac{  \mu r }{ n }  }
\end{align}
The particular condition allows to control the perturbation of the linear form of eigenvectors. Thus, the perturbation upper bound grows as either the rank $r$ or the condition number $\kappa$ increases. 
\begin{corollary}[\cite{chen2021asymmetry}]
Consider the $\ell-$th eigenvalue $\lambda_{\ell}$ for $\ell \in \left\{ 1,..., r \right\}$ of the matrix $\mathbf{M}$. Under the assumptions of theorem 4, with probability at least $1 - \mathcal{O} \left( n^{-10} \right)$, there exists $1 \leq j \leq r$ 
\begin{align}
\left| \lambda_{\ell} - \lambda_{j}^{\star} \right| \leq \text{max} \left\{ \sigma \sqrt{n \text{log} (n) } , B \text{log}(n)  \right\} \kappa r \sqrt{ \frac{  \mu }{ n }  }, \ \ \frac{ \text{max} \left\{ \sigma \sqrt{n \text{log} (n) } , B \text{log}(n)  \right\}  }{ \lambda_{ \text{max} }^{\star } }  \leq c1 / \kappa^2 
\end{align}
for sufficiently small constant $c_1 > 0$. 
\end{corollary}

\newpage

In comparison, the Bauer-Fike theorem (Lemma 2) together with Lemma 1 gives a perturbation bound
\begin{align}
\left| \lambda_{\ell} -  \lambda_{j}^{\star} \right| \leq \norm{ \mathbf{H} } \leq  \text{max} \left\{ \sigma \sqrt{n \text{log} (n) } , B \text{log}(n)  \right\} \ \ \text{for some} \ \ 1 \leq j \leq r. 
\end{align}
For the low-rank case where $r << \sqrt{n}$, the eigenvalue pertrubation bound derived in Corollary 5 can be much sharper than the Bauer-Fike theorem. 

\begin{corollary}[\cite{chen2021asymmetry}]
Under the same setting of Theorem 4, with probability $1 - \mathcal{O} \left( n^{-9} \right)$, 
\begin{align}
\norm{ \mathbf{a}^{\top} \mathbf{U} }_2 \leq \kappa \sqrt{r} \norm{  \mathbf{a}^{\top} \mathbf{U}^{\star} }_2 + \frac{ \text{max} \left\{ \sigma \sqrt{n \text{log} (n) } , B \text{log}(n)  \right\}  }{ \lambda_{ \text{max} }^{\star } }  \kappa^2 r \sqrt{ \frac{\mu}{n} }. 
\end{align}
\end{corollary}
Consequently, by taking $\mathbf{a} = \mathbf{e}_i$ for $1 \leq i \leq n$ in Corollary 6, we arrive at the following statement regarding the alternative definition of the incoherence of the eigenvalue matrix $\mathbf{U}$. 

\begin{corollary}
Under the same setting of Theorem 4, with probability $1 - \mathcal{O} \left( n^{-8} \right)$ we have  
\begin{align}
\norm{ \mathbf{U} }_{2, \infty } kr \sqrt{ \frac{ \mu }{ n } }. 
\end{align}
\end{corollary}
\begin{proof}
Given that $\norm{ \mathbf{U} }_{2, \infty } = \text{max}_{ 1 \leq i \leq n }  \norm{ \mathbf{e}_i^{\top} \mathbf{U} }_2$ and recalling our assumption implies that the following condition holds $\norm{ \mathbf{U} }_{2, \infty } \leq \sqrt{ \mu r / n}$, we can invoke Corollary 6 and the union bound to derive the advertised entrywise bounds. 
\end{proof}

\paragraph{Proof of Theorem \eqref{rank-r}}

Without loss of generality, we shall assume that $\lambda_{\text{max}}^{\star} = \lambda_1^{\star}  = 1$ throughout the proof. To begin with, Lemma 2 implies that for all $1 \leq \ell \leq r$, 
\begin{align}
\left| \lambda_{\ell} \right| \leq \left| \lambda_{ \text{min} }^{\star } \right| - \norm{ \mathbf{H} } > 1 / ( 2 \kappa ) > \norm{ \mathbf{H} }
\end{align} 
as long as $\norm{ \mathbf{H} } < 1 / ( 2 \kappa )$. In view of the Neumann trick (Theorem 2), we can derive that 
\begin{align*}
\left| \mathbf{a}^{\top} \mathbf{u}_{\ell} - \sum_{ j = 1}^r \frac{ \lambda_j^{\star} \mathbf{u}_j^{\star \top } \mathbf{u}_{\ell}  }{ \lambda_j } \mathbf{a}^{\top} \mathbf{u}_j^{\star} \right|
&=  
\left| \sum_{ j = 1}^r \frac{ \lambda_j^{\star } }{ \lambda_j } \left( \mathbf{u}_j^{\star \top} \mathbf{u}_{\ell} \right) \left\{ \sum_{ s = 1}^{ \infty } \frac{1}{ \lambda^s_{\ell} } \mathbf{a}^{\top} \mathbf{H}^s \mathbf{u}_j^{\star} \right\} \right| 
\\
&\leq \left( \sum_{ j = 1}^r \frac{ \left| \lambda_j^{\star } \right| }{ \left| \lambda_j \right| } \left| \mathbf{u}_{j}^{\star \top }  \mathbf{u}_{ \ell } \right| \right) \left\{ \underset{ 1 \leq j \leq r }{ \text{max} } \sum_{ j = 1}^{ \infty } \frac{1}{ \left| \lambda_{\ell} \right|^s } \left| \mathbf{a}^{\top} \mathbf{H}^s \mathbf{u}_j^{\star} \right| \right\}
\\
&\leq \sqrt{ r \sum_{ j = 1}^r \left| \mathbf{u}_j^{\star \top } \mathbf{u}_{\ell} \right|^2 } \left\{ \underset{ 1 \leq j \leq r }{ \text{max} } \frac{ \left| \lambda_j^{\star } \right| }{ \left| \lambda_j \right| } \right\}  \left\{ \underset{ 1 \leq j \leq r }{ \text{max} } \sum_{ j = 1}^{ \infty } \frac{1}{ \left| \lambda_{\ell} \right|^s } \left| \mathbf{a}^{\top} \mathbf{H}^s \mathbf{u}_j^{\star} \right| \right\} 
\\
&\leq \sqrt{r} . \frac{1}{ \left| \lambda_{\ell} \right| } . \left\{    \underset{ 1 \leq j \leq r }{ \text{max} } \sum_{ j = 1}^{ \infty } \frac{1}{ \left| \lambda_{\ell} \right|^s } \left| \mathbf{a}^{\top} \mathbf{H}^s \mathbf{u}_j^{\star} \right| \right\}
\end{align*}

\newpage

where the third line follows since $\sum_{ j = 1}^r \left| \mathbf{u}_j^{ \star } \mathbf{u}_{\ell} \right|^2 \leq \norm{ \mathbf{u}_{\ell} }_2^2 = 1$, and the last inequality makes use of (49). Apply Corollary 4 in  \cite{chen2021asymmetry} to obtain a bound as below 
\begin{align*}
&\leq \frac{ \sqrt{r} }{ \left| \lambda_{\ell} \right| } \sum_{ j = 1}^{ \infty } \left( 2 c_2 \kappa \ \text{max} \left\{ B \text{log}(n), \sqrt{n \sigma^2 \text{log} (n) } \right\} \right)^s \sqrt{ \frac{ \mu }{ n } } 
\\
&\leq \frac{ \kappa }{ \left| \lambda_{\ell} \right| } \text{max} \left\{ B \text{log}(n), \sqrt{n \sigma^2 \text{log} (n) } \right\} \sqrt{ \frac{ \mu r }{ n } } 
\\
&\leq \kappa^2  \text{max} \left\{ B \text{log}(n), \sqrt{n \sigma^2 \text{log} (n) } \right\} \sqrt{ \frac{ \mu r }{ n } } 
\end{align*}
with the provision that $\left| \lambda_{\ell} \right| > 1 / ( 2 \kappa )$ and $\text{max} \left\{ B \text{log}(n), \sqrt{n \sigma^2 \text{log} (n) } \right\} \leq c_1 / \kappa$ for some sufficiently small constant $c_1 > 0$. The condition $| \lambda_{\ell} | > 1 / (2 \kappa )$ follows immediately by combining Lemma 2, Lemma 1 and the condition (34). 

\paragraph{Appendix}

We justify the existence and uniqueness of an eigenvalue in $\mathcal{B} \left( 1, \norm{ \mathbf{H} } \right)$. Denote with $\Lambda \left( \textbf{M} \right) = \left\{ \lambda_1,..., \lambda_n \right\}$ and define a set of auxiliary matrices
\begin{align}
\mathbf{M}_t = \mathbf{M}^{\star} + t \mathbf{M}, \ \ 0 \leq t \leq 1.  
\end{align}
As we can see the set of eigenvalues of $\mathbf{M}_t$ depends continuously on $t$, we can write as below 
\begin{align}
\Lambda \left( \mathbf{M}_t \right) = \left\{ \lambda_1(t), \lambda_2(t),..., \lambda_n(t) \right\}, 
\end{align}
with each $\lambda_j (t)$, $1 \leq j \leq n$ being a continuous function in $t$. Meanwhile as $\norm{ \mathbf{H} } < 1 / 2$ and $0 \leq t \leq 1$, the two disks $\mathcal{B} \left( 1, t \norm{ \mathbf{H} } \right)$ and $\mathcal{B} \left( 0, t \norm{ \mathbf{H} } \right)$ are always disjoint sets. Therefore, the continuity of the spectrum with respect to $t$ requires that $\lambda_j(t)$ to always stay within the same disk where $\lambda_j(0) \in \left\{ 0, 1 \right\}$ lies, namely, such that
\begin{align}
\lambda_j (t) \in \mathcal{B} \left( \lambda_j(0), t \norm{ \mathbf{H} } \right). 
\end{align}
Given that $\mathbf{M}^{\star}$ has $n - 1$ eigenvalues equal to 0 and one eigenvalue equal to 1, we establish the lemma for the rank-1 case. 

\paragraph{ Rank-r } Follow the above argument for the case where we have a rank$-1$ matrix, we can immediately show that: if $\norm{ \mathbf{H} } < \lambda_r^{ \star } / 2$, then (i) there are exactly $n - r$ eigenvalues lying within $\mathcal{B} \left( 0, \norm{ \mathbf{H} } \right)$, (ii) all other eigenvalues lie within $\cup_{ 1 \leq j \leq r} \mathcal{B} \left( \lambda_j^{\star} , \norm{ \mathbf{H} } \right)$, which are exactly top$-r$ leading eigenvalues of $\mathbf{M}$. 

\newpage

\paragraph{Proof of Theorem 2 
}[\cite{chen2021asymmetry}] Following the definition of eigenvectors, we have that 
\begin{align}
\left( \mathbf{M}^{\star} + \mathbf{H} \right) = \lambda_{\ell} \mathbf{u}_{ \ell }, \ \ \text{or equivalently},  \ \ \frac{1}{ \lambda_{\ell} } \mathbf{M}^{\star} \mathbf{u}_{\ell} = \left( \mathbf{I} - \frac{1}{ \lambda_{\ell} } \mathbf{H} \right) \mathbf{u}_{\ell}
\end{align}
When $\norm{ \mathbf{H} }_2 < \left| \lambda_{\ell} \right|$, one can invert $\left( \mathbf{I} - \frac{1}{ \lambda_{\ell}  } \mathbf{H}   \right)$ and therefore we obtain that 
\begin{align*}
\mathbf{u}_{ \ell } 
= \left( \mathbf{I} - \frac{1}{ \lambda_{\ell}  } \mathbf{H}  \right)^{-1} \frac{1}{ \lambda_{\ell} } \mathbf{M}^{\star} \mathbf{u}_{\ell}  
=  \frac{1}{ \lambda_{\ell} } \left( \mathbf{I} - \frac{1}{ \lambda_{\ell}  } \mathbf{H}  \right)^{-1} \left( \sum_{ j = 1}^r \lambda_j^{\star} \mathbf{u}_j^{\star} \mathbf{u}_j^{\star \top}   \right) \mathbf{u}_j
= \sum_{ j = 1}^r \frac{ \lambda_j^{\star} }{ \lambda_j } \left( \mathbf{u}_j^{\star} \mathbf{u}_j^{\star \top} \right) \left( \mathbf{I} - \frac{1}{ \lambda_{\ell}} \mathbf{H} \right)^{-1} \mathbf{u}_j^{\star},   
\end{align*}
where the last term follows by rearranging terms. Finally, replacing $\left( \mathbf{I} - \frac{1}{ \lambda_{\ell}} \mathbf{H} \right)^{-1}$ with the Neumann series we obtain that $\sum_{ s = 0}^{\infty} \frac{1}{ \lambda^s } \mathbf{H}^s$, we establish the theorem.

\paragraph{Proof of Lemma 3}[\cite{chen2021asymmetry}] We start with the rank-1
case. Towards this, we resort to the Neumann trick in Theorem 2, which in the rank-1 case gives the following expression 
\begin{align}
\mathbf{u} = \frac{1}{ \lambda } \left( \mathbf{u}_j^{\star  \top} \mathbf{u}_j^{\star} \right) \sum_{ s = 0 }^{ \infty } \left( \frac{1}{ \lambda } \mathbf{H} \right)^s \mathbf{u}^{\star}. 
\end{align}
From Lemma 2, we know that $\lambda$ is real-valued and that $\lambda > 1 - \norm{ \mathbf{H} } \geq 3/4 > \norm{ \mathbf{H} }$ under our assumption. This together yields that 
\begin{align}
\norm{ \mathbf{u} - \frac{ \mathbf{u}^{\star \top } \mathbf{u}}{    \lambda } \mathbf{u}^{\star} }_2 \leq \frac{1}{ \lambda } \sum_{ s = 0 }^{ \infty } \norm{ \frac{1}{\lambda} \mathbf{H} }^s \norm{ \mathbf{u}^{\star} }_2 = \frac{1}{ \lambda } \sum_{ s = 0 }^{ \infty } \norm{ \frac{1}{ \lambda } \mathbf{H} }^s = \frac{ \norm{ \mathbf{H} } }{ \lambda \left( \lambda - \norm{ \mathbf{H} }  \right) } \leq \frac{8}{3} \norm{\mathbf{H}} , 
\end{align}
where the last inequality holds since $\lambda \geq 3 / 4$ and $\lambda - \norm{\mathbf{H}} \geq 1 - 2 \norm{\mathbf{H}} \geq 1 / 2 $. Next, by decomposing $\mathbf{u}$ into two orthogonal components
\begin{align}
\mathbf{u} = \left( \mathbf{u}^{\star  \top} \mathbf{u} \right) \mathbf{u}^{\star} + \left( \mathbf{u} - \left( \mathbf{u}^{\star \top} \mathbf{u} \right) \mathbf{u}^{\star} \right)
\end{align}
we obtain that 
\begin{align*}
\left| \mathbf{u}^{\star  \top} \mathbf{u} \right| 
= \norm{ \left( \mathbf{u}^{\star \top} \mathbf{u} \right) \mathbf{u}^{\star} }_2 = \sqrt{ 1 - \norm{ \mathbf{u} - \left( \mathbf{u}^{\star \top} \mathbf{u} \right) \mathbf{u}^{\star} }_2^2 }
&\geq 
1 - \norm{ \mathbf{u} -  \left( \mathbf{u}^{\star \top} \mathbf{u} \right) \mathbf{u}^{\star} }_2^2
\geq 1 - \norm{ \mathbf{u} -  \frac{ \mathbf{u}^{\star \top} \mathbf{u} }{ \lambda } \mathbf{u}^{\star} }_2^2 
\\
&\geq 1 - \frac{64}{9} \norm{ \mathbf{H} }^2.
\end{align*}
the particular inequality holds since $\left( \mathbf{u}^{\star \top} \mathbf{u} \right) \mathbf{u}^{\star}$ is orthogonal projection of $\mathbf{u}$ onto the subspace spanned by $\mathbf{u}^{\star}$, and hence 
\begin{align}
\norm{ \mathbf{u} - \left( \mathbf{u}^{\star \top} \mathbf{u} \right)     \mathbf{u}^{\star} }_2 \leq \norm{ \mathbf{u} - \frac{1}{\lambda}  \left( \mathbf{u}^{\star \top} \mathbf{u} \right) \mathbf{u}^{\star}      }_2. 
\end{align} 

\newpage

Moreover, since $\mathbf{u}$  is real-valued we obtain the following bound: 
\begin{align}
\text{min} \big\{ \norm{ \mathbf{u} - \mathbf{u}^{\star} }_2 ,   \norm{ \mathbf{u} + \mathbf{u}^{\star} }_2 \big\} = \sqrt{  \norm{ \mathbf{u} }_2^2 + \norm{ \mathbf{u}^{\star} }_2^2 - 2 \left| \mathbf{u}^{\star \top} \mathbf{u} \right| } \leq \frac{8 \sqrt{2} }{3} \norm{ \mathbf{H} }. 
\end{align}
Next, we consider the rank$-r$ case, since for any $1 \leq r \leq r$ we have that 
\begin{align*}
\sum_{j=1}^r \left| \mathbf{u}_j^{\star \top} \mathbf{u}_{\ell}  \right|^2 
=
 \norm{ \sum_{j=1}^r  \left( \mathbf{u}_j^{\star \top} \mathbf{u}_{\ell} \right)   \mathbf{u}_j^{\star} }_2^2 
=
1 - \norm{ \mathbf{u}_{\ell} - \sum_{j=1}^r \left( \mathbf{u}_j^{\star \top} \mathbf{u}_{\ell} \right) \mathbf{u}_j^{\star} }_2^2 
\geq 1 - \norm{ \mathbf{u}_{\ell} - \sum_{j=1}^r \frac{ \lambda_j^{\star} \mathbf{u}_j^{\star \top} \mathbf{u}_{\ell} }{ \lambda_j} \mathbf{u}_j^{\top} }_2^2 
\end{align*}

where the inequality arises since $\sum_{j=1}^r \left( \mathbf{u}_j^{\star \top} \mathbf{u}_{\ell} \right) \mathbf{u}_j^{\star}$ is the Euclidean projection of $\mathbf{u}_{\ell}$ onto the space of $\left\{ \mathbf{u}_{1}^{\star},..., \mathbf{u}_{r}^{\star} \right\}$. Furthermore, we observe that 
\begin{align}
\norm{ \sum_{j=1}^r \lambda_j^{\star}  \left( \mathbf{u}_j^{\star \top} \mathbf{u}_{\ell} \right) \mathbf{u}_j^{\star} }_2 
\leq \sqrt{ \sum_j \left( \lambda_j^{\star} \right)^2 \left|   \mathbf{u}_j^{\star \top} \mathbf{u}_{\ell}  \right|^2  } \leq \lambda_{ \text{max} }^{\star} \sqrt{ \sum_j \left| \mathbf{u}_j^{\star \top} \mathbf{u}_{\ell} \right|^2 } = \lambda_{ \text{max} }^{\star} \norm{ \mathbf{u}_{\ell} }_2 = 1.  
\end{align}
This taken collectively with Theorem 2 leads to 
\begin{align*}
\norm{ \mathbf{u}_{\ell} - \sum_{j=1}^r \frac{ \lambda_{ j }^{\star}  \mathbf{u}_{j}^{\star \top} \mathbf{u}_{ \ell} }{ \lambda_{ \ell} } \mathbf{u}_{j}^{\star} }_2 
&= 
\norm{ \sum_{j=1}^r \frac{ \lambda_j^{\star} }{ \lambda_{\ell} } \left( \mathbf{u}_j^{\star \top} \mathbf{u}_{\ell} \right) \left\{ \sum_{s=1}^{\infty } \frac{1}{ \lambda_{\ell}^2 } \mathbf{H}^s \mathbf{u}_j^{\star} \right\} }_2 
= 
\norm{ \frac{1}{ \lambda_{\ell} } \sum_{j=1}^r \frac{1}{ \lambda_{\ell}^s } \mathbf{H}^s \left\{ \sum_{j=1}^r \lambda_j^{\star}  \left( \mathbf{u}_j^{\star \top} \mathbf{u}_{\ell} \right)  \mathbf{u}_j^{\star} \right\} }_2  
\\
&\leq \frac{1}{ \left| \lambda_{\ell} \right|} \sum_{s=1}^{ \infty } \frac{1}{ \left| \lambda_{\ell} \right|^s } \norm{ \mathbf{H} }^s \norm{ \sum_{j=1}^{ r } \lambda_{j}^{\star}  \left( \mathbf{u}_j^{\star \top} \mathbf{u}_{\ell} \right)  \mathbf{u}_j^{\star} }_2    
\leq \frac{1}{ \left| \lambda_{\ell} \right|} . \frac{ \norm{ \mathbf{H} } }{ | \lambda_j | - \norm{ \mathbf{H} } } 
\leq \frac{ 8 \kappa^2 }{3} \norm{ \mathbf{H} }. 
\end{align*}
The last two lines follow since, when $\norm{ \mathbf{H} } < 1 / (4 \kappa )$. We can show that $| \lambda_{\ell} | > | \lambda_{ \text{min} }^{ \star } | - \norm{ \mathbf{H} } \geq 3 / (4 \kappa )$ and (ii) $\left| \lambda_{\ell} \right| - \norm{ \mathbf{H } } \geq \left| \lambda_{ \text{min} }^{ \star } \right| - 2 \norm{ \textbf{H} } \geq 1 / (2 \kappa)$.

\paragraph{Proof of Lemma 5 on Edge and Vertices}[\cite{chen2021asymmetry}]

To establish this lemma, we exploit entrywise independence of $\mathbf{H}$ and develop a combinatorial trick. To begin with, we expand the quantity of interest as below
\begin{align}
\left( \mathbf{a}^{\top} \mathbf{H}^s \mathbf{u}^{\star} \right)^k = \underset{ 1 \leq i_t^{(b)} \leq n, 0 \leq t \leq s, 1 \leq b \leq k }{ \sum } \prod_{ b = 1 }^k a_{ i_0^{(b)} } \left(  H_{i_{t-1}^{(b)}  i_t^{(b)} } \right) u_{ i_s^{(b)} }^{\star} 
\end{align}
In this section, we use that $
\mathcal{I} := \left\{ i_t^{(b)} | 0 < t < s, 1 < b < k \right\} \in [ n ]^{ (s+1) k }    
$, to denote such a collection of $( s + 1 ) k$ indices. Therefore, one can write that 
\begin{align}
\mathbb{E} \left[ \left( \mathbf{a}^{\top} \mathbf{H}^s \mathbf{u}^{\star} \right)^k \right] = \underset{  \mathcal{I} \in [n]^{(s+1)k}   }{ \sum } \mathbb{E} \left[  \prod_{b=1}^k a_{ i_0^{(b)}} \left( \prod_{t=1}^s H_{ i_{t-1}^{(b)} i_t^{(b)} } \right)  \mathbf{u}^{\star}_{ i_{t-1}^{(b)} } \right]
\end{align}
Further details of these proofs can be found in the Appendix of \cite{chen2021asymmetry}.

\newpage

\subsection{Limit theory for Functional Data}

\subsubsection{Simultaneous Diagonalization}

In this section we briefly study the framework proposed by \cite{yuan2010reproducing}. In particular, before studying the asymptotic properties of the regularized estimators, we first investigate the relationship between the eigen structures of the covariance operator for $X(.)$ and the reproducing kernel of the functional space $\mathcal{H}$. As observed in earlier studies, eigenvector structures play prominent roles in determining the nature of the estimation problem in functional linear regression. 

Recall that $K$ is the reproducing kernel of $\mathcal{H}_1$. Because $K$  is continuous and square integrable, it follows from Mercer's theorem that $K$ admits the following spectral decomposition
\begin{align}
K(s,t) = \sum_{k = 1}^{\infty } \rho_k \psi_k (s) \psi_k(t).
\end{align}

Here $\rho_1 \geq \rho_2 \geq ... $ are the eigenvalues of $K$, and $\left\{ \psi_1, \psi_2,...  \right\}$ are the corresponding eigenfunctions,  
\begin{align}
K \psi_k = \rho_k \psi_k, \ \ k = 1,2,...
\end{align}

Moreover, 
\begin{align}
\langle \psi_i , \psi_j \rangle_{ \mathcal{L}_2 } = \delta_{ij} \ \ \text{and} \ \ \langle \psi_i , \psi_j \rangle_{ \mathcal{L}_2 } = \delta_{ij} / \rho_{j}, 
\end{align}
where $\delta_{ij}$ is the Kronecker's delta.

\paragraph{Convergence rates} 

We now turn to the asymptotic properties of the smoothness regularized estimators. To fix ideas, in what follows, we shall focus on the squared error loss. Recall in this case we have that 
\begin{align}
\left( \hat{\alpha}_{n \lambda}, \hat{\beta}_{n \lambda} \right) = \underset{ \alpha \in \mathbb{R}, \beta in \mathcal{H}  }{ \text{arg min} } \left\{ \frac{1}{n} \sum_{i=1}^n \left[ y_i - \left( \alpha + \int_{\mathcal{I} } x_i(t) \beta(t) dt  \right) \right]^2 + \lambda J( \beta) \right\}.
\end{align}

\medskip

\begin{remark}
Due to the fact that the risk matrix employed in the optimal portfolio problem is not a typical covariance operator, then in order to investigate the convergence rates and asymptotic properties of the elements of the risk matrix and its eigenvalues, we can simplify the problem by assuming that the entries of the matrix have a linear process representation. This simplification allows to model the heavy tailed entries which represent the tail-risk measures without considering the distributional properties of the forecasts based on the quantile specifications mentioned above (see, \cite{muller2005generalized}).  
\end{remark}

\newpage

Furthermore, we have that 
\begin{align}
\ell_n ( \beta ) = \frac{1}{n} \sum_{ i = 1}^n \left( y_i - \int_{\mathcal{I} } x_i(t) \beta(t) dt \right)^2. 
\end{align}

Observe that 
\begin{align*}
\ell_{ \infty } 
&:= 
\mathbb{E} \left[ \ell_n ( \beta ) \right] = \mathbb{E} \left[ Y - \int_{\mathcal{I} } X(t) \beta( t) dt \right]^2
\\
&= 
\sigma^2 + \int_{\mathcal{I} } \int_{\mathcal{I} } \big[ \beta(s) - \beta_0(s) \big] C(s,t) \big[ \beta(s) - \beta_0(s) \big] ds dt
\\
&=
\sigma^2 + \norm{ \beta - \beta_0 }^2_0.
\end{align*}

Write
\begin{align}
\bar{ \beta }_{ \infty \lambda } = \underset{ \beta \in \mathcal{H} }{  \text{arg min} } \bigg\{ \ell_{ \infty }( \beta ) + \lambda J ( \beta ) \bigg\}. 
\end{align}

For instance, we have the following decomposition
\begin{align}
\hat{ \beta }_{ n \lambda } - \beta_0 = \left( \hat{ \beta }_{ n \lambda } - \bar{ \beta }_{ \infty \lambda } \right) +  \left( \bar{ \beta }_{ \infty \lambda } - \beta_0 \right).  
\end{align}
we refer to the terms above on the right hand side as the stochastic error and deterministic error respectively.

\paragraph{Deterministic Error} 

Write
\begin{align}
\beta_0(.) = \sum_{k = 1}^{\infty} a_k \omega_k(.) \ \ \text{and} \ \ \beta(.) = \sum_{k = 1}^{\infty} b_k \omega_k(.) 
\end{align}
Then, Theorem 3 above implies that 
\begin{align}
\ell_{ \infty } ( \beta ) = \sigma^2 + \sum_{k=1}^{ \infty } \left( b_k - a_k \right)^2, \ \ \ J( \beta ) = \sum_{k=1}^{ \infty } \gamma_k^{-1} b_k^2. 
\end{align}
Therefore, 
\begin{align}
\bar{ \beta }_{ \infty \lambda } (.) = \sum_{ k = 1}^{ \infty } \frac{ a_k }{ 1 + \lambda \gamma_k^{-1} } \omega_k (.) =: \sum_{ k = 1}^{ \infty } \bar{b}_k \omega_k(.).
\end{align}

\newpage

\subsubsection{Common Functional Principal Components}

According to \cite{benko2009common}, the elements of the risk matrix $\Gamma$ are functionals of the model estimates which can be estimated with asymptotically negligible bias and a parametric rate of convergence $T_i^{- 1 / 2}$. For example when the data are generated from a balanced, equidistance design, then it can be easily seen that for $i \neq j$ this rate of convergence is achieved by the following estimator. Notice that a bias corrected estimator may yield negative eigenvalues. However, in practise these values are small and can be interpreted close to zero. For instance, when evaluating eigenvalues and eigenfunctions of the empirical covariance operator of nonparametrically estimated curves $\hat{X}_i$, then for fixed $r \leq r_0$ the above rate of convergence for the estimated eigenfunctions may well be achieved for a suitable choice of smoothing parameters. For example, when using standard methods it does not seem to be possible to obtain a corresponding rate of convergence, since any smoothing bias $\left| \mathbb{E} \left[ \hat{X}_i \right] - X_i(t) \right|$ will invariably affect the quality of the corresponding estimate of $\hat{\lambda}_r$.

\paragraph{Main Results}

In the following, we have that $\norm{ v } = \left( \int_0^1 v(t)^2 dt   \right)^{1 / 2}$ will denote the $L^2-$norm for any square integrable function $v$. Consider that $\hat{\Sigma}_m$ is the $m \times m$ matrix with elements given by 
\begin{align}
\left\{ \displaystyle \frac{1}{n} \sum_{i=1}^n \left( \beta_{ji} - \bar{\beta}_j \right)  \left( \beta_{ki} - \bar{\beta}_k \right) \right\}_{j,k = 1}^m
\end{align}
Let $\lambda_1 \left( \hat{\Sigma}_m \right) \geq \lambda_2 \left( \hat{\Sigma}_m \right) \geq ... \geq \lambda_m \left( \hat{\Sigma}_m \right)$ and $\hat{\zeta}_{1,m},...,\hat{\zeta}_{m,n}$ denote eigenvalues and corresponding eigenvectors of $\hat{\Sigma}_m$. Therefore, some straightforward algebra shows that 
\begin{align}
\hat{ \lambda }_{r,m} = \lambda_r \left( \hat{\Sigma}_m \right), \ \ \hat{ \gamma } = g_m(t)^{ \top } \hat{\zeta}_{r,m}.  
\end{align}
We will use $\Sigma_m$ to represent the $m \times m$ diagonal matrix with diagonal entries given by $\lambda_1 \leq ... \leq \lambda_m$. Then, the corresponding eigevectors are given by the $m-$dimensional unit vectors denoted by $e_{1,m},..., e_{m,m}$. Then, by Lemma A of Kneip and Utikal we obtain that the differences between the eigenvalues and eigenvectors of $\Sigma_m$ and $\hat{\Sigma}_m$ can be bounded by 
\begin{align}
\big( \hat{ \lambda }_{r,m} - \lambda_r \big) = \text{trace} \big\{ e_{r,m} e_{r,m}^{\top} \left( \hat{\Sigma}_m - \Sigma \right) \big\} + \tilde{R}_{r,m}, \ \ \ \tilde{R}_{r,m} \leq \frac{ \displaystyle 6 \underset{ \norm{a} = 1 }{ \text{sup} } a^{\top} \left( \hat{\Sigma}_m - \Sigma \right) ^2 a   }{ \displaystyle  \underset{ s }{ \text{min} } \left| \lambda_s - \lambda_r   \right| }
\end{align}
Moreover, we have that 
\begin{align}
\hat{ \zeta }_{r,m} - e_{r,m} = - S_{ r,m } \left( \hat{\Sigma}_m - \Sigma \right) e_{r,m} + R_{r,m}^{*}, \ \ \ \norm{ R_{r,m}^{*} } \leq \frac{ \displaystyle 6 \underset{ \norm{a} = 1 }{ \text{sup} } a^{\top} \left( \hat{\Sigma}_m - \Sigma \right) ^2 a   }{ \displaystyle  \underset{ s }{ \text{min} } \left| \lambda_s - \lambda_r   \right| }, 
\end{align}
where we denote with $S_{r,m} = \sum_{ s \neq r } \frac{1}{ \lambda_s - \lambda_r } e_{s,m} e_{s,m}^{\top}$.

\newpage

Assumption 1 implies that $\mathbb{E} \left( \hat{ \beta }_r \right) = 0$ and Var$\left( \hat{ \beta }_r \right) = \frac{ \lambda_r }{ r }$ and with $\delta_{ii} = 1$ as well as $\delta_{ij} = 0$ for $i \neq j$, we obtain that 
\begin{align*}
\mathbb{E} \left\{ \underset{ \norm{ a } = 1 }{ \text{sup} } a^{\top} \left( \hat{\Sigma}_m - \Sigma \right) \right\} 
&\leq 
\mathbb{E}  \bigg\{  \text{trace} \left[ \left( \hat{\Sigma}_m - \Sigma \right)^2 \right]  \bigg\} 
\\
&= 
\mathbb{E} \left\{ \sum_{j,k = 1}^m \left[ \frac{1}{n} \left( \beta_{ji} - \bar{\beta}_j \right) \left( \beta_{ki} - \bar{\beta}_k \right) - \delta_{jk} \lambda_j \right]^2 \right\}
\\
&\leq 
\mathbb{E} \left\{ \sum_{j,k = 1}^{ \infty } \left[ \frac{1}{n} \left( \beta_{ji} - \bar{\beta}_j \right) \left( \beta_{ki} - \bar{\beta}_k \right) - \delta_{jk} \lambda_j \right]^2 \right\}
\\
&= 
\frac{1}{n} \left( \sum_j \sum_k \mathbb{E} \left( \beta_{ji}^2 \beta_{ki}^2 \right) \right) + o( n^{-1} ) = \mathcal{O} \left( n^{-1} \right)
\end{align*}
for all $m$. Furthermore, since  
\begin{align}
\text{trac} \big\{ e_{r,m} e_{r,m}^{\top} \left( \hat{\Sigma}_m - \Sigma_m \right) \big\} = \frac{1}{n} \sum_{ i = 1 }^n \left( \beta_{ri} - \bar{ \beta }_r \right)^2 - \lambda_r
\end{align}
Therefore after applying the central limit theorem we obtain that
\begin{align*}
\sqrt{n} \left( \hat{\lambda}_r - \lambda_r \right) 
&= 
\frac{1}{ \sqrt{n} } \sum_{ i = 1 }^n \left( \beta_{ri} - \hat{\beta}_r \right)^2 - \lambda_r + \mathcal{O}_p \left( n^{- 1 / 2} \right)
\\
&= 
\frac{1}{ \sqrt{n} } \sum_{ i = 1 }^n \left\{ \left( \beta_{ri} \right)^2 - \mathbb{E} \left[  \left( \beta_{ri} \right)^2 \right] \right\} + \mathcal{O}_p \left( n^{- 1 / 2} \right)
\\
&\to \mathcal{N} \left( 0, \Lambda_r \right).
\end{align*}
Obviously the event $\hat{\lambda}_{r-1} > \hat{\lambda}_{r} > \hat{\lambda}_{r+1}$ occurs with probability 1. 

\medskip

\begin{remark}
Therefore, via the asymptotic analysis and the econometric model above we aim to verify that the position of the node in the network in terms of centrality, such as for example highly central nodes versus nodes which are on the periphery of the network can have an impact on the stability properties of a community of nodes. For example, if a node is quite central then the impact of the addition of a node to the particular node can dramatically change the dynamics in the network in the sense that that node as well as the structure of the remaining nodes can have higher inteconnectedness. Even though this outlier node is not highly interconnected affects the stability of the remaining of the nodes. 
\end{remark}

\newpage

\subsection{Limit theorems for Dependent Sequences}

In this Section, we consider some important results presented in the literature for dependent sequences. We consider a "high-dimensional" probability space $\left( \Omega, \mathcal{F}, \mathbb{P} \right)$. For functions of dependent random variables $( X_i )_{ i \in \mathbb{N} }$ the main challenge is often to quantify and bound the dependence among random variables $X_i$, in terms of various types of mixing coefficients (see, \cite{kontorovich2008concentration}). To establish concentration bounds a sufficiently rapid decay of the mixing coefficients is necessary.  

\subsubsection{Method of bounded martingale differences}

The particular methodology presented in the paper of \cite{kontorovich2008concentration} provides a suitable mechanism of establishing concentration inequalities (see also \cite{doukhan1995invariance}). 
\begin{theorem}
Suppose that $\mathcal{S}$ is a countable space, $\mathcal{F}$ is a set of all subsets and $\mathbb{P}$ is a probability measure on $\left( \mathcal{S}^n, \mathcal{F} \right)$ and $\varphi: \mathcal{S}^n \to \mathbb{R}$ is a \textit{Lipschitz} function (with respect to the Hamming metric) on $\mathcal{S}^n$ for some $c > 0$. Then, for any $t > 0$,
\begin{align}
\mathbb{P} \left( \left| \varphi - \mathbb{E} \varphi \right| \geq t  \right) \leq 2 \mathsf{exp} \left( - \frac{ t^2 }{ 2n c^2 \norm{ \Delta_n }_{\infty}^2 } \right).    
\end{align}
\end{theorem}

\begin{remark}
Notice that the bound can be simplified further. More precisely, given any intial probability distribution $p_0(.)$ and stochastic transition kernels $p_i (.|.), 1 \leq i \leq n - 1$, then the probability measure is 
\begin{align}
P \big( ( X_1,..., X_i ) = x \big) = p_0 (  x_1 ) \prod_{j=1}^{i-1} p_j \big( x_{j+1} | x_j \big), \ \ \forall \  1 \leq i \leq n
\end{align}
(see also \cite{finner1992generalization}).
\end{remark}

\subsubsection{Concentration Inequalities for Dependent Random Variables}

We consider the concept of \textit{pointwise inequalities}, that is, inequalities that hold uniformly for any parameter $\theta \in \Theta$. Define the function (see, \cite{van2002hoeffding})
\begin{align}
\psi_{\alpha} (x) := \mathsf{exp} \left( x^{\alpha} \right) - 1, \ \ \ \text{for any} \ \ x > 0.     
\end{align}
For a real-valued random variable $\xi$, we define with
\begin{align}
\norm{ \xi }_{\psi_{\alpha} } := \mathsf{inf} \bigg\{ \lambda > 0: \mathbb{E} \left[ \psi_{\alpha} \left(  \frac{ |\xi| }{\lambda} \right) \right] \leq 1 \bigg\}    
\end{align}
Moreover, we write that $\xi \in \mathcal{L}^q$ for some $q > 0$ if it holds that 
\begin{align}
\norm{\xi}_q := \left\{ \mathbb{E} \left( \left| \xi \right|^q \right) \right\}^{1/q}     
\end{align}

\newpage

\begin{definition}[Orlicz-norm]
For any convex function $\psi: \mathbb{R}^{+} \to \mathbb{R}^{+}$ such that $\psi(0) = 0$ and $\psi(x) \to \infty$ as $x \to \infty$ and (real-valued) random variable $X$, we denote with $\norm{ x }_{\psi}$ the Orlicz-norm, which is defined by 
\begin{align}
\norm{ X }_{ \psi} := \mathsf{inf} \left\{ C > 0: \mathbb{E} \left[ \psi \left( \frac{ | X | }{ C } \right) \right] \leq 1  \right\} .   
\end{align}
\end{definition}

\begin{itemize}
    \item Denote the $\ell^p$ Orlicz-norm of $X$ by $\norm{ X }_{p}$ for $p \in [ 0, + \infty)$ by setting $\psi (x) = x^p$ and $\norm{ X }_{ e^{\gamma} }$ the exponential Oricz-norm for $\gamma > 0$ by setting $\psi(x) = \mathsf{exp} ( x^{\gamma} ) - 1$ for some $\gamma \geq 1$. 
    
    \item The function $\psi(x)$ is the convex hull of $x \mapsto \mathsf{exp} ( x^{\gamma} ) - 1$ for some $\gamma \in (0,1)$, which ensures convexity. 
    
    \item Moreover, when $\boldsymbol{X}$ is a random vector, we define its Orlicz-norm by $\norm{ X }_{\psi} := \mathsf{sup}_{ \norm{ \boldsymbol{u} } \leq 1 } \norm{ \boldsymbol{u}^{\prime} \boldsymbol{X} }_{\psi}$.
    
\end{itemize}

\subsubsection{Central limit theorems for high dimensional dependent data}

Following the framework proposed by  \cite{chang2021central} , recall that we define with $S_{n,x} = n^{-1/2} \sum_{t=1}^n X_t$. Let $\mathcal{G} \sim \mathcal{N} (0, \Xi)$ where $\Xi := \mathsf{Cov} \left( n^{-1/2} \sum_{t=1}^n X_t \right)$. Without loss of generality we assume that $\mathcal{G}$ is independent of $\mathcal{X} = \left\{ X_1,..., X_n \right\}$. We write with $X_t = \left( X_{t,1},..., X_{t,p}     \right)^{\prime}$. Then, the long-run variance of the $j-$th coordinate marginal sequence $\left\{ X_{t,j} \right\}_{t=1}^n$ is defined as 
\begin{align}
V_{n,j} = \mathsf{Var} \left( \frac{1}{\sqrt{n} } \sum_{t=1}^n X_{t,j} \right).    
\end{align}
Therefore, in order to determine the convergence rate of $\rho_n$ for the $\alpha-$mixing sequence $\left\{ X_t \right\}$, we impose additional regularity conditions.
The above condition assumes that the partial sum $\frac{1}{\sqrt{n} } \sum_{t=1}^n X_{t,j}$ is non-degenerated which is necessary to bound the probability of a Gaussian vector taking values in a small region. When $\left\{ X_{t,j} \right\}_{ t \geq 1 }$ is stationary, then it holds that
\begin{align}
V_{n,j} := \Gamma_j(0) + 2 \sum_{k=1}^{n-1} \left( 1 - \frac{k}{n} \right) \Gamma_j (k)     
\end{align}
where $\Gamma_k(k) = \mathsf{Cov} \left( X_{1,j}, X_{k+1,j} \right)$ is the autocovariance of $\left\{ X_{t,j} \right\}_{ t \geq 1}$ at lag $k$.

\begin{assumption}[Subexponential moment] There exists a sequence of constants $B_n \geq 1$ and a universal constant $\gamma_1 \geq 1$ such that $\norm{ X_{t,j} }_{ \psi_{\gamma_1} } \leq B_n$ for all $t \in [n]$ and $j \in [p]$.
\end{assumption}

\begin{assumption}[Decay of $\alpha-$mixing coefficients] There exist some universal constants $K_1 > 1, K_2 > 0$ and $\gamma_2 > 0$ such that $\alpha_n(k) \leq K_1 e^{ \left( - K_2 k^{\gamma_2}  \right) }$ for any $k \geq 1$.
\end{assumption}

\begin{assumption}[Non-degeneracy] 
There exists a universal constant $K_3 > 0$ such that $\mathsf{min}_{ j \in [p] } V_{n,j} \geq K_3$.
\end{assumption}

\newpage

\begin{example}
Consider the inverse of the covariance matrix such that 
\begin{align}
\vertiii{  \norm{ \widehat{\boldsymbol{\Sigma}}_i^{-1} }  }_{\psi} \leq C. 
\end{align}
For the polynomial case, applying the union bound followed by Markov's inequality we conclude that 
\begin{align}
\underset{ i }{ \mathsf{max} } \norm{ \widehat{\boldsymbol{\Sigma}}_i^{-1} } \leq_{\mathbb{P}} n^{1 / p} \ \ \ \text{and} \ \ \ \underset{ i,t,j }{ \mathsf{max} } \left| X_{i,t}^{(j)} \right| \leq_{\mathbb{P}}   \left( nkT \right)^{1 / p}.
\end{align}
\end{example}

\begin{lemma}
Let $X$ and $Y$ be random elements defined in the same probability space $( \Omega, \mathcal{F}, \mathbb{P} )$ taking values in the metric space $( S, d)$. Then for measurable $A$ and $\delta \geq 0$
\end{lemma}

Due to the fact that $x \mapsto \mathsf{exp} \left[ \left( x \big/ \vertiii{X}_{ e^{\gamma} } \right)^{\gamma}  \right]$ is non-decreasing then
\begin{align*}
\mathbb{P} \left( | X | \geq x \right) &= \mathbb{P} \bigg( \mathsf{exp} \left[ \left( |X| \big/ \vertiii{X}_{e^{\gamma}} \right)^{\gamma} \right] \geq \mathsf{exp} \left[ \left( x \big/ \vertiii{X}_{e^{\gamma}} \right)^{\gamma} \right] \bigg) 
\\
&\leq \mathsf{exp} \left[ - \left( x \big/ \vertiii{X}_{e^{\gamma}} \right)^{\gamma}   \right] \mathbb{E} \mathsf{exp} \left[ \left( | X | \big/ \vertiii{X}_{e^{\gamma}} \right)^{\gamma} \right].
\end{align*}
It holds that, 
\begin{align}
\psi_{e^{\gamma}}(x) = K_{\gamma} x \mathbf{1} \left\{ 0 \leq x \leq a_{\gamma} \right\} + \left[ \mathsf{exp} ( x^{\gamma} ) - 1 \right] \mathbf{1} \left\{ x \geq a_{\gamma} \right\}    
\end{align}
where $K_{\gamma} := \frac{ ( \mathsf{exp} a_{\gamma}^{\gamma} - 1) }{ a_{\gamma} }$ and $a_{\gamma}$ is defined as below 
\begin{align}
a_{\gamma} := \mathsf{inf} \left\{ x \in \mathbb{R}_{+} : x \geq \left(  \frac{1 - \gamma}{ \gamma } \right)^{ 1 / \gamma }  \right\}    
\end{align}
Moreover, it holds that 
\begin{align}
\left( \frac{1 - \gamma}{ \gamma } \right)^{ 1 / \gamma}  \leq a_{\gamma} \leq \left( \frac{1}{\gamma} \right)^{ 1 / \gamma }.    
\end{align}

\begin{example}
Consider the martingale sequence 
\begin{align}
S_n = \sum_{i=1}^n X_i, n \geq 1.    
\end{align}
Consider the $\mathcal{F}_{i-1}$measurable random variables $K_i > 0$, for $i = 1,2,...$. Define with $B_0^2 = 0$ and for any $n \geq 1$ such that
\begin{align}
B_n^2 = \sum_{i=1}^n K_i^2 \left\{ 1 + \mathbb{E} \left[ \psi \left(  \frac{|X_i|}{K_i} \right) \big| \mathcal{F}_{i-1} \right]  \right\}    
\end{align}    
\end{example}

\newpage

\begin{theorem}
Let $\psi$ be an Orlicz function such that it holds that $\mathsf{sup}_{ x, y \to \infty } \psi(x) \psi(y) / \psi (cxy) < \infty$, for some constant $c$. Suppose that $\left\{ Z_{\theta} : \theta \in \Theta \right\}$ is a separable stochastic process indexed by $\theta$ in the pseudo-metric space $( \Theta, \tau )$. Assume that 
\begin{align}
\norm{ Z_{\theta} - Z_{ \vartheta } }_{\psi} \leq C^{\prime} \int_0^{ \mathsf{diam}(\Theta) } \psi^{-1} \big( D ( \delta ) \big) d \delta   
\end{align}
where $\mathsf{diam}(\Theta)$ is the diameter of $\Theta$ and $D(\delta)$ is the $\delta-$packing number.
\end{theorem}

\begin{corollary}
Let $W_i$ be $\mathcal{F}_i-$measurable and $\mathbb{E} ( W_i | \mathcal{F}_{i-1} ) = 0$ for $i \geq 1$. Suppose that for some constant $c < \infty$ it holds that 
\begin{align}
\mathbb{E} \left( \psi \left( \frac{ | W_i | }{ c} \right) \big| \mathcal{F}_{i-1} \right) \leq 1, \ \ \text{almost surely} \ i = 1,2,...      
\end{align}
\end{corollary}

\begin{remark}
Partitioning entropy could be applied to nonstationary time series? This could be the case when considering a discretenized method, such as block of nonstationary time series (i.e., $m-$dependence). Notice that this paper doesn't have in depth explanation of the dependence structure. However, the Orlicz norm provides related moment condition for understanding the asymptotic behaviour.  
\end{remark}
We define with $\phi(d)$ the following quantity
\begin{align}
\phi(d) = \int_0^d H^{1/2} ( \delta, d ) d \delta \vee d := \mathsf{min} \left\{ \int_0^d H^{1/2} ( \delta, d ) d \delta, d \right\},    
\end{align}

\begin{itemize}
    
\item What type of dependence structure does the entropy integral $\phi(d)$ introduce? For example, what form this integral would have in the case of Garch processes or for the autoregressive model? Are there any related results to Hoeffding's inequality for $\beta-$mixing sequences? To derive the proofs of main results presented in the paper we use that $P(A) \leq \mathsf{exp} \left\{ - \beta \alpha + 2 \beta^2 b^2 \right\}$. 

\item All probability bounds are derived with respect to $S_n$, which the sum of stationary martingale differences. Under the assumption of stationary sequences we assume that sub-Gaussianity condition holds in order to obtain probability bounds. Define with $g ( y_1,..., y_n )$ a measurable function of the data, which for example could be extended to sample moments of estimators. 
    
\item An important related assumption is the Geometric ergodicity which along with $\beta-$mixing can facilitate the development of further the asymptotic theory in time series model. Moreover the partitioning entropy condition clearly holds in the case of stationary sequences but the main challenge in the case of nonstationary time series with a LUR process representation is the presence of the nuisance parameter of persistence. 
    
\item The theoretical framework presented in the paper shows that the theory can be also extended to the case of M estimators (such as quantile autoregression) using suitable smoothing conditions and deriving the corresponding probability bounds.  
    
\end{itemize}

\newpage

\subsubsection{Sub-Weibull random vectors under $\beta-$mixing}

Most analysis on lasso assume that data that have sub-Gaussian or subexponential tails. These assumptions ensure that the moment generating function exists, at least for some values of the free parameter. Nonexistence of the moment generating function is often taken as a definition of having a heavy tail.  

\begin{definition}[Sub-Weibull random variables] A sub-Gaussian random variable $X$ can be defined as one for which  
\begin{align}
\mathbb{E} \left( | X |^p \right)^{1/p} \leq K \sqrt{p}, \ \ \forall \ p \geq 1, \ \text{where} \ K \ \text{is constant}.    
\end{align}
\end{definition}
A natural realization which allows for heavier tails is as follows. Fix some $\gamma > 0$, and require
\begin{align}
\norm{ X }_p := \left( \mathbb{E} | X |^p \right)^{1/p} \leq K p^{1 / \gamma} \ \forall \ p \geq 1 \vee \gamma.    
\end{align}
Notice that the condition above requires that the tail is no heavier than that of a Weibull random variable with parameter $\gamma$.

\medskip

\begin{definition}[Sub-Weibull Random Variables Properties]
Let $X$ be a random variable. Then, the following statements are equivalent for every $\gamma > 0$. The constants $K_1, K_2, K_3$ differ from each other at most by a constant depending only on $\gamma$.
\begin{itemize}
\item[(\textit{i})]  The tails of $X$ satisfies 
\begin{align}
\mathbb{P} \left( \left| X \right| > t \right) \leq 2 \mathsf{exp} \left\{ - \left( \frac{t}{K_1}^{\gamma} \right)^{\gamma} \right\} \ \ \forall \ \ t \geq 0.
\end{align}
    
\end{itemize}    
\end{definition}

\subsubsection{Approximation Theorems for Strongly Mixing Random Variables}

\begin{theorem}
Suppose that $( X_k )$ is a strictly stationary sequence of real-valued Random Variables with $\mathbb{E} (X_k) = 0$, $\mathbb{E} \left(X_k^2\right) < \infty$ and $\mathsf{Var} \left(S_n\right) \to \infty$ as $n \to \infty$. Suppose that $\delta > 0$ and that $\lambda > 1 + 3 / \delta$ are real numbers, such that $\alpha (n) = o \left( \left( \mathsf{log} \right)^{- \lambda} \right)$ as $n \to \infty$, and
\begin{align}
\underset{ n \to \infty }{\mathsf{sup} }  \  \frac{ \mathbb{E} | S_n |^{2 + \delta} }{ \left( \mathsf{Var} S_n \right)^{ \frac{(2 + \delta)}{2} }  } < \infty.    
\end{align}
Then, $\exists \  \sigma^2$, such that $0 < \sigma^2 < \infty$, such that $\underset{ n \to \infty }{ \mathsf{lim} } \frac{1}{n} \mathsf{Var} (  S_n ) = \sigma^2$. Thus, without changing its probability process $\big( S(t), t \geq 0 \big)$ can be redefined on another probability space, together with a Wiener process $\big( W(t), t \geq 0 \big)$ 
\begin{align}
\mathbb{P} \big( \left| S(t) - W( \sigma^2 t ) \right| \big) = o \left( t^{1/2} \left( \mathsf{log}  \mathsf{log} t \right)^{-1/2} \right) \ \ \ \text{as} \ \ t \to \infty   
\end{align}
\end{theorem}

\newpage 

\subsection{Discussion}
The various mixing properties of time series (either contemporaneously or temporally dependent) are important for understanding how these particular features affect the asymptotic theory for model estimators, test statistics and the relevant properties of model selection methods based on dependent data which exhibit these features. Some potential data specific features such as the presence of serial dependence among errors and regressors, heteroscedasticity and fat tails can lead to misspecified time series models which especially in a high dimensional setting require to establish validity via error bound determination (see, \cite{adamek2023lasso} and \cite{wong2020lasso} among others). 

In particular, shrinkage methodologies assume a certain structure on the unknown parameter vector of interest.  Generally, the sparsity condition implies that a small but unknown subset of the high dimensional vector of covariates is considered to have "significantly different than zero" coefficients, while the remaining subset of covariates have negligible, or even exactly zero, coefficients. In other words, for Lasso shrinkage methodologies to render meaningful inference the penalty function exploits the underline sparsity condition. Specifically, \cite{adamek2023lasso} obtain novel theoretical results for both point estimation and inference via the desparsified lasso. Furthermore, they consider a general time series framework where the regressors and error terms are allowed to be non-Gaussian, serially correlated and heteroscedastic and the number of variables can grow faster than the time dimension.  

\begin{example}[Nuclear Norm Regularized Estimation]
\

Denote the Frobenius norm of an $( N \times T )$ matrix $A$ be a $\norm{ A }_2 := \left( \sum_{i=1}^N \sum_{t=1}^T A_{it}^2 \right)^{1 / 2}$ which implies that $\norm{ A }_2^2 = \left( \sum_{i=1}^N \sum_{t=1}^T A_{it}^2 \right)$. Then, the OLS estimator of $\beta$ is given by 
\begin{align}
\widehat{\beta}_{OLS} := \underset{ \lambda \in \mathbb{R}^{N \times R}, f \in \mathbb{R}^{ T \times R }  }{ \mathsf{min} } \  \frac{1}{2 NT} \norm{ Y - \beta . X - \lambda f^{\prime} }_2^2,
\end{align}
Relevant research questions of interest: 

\begin{itemize}

\item How does the limiting distribution of the OLS estimator in the above high-dimensional environment is affected under different mixing conditions and distributional assumptions on the error term?

\item What approach we need to follow in order to obtain error bounds on relevant statistical quantities? 

\end{itemize}

\end{example}

\begin{example}[Prewhitening Estimators]
The idea behind prewhitening goes like this: Suppose one is nonparametrically estimating a function $f( \lambda )$ at $\lambda_0$ by taking unbiased estimates of $f( \lambda )$ at a number of points $\lambda$ in a neighbourhood of $\lambda_0$ and averaging them. Additionally, if the function $f( \lambda )$ is flat in this neighbourhood, then this procedure yields an unbiased estimator of $f( \lambda_0 )$. If $f( \lambda )$ is not flat in this neighbourhood, then the procedure is biased and the magnitude of the bias depends on the degree of non-constancy of $f( \lambda )$. In the time series literature, the idea of prewhitening has been applied to nonparametric estimators of the spectral density function. In this case, one tries to transform (filter) the data in such a way that the transformed data is uncorrelated, since an uncorrelated sequence has a flat spectral density function. Understanding the technical tools of learning theory and mixing conditions are useful for this application as well. 
\end{example}


\newpage

\section{Time Series Regression Models with Many Covariates}

A recent growing literature develops econometric frameworks for estimation and inference in regression models with many covariates. Relevant studies include among others
\cite{cattaneo2018inference}, \cite{karmakar2022long} and \cite{wei2023inference} (see, also \cite{farrell2015robust} although not in a time series setting).

\subsection{A forecasting Application}

Let $\hat{y}_t$ be the forecast of $y_t$ based upon information up to $t-1$. Notice that when the interested of the researcher is the one-period ahead forecast, $\left( y_t - \hat{y}_t \right)^2$ is the cost to be minimized. However, there are two situations where the accumulated cost function, denoted with $\sum_{ j = 1 }^t \left( y_t - \hat{y}_t \right)^2$ is more appropriate. For example, in the sequential forecast case, the forecaster are updated sequentially over many periods and therefore the accumulated cost function is the target to be minimized. 

Consider for example the one-period expected loss function $\mathbb{E}\left( y_{T+1} - \hat{y}_{T+1} \right)^2$. For instance, for the AR(1) model, under the assumption that $\mathbb{E} \left( \epsilon_t^2 | \mathcal{F}_{t-1} \right) = \sigma^2$ almost surely for all $t$, then it can be shown that under appropriate assumptions that $\frac{1}{T} \sum_{t=1}^T \left( y_{T+1} - \hat{y}_{T+1} \right)^2$ almost surely. 

Let $\hat{\beta}_t$ be the Least square estimate of $\beta$, then we have that 
\begin{align}
\hat{\beta}_t = \left( \sum_{j=1}^t \mathbf{Y}_{j-1} \mathbf{Y}_{j-1}^{\prime} \right)^{-1} \left( \sum_{j=1}^t  \mathbf{Y}_{j-1}^{\prime} y_k \right)
\end{align}
where $Y_t = \left\{ y_1,...,y_t \right\}^{\prime}$, then $\hat{y}_t = \hat{\beta}_{t-1}^{\prime} Y_{t-1}$ is the least square prediction of $y_t$ at time $t-1$. 

\begin{theorem}[\cite{Dhrymes2013mathematics}]
Assume that $\epsilon_t$ are \textit{i.i.d} random variables with $\mathbb{E} \left( \epsilon_t \right) = 0$ such that it holds $0 < \mathbb{E} \left( \epsilon_t^2 \right) < \sigma^2 < \infty$. Moreover, let $\mathbf{X}_t = \left( x_{t-1}, ...., x_{t-p} \right)^{\prime}$, $S_T = \sum_{t=1}^T \epsilon_t$. 
\end{theorem}

\begin{lemma}[\cite{Dhrymes2013mathematics}]
Assume that $\mathbf{X}_{t+1} = A \mathbf{X}_{t} + \boldsymbol{\epsilon}_t$, where $\boldsymbol{\epsilon}_t = \left( \epsilon_t, 0,...., 0 \right)^{\prime}$ and the eigenvalues of $A$ are all inside the unit circle. Then, we have that 
\begin{align}
\underset{ T \to \infty }{ \text{lim} } \ \frac{ \displaystyle \sum_{t=1}^T \mathbf{X}_{t} S_t }{ \displaystyle \sqrt{ T \sum_{t=1}^T S_t^2 } } = 0, \ \ \text{almost surely}, 
\end{align}
\end{lemma}

\begin{proof}
It is known that 
\begin{align}
\underset{ T \to \infty }{ \text{lim} } \frac{1}{T} \sum_{t=1}^T \mathbf{X}_t  \mathbf{X}_t^{\prime} = \Sigma, \ \ \text{almost surely}, 
\end{align}
where $\Sigma$ is a positive definite matrix. Further details can be found in \cite{Dhrymes2013mathematics}.
\end{proof}

\newpage 

\subsubsection{Best Linear Predictor}

The main tool of prediction, or forecasting, in time series relies on the concept of the best linear predictor which is often employed in order to obtain the limiting distribution of an arbitraty size vector of autocovariance estimators (see, \cite{Dhrymes2013mathematics}).   

\medskip

\begin{definition}
Let $\left\{ X_t: t \in \mathcal{T} \right\}$ be a zero mean stationary time series indexed on the linear index set. The best linear predictor of $X_{t+h}, h \geq 1$, given $\left\{ X_1,..., X_n \right\}$, is the function given by 
\begin{align}
\hat{X}_{t+h} = \sum_{j=1}^n \alpha_j X_{n+1-j}    
\end{align}
which minimizes 
\begin{align}
S = \mathbb{E} \left[  X_{n+h} - \sum_{j=1}^n \alpha_j X_{n+1-j} \right]^2   
\end{align}
\end{definition}

The first order conditions are given by 
\begin{align}
\frac{ \partial S }{ \partial \alpha} = - 2 \mathbb{E} \left[ X_{n+h} - \alpha^{\prime} X_{(n)} \right] X_{(n)}^{\prime} = 0, \ \ \ 
X_{(n)} 
= 
\begin{bmatrix}
X_n
\\
X_{n-1}
\\
\vdots
\\
X_1
\end{bmatrix}
\end{align}
Thus, by rearranging we get that 
\begin{align}
\mathbb{E} \left[ X_{n+h} X_{(n)}^{\prime} \right] = \alpha^{\prime} \mathbb{E} \left[ X_{(n)} X_{(n)}^{\prime}  \right]
\end{align}
or $\alpha = C_n(i-j) c_{nh}$. 

\medskip

\begin{proposition}[\cite{Dhrymes2013mathematics}]
The limiting distribution of the coefficient vector of the best linear predictor (BLP) obeys the following limit
\begin{align}
\sqrt{n} \left( \hat{\alpha} - \alpha \right) \overset{d}{\to} \mathcal{N} \left( 0, G V G^{\prime}\right), \ \ \ \   
\end{align}
where the matrix $G$ is defined as below:
\begin{align}
G = C_n^{-1} \big[ F - \left( \alpha^{\prime} \otimes I_n \right) \left( B_{ST}, 0 \right) \big] .   
\end{align}
\end{proposition}

\newpage

\subsubsection{Forecast Combination}

Assume that we are interested to construct a forecast combination criterion based a specific econometric model of interest. Then, a simple approach would be to consider the construction of a criterion which uses OOS forecast combinations. 
Let $\mathbf{f}_t$ be a sequence of out-of-sample forecasts for $y_{t+1}$ derived using a set $m$ different econometric models. The combination forecast is then defined as $f( \mathbf{w} ) = \mathbf{w}^{'} \mathbf{f}$. 

Therefore, the OOS empirical MSFE is given by 
\begin{align}
\hat{ \sigma}^2 = \frac{1}{P} \sum_{t = n - P }^n \left( y_{t+1} - \mathbf{w}^{'} \mathbf{f}_t \right)^2
\end{align}
Furthermore, the Granger-Ramanathan forecast combination method implies to select $\mathbf{w}$ to minimize the OOS MSFE. However,  minimization over $\mathbf{w}$ is equivalent to the least-squares regression over $y_t$ on the forecasts 
$y_{t+1} = \mathbf{w}^{'} \mathbf{f} + \epsilon_{t+1}$. Then, the unrestricted least-squares gives the following vector of weights
\begin{align}
\hat{ \mathbf{w} } = \left( \sum_{t = n - P }^n \mathbf{f}_t \mathbf{f}^{'}_t  \right)^{-1} \sum_{t = n - P}^n \mathbf{f}_{t} y_{t+1}
\end{align}
Notice that the above unrestricted least-squares approach can produce weights which are far outside $[0,1]$, thus an alternative representation can be constructed by subtracting $y_{t+1}$ from each side, such that $0 = \mathbf{w}^{'} \mathbf{f} - y_{t+1} + \epsilon_{t+1}$, and by defining $\mathbf{e}_{t+1} = \left( \mathbf{f}_t  - y_{t+1} \right)$ to be the negative forecast errors. Then, it holds that $
0 = \mathbf{w}^{'} \mathbf{e}_{t+1} + \epsilon_{t+1}$, which is the regression of 0 on the forecast errors. Therefore, the constrained GR weights solve the following problem. 
\begin{align}
\underset{ \mathbf{w} }{ \mathsf{min}} \left\{ \mathbf{w}^{'} \mathbf{ A } \mathbf{w} \right\} \ \text{subject to} \ \sum_{m=1}^M w(m) = 1 \ \text{and} \ 0 \leq w(m) \leq 1  
\end{align}
where $\mathbf{A} = \displaystyle \sum_t \mathbf{e}_{t+1} \mathbf{e}^{'}_{t+1}$, the $ M \times M $ matrix of forecast error empirical variances/covariances. 

\begin{proposition}
In linear models, the combination forecast is the same as the forecast based on the weighted average of the parameter estimates across the different models. That is, 
\begin{align}
\hat{y}_{n+1}( \mathbf{w} ) &= \sum_{m=1}^M w(m) \hat{y}_{n+1}(m)
= \sum_{m=1}^M w(m) \mathbf{x}_n(m)^{'} \widehat{\boldsymbol{\beta}}(m) 
\equiv
\mathbf{x}_n(m)^{'} \widehat{\boldsymbol{\beta}}( \mathbf{w})
\end{align}
where $\widehat{\boldsymbol{\beta}}( \mathbf{w}) = \displaystyle \sum_{m=1}^M w(m) \widehat{\boldsymbol{\beta}}(m)$.  
\end{proposition}
Relevant studies include \cite{hansen2007least}, \cite{hansen2008least} and \cite{cheng2015forecasting} while a recent approach related to forecasting can be found in \cite{boot2019forecasting} and \cite{swanson2020predicting}. For nonstationary time series models relevant studies include \cite{hansen2010averaging},  \cite{coroneo2020comparing} and \cite{kejriwal2021generalized} among others.

\newpage

\subsection{Statistical Theory for Lasso Regression Models}

\subsubsection{Statistical Principles}

\begin{example}(Predicting inflation rates)

\begin{itemize}
    \item Minimization of the empirical risk for a sample $\left\{ y_1,..., y_n  \right\}:$
    \begin{align}
        \widehat{\mu} = \underset{ \mu \in \mathbb{R} }{ \mathsf{arg \ min} } \frac{1}{T} \sum_{t=1}^n \left( y_t - \mu \right)^2. 
    \end{align}
    
    \item Taking the first-order conditions give:
    \begin{align}
        \frac{ \partial  }{ \partial \mu } \left\{ \frac{1}{T} \sum_{t=1}^n \left( y_t - \mu \right)^2 \right\} \bigg|_{ \mu = \widehat{\mu} } = - \frac{2}{n} \sum_{t=1}^n  \left( y_t - \widehat{\mu} \right) = 0
    \end{align}
    
\end{itemize}
Thus we obtain that, $\displaystyle  \widehat{\mu} := \frac{1}{T} \sum_{t=1}^n y_t \equiv \bar{y}$. 
\end{example}

\begin{corollary}[Shrinkage method]
If we drop the requirement of unbiasedness, can we find estimators $\widehat{\beta}$ with the following property:
\begin{align}
\label{condition}
\sum_{j=1}^K MSE \left( \widehat{\beta}_j \right) \leq \sum_{j=1}^K MSE \left( \widehat{\beta}^{ols}_j \right)    
\end{align}

where $MSE \left( \widehat{\beta}_j \right) = \mathbb{E} \left[ \left( \widehat{\beta}_k - \beta_k  \right)^2 \right]$. 
\end{corollary}

\begin{itemize}
    \item If we aim for \eqref{condition} to hold for all possible $\beta \in \mathbb{R}^K$ then, this is equivalent to imposing a condition that the OLS estimator is admissible. 
    
    \item If we aim for \eqref{condition} to hold for some given fixed $\beta \in \mathbb{R}^K$, then finding such a $\widehat{\beta}$ that improves on $\widehat{\beta}^{ols}$ is relatively easy.   
    
\end{itemize}

\begin{corollary}[Shrinkage: Ridge Regression]
Suppose that there exists a given penalty parameter $\lambda \geq 0$ the ridge regression estimators is given by the following expression 
\begin{align}
\widehat{\beta}^{\ell_2 } (\lambda) = \underset{ b \in \mathbb{R}^K }{ \mathsf{arg \ min} } \left[ \frac{1}{n} \sum_{t=1}^n \left( y_t - x_t b \right)^2 + \lambda \norm{b}_2^2 \right]  \equiv \left( \boldsymbol{X}^{\prime} \boldsymbol{X} + n \boldsymbol{I}_k \right)^{-1} \left( \boldsymbol{X}^{\prime} \boldsymbol{y} \right)  
\end{align}
\end{corollary}

Recall that the Euclidean norm is given by: $\norm{b}_2 = \sqrt{ b^{\prime} b }$. In other words, the inverse of $\left( \boldsymbol{X}^{\prime} \boldsymbol{X} \right)$ is regularized by $n \lambda \boldsymbol{I}_k$ (called Tikhonov regularization for ill-posed problems).

\newpage

Special cases of Ridge Regression include:
\begin{itemize}
    
    \item When the design matrix $\boldsymbol{X}^{\prime} \boldsymbol{X}$ is singular then it corresponds to a large variance of the OLS estimator.
    
    \item When $n^{-1} \boldsymbol{X}^{\prime} \boldsymbol{X} = \boldsymbol{I}_p$, then simply $\widehat{\beta}^{\ell_2} (\lambda) = \left( 1 + \lambda \right)^{-1} \widehat{\beta}^{ols}$.
    
\end{itemize}

\begin{corollary}[Bias of the Ridge Estimator when $p < n$]

Consider the ridge estimator as below:
\begin{align*}
\widehat{\boldsymbol{\beta}}_{ridge} (\lambda) 
&= 
\left( \boldsymbol{X}^{\prime} \boldsymbol{X} + \lambda \boldsymbol{I} \right)^{-1} \left( \boldsymbol{X}^{\prime} \boldsymbol{Y} \right)      
=
\left( \boldsymbol{R} + \lambda \boldsymbol{I} \right)^{-1} \boldsymbol{R} \left( \boldsymbol{R}^{-1} \boldsymbol{X}^{\prime} \boldsymbol{Y} \right)
\\
&=
\big[ \boldsymbol{R}  \left( \boldsymbol{I} + \lambda \boldsymbol{R}^{-1} \right) \big]^{-1}  \boldsymbol{R} \big[ \left( \boldsymbol{X}^{\prime} \boldsymbol{X} \right)^{-1} \left( \boldsymbol{X}^{\prime} \boldsymbol{Y} \right) \big]
\\
&=
\left( \boldsymbol{I} + \lambda \boldsymbol{R}^{-1} \right)^{-1} \boldsymbol{R}^{-1} \boldsymbol{R} \widehat{\boldsymbol{\beta}}_{ols}
=
\left( \boldsymbol{I} + \lambda \boldsymbol{R}^{-1} \right)^{-1} \widehat{\boldsymbol{\beta}}_{ols}. 
\end{align*}
\end{corollary}
Therefore, it holds that 
\begin{align*}
\mathbb{E} \left[ \widehat{\boldsymbol{\beta}}_{Ridge} (\lambda) | \boldsymbol{X} \right] 
&= 
\mathbb{E} \big[ \left( \boldsymbol{I} + \lambda \boldsymbol{R}^{-1} \right)^{-1} \widehat{\boldsymbol{\beta}}_{ols} | \boldsymbol{X} \big]    
= 
\left( \boldsymbol{I} + \lambda \boldsymbol{R}^{-1} \right)^{-1} \boldsymbol{\beta} 
\\
\mathbb{E} \left[ \widehat{\boldsymbol{\beta}}_{Ridge} (\lambda) \right] 
&= 
\mathbb{E} \left[ \left( \boldsymbol{I} + \lambda \boldsymbol{R}^{-1} \right)^{-1} \right] \boldsymbol{\beta} \neq \boldsymbol{\beta}. 
\end{align*}
Based on the above results we can introduce the model selection idea with an $\ell_0-$penalty. Thus, in the literature various studies mention that one of the main advantages of the Lasso shrinkage is that it "bets on sparsity". Therefore, using an $\ell_0$ penalty term we have that
\begin{align}
\norm{\beta}_0 = \sum_{j=1}^p   \boldsymbol{1} \left\{ \beta_j \neq 0 \right\} 
\end{align}
which implies that
\begin{align*}
\widehat{ \boldsymbol{\beta}}^{\ell_0 } (\lambda) 
&= \underset{  \boldsymbol{\beta} \in \mathbb{R}^p }{ \mathsf{arg \ min} } \left[ \frac{1}{n} \sum_{t=1}^n \left( Y_t - \boldsymbol{\beta}^{\prime} \boldsymbol{X}_t \right)^2 + \lambda  \sum_{j=1}^p \boldsymbol{1} \left\{ \beta_j \neq 0 \right\}    \right] 
= 
\underset{  \boldsymbol{\beta} \in \mathbb{R}^p }{ \mathsf{arg \ min} } \left[ \frac{1}{n} \norm{ \boldsymbol{Y} - \boldsymbol{X} \boldsymbol{\beta} }_2^2  + \lambda \norm{ \boldsymbol{\beta} }_0 \right].
\end{align*}
where $\lambda \geq 0$ is the tuning (penalty) parameter which needs to be chosen. Similarly, using an $\ell_1$ penalty
\begin{align*}
\widehat{ \boldsymbol{\beta}}^{\ell_1} (\lambda) 
=
 \underset{  \boldsymbol{\beta} \in \mathbb{R}^p }{ \mathsf{arg \ min} } \left[ \frac{1}{n} \sum_{t=1}^n \left( Y_t - \boldsymbol{\beta}^{\prime} \boldsymbol{X}_t \right)^2 + \lambda  \sum_{j=1}^p \left| \beta_j \right| \right] 
=
 \underset{  \boldsymbol{\beta} \in \mathbb{R}^p }{ \mathsf{arg \ min} } \left[ \frac{1}{n} \norm{ \boldsymbol{Y} - \boldsymbol{X} \boldsymbol{\beta} }_2^2  + \lambda \norm{ \boldsymbol{\beta} }_1 \right]  
\end{align*}
where $\lambda \geq 0$ is the tuning (penalty) parameter which needs to be chosen. 

\medskip

The lasso shrinkage approach was proposed in the seminal study of \cite{tibshirani1996regression}. Various studies examine its statistical properties such as \cite{fan2001variable} and \cite{zhang2010nearly} among others, while many variants of the lasso shrinkage have been proposed in the literature such as \cite{zou2005regularization},  \cite{zou2006adaptive}, \cite{huang2008adaptive},  \cite{park2008bayesian} and \cite{james2009dasso}. A relevant empirical study is present by \cite{katsouris2021forecast}. Recently, the literature focuses on applications of the lasso shrinkage approach with group structure (see, also \cite{huang2012selective} and \cite{bing2022inference}).

\newpage

In particular, \cite{klau2018priority} propose the priority lasso, which considers the use of different penalty across a baseline group of covariates (see, \cite{breheny2015group} for case with grouped predictors) and the remaining set of variables while  \cite{campbell2017within} propose a variable selection method using the exclusive Lasso. Now, the Lasso estimates of the slope coefficients in a linear regression model solve the $\ell_1-$penalized least regression problem:
\begin{align}
\underset{ \beta }{ \mathsf{min} } \sum_{i=1}^n \left(  Y_i - \sum_{j=1}^p \beta_j X_{ij} \right)^2 \ \ \text{subject to} \ \  \sum_{j=1}^p \left| \beta_j \right| \leq s,   
\end{align}
or, equivalently we have that 
\begin{align}
\underset{ \beta }{ \mathsf{min} } \sum_{i=1}^n \left(  Y_i - \sum_{j=1}^p \beta_j X_{ij} \right)^2 + \lambda \sum_{j=1}^p \left| \beta_j \right| 
\end{align}
where $\boldsymbol{\beta} = \left( \beta_1,..., \beta_p \right)^{\prime}$, and $s$ and $\lambda$ are tuning parameters.  The lasso uses a constraint in the form of $\ell_1-$norm: $\sum_{j=1}^p \left| \beta_j \right| \leq s$. Therefore, by using the $\ell_1-$penalty, the Lasso achieves variable selection and shrinkage simulatneously, while $\lambda$ controls the amount of regularization. 

Denote with $\boldsymbol{\beta}^0$ the true vector of parameters. There are three types of errors of interest in LASSO regression, that is, (i) the prediction error: $
\norm{ \boldsymbol{X} \left( \widehat{\boldsymbol{\beta}} - \boldsymbol{\beta}^0  \right) }_2^2$, (ii) the parameter estimation error: $
\norm{ \left( \widehat{\boldsymbol{\beta}} - \boldsymbol{\beta}^0  \right) }_2^2$, and (iii) the model selection error: $\mathbb{P} \left( \mathsf{supp} \left( \widehat{\boldsymbol{\beta}} \right) = \left( \boldsymbol{\beta}^0 \right) \right)$. In other words, the prediction problem implies that given a random sample $\left( y_i, x_i \right), i = 1,...,n$ and the covariates $x_{n+1}$ of an additional observation, we aim to predict the unobserved outcome $y_{n+1}$.  Choice of $\lambda$ is achieved using Cross-Validation method: 
\begin{align}
\lambda^{*} = \underset{ \lambda > 0 }{ \mathsf{\arg \ min} } \sum_{i=1}^n \left( y_i - x_i \widehat{\beta}^{(-i)} \left(\lambda\right) \right)^2,
\end{align}
where $\widehat{\beta}^{(-i)}(\lambda)$ is the estimator $\widehat{\beta}^{\ell_1 } (\lambda)$
obtained after dropping observation $i$ from the sample. Notice that cross-validation is usually used to choose the tuning parameter $\lambda$ for in-sample model estimation and variable selection. However, there are theoretical results that show that LASSO performs well for prediction even for large number of covariates $p$, as long as the true parameter $\beta$ is sparse. Under appropriate regularity conditions with high probability (for large $n$) we have that
\begin{align}
\frac{1}{n} \norm{ X \widehat{\beta}^{\ell_1}(\lambda^{*}) - X \beta }_2^2 \leq C \norm{\beta}_1 \sqrt{ \frac{\mathsf{log}(p)}{n} }
\end{align}
The cross-validation is usually used to choose the tuning parameter $\lambda$ for in-sample model estimation and variable selection. Furthermore, a key property for model selection implies that under appropriate regularity conditions with high probability (large $n$) one can show the following relation  
\begin{align}
\widehat{S} := \bigg\{ j \in \left\{ 1,..., p \right\} : \widehat{\beta}^{\ell_1}_j (\lambda) \neq 0 \bigg\} = \bigg\{ j \in \left\{ 1,..., p \right\} : \beta_j \neq 0  \bigg\}.
\end{align}

\newpage

\begin{lemma}[Maximum Inequality for Gaussians]
For a sample of $n$ Gaussian random variables $Z_i$, for $i = 1,...,n$, such that $\mathbb{E} ( Z_i ) = 0$ and $\mathbb{E} ( Z^2_i ) \leq \bar{\sigma}^2_z, \ \forall \ i$, then it holds that
\begin{align}
\mathbb{P} \left( \underset{ i = 1,...,n }{ \mathsf{max} } \left| Z_i \right| \geq z \right) \leq 2 n e^{ - \frac{n z^2}{2 \bar{\sigma}^2_z } }. 
\end{align}
\end{lemma}
Under our assumptions, $\frac{ \boldsymbol{U}^{\prime} \boldsymbol{X}_j }{T}$, $j = 1,...,p$, is Gaussian with mean zero and variance bounded by $\frac{ \mathcal{C} \sigma^2 }{n}$. Then, it holds that 
\begin{align}
\mathbb{P} \left( \norm{  \frac{\boldsymbol{U}^{\prime} \boldsymbol{X}}{n} }_{\infty} \geq z \right) \leq 2 p e^{ - \frac{n z^2}{2 \mathcal{C} \bar{\sigma}^2 } }.        
\end{align}
The aim is to make (bound) the above probability as small as possible. Moreover, to obtain convergence rates and bounds we can use the following definitions $z = \sigma \sqrt{ \frac{2 \mathsf{log} \left( \frac{ep}{\delta} \right) }{n} }$, where $\delta > 0$ and $e := \mathsf{exp}(1)$. Set also $\mathcal{C} = 1$. Therefore, we obtain that 
\begin{align*}
\mathbb{P} \left( \norm{  \frac{\boldsymbol{U}^{\prime} \boldsymbol{X}}{n} }_{\infty} \geq \sqrt{ \frac{2 \mathsf{log} \left( \frac{ep}{\delta} \right) }{n} } \right) 
&\leq 
2 e^{ - \mathsf{log} \left( \frac{ep}{\delta} \right) + \mathsf{log}(p) }
= 
2 e^{ - \mathsf{log}(e) +  \mathsf{log}(\delta) } = \frac{2}{e} \delta < \delta. 
\end{align*}
Therefore, we get that 
\begin{align}
\frac{1}{n} \norm{ \boldsymbol{X} \left( \widehat{\boldsymbol{\beta}} - \boldsymbol{\beta}^0 \right) }_2^2 \leq 4 \sigma \norm{ \boldsymbol{\beta}^0 }_1 \sqrt{ \frac{ 2 \mathsf{log} \left( ep / \delta \right) }{n} }.
\end{align}
Another important aspect related to the statistical properties of the lasso shrinkage approach is the \textit{sign consistency}. Specifically, \textit{Strong Sign Consistency:} implies that one can use a pre-selected $\lambda$ to achieve consistent model selection via the LASSO. On the other hand, a \textit{General Sign Consistency:} means that for a random realization there exists a correct amount of regularization that selects the true model. Further discussion and related asymptotic results on the consistency and regularization property\footnote{Professor Marcelo C. Medeiros gave a seminar with title: "Bridging Factor and Sparse Models" (see,  \cite{fan2021bridging}) at the Department of Economics, University of Southampton on the 6th of October 2021.} under different distributional assumptions can be found in \cite{medeiros2016}.    

\medskip

\begin{lemma}[Oracle Property] Let $\widehat{\boldsymbol{\beta}}_{ols, S}$ denote the OLS estimator of $\boldsymbol{\beta}^0_S$. Suppose that a lower bound for $\beta$ can be established such that $\beta_{\text{min}} > \left( \lambda / n^{1 - \xi/2}   \right) \left( s^{1/2} / \phi_{\text{min}} \right)$, then it holds that 
\begin{align}
\sqrt{n} \boldsymbol{\alpha}^{\prime} \left[ \widehat{\boldsymbol{\beta}}_S - \boldsymbol{\beta}_S^0  \right] = \sqrt{n} \boldsymbol{\alpha}^{\prime} \left[ \widehat{\boldsymbol{\beta}}_{ols,S} - \boldsymbol{\beta}_S^0  \right]     
\end{align}
for any $s-$dimensional vector $\boldsymbol{\alpha}$ with Euclidean norm 1. 
\end{lemma}

\newpage 

\begin{example}[Modelling Unobserved Heterogeneity] Consider the linear regression model below: 
\begin{align}
y = \sum_{j=1}^p x_j \beta_j^0 + \varepsilon, \ \ \ y = X \beta^0 + \varepsilon, 
\end{align}
where $X = \left( x1,..., x_p \right)$ is an $\left( n \times p \right)$ matrix, $y = \left( y_1,..., y_n \right)^{\prime}$ is an $( n \times 1 )$ vector and $\beta^0 = \left( \beta_1^0,..., \beta_p^0 \right)$ is the true parameter vector. Then, the \textit{homogeneity assumption} implies that the regression coefficients $\beta_j$ share the same value in their unknown clusters such that under the null hypothesis: 
\begin{align}
H_0: \ \beta_j^0 = \beta_{A,k}^0 \ \ \ \text{for all} \ j \in A_k
\end{align}
which also demonstrates the variable selection consistency property.
\end{example}

\subsubsection{Dantzig Selector}

A large stream of literature has focused on the $\ell_1$ penalized LASSO estimator of parameters in high-dimensional linear regression when the number of variables can be much larger than the sample size. We consider the linear regression with many covariates such as 
\begin{align}
\boldsymbol{ y } = X \beta + \boldsymbol{ \epsilon } 
\end{align}
where $X$ is the $n \times M$ deterministic design matrix, with $M$ possibly much larger than $n$, and $\boldsymbol{ \epsilon }$ is a vector of $\textit{i.i.d}$ standard normal random variables. In particular, we are interested in the case of sparsity parameters, which implies that the high-dimensional vector $\beta$ has coefficients that are mostly $0$. Denote with $\widehat{S} \left( \beta \right)$, the residual sum of squares such that
\begin{align}
\widehat{S} \left( \beta \right) = \frac{1}{n} \sum_{i=1}^n \bigg( Y_i - f_{\beta} \left( Z_i \right) \bigg)^2
\end{align}
for all $\beta \in \mathbb{R}^{M}$. Define the Lasso solution such that $\widehat{\beta}_L = \left( \widehat{\beta}_{1,L},....,  \widehat{\beta}_{M,L} \right)$ by the following expression 
\begin{align}
\widehat{\beta}_L = \underset{ \beta \in \mathbb{R}^M  }{ \text{arg min} } \left\{ \widehat{S} \left( \beta \right) + 2r \sum_{ j = 1 }^M \norm{ f_j }_n \left| \beta_j \right| \right\}, 
\end{align}
where $r > 0$ is some tuning constant. Then, the corresponding Lasso estimator 
\begin{align}
 \widehat{f}_L (x) = f_{ \widehat{ \beta }_{L} } (x) = \sum_{j=1}^M \widehat{ \beta }_{j,L} f_j (z)
\end{align}
A necessary and sufficient condition of the minimizer is that 0 belongs to the subdifferential of the convex function $\beta \mapsto n^{-1} \left| y - X \beta \right|_2^2 + 2 r \left| D^{1 / 2} \beta \right|_1$.

\newpage

This implies that the Lasso selector $\widehat{ \beta }_{L}$ satisfies the constraint
\begin{align}
\left| \frac{1}{n} D^{ - 1 / 2} X^{\top} \left( y - X \widehat{ \beta }_{L} \right) \right|_{ \infty } \leq r, \ \ \ D = \text{diag} \left\{ \norm{ f_1 }_n^2, ... , \norm{ f_M }_n^2 \right\}.
\end{align}   
Various studies in the literature present relevant applications and statistical theory of the Dantzig selector such as \cite{osborne2000lasso}, \cite{candes2007dantzig}, \cite{bickel2009simultaneous}, \cite{koltchinskii2009dantzig} and \cite{james2009dasso} among others. We say that $\beta \in \mathbb{R}^M$ satisfies the Dantzig constraint if $\beta$ belongs to the set
\begin{align}
\left\{ \beta \in \mathbb{R}^M : \left| \frac{1}{n} D^{1 / 2} X^{\top}    \left( y - X \beta_{L} \right) \right|_{ \infty } \leq r \right\}. 
\end{align}
The Dantzig estimator is defined by the following expression 
\begin{align}
\widehat{ f }_D (z) = f_{ \widehat{ \beta }_D } = \sum_{j=1}^M \widehat{ \beta }_{ j, D } f_j (z), 
\end{align}
where $\widehat{\beta}_D = \left( \widehat{\beta}_{1,D},..., \widehat{\beta}_{M,D} \right)$ is the Dantzig selector. By the definition of Dantzig selector, we have that $\left| \widehat{\beta}_{D}  \right|_1 \leq \left| \widehat{\beta}_{L}  \right|_1$. Notice that the Dantzig selector is computationally feasible, since it reduces to a linear programming problem. Finally, for any $n \geq 1$, $M \geq 2$, we consider the Gram matrix as below
\begin{align}
\Psi_n = \frac{1}{n} X^{\prime} X = \left( \frac{1}{n} \sum_{i=1}^n f_j (Z_i) f_{ j^{\prime} } (Z_i) \right)_{ 1 \leq j, j^{\prime} \leq M }
\end{align}
and let $\phi_{\text{max}}$ denote the maximal eigenvalue of $\psi_n$.

\subsubsection{Oracle inequalities for prediction loss}

In this section, we prove sparsity oracle inequalities for the prediction loss of the Lasso and Dantzig estimators. These inequalities allow us to bound the difference between the prediction errors of the estimators and the best sparse approximation of the regression function (e.g., by an oracle that knows the truth but is constrained by sparsity). We demonstrate that the distance between the prediction losses of the Dantzig and Lasso estimators is of the same order as the distances between them and their oracle approximations. Recall that an estimator has the oracle property if it is variable selection consistent and the limiting distribution of its subvector corresponding to the non-zero coefficients is the same as if their set were known prior to estimation (see, \cite{giurcanu2016thresholding}).    

\begin{theorem}
Let $W_i$ be independent $\mathcal{N} \left( 0, \sigma^2 \right)$ random variables with $\sigma^2 > 0$. Fix some $\epsilon > 0$ and integers $n \geq 1$, $M \geq 2$, $1 \leq s \leq M$. Let Assumption RE $(s, 3 + 4/ \epsilon )$ be satisfied. Consider the Lasso estimator $\widehat{f}_L$ with $r = A \sigma \sqrt{ \frac{\text{log} M }{n} }$, for some $A > 2 \sqrt{2}$. Then, with probability at least $1 - M^{1 - A^2 / 8}$, it can be proved that $
\norm{ \widehat{f}_L - f }_n^2$ is bounded.
\end{theorem}

\newpage

\subsection{Parameter Estimation and Model Selection Consistency Properties}

Consider the linear regression model as below (see, \cite{amann2018uniform})
\begin{align}
y = X \beta + \epsilon    
\end{align}
where $y \in \mathbb{R}^n$ is the response vector, $X \in \mathbb{R}^{ n \times p}$ the nonstochastic regressor matrix which is assumed to have full column rank, $\beta \in \mathbb{R}^p$ the unknown parameter vector, and $\epsilon \in \mathbb{R}^n$ the unobserved stochastic error term consisting of i.i.d distributed components with mean zero and finite second moments, defined on some probability space $\left( \Omega, \mathcal{F}, \mathbb{P} \right)$. Model consistency results can be found in \cite{shibata1986consistency} and \cite{potscher1991effects}. Moreover, to define the adaptive Lasso estimator, we consider the following expression 
\begin{align}
L_n( b ) = || y - X b ||^2 + 2 \sum_{j=1}^p \lambda_j  \frac{ | b_j | }{ | \hat{b}_j | }   
\end{align}
where $\hat{\beta}$ denotes the OLS estimator. In particular, we assume that the event $\left\{ \hat{\beta}_j = 0 \right\}$ to have zero probability, for all $j = 1,...,p$ and thus we do not consider this event occurring in the subsequent analysis.

\begin{lemma}[Equivalence to LS estimator, \cite{amann2018uniform}]
\

If $\lambda^* \to 0$, then $\hat{\beta}_{AL}$ and $\hat{\beta}_{LS}$ are asymptotically equivalent in the sense that 
\begin{align}
\sqrt{n} \left( \hat{\beta}_{AL} - \hat{\beta}_{LS} \right) \to 0, \ \ \ \text{as} \ \ n \to \infty \ \ \ \text{for all} \ \ \omega \in \Omega.     
\end{align}
\end{lemma}

\begin{remark}
The above Lemma shows that in the case that $\lambda^* \to \infty$, then the adaptive Lasso estimator is asymptotically equivalent to the LS estimator. Furthermore, in terms of consistency in parameter estimation it can be shown that the adaptive lasso estimator is both pointwise and uniformly consistent for the unknown parameter vector of the possibly high dimensional linear model. 
\end{remark}

\subsubsection{Consistency Properties}

Consider any sequence $\left( \beta_n  \right)_{ n \in \mathbb{N} } \subset \mathbb{R}^p$ converging to $\phi$ and let 
\begin{align}
f_n ( \beta ) = \mathbb{P}_{ \beta } \left(  \beta \in \hat{\beta} - \sqrt{ \frac{ \lambda^* }{ n } } \mathcal{M}_d \right).
\end{align}
Then, by the Portmanteau Theorem, we have that 
\begin{align*}
0 \leq \underset{ n }{ \mathsf{lim \ inf} } \ \underset{ \beta \in \mathbb{R}^p }{ \mathsf{lim} } \ f_n ( \beta )  
&\leq 
\underset{ n }{ \mathsf{lim \ sup} } \ \underset{ \beta \in \mathbb{R}^p }{ \mathsf{lim} } \  f_n ( \beta )  
\leq    
\underset{ n }{ \mathsf{lim \ sup} } \ \mathbb{P}_{ \beta_n } \left(  \sqrt{ \frac{ n }{ \lambda^* } } \left( \hat{\beta} - \beta_n \right) \in \mathcal{M}_d \right)
\\
&\leq  
\mathbb{P}_{ \phi } \big(  \underset{ u }{ \mathsf{arg \ min} } \ V_{\phi} (u) \in \mathcal{M}_d \big) = \mathbf{1} \left\{ m \in \mathcal{M}_d \right\} = 0.
\end{align*}

\newpage

Then, for any $\omega \in \Omega$, we have that 
\begin{align}
\mathcal{M} =  \underset{ \phi \in \bar{\mathbb{R}}^p }{ \bigcup } \underset{ u \in \mathbb{R}^p }{  \mathsf{ arg \ min } } V_{\phi} (u)( \omega ). 
\end{align}

In other words, while the limit of $\sqrt{ \frac{ n }{ \lambda^* } } \left( \hat{\beta} - \beta_n \right)$ will in general be random, the set $\mathcal{M}$ is not random. In particular, Proposition 8 shows that, for any $\omega$, the union of limits over all possible sequences of unknown parameters is always given by the same compact set $\mathcal{M}$.    

Specifically, this observation is central for the construction of confidence regions in the following section. It also shows that while in general, a stochastic component will survive in the limit, it is always restricted to have bounded support that depends on the regressor matrix and the tuning parameter through the matrix $C$ and the quantities $\psi$ and $\lambda^0$. Interestingly, $\mathcal{M}$ only depends on $\psi$ for the components where $\psi_j = \infty$, in which case the set $\mathcal{M}$ loses a dimension. In particular, this can be seen as a result of the $j-$th component being penalized much less than the maximal one, so that the scaling factor used in Theorem 7 is not large enough for this component to survive in the limit.

Consider the following expression (see, \cite{amann2018uniform})
\begin{align}
\left|  \left( \frac{ X X^{\prime} }{ n_k } \left(  \hat{\beta}_{AL} - \hat{\beta}_{LS} \right)  \right)_j  \right|  = \frac{ \lambda_j }{ n_k } \frac{1}{ \left| \hat{\beta}_{LS,j} \right| } 
\end{align}

where $\hat{\beta}_{AL} \neq 0$. Notice that the left-hand side is bounded by $L$, whereas the right-hand side converges to $\frac{ c }{ | \beta_j | }$ in probability. We therefore get $\mathbb{P}_{ \beta } \left( \hat{\beta}_{AL,j} = 0 \right) \to 1$, for all $\beta_j \in \mathbb{R}$, satisfying the condition $\left| \beta_j \right| < \frac{ c }{ L }$. The set $\mathcal{M} = \mathcal{M}_1$ acts as a benchmark for confidence sets in the sense that if we take a "slightly larger" set, multiplied with the appropriate factor and centered at the adaptive Lasso estimator, we get a confidence region with minimal asymptotic coverage probability equal to 1. But if we base the region on a "slightly larger" set than $\mathcal{M}$, we end up with a confidence set of asymptotic minimal coverage 0. 

\medskip

\begin{remark}[\cite{amann2018uniform}]
In order to provide some further insights of the main pitfalls, we focus on the case where $\lambda^0 \in (0,1]^p$, that is, the case where all components of $\lambda^0$ are nonzero (implying that $\psi = 0$). In other words, this implies that all components are penalized at the same rate, which is obviously fulfilled for uniform tuning. However, in the case of uniform tuning (i.e., parameter penalization with the same rate), then the asymptotic distribution is mere point-mass with no stochastic part surviving in the limit. The reason for this, is the fact that when controlling for the bias of the estimator, the stochastic part vanishes asymptotically. In other words, the appropriate scaling factor is simply not large enough to keep the random component in the limit (i.e., it is asymptotically negligible for large samples). This basically illustrates that the bias is of larger order than the stochastic component when viewed under uniform lens - a fact that is generally inherent to penalized estimators. Moreover, the aspect of omitted variable bias in high dimensional settings is discussed by \cite{wuthrich2023omitted}.  
\end{remark}

\newpage 

\begin{example}
Suppose that the data $\left\{ y_t, t = 1-p,..., n   \right\}$ is generated by the model 
\begin{align}
y_t = \sum_{j = 1}^{ p_n  } \Phi_j y_{t-j} + \epsilon_t, \ \ \ t = 1,..., n   
\end{align}
where $y_t = \big(  y_{t,1},..., y_{t,k}  \big)$, a $(k \times 1)$ vector of variables in the model, where $\epsilon_t$ is a sequence of $\textit{i.i.d}$ error terms with $\mathcal{N}( 0, \Sigma )$ distribution. Furthermore, all roots of
\begin{align*}
\left| I_k - \sum_{j=1}^p \Phi_j z^j \right|
\end{align*}
are assumed to lie outside the unit disc. 

\begin{definition}[Restricted Eigenvalue Condition]
The restricted eigenvalue condition RE(r) is said to satisfied for some $1 \leq r \leq kp$ if
\begin{align}
k^2_{ \psi_n } (r) 
:=
\underset{ R \subset \left\{ 1,..., kp \right\} }{ \mathsf{min} } \ \underset{ \delta \in \mathbb{R}^{ kp \backslash \left\{ 0 \right\} }   }{ \mathsf{min} } \ \frac{ \delta^{\prime} \Phi_n \delta }{ \norm{ \delta_R }^2  }  > 0.
\end{align}
where $R \subset \left\{ 1,..., kp \right\}$ and $|R|$ is its cardinality. 
\end{definition}
Then, one is interested to investigate the properties of the Lasso shrinkage norm when applied to each equation $i \in \left\{ 1,.., k \right\}$ separately. The Lasso estimates are obtained by minimizing the objective function: 
\begin{align}
L ( \beta_i ) = \frac{1}{n} \norm{ y_i - X \beta_i  }^2 + 2 \lambda_n \norm{ \beta_i }_{ \ell_1 }    
\end{align}
Let $J ( \hat{\beta}_i ) = \big\{ j: \hat{\beta}_{i,j} \neq 0 \big\}$ be the indices of the parameters for which the estimator is non-zero. 
\end{example}

\medskip

\begin{example}
Consider the time series vector $\boldsymbol{Y}_t$, which is an $N-$dimensional random vector generated by the VAR model such that $\boldsymbol{Y}_t = A_1 \boldsymbol{Y}_{t-1} + ... + A_p \boldsymbol{Y}_{t-p} + \boldsymbol{u}_t, \ \ \ t \in \left\{ 1,..., T \right\}$. Define the $N (p+1)$ vector $\boldsymbol{X}_t = \big( \boldsymbol{Y}_{t-p}^{\top},...,  \boldsymbol{Y}_{t-1}^{\top}, \boldsymbol{Y}_t^{\top} \big)^{\top}$ and let $\Sigma_x = Var ( \boldsymbol{X}_t ) = \mathbb{E} \left[ \boldsymbol{X} \boldsymbol{X}^{\top} \right]$ and $\Gamma_i = \mathbb{E} \left[ \boldsymbol{Y}_t \boldsymbol{Y}_{t-i}^{\top} \right]$ the autocovariance matrix. An equivalence relation between the matrix coefficients of the two multivariate regression models 
\begin{align}
B_{i, k \ell} = 0 \iff Corr \big( \boldsymbol{Y}_{k,t},  \boldsymbol{Y}_{\ell,t-i} \big| \big\{ \boldsymbol{X}_t \backslash \left\{ \boldsymbol{Y}_{k,t}, \boldsymbol{Y}_{\ell,t-1}  \right\} \big\} \big)
\end{align}
for $i \in \left\{ 1,..., p \right\}$ (see, \cite{poignard2023estimation}). Thus, to derive the partial correlation coefficient between two variables in $\boldsymbol{X}_t$, we rely on the inverse of $\Sigma_x$. Denoting the $(k. \ell)-$th element of $\Sigma_x^{-1}$ by $\sigma_x^{ k \ell }$, the partial correlation coefficient between the $k-$th and the $\ell-$th elements of $\boldsymbol{X}_t$ is then $\rho_x^{ k \ell } = - \frac{ \sigma_x^{ k \ell } }{ \sqrt{ \sigma_x^{k k } \sigma_x^{ \ell \ell } }   }$ for $k \neq \ell$. Also, we denote with $\sigma_u^{ k \ell }$ the $( k, \ell)-$th element of the matrix $\Sigma_u^{-1}$. Then, it can be obtained that
\begin{align}
B_{i, k \ell} = \rho_x^{rs} \left\{ \frac{ \displaystyle \Sigma_{x,ss} ( 1 - \rho^2_{ s \backslash r} ) }{ \displaystyle \Sigma_{x,rr} ( 1 - \rho^2_{ r \backslash s} ) }  \right\}^{1/2}.  
\end{align}
\end{example}

\newpage 

\subsection{A Lasso-based Time Series Regression Model}

\begin{example}
Consider the following Dickey-Fuller regression model 
\begin{align}
\Delta y_t = \rho^{*} y_{t-1} + \sum_{j=1}^p \Delta y_{t-j} + \epsilon_t.
\end{align}
Notice that in the case that the autoregression parameter $\rho^{*} = 0$, then the model is said to have a unit root and is said to be nonstationary. According to \cite{kock2016consistent} it can be shown that: 

\begin{itemize}
\item[(i)] The adaptive Lasso possesses the oracle property in stationary and nonstationary autoregressions (see, \cite{kwiatkowski1992testing}, \cite{muller2008impossibility}, \cite{nielsen2009powerful}). Hence, the ALasso shrinkage approach can distinguish between stationary and nonstationary autoregressions which is extremely important when choosing the right model for forecasting;

\item[(ii)] Show that choosing the tuning parameter by BIC results is consistent model selection; 

\item[(iii)] Analyze the asymptotic behaviour of the probability of classifying $\rho^{*}$ as 0 in the stationary, nonstationary and local to unity setting such that $\rho^{*} = c / T$ (see also \cite{caner2013alternative}).  

\end{itemize}

In the nonstationary setting the problem due to nonuniformity in the asymptotics, can be alleviated if one is willing to tune the adaptive Lasso to perform conservative model selection instead of consistent model selection. We employ the following variant of the adaptive Lasso which is defined as the minimized of 
\begin{align}
\Psi_T \left( \rho, \beta \right) = \sum_{t=1}^T \left( \Delta y_t - \rho y_{t-1} - \sum_{j=1}^p \beta_j \Delta y_{t-j} \right)^2 + \lambda_T w_1^{ \gamma_1 } | \rho | + \lambda_T \sum_{j=1}^p w_1^{ \gamma_2 } \left| \beta_j \right|, 
\end{align}
where $\gamma_1, \gamma_2 > 0$ and $w_1 = 1 / | \hat{ \rho }_I |$ and $w_2 = 1 / | \hat{ \beta }_{I,j} |$ for $\hat{\rho}_I$ and $\hat{\beta}_{I, j}$ denote some initial estimators of the parameters of the model. Furthermore, notice that the objective function is modified compared to the usual adaptive Lasso since it penalized $\rho$, the coefficient on the potentially nonstationary variable variable $y_t$, different from the coefficients on the stationary variables. 
\end{example}
We denote with $\theta = \left( \rho^{*}, \beta  ^{* \prime} \right)^{\prime}$ the set of model parameters. The set of active variables is denoted with 
\begin{align}
\mathcal{A} = \big\{ 1 \leq j \leq p + 1: \theta_j^{*} \neq 0 \big\},
\end{align}
Moreover, denote with $\mathcal{S} = \text{diag} \left( T, \sqrt{T}, ..., \sqrt{T} \right)$ denotes a $\left( p + 1 \times p + 1 \right)$. Denote with $\hat{ \theta } = \left( \hat{\rho}, \hat{\beta}^{\prime}    \right)^{\prime}$ denote the minimizer of the objective function. Furthermore, denoting $\mathcal{M}_0$ the true model and $\hat{\mathcal{M} }$ the estimated model, then we can say that the shrinkage methodology is consistent if for all $\left( \rho^{*}, \hat{\beta}^{*} \right)$ it holds that $\mathbb{P} \left( \hat{\mathcal{M}} = \mathcal{M}_0 \right) \to 1$. Thus, the shrinkage selection methodology is said to be conservative if for all $\left( \rho^{*}, \hat{\beta}^{*} \right)$ it holds that $\mathbb{P} \left( \hat{ \mathcal{M} }_0 \subset \hat{ \mathcal{M} } \right) \to 0$, which implies that the probability of excluding relevant variables tend to zero.

\newpage

\subsubsection{Oracle Property}

We discuss the oracle property of the adaptive Lasso for stationary and nonstationary autoregressions. 
\begin{theorem}
(Consistent model selection). Assume that $\epsilon_t$ is $\textit{i.i.d}$ with $\mathbb{E} \left( \epsilon_1 \right) = 0$ and $\mathbb{E} \left( \epsilon_1^4 \right) < \infty$. 

\begin{enumerate}

\item[\textbf{(A)}] \textbf{(Nonstationary Case)} Then, if $\rho^{*} = 0$, and it holds that $\frac{ \lambda_T }{ T^{ 1 - \gamma_1 } } \to \infty$, $\frac{ \lambda_T }{ T^{ 1/2 - \gamma_2 /2  } } \to \infty$, and $\frac{ \lambda_T }{ T^{1/2} } \to 0$, the following properties hold 
\begin{enumerate}

\item[(i)] Consistency: $\norm{ S_T \left[ \left( \hat{\rho}, \hat{\beta}^{\prime} \right)^{\prime} - \left( 0, \beta^{* \prime}       \right)^{\prime} \right] }_{ \ell_2 } \in \mathcal{O}_p (1)$.

\item[(ii)] Oracle I: $\mathbb{P} \left( \hat{\rho} = 0 \right) \to 1$ and $\mathbb{P} \left( \hat{\beta}_{ \mathcal{A}^c } = 0 \right) \to 1$. 

\item[(iii)] Oracle II: $\sqrt{T} \left( \hat{\beta}_T - \beta_A  \right) \to \mathcal{N} \left( 0, \sigma^2 \left[ \Sigma_{\mathcal{A}} \right]^{-1} \right)$. 
\end{enumerate}

\item[\textbf{(B)}] \textbf{(Stationary Case)} If $y_t$ is stationary such that $\rho^{*} \neq 0$, $\frac{ \lambda_T }{ T^{-1/2 - \gamma_2 / 2} } \to \infty$, and $\frac{ \lambda_T }{ T^{1/2} } \to 0$, 

\begin{enumerate}

\item[(i)] Consistency: $\norm{ S_T \left[ \left( \hat{\rho}, \hat{\beta}^{\prime} \right)^{\prime} - \left( \rho^{*} , \beta^{* \prime}       \right)^{\prime} \right] }_{ \ell_2 } \in \mathcal{O}_p (1)$. 

\item[(ii)] Oracle I: $\mathbb{P} \left( \hat{\rho} = 0 \right) \to 1$ and $\mathbb{P} \left( \hat{\beta}_{ \mathcal{A}^c } = 0 \right) \to 1$. 

\item[(iii)] Oracle II: $\begin{bmatrix}
\sqrt{T} \left( \hat{\beta}_T - \beta_A  \right)
\sqrt{T} \left( \hat{\beta}_T - \beta_A  \right) \\
\end{bmatrix} \to \mathcal{N} \left( 0, \sigma^2 \left[ \Sigma_{\mathcal{A}} \right]^{-1} \right)$. 
\end{enumerate}

\end{enumerate}
\end{theorem}

\medskip

\begin{example}
Consider the framework proposed by \cite{wong2020lasso}. The particular modelling environment corresponds to (i) stationary Gaussian processes with suitably decaying $\alpha-$mixing coefficients, and (ii) stationary processes with sub-Weibull marginals and geometrically decaying $\beta-$mixing coefficients. Moreover, it is well known that guarantees for lasso follow if one can establish the restricted eigenvalue (RE) conditions and provide deviation bounds for the correlation between noise and the regressors. Then, statistical estimation can be achieved as below:
\begin{align}
\Theta^{\star} = \underset{ \Theta \in \mathbb{R}^{p \times q} }{ \mathsf{arg min} } \ \mathbb{E} \big[ \norm{ Y_t - \Theta^{\prime} X_t }_2^2 \big].    
\end{align}
Then, the $\ell_1$penalized least squares estimator $\widehat{\Theta} \in \mathbb{R}^{ p \times q }$ is defined as
\begin{align}
\widehat{\Theta} = \underset{ \Theta \in \mathbb{R}^{p \times q} }{ \mathsf{arg min} } \ \frac{1}{T} \norm{ \mathsf{vec} \left( \boldsymbol{Y} - \boldsymbol{X} \Theta \right) }_2^2 + \lambda_T \norm{ \mathsf{vec} \left( \Theta \right) }_1    
\end{align}
\end{example}

\begin{remark}
The $\beta-$mixing condition has been of interest in statistical learning theory for obtaining finite sample generalization error bounds for empirical risk minimization. Specifically, the usefulness of $\beta-$mixing lies in the fact that by using a simple blocking technique, (e.g., by \cite{yu1994rates}), one can reduce the situation to the $\textit{i.i.d}$ setting. However, there are no results showing that the RE and DB conditions holds under mixing conditions. 
\end{remark}

\newpage

\subsection{Lasso shrinkage with long memory regression errors}

In many problems of practical interest in which the lasso shrinkage is applied, such as for example when modelling and forecasting Realized Volatility measures, it is reasonable to consider the existence of long memory errors. We investigate the asymptotic behaviour of Lasso in regression models with long memory errors. Consider $X_i = \left( x_{i1},...., x_{ip} \right)^{\prime}$ for $i = 1,...,n$ to be the vector of design matrices and $Y_i$ to denote its response variable. Thus we have the model (see, \cite{kaul2014lasso}) 
\begin{align}
Y_i = X_i^{\prime} \beta + \epsilon_i , \ \ \text{for some} \ \ \beta \in \mathbb{R}^p, \ 1 \leq i \leq n. 
\end{align}

The errors $\epsilon_i$ are assumed to be long memory moving average with i.i.d innovations, that is, 
\begin{align}
\epsilon_i = \sum_{ k = 1 }^{ \infty } a_k \zeta_{i - k} = \sum_{k = - \infty }^i a_{i-k} \zeta_k,
\end{align}

where $a_k = c_0 k^{-1 + d}, \forall k \geq 1$, $0 \leq d \leq \frac{1}{2}$ and some constant $c_0 > 0$, and $a_k = 0$ for $k \leq 0$. Moreover, we have that $\zeta_j, j \in \mathbb{Z} := \left\{  0, \pm 1, \pm 2,...  \right\}$, are \textit{i.i.d} RV's with mean zero and variance $\sigma_{ \zeta }^2$. Without loss of generality we assume that $\sigma_{ \zeta }^2 = 1$. We denote with $X = \left( x_{ij} \right)_{ n \times p}$ as the design matrix, and $\epsilon := \left( \epsilon_1,..., \epsilon_n \right)^{\prime}$. Moreover, $\left\{ \epsilon_i, i \in \mathbb{Z}   \right\}$ is a stationary process with autocovariance function 
\begin{align}
\gamma_{\epsilon} (k) = \sum_{j=1}^{ \infty } a_j a_{j+k} = k^{-1 + 2d} B ( d, 1 - 2d) 
\end{align}

Moreover, the Lasso estimate of $\beta$ is defined as follows
\begin{align}
\hat{ \beta }^n( \lambda ) = \underset{ \beta }{ \text{arg min} }  \left\{ \frac{1}{n} \norm{ Y - X^{\prime} \beta }_2^2 + \lambda_n \norm{ \beta }_1 \right\}, \ \ \lambda > 0, 
\end{align}

where $Y = \left( Y_1,..., Y_n \right)^{\prime}$ and $\norm{ \beta }_1 := \sum_{j=1}^p \left| \beta_j  \right|$ denotes $\ell_1$ norm of $\beta = \left( \beta_1,..., \beta_p \right)^{\prime}$. 

Notice that the literature in the area of regularized estimation with dependence considerations is scarce. In this paper, we we investigate the asymptotic behaviour of Lasso under strong dependence structure and less restrictive model assumptions. In particular, we assign a long memory structure on the model errors $\epsilon$, that is, $\sum_{ k = 1}^{ \infty } \left| \gamma_{\epsilon} (k) \right| = \infty$. We provide restrictions on the rate of increase of the design variables as well as the rate of increase of the dimension of $p$ in order to obtain the corresponding finite sample error bounds. Furthermore, we allow the design variables to grow with the restriction $\sum_{1 \leq i \leq n} x_{ij}^2 = \mathcal{O}(n)$, and hence the results obtained can also easily be extended to the case of Gaussian random designs. 

Notice that the framework proposed by \cite{babii2022machine} considers a machine learning application for high-frequency time series panel data. The particular time series filters are constructed based on mixed-frequency data and therefore assumptions on the dependence structure of the errors such as persistence, mixing and long-memory are indeed plausible in such economic and finance studies. 

\newpage 

\subsubsection{Results with finite sample}

In this Section we prove a finite sample oracle inequality for the Lasso solution when the design matrix is non-random. We define with 
\begin{align}
W_{nj} = n^{ - ( 1 / 2 + d )} \sum_{ i = 1}^n x_{ij} \epsilon_i
= n^{ - ( 1 / 2 + d )} \sum_{ i = 1}^n \sum_{ s =  - \infty }^i x_{ij} a_{i - s} \zeta_s  =  \sum_{ s =  - \infty}^n c_{ns,j} \zeta_s,
\end{align} 
where we have that 
\begin{align}
 c_{ns,j} &:= n^{ - ( 1 / 2 + d )} \sum_{ i = 1}^n x_{ij} a_{i - s}, \ \ \ s \in \mathbb{Z}, j = 1,...,p, 
\\
c_{n,j} &:= \underset{ - \infty < s \leq n }{ \text{sup} }  \left|  c_{ns,j} \right|, \ \ \ c_n = \underset{ 1 \leq j \leq p  }{ \text{max} } c_{n,j}.
\end{align}
Moreover, we denote with 
\begin{align}
\sigma_{n,j}^2 := \text{Var} \left( W_{nj} \right) , \ \ \sigma_n^2 = \underset{ 1 \leq j \leq p  }{ \text{max} } \sigma_{n,j}^2. 
\end{align}
Therefore, we shall prove that with an appropriate choice of $\lambda_n$, the Lasso solution obeys the following oracle inequality in the long memory case, for any $n \geq 1$, 
\begin{align}
\frac{1}{n} \norm{ X \left( \hat{\beta} - \beta \right)}_2^2  + \lambda_n \norm{ \hat{\beta} - \beta }_1 \leq \frac{ 4 \lambda_n^2 s_0 }{ \phi_0^2 }
\end{align}

where $\lambda_n = \left( \mathcal{O}(1) \right) \text{log} (p) / n^{1/2 - d}$, under some conditions on the design matrix. Moreover, $s_0$ denotes the cardinality of the set of non-zero components of $\beta$ and $\phi_0$ is a constant depending on the design matrix $X$. Therefore, in order to prove the result we need to obtain a probability bound for the set as below
\begin{align}
\Lambda = \left\{ \underset{ 1 \leq j \leq p }{ \text{max} } \frac{2}{n} \left| \sum_{i=1}^n x_{ij} \epsilon_i \right| \leq \lambda_{0n} \right\},
\end{align}
for a proper choice of $\lambda_{0n}$. Thus, once this probability bound is obtained, the oracle inequality follows by deterministic arguments.

\begin{remark}
The presence of long memory errors can affect the convergence rates of estimators as well as the variable selection procedure. Moreover, the presence of serial correlation can be captures by modeling the error term such that $\varepsilon_t = \rho \varepsilon_{t-1} + u_t$. In particular, there are cases in which serial correlation can manifest as structural breaks in high dimensional models (see,  \cite{kapetanios2018time}). On the other hand, the long memory property is a common feature of mean-revering processes, which implies the significance of sample autocorrelations at large lags. In other words, long memory processes implies that when predicting future values this will depend persistently from the past observations. For instance, $d = 0.5$ is refereed to the fractional or long-memory parameter of the stochastic process.
\end{remark}

\newpage 

\begin{theorem}
\label{TheoremA}
For the long memory regression model suppose that the design variables satisfy the model assumptions. Further, suppose that the tuning parameter $\lambda_n$ is such that $\lambda_n \to \lambda_0 \geq 0$, then we have that 
\begin{align}
\hat{\beta}_n &\overset{ p }{ \to } \mathsf{arg min} \big( Z( \phi) \big)
\\
Z ( \phi ) &= \left( \phi - \beta \right)^{\prime} C \left( \phi - \beta \right) + \lambda_0 \sum_{j = 1}^p \left| \phi_j \right|, \ \ \phi \in \mathbb{R}^p
\end{align}
Thus, if $\lambda_n = o(1)$ then $\mathsf{arg \ min}_{ \phi } \left( Z(\phi)  \right) = \beta$ and $\hat{\beta}_n ( \lambda_n )$ is consistent for the unknown parameter $\beta$.  
\end{theorem}

\begin{proof}
To prove this theorem we consider the following objective function
\begin{align}
Z_n ( \phi ) = \frac{1}{n} \sum_{i = 1}^n \left( Y_i - X_i^{\prime} \phi \right)^2 + \lambda_n \sum_{j = 1}^n \left| \phi_j \right|, 
\end{align}
then $Z_n ( \phi )$ is convex. Therefore, we need to show the pointwise convergence (in probability) of $Z_n ( \phi )$ to $Z ( \phi ) + k^2$ for some constant. Clearly, it holds that
\begin{align}
\lambda_n \sum_{j = 1}^n \left| \phi_j \right| \to \lambda_0 \sum_{j = 1}^n \left| \phi_j \right|.
\end{align}
Consider expanding the expression,
\begin{align*}
\frac{1}{n} \sum_{i = 1}^n  \left( Y_i - X_i^{\prime} \phi \right)^2 
&= \frac{1}{n} \sum_{i = 1}^n \left[ \epsilon_i - X_i^{\prime} \left( \phi - \beta \right) \right]^2 
\\
&= \frac{1}{n} \sum_{i = 1}^n \epsilon_i^2 + \frac{1}{n} \sum_{i = 1}^n \left( \phi - \beta \right)^{\prime} X_i X_i^{\prime} \left( \phi - \beta \right) - \frac{2}{n} \left( \phi - \beta \right)     ^{\prime} \sum_{i = 1}^n X_i \epsilon_i, 
\\
&= \frac{1}{n} \sum_{i = 1}^n \epsilon_i^2 + \frac{1}{n} \sum_{i = 1}^n \left( \phi - \beta \right)^{\prime} X_i X_i^{\prime} \left( \phi - \beta \right) - \frac{2}{n} \left( \phi - \beta \right)     ^{\prime} X^{\prime} \epsilon,
\end{align*} 
Notice that the first term in the above expression converges to $k^2$ by the ergodic theorem, since the error sequence $\left\{ \epsilon_i \right\}$ forms a stationary ergodic sequence. The second term converges to $ \left( \phi - \beta \right)^{\prime} C \left( \phi - \beta \right)$ and the last term converges to zero in probability, since we have that $\frac{1}{ n^{ 0.5 + d } } X^{\prime} \epsilon$ converges in distribution. 
\end{proof}

\medskip

\begin{theorem}
\label{TheoremB}
For the long memory regression model assume that the design variables satisfy the model assumptions. Suppose that $n^{1/2 - d} \lambda_n \to \lambda_0 \geq 0$ as $n \to \infty$, then
\begin{align}
n^{1 / 2 - d} \left( \hat{\beta}^n - \beta \right) &   \to_D \underset{ u }{ \mathsf{arg \ min} } \ V(u),  
\\
V(u) &= - 2u^{ \prime } \mathcal{W} + u^{ \prime }  C u + \lambda_0 \sum_{j=1}^p \left[ u_j \text{sign} \left( \beta_j \right) \mathbf{I}_{ \left[ \beta_j \neq 0 \right]} +  \left| u_j \right| \mathbf{I}_{ \left[ \beta_j \neq 0 \right]} \right]
\end{align}
such that $\mathcal{W}$ is an $\mathcal{N} \left( 0, \Sigma \right)$ random variable. 
\end{theorem}

\newpage

\begin{proof}
We define the following expression 
\begin{align}
V_n(u)  = n^{ 1 - 2d } \left\{ \sum_{i = 1}^n \frac{1}{n} \left[ \left( \epsilon_i - \frac{ X_i^{\prime} u }{ n^{ \frac{1}{2} - d } } \right)^2 - \epsilon_i^2 \right] + \lambda_n \sum_{j = 1}^p \left[ \left| \beta_j + \frac{ u_j }{ n^{ \frac{1}{2} - d } }  \right| \right] - \left| \beta_j \right|  \right\}.
\end{align}
Furthermore, we denote the first term of the above expression by (I) and the second term by (II). Then, we obtain that
\begin{align*}
(I) 
= 
n^{ 1 - 2d } \left\{ \frac{1}{n} \sum_{i = 1}^n \epsilon_i^2 - 2 \frac{ \sum_{i = 1}^n X_i^{\prime} u \epsilon_i }{ n n^{ 1 - 2d } }    + \frac{ u^{\prime} \sum_{i = 1}^n X_i X_i^{\prime} u }{ n n^{ \frac{1}{2} - d } } - \frac{1}{n} \sum_{i = 1}^n \epsilon_i^2 \right\}
&= 
\left\{ \frac{ u^{\prime} \sum_{i = 1}^n X_i X_i^{\prime} u }{ n }      - 2 \frac{ \sum_{i = 1}^n X_i^{\prime} u \epsilon_i }{ n^{ \frac{1}{2}  + d } } \right\}
\\
& \to u^{\prime} C u - 2 u^{\prime} \mathcal{W} , \ \ \text{as} \ n \to \infty, 
\end{align*}
where $\mathcal{W} \sim \mathcal{N} \left( 0, \Sigma \right)$. 
Moreover, we have that 
\begin{align*}
(II) 
= 
n^{ \frac{1}{2} - d } \lambda_n \sum_{j=1}^p \left[ \left| n^{ \frac{1}{2} - d } \beta_j + u_j \right| - n^{ \frac{1}{2} - d } \left| \beta_j \right| \right] 
\to \lambda_0 \sum_{j=1}^p \left[ u_j \text{sign} \left( \beta_j \right) \mathbf{I}_{ \left[ \beta_j \neq 0 \right]} +  \left| u_j \right| \mathbf{I}_{ \left[ \beta_j \neq 0 \right]} \right].  
\end{align*}
\end{proof}

Notice that the notion of weak dependence for stationary time series is measured in terms of covariance functions. The following lemma given by \cite{gupta2012note} is useful. 

\begin{lemma}
For each fixed $n$, let
\begin{align*}
A&:= \left\{ \left| \sum_{i = - \infty }^n Y_{ni} \right| > r \right\}, \ \ \ B_m = \left\{ \left| \sum_{i = - m }^n Y_{ni} \right| > r - \delta \right\}, \ \ r > 0, \delta > 0 , m = 1,2,...
\\
B &= \underset{ m \to \infty }{ \text{lim inf} } \ B_m.
\end{align*}
If $\left| \displaystyle \sum_{i = - \infty }^n  Y_{ni} \right| < \infty$, almost surely, then, for each fixed $n$, $A \subset B$. 
\end{lemma}

For the proof see \cite{gupta2012note}. The argument is to let $\omega \in A$, then it follows that 
\begin{align*}
\left| \displaystyle \sum_{i = - \infty }^n  Y_{ni} \right| > r, \ \ \ \left| \displaystyle \sum_{i = - \infty }^n  Y_{ni} \right| < \infty.
\end{align*}

\medskip

\begin{remark}
Another relevant aspect for the lasso shrinkage approach regardless of the properties of the error terms is the cross-validation consistency (see, \cite{chetverikov2021cross} and \cite{yang2007consistency}). Moreover, although we discuss various cases of the Lasso shrinkage algorithm, those correspond to a different penalization property as well as possibly different-type of thresholding, especially when one considers the pathwise selection approach against the cross-validation approach. 
\end{remark}

\newpage

\subsection{Residual empirical process based on the ALasso}

Moreover, the framework proposed by \cite{chatterjee2015residual} provides regularity conditions and assumptions for deriving the asymptotic behaviour the adaptive Lasso estimator. Specifically, the Alasso estimator of $\beta$ is defined as the minimizer of the weighted $\ell_1-$penalized least squares criterion function
\begin{align}
\widehat{ \boldsymbol{\beta} }_n = \underset{ \mathbf{u} \in \mathbb{R}^p }{ \text{arg min} } \ \sum_{i=1}^n \left( y_i - \mathbf{x}_i^{\prime} \mathbf{u} \right)^2 + \lambda_n \sum_{j=1}^p \frac{ | u_j | }{ \tilde{\beta}_{j,n} }^{\gamma}, 
\end{align}

where, $\lambda_n > 0$ is a regularization parameter, $\gamma > 0$ and$\tilde{\beta}_{j,n}$ is the $j$th component of $\widetilde{ \boldsymbol{\beta} }_n$, a consistent preliminary estimator of $\boldsymbol{\beta}$.

\subsubsection{Main results}

\paragraph{Asymptotic uniform linearity in high dimensions} 

For $t \in [0,1]$, we define
\begin{align}
\widehat{Z}_n(t) 
&= 
\frac{1}{ \sqrt{n } } \sum_{i=1}^n \left[ \mathbf{1} \left( F(e_i) \leq t \right) - t  \right]
\\
Z_n(t)
&= 
\frac{1}{ \sqrt{n } } \sum_{i=1}^n \left[ \mathbf{1} \left( F( \epsilon_i ) \leq t \right) - t  \right]
\end{align}
We set $J(t) = f \left( F^{-1} (t) \right)$. The first result of this section proves the AUL property of the ALASSO based residual empirical distribution function for p. 
\begin{theorem}
\label{theoremABC}
\begin{align}
\underset{ t \in [0,1] }{ \text{sup} } \ \left| \widehat{Z}_n(t)  - \left[ Z_n(t) - \bar{\mathbf{x}}_n^{\prime} \sqrt{n} \left( \widehat{ \boldsymbol{\beta} }_n - \boldsymbol{ \beta } \right) J(t) \right] \right| = o_p(1).
\end{align}
\end{theorem}
Notice that Theorem \ref{theoremABC} shows that the AUL property holds for the empirical distribution function of the residuals based on the ALASSO fit even for $p >>n$, provided the regularity conditions hold and the regression model has enough sparsity.

\paragraph{Functional oracle property of the ALASSO}

The AUL property has various important applications in the context of statistical inference in high dimensional regression. For instance, it allows to establish the asymptotic distribution of the residual edf which, in turn, can be used to carry the goodness of fit tests on $F$ and set confidence bands for $F$.
Under standard weak convergence theory we have that 
\begin{align}
\mathcal{Z}_n^{OR} = n^{1 / 2} \left( \widehat{F}_n^{OR} - F \right) \to_w \mathcal{Z}_{\infty} \ \ \text{on} \ \ L^{\infty} \left( \left[ - \infty, + \infty \right] \right)
\end{align}
where $\widehat{F}_n^{OR}$ is the edf of the residuals based in OLS under the oracle property.

\newpage 

Moreover, $\mathcal{Z}_{\infty}$ is a zero mean Gaussian process with covariance function as below
\begin{align}
\tau( x, y ) = \text{Cov} \left( \mathcal{Z}_{\infty}(x), \mathcal{Z}_{\infty}(y) \right), \ \ \ x, y \in \left( - \infty, + \infty \right)
\end{align}
We also define the ALASSO based residual empirical process as below
\begin{align}
\mathcal{Z}_n (x) = \sqrt{x} \left[ \widehat{F}_n(x) - F(x) \right], \ \ x \in \left( - \infty, + \infty \right).
\end{align}

\begin{proof}
We denote with
\begin{align}
\Gamma_n \left(t , \mathbf{u} \right) = \frac{1}{ \sqrt{n} } \sum_{ i = 1 }^n \left[ \mathbf{1} \left( F \left( \epsilon_i - n^{1 / 2} \mathbf{x}_i^{\prime} \mathbf{u} \right) \leq t \right) \right] , \ \ t \in [0,1], \ \text{and} \ \mathbf{u} \in \mathbb{R}^p.
\end{align}
Next, we define the random vector $\mathbf{W}_n = n^{- 1/ 2} \sum_{i = 1}^n \mathbf{x}_i^{\prime} \epsilon_i$ and write $\mathbf{W}_n = \left(   \mathbf{W}_n^{ (1) \prime }, \mathbf{W}_n^{ (2) \prime } \right)$. 
 
\end{proof}

\subsubsection{Sparsity in high dimensional estimation problems}

The Lasso estimator appears to be suitable for high dimensional inference problems such as the case when $p$ is much larger than the time series observations $n$. Some useful results are as below:

\paragraph{Regression in the iid case} Under the usual assumption that the $\epsilon_i$ are iid and subGaussian, 
\begin{align}
\forall \ s , \mathbb{E} \left[ \text{exp} \left( s \epsilon_i^2\right) \right] \leq \text{exp} \left( \frac{s^2 \sigma^2 }{ 2} \right)
\end{align}
for some known $\sigma^2$, then we have that 
\begin{align}
\mathbb{P} \left( \left| \frac{1}{n} \sum_{ i = 1}^n W_i^{(j)}     \right| \geq \frac{t}{ \sqrt{n} } \right) \leq \psi(t) = \text{exp} \left( - \frac{t^2}{2 \sigma^2 } \right).
\end{align}

\subsection{Lasso shrinkage for Predictive Regression}

In this model, the unknown coefficients $\beta_n^{*}$ can be obtained from the data by running OLS 
\begin{align}
\widehat{ \beta }^{OLS} = \underset{ \beta }{ \text{arg min} } \norm{ y - X \beta }^2,
\end{align}
The asymptotic behaviour of the OLS estimator has been extensively examined in the time series econometrics literature (see, \cite{koo2020high}, \cite{lin2020robust}, \cite{yousuf2021boosting}, \cite{lee2022lasso}). A relevant framework for lasso time series regression is proposed by \cite{adamek2023lasso}.

\newpage

In particular, the following functional central limit theorem holds
\begin{align}
\frac{1}{ \sqrt{n} } \sum_{j=1}^{ \floor{nr} } 
\begin{pmatrix}
e_{j.}^{\prime} \\
u_j
\end{pmatrix}
\Rightarrow
\begin{pmatrix}
B_e(r) \\
B_u(r) 
\end{pmatrix}
\equiv
\text{BM} \left( \Sigma \right). 
\end{align}
To represent the asymptotic distribution of the OLS estimator, define $u_i^{+} = u_i - \Sigma_{eu}^{\prime} \Sigma_{ee} e_i^{\prime}$. By definition, Cov$\left( e_{ij}, u_i^{+} \right) = 0$ for all $j$ so that 
\begin{align}
\frac{ X^{\prime} u }{n} \Rightarrow \zeta:= \int_0^1 B_e(r) dB_{u^{+}} (r) + \int_0^1 B_e(r) \Sigma^{\prime}_{eu} \Sigma_{ee}^{-1} dB_e(r)^{\prime}
\end{align} 
which is the sum of a mixed normal random vector and a non-standard random vector. The OLS limit distribution is then given by 
\begin{align}
n \left( \hat{\beta}^{OLS} - \beta_n^{*} \right) = \left( \frac{X^{\prime} X }{ n^2 } \right)^{-1} \frac{X^{\prime} u }{ n } \Rightarrow \Omega^{-1} \zeta, 
\end{align}
where $\Omega: = \displaystyle \int_0^1 B_e(r) B_e (r)^{\prime} dr$. This implies that when we inflate $\widehat{ \beta }_j^{OLS}$ by the factor $n^{ \delta_j }$ so that its magnitude is comparable to the constant $\beta_j^{0*}$, we attain consistency in that  
\begin{align}
n^{ \delta_j } \left( \widehat{\beta}^{OLS} - \beta_{jn}^{*} \right) = n^{ \delta_j } \widehat{\beta}^{OLS} - \beta_j^{0*} = \mathcal{O}_p \left( n^{ \delta_j - 1} \right) = o_p(1) \ \ \text{for all} \ j \leq p.  
\end{align}
Moreover, if $\beta_j^{0*} \neq 0$, when $\delta_j$ is close to 1, the signal of $x_j$ is weak, and therefore the rate of convergence is slow. However, notice that some true coefficients $\beta_j^{*}$ could be exactly zero, where the associated predictors would be redundant (inactive) in the regression. Let $M^{*} = \left\{ j : \beta_j^{0*} \neq 0 \right\}$ be the index set of regressors relevant to the regression, $p^{*} = \left| M^{*} \right|$, and $M^{* c} = \left\{ 1,..., p \right\} M^{*}$ be the set of redundant regressors.  For simplicity, we refer to $M^{*}$ as the active set, meaning that it plays an active role in the regression, and we call $M^{* c}$ the inactive set. Define $\widehat{\beta}^{ora} = \left( \widehat{\beta}^{ \text{ora} \prime}_{ M^{*} }, \widehat{\beta}^{ \text{ora} \prime}_{ M^{*}} \right)$, where
\begin{align}
\widehat{\beta}^{ \text{ora} }_{ M^{*} } = \underset{ \beta }{ \text{arg min} } \norm{ y - \sum_{ j \in M^{*} } x_j \beta_j }^2
\end{align}
and the complement of the oracle estimator is just the remaining coefficients which are identical equal, that is,  $\widehat{\beta}^{ \text{ora} }_{ M^{*} c } = 0$. Notice that the oracle estimator is the infeasible oracle information of $M^{*}$. The asymptotic distribution of the estimator is given by 
\begin{align}
n \left( \widehat{\beta}^{ \text{ora} } - \beta_{n}^{*} \right)_{ M^{*} } \Rightarrow \Omega_{ M^{*} } \zeta_{ M^{*} }, 
\end{align}
where $\Omega_{ M^{*} }$ is the $p^{*} \times p^{*}$ submatrix $\left( \Omega_{j,j^{\prime} } \right)_{ j,j^{\prime} \in M^{*} }$ and $\zeta_{ M^{*} }$ is the $p^{*} \times 1$ subvector $\left( \zeta_j \right)_{ j \in M^{*} }$.

\newpage

\subsection{Lasso estimation with a structural break}
\label{lassobreaks}

Relevant studies with frameworks for lasso estimation robust to the presence of structural breaks in the covariates include  \cite{chan2014group}, \cite{cheng2016shrinkage}, \cite{qian2016shrinkage} and \cite{smith2019variable}. For any $n-$dimensional vector $W = \left( W_1,..., W_n   \right)^{\prime}$, define the empirical norm as $\norm{ W }_n := \left( \frac{1}{n} \sum_{i=1}^n W_i^2 \right)^{1/2}$. 
\begin{align}
f_{ ( \alpha, \uptau ) } (x, q) &:= x^{\prime} \beta + x^{\prime} \delta \mathbf{1} \big\{ q < \uptau \big\}
\\
f_{0} (x, q) &:= x^{\prime} \beta_0 + x^{\prime} \delta_0 \mathbf{1} \big\{ q < \uptau_0 \big\}
\\
\hat{f}_{ ( \alpha, \uptau ) } (x, q) &:= x^{\prime} \hat{\beta} + x^{\prime} \hat{\delta} \mathbf{1} \big\{ q < \hat{\uptau} \big\}
\end{align}    
Then, we define the prediction risk as below
\begin{align}
\norm{ \hat{f} - f_0 }_n := \left[ \frac{1}{n} \sum_{i=1}^n \bigg(  \hat{f} ( X_i, Q_i ) - f_0 ( X_i, Q_i ) \bigg)^2  \right]^{1/2}.
\end{align}
Define with $\alpha_0 = \left( \beta_0^{\prime}, \delta_0^{\prime}  \right)^{\prime}$. Then, we can rewrite the model as below $Y_i = \boldsymbol{X}_i ( \uptau_0 )^{\prime} \alpha_0 + U_i$, for $i = 1,...,n$. Let $\boldsymbol{y} \equiv \left( Y_1,..., Y_n \right)^{\prime}$. Then, for any fixed $\uptau \in \mathcal{T}$ we consider the residual sum of squares given by 
\begin{align}
S_n ( \alpha, \uptau ) &= \frac{1}{n} \sum_{i=1}^n \bigg( Y_i - X_i^{\prime} \beta - X_i^{\prime} \mathbf{1} \left\{ Q_i < \uptau  \right\} \bigg)^2
\equiv 
\norm{ \boldsymbol{y} - \boldsymbol{X} ( \uptau ) \alpha }_n^2. 
\end{align} 
\begin{remark}
The examples discussed in the previous sections are concerned with the use of the shrinkage approach under the assumption of sparsity in high dimensional time series regression model for estimation purposes (see, also \cite{basu2015regularized} and \cite{garcia2022variable}). Consider that we have $p-$variate time series then we can quantify the stability of the process based on the largest eigenvalue of the transition matrix. Since the largest eigenvalue is bounded, then a condition to roll out the cross-sectional and serial correlation on the data is ensured. The literature on high-dimensional time series modelling allows for dominant cross-sectional autocorrelations to be accounted for, using common factors. In particular, \cite{cho2023high} propose the factor-adjusted VAR model, in which case the time series is decomposed into two latent components. A relevant framework for robust estimation of high-dimensional VAR Models is given by \cite{wang2023rate}. The frameworks mentioned in Section \ref{lassobreaks} are suitable for statistical estimation in high dimensional regression models under the presence of structural breaks\footnote{Dr. Haeran Cho (University of Bristol) gave a seminar titled: "High-dimensional time series segmentation via factor-adjusted VAR modelling", at the S3RI Departmental Seminar Series, University of Southampton on the 24th of March 2022. \\}. One may be interested to develop an econometric framework for structural break detection within these data-dependent model specifications such as as the piecewise stationary factor-adjusted VAR models (see, \cite{breitung2011testing}). An alternative algorithmic procedure\footnote{Professor Degui Li (University of York) gave a seminar titled: "Detection of Multiple Structural Breaks in Large Covariance Matrices", at the S3RI Departmental Seminar Series, University of Southampton on the 2nd of December 2021.} for structural change detection in high dimensional time series models is proposed by \cite{li2021group}. 
\end{remark}

\newpage

\subsection{Classical Shrinkage Type Estimation Approach}

In this section, we present the Score function approach proposed in the framework of \cite{sen1987preliminary}, although the particular methodology doesn't correspond to a high-dimensional environment, it has many possible applications such as a "Shrinkage type Structural Break Estimation in High-Dimensional Predictive Regressions" or "Testing for Predictability in High Dimensional Environments via a Shrinkage type Structural Break Estimator".

Consider the score function given by $\psi = \left\{ \psi(x), x \in \mathbb{R}  \right\}$ required for the definition of M-estimators. Moreover, suppose that $\psi (x) = \big\{ \psi_1 (x) + \psi_2(x) \big\}$, where $\psi_1(.)$ and $\psi_2(.)$ are both non-decreasing and skew-symmetric such that, $\psi_j (x) + \psi_j (-x) = 0, \forall \ x \in \mathbb{R}, j = 1,2$. Denote with $\boldsymbol{A}_T = \left( \boldsymbol{a}_1,.., \boldsymbol{a}_T \right)$ and for every $\boldsymbol{b} \in \mathbb{R}^p$ we define the following $p-$dimensional functional with $T$ elements in each dimension
\begin{align}
\boldsymbol{M}_T \left( \boldsymbol{b} \right) = \big( \boldsymbol{M}_{T1} \left( \boldsymbol{b} \right),..., \boldsymbol{M}_{Tp} \left( \boldsymbol{b} \right) \big)^{\prime} = \sum_{t=1}^T \boldsymbol{a}_t \psi \big( \boldsymbol{y}_t - \boldsymbol{a}_t^{\prime} \boldsymbol{b} \big), \ \ \boldsymbol{b} \in \mathbb{R}^p.
\end{align}
Then, under the assumption of stationarity then, the RME for $\hat{\boldsymbol{\beta}}_{1T}$ of the model parameter $\boldsymbol{\beta}_{1}$ is a solution to the following equation
\begin{align}
\boldsymbol{M}_{T(1)} \big( \boldsymbol{b}_1,  \boldsymbol{0} \big) = \boldsymbol{0}.
\end{align}
We introduce a suitable $M-$test statistic for testing the null hypothesis $\mathbb{H}_0: \boldsymbol{\beta}_2 = \boldsymbol{0}$,
\begin{align}
\hat{ \boldsymbol{M} }_{T(2)} = \boldsymbol{M}_{T(2)} \big( \hat{\boldsymbol{\beta}}_1, \boldsymbol{0} \big),
\end{align} 
where $\hat{\boldsymbol{\beta}}_{1T}$ is the restrained  M-estimator (RME) of $\boldsymbol{\beta}_1$. Notice that under stationarity, it can be proved that the restrained  M-estimator (RME) generally performs better than the unrestrained M-estimator (URE) especially when  $\boldsymbol{\beta}_2$ is close to $\boldsymbol{0}$ (see, \cite{sen1987preliminary}). Furthermore, denote with 
\begin{align}
\label{St}
\mathcal{S}_T^2 &= T^{-1} \sum_{t=1}^T \psi^2 \bigg( \boldsymbol{y}_t - \hat{\boldsymbol{\beta}}_{1T}^{\prime} \boldsymbol{x}_{(j)t-1} \bigg), \ \ \ \ \boldsymbol{x}_{ t-1 } = \left( \boldsymbol{x}_{(1)t-1}, \boldsymbol{x}_{(2)t-1} \right),
\\
\boldsymbol{C}_{ii.j} &=  \boldsymbol{C}_{ii} - \boldsymbol{C}_{ij}  \boldsymbol{C}_{jj}^{-1} \boldsymbol{C}_{ji}, \ \ \text{for} \ \ i \neq j = \left\{ 1,2 \right\}.
\end{align}
Therefore, an appropriate aligned $M-$test is given by 
\begin{align}
\mathcal{L}_T = S_T^{-2} \bigg\{ \hat{ \boldsymbol{M} }_{T(2)} \boldsymbol{C}_{22.1} \hat{ \boldsymbol{M} }_{T(2)} \bigg\}. 
\end{align}
Under the null hypothesis, $\mathbb{H}_0$, $\mathcal{L}_T$ has asymptotically  the chi-square distribution with $p_2$ degrees of freedom. Therefore, based on the significance level $\alpha$ such that $0 < \alpha < 1$, the preliminary test under the null, $\mathbb{H}_0$, is rejected based on the critical values. 

\newpage 

Considering for example one is interested to the conditional mean specification of a predictive regression model, focusing on the OLS against the IVX estimation approaches based on the full set of regressors in the model. The, the corresponding expressions are:  
\begin{align}
\mathcal{S}_T^2 = T^{-1} \sum_{t=1}^T \psi^2 \bigg( \boldsymbol{y}_t - \hat{\boldsymbol{\beta}}_{1T,ols}^{\prime} \boldsymbol{x}_{(j)t-1} \bigg), \ \ \ \ \boldsymbol{x}_{ t-1 } = \left( \boldsymbol{x}_{(1)t-1}, \boldsymbol{x}_{(2)t-1} \right), \ j \in \left\{ 1,2 \right\} ,
\end{align}
and 
\begin{align}
\widetilde{\mathcal{S}}_T^2 = T^{-1} \sum_{t=1}^T \psi^2 \bigg( \boldsymbol{y}_t - \hat{\boldsymbol{\beta}}_{1T,ivx}^{\prime} \boldsymbol{x}_{(j)t-1} \bigg), \ \ \ \ \boldsymbol{x}_{ t-1 } = \left( \boldsymbol{x}_{(1)t-1}, \boldsymbol{x}_{(2)t-1} \right), \ j \in \left\{ 1,2 \right\}. 
\end{align}
In other words, the authors of the particular framework focus on preliminary test M-estimation (PTME) formulation, thus the $\hat{\boldsymbol{\beta}}_{1n}^{PT}$ is chosen as the RME or UME, according as this preliminary test leads to the acceptance or rejection of the null hypothesis $\mathbb{H}_0$. Moreover, the shrinkage $M-$estimator (SME), based on the usual James-Stein rule incorporates the same test statistic in a smoother manner. When $\boldsymbol{\beta}_2$ is very close to $\boldsymbol{0}$, generally both the PTME and SME perform better than the UME, but the RME may still be better than either of them. On the other hand, for $\boldsymbol{\beta}_2$ away from $\boldsymbol{0}$, the RME may perform rather poorly, while both the PTME and SME are robust. This relative picture on the performance characteristics of all four versions of M-estimators can be best studies within an asymptotic framework. One can consider the notion of asymptotic distributional risk (ADR) as well as the asymptotic risk efficiency (ARE) for the various versions of the M-estimators (see, \cite{sen1987preliminary}).

\subsubsection{Limiting Distributional Risk}

In the multivariate location model \cite{sen1985some} have pointed out that shrinkage estimation works out well only in a shrinking neighbourhood of the pivot. Since for $\boldsymbol{\beta}_2$ the pivot is taken as $\boldsymbol{0}$, we consider a shrinking neighbourhood of $\boldsymbol{0}$ and we consider a sequence $\left\{ \mathcal{K}_T \right\}$ of alternatives,
\begin{align}
\mathcal{K}_T : \boldsymbol{\beta}_2 = \boldsymbol{\beta}_{2(T)} = T^{- 1/ 2} \boldsymbol{\xi}, \ \ \  \boldsymbol{\xi} = \big( \xi_{p_1 + 1},..., \xi_{p} \big)^{\prime} \in \mathbb{R}^{p_2},    
\end{align}  
which implies that the null hypothesis $\mathbb{H}_0$ reduces to $\mathbb{H}_0: \boldsymbol{\xi} = \boldsymbol{0}$. Using a suitable estimator $\boldsymbol{\beta}^{\star}_{1T}$ of $\boldsymbol{\beta}_{1}$, we denote by $\boldsymbol{G}^{\star} ( \boldsymbol{\mathsf{u}} ) = \underset{ T \to \infty }{ \mathsf{lim} } \ \mathbb{P} \bigg(  T^{- 1/ 2} \big( \boldsymbol{\beta}^{\star}_{1T} - \boldsymbol{\beta}_{1} \big) \leq \boldsymbol{\mathsf{u}} \big| \mathcal{K}_T \bigg), \ \ \ \boldsymbol{\mathsf{u}} \in \mathbb{R}^{p_1}$, where we assume that $\boldsymbol{\mathsf{u}}$ is nondegenerate. Then, with a quadratic error loss such that $T \big( \boldsymbol{\beta}^{\star}_{1T} - \boldsymbol{\beta}_{1} \big)^{\prime} \boldsymbol{W} \big( \boldsymbol{\beta}^{\star}_{1T} - \boldsymbol{\beta}_{1} \big)$, for a suitable probability distribution function $\boldsymbol{W}$, the ADR of $\boldsymbol{\beta}^{\star}_{1T}$ is defined as below
\begin{align}
\mathcal{R} \big( \boldsymbol{\beta}^{\star}_{1T} ; \boldsymbol{W}  \big) = \mathsf{trace} \left\{ \boldsymbol{W} \int_{ \mathbb{R}^{p} }  ... \int_{ \mathbb{R}^{p} } \boldsymbol{\mathsf{u}} \boldsymbol{\mathsf{u}}^{\prime} d \boldsymbol{G}^{\star} ( \boldsymbol{\mathsf{u}} ) \right\} \equiv \mathsf{trace} \left\{ \boldsymbol{W} \boldsymbol{V}^{\star} \right\},
\end{align} 
where $\boldsymbol{V}^{\star}$ is the dispersion matrix for the asymptotic distribution of $\boldsymbol{G}^{\star} ( . )$.

\newpage

\subsubsection{Asymptotic representation of estimators}

We follow the framework presented in \cite{mukherjee1998preliminary}. 

\paragraph{Notation} The $i-$th row and the $j-$th column of any matrix $\boldsymbol{S}$ are denoted by $s_{i \bullet}$ and $s_{ \bullet j}$, respectively. Then if $\boldsymbol{S}$ has $p$ columns, $\boldsymbol{D}_s$ denotes the maximal diagonal $\big( \norm{s_{ \bullet 1} }, ..., \norm{s_{ \bullet p} } \big)$, where $\norm{s_{ \bullet j} }$ is the usual Euclidean norm of the $j-$th column of $\boldsymbol{S}$. Moreover, the usual partitions, $\boldsymbol{S}_{ii.j}$ stands for the matrix $\boldsymbol{S}_{ii} - \boldsymbol{S}_{ij} \boldsymbol{S}_{jj} \boldsymbol{S}_{ji}$. Then, we define the following matrices which appear as the scaling factors in the asymptotic distributions of the different estimators. We denote with $\boldsymbol{A}_x = n^{-d} \boldsymbol{D}_x$, $\boldsymbol{B}_x = n^{d} \boldsymbol{D}_x$, $\boldsymbol{A}_{x1} = n^{-d} \boldsymbol{D}_{x1}$. Moreover, we have that $\underset{ n \to \infty }{ \mathsf{lim} } n^{-1} \boldsymbol{X}^{\prime} \boldsymbol{X} \ \ \text{exists and equals} \ \boldsymbol{C}$, with $\boldsymbol{C}$ being positive definite. Consequently,  
\begin{align}
\boldsymbol{R}_n := \boldsymbol{D}_x^{-1} \boldsymbol{X}^{\prime} \boldsymbol{X} \boldsymbol{D}_x^{-1} \overset{ p }{ \to } \boldsymbol{R} := \boldsymbol{D}_c^{-1} \boldsymbol{C} \boldsymbol{D}_c^{-1} 
\end{align}
Denote with $\boldsymbol{\Sigma}_n$ to denote the dispersion matrix of $\boldsymbol{\epsilon}$. Then, the dispersion of $\boldsymbol{B}_x^{-1} \boldsymbol{X}^{\prime} \boldsymbol{\Sigma}_n \boldsymbol{X}  \boldsymbol{B}_x^{-1}$  and it follows that the limit of $\boldsymbol{E}_n$ exists and equals $\boldsymbol{E}$, say. As in \cite{sen1987preliminary}, the shrinkage least-squares estimator (SLSE) $\hat{\beta}_s$ can be motivated as a smoothed version of a preliminary test estimator. Thus, for each $c > 0$, this is defined as below   
\begin{align}
\hat{\beta}_s = \hat{\beta}_1 + \left( 1 - \frac{ \kappa }{  \mathcal{L}_T  } \right) \left( \tilde{\beta}_1 - \hat{\beta}_1 \right),  c > 0.
\end{align}
For large values of $\mathcal{L}_T$, $\hat{\beta}_s$ yields to $\tilde{\beta}_1$, whereas it performs quite differently from $\hat{\beta}_{pt}$ form small values of $\mathcal{L}_T$. Define the dispersion matrix
\begin{align}
\boldsymbol{H} := \big[ \boldsymbol{0} \ \vdots \ \boldsymbol{R}_{22.1} \big] \boldsymbol{R}^{-1} \boldsymbol{E} \boldsymbol{R}^{-1} \big[ \boldsymbol{0} \ \vdots \ \boldsymbol{R}_{22.1} \big] 
\end{align}
Then, it can be proved that the second term converges to  $\boldsymbol{R}_{22.1} \boldsymbol{D}_{c2} \boldsymbol{\eta}$. Therefore, we define with 
\begin{align}
\boldsymbol{H}_n := \big[ \boldsymbol{0} \ \vdots \ \boldsymbol{R}_{n22.1} \big] \boldsymbol{R}_n^{-1} \boldsymbol{E}_n \boldsymbol{R}_n^{-1} \big[ \boldsymbol{0} \ \vdots \ \boldsymbol{R}_{n22.1} \big] 
\end{align}
implying that $\boldsymbol{H}_n$ converging to $\boldsymbol{H}$ and $\boldsymbol{\xi} := \boldsymbol{D}_{c2} \boldsymbol{\eta} \in \mathbb{R}^{p_2}$. Then, $\mathcal{L}_n$ is defined as
\begin{align}
\mathcal{L}_n := \norm{ \boldsymbol{H}_n^{- 1/ 2} \boldsymbol{B}_{x2}^{-1} \boldsymbol{X}_2^{\prime} \big( \boldsymbol{Y} - \boldsymbol{X}_1 \hat{ \boldsymbol{\beta} }_1 \big) }^2 
\end{align}  
which converges to a noncentral chi-square random variable with $p_2$ degrees of freedom and noncentrality parameter $\boldsymbol{\delta} := \boldsymbol{\xi}^{\prime} \boldsymbol{L}_{22}^{-1} \boldsymbol{\xi}$. Therefore, under the null hypothesis $\boldsymbol{\eta} = \boldsymbol{0}$ or equivalently, $\boldsymbol{\xi} = \boldsymbol{0}$, the test statistic $\mathcal{L}_n$ converges to a chi-square random variable with $p_2$ degrees of freedom. Define the projection matrix $\boldsymbol{P}_{x1} = \boldsymbol{X}_1 \left( \boldsymbol{X}_1^{\prime} \boldsymbol{X}_1  \right) \boldsymbol{X}_1^{\prime}$. Let $\boldsymbol{D}_n := \boldsymbol{R}_n^{-1}$ and $\boldsymbol{D} := \boldsymbol{R}^{-1}$. Then, the covariance matrices of the random vectors $\boldsymbol{D}_n \boldsymbol{U}_n$ and $\boldsymbol{D} \boldsymbol{U}$ are denoted by $\boldsymbol{L}_n$ and $\boldsymbol{L}$ respectively.

\newpage

\subsection{Lasso Inference for High Dimensional Time Series under NED} 

Following the framework proposed by \cite{adamek2023lasso}, consider the linear time series regression 
\begin{align}
y_t = \boldsymbol{x}_t^{\prime} \boldsymbol{\beta}^0 + u_t, \ \ \ \ t = 1,..., T,    
\end{align}
where $\boldsymbol{x}_t = \left( x_{1,t},..., x_{N,t} \right)^{\prime}$ is an $(N \times 1)$ vector of explanatory variables and we assume that we have a high-dimensional time series model where $N$ can be larger than $T$.  

\begin{assumption}
Let $\boldsymbol{z}_t = \left( \boldsymbol{x}_t^{\prime}, u_t \right)^{\prime}$ and let there exists some constants $\bar{m} > m > 2$ and $d \geq \mathsf{max} \left\{ 1, \frac{ \frac{\bar{m} }{m-1}   }{  \bar{m} - 2 }  \right\}$,  
\begin{align}
\mathbb{E} \big[ \boldsymbol{z}_t \big] = \boldsymbol{0}, \ \ \ \mathbb{E} \big[ \boldsymbol{x}_t u_t \big] = \boldsymbol{0} \ \ \ \text{and} \ \ \ \underset{ 1 \leq j \leq N + 1, 1 \leq t \leq T }{ \mathsf{max} } \ \mathbb{E} | z_{j,t}  |^{ 2 \bar{m} } \leq C.     
\end{align}
\end{assumption}

\begin{assumption}
Let $\boldsymbol{s}_{T,t}$ denote a $k(T)-$dimensional triangular array that is $\alpha-$mixing of size 
\begin{align}
\frac{ d }{ \left( \frac{1}{m} - \frac{1}{ \bar{m} } \right) }, \ \ \ \text{with} \ \ \sigma-\text{field} \ \mathcal{F}_t^s := \sigma \left\{ \boldsymbol{s}_{T,t}, \boldsymbol{s}_{T,t-1},... \right\}
\end{align}
such that $\boldsymbol{z}_t$ is $\mathcal{F}_t^s-$measurable. The process $\left\{ z_{j,t} \right\}$ is $L_{2m}-$near-epoch dependent (NED) of size $-d$ on $\boldsymbol{s}_{T,t}$ with positive bounded NED constants, uniformly over $j \in \left\{ 1,..., N + 1 \right\}$. 
\end{assumption}

\begin{remark}
The first assumption above ensures that the error terms are contemporaneously uncorrelated with each of the regressors, and that the process has finite and constant unconditional moments. The second assumption above implies that $\boldsymbol{s}_{T,t}$ is an underlying shock process driving the regressors and errors in $\boldsymbol{z}_t$, where we assume $\boldsymbol{z}_t$ to depend almost nearly entirely on the "near epoch" of $\boldsymbol{s}_{T,t}$. In other words, the near epoch dependence of the vector $\boldsymbol{z}_t$ can be interpreted as $\boldsymbol{z}_t$ being "approximately" mixing, in the sense that it can be well-approximated by a mixing process. The NED framework allows for general forms of dependence and that are often encountered in econometrics applications including, but not limited to, strong mixing processes, linear processes including ARMA models, various types of stochastic volatility and GARCH models as well as nonlinear processes (see, \cite{davidson2002establishing}). 
\end{remark}

The framework proposed by \cite{adamek2023lasso} is the first to utilize the NED framework for establishing uniformly valid high-dimensional inference. For instance, \cite{wong2020lasso} consider time series models with $\beta-$mixing errors, which has the advantage of allowing for general forms of dynamic misspecification resulting in serially correlated error terms, but, rules out several relevant data generating processes, and is in addition typically difficult to verify. On the other hand, \cite{masini2022regularized} use an \textit{m.d.s} assumption on the innovations in combination with sub-Weibull tails and a mixing assumption on the conditional covariance matrix. However, the particular \textit{m.d.s} assumption does not allow for dynamic misspecification of the full model. Importantly, the NED assumption on $u_t$ does allow for misspecified models as well, in which case we view $\boldsymbol{\beta}^0$ as the coefficients of the pseudo-true model when restricting the class of models to those linear in $\boldsymbol{x}_t$.

\newpage

Therefore, this econometric intuition allows to view the high-dimensional time series model of \cite{adamek2023lasso} as simply the linear projection of $y_t$ on all the variables in $\boldsymbol{x}_t$, with $\boldsymbol{\beta}^0$ in that case representing the best linear projection coefficients, which implies that $\mathbb{E} [ u_t ] = 0$ and $\mathbb{E} [ u_t x_{j,t} ] = 0$. On the other hand, $u_t$ is not likely to be an $\textit{m.d.s}$. Thus, allowing for misspecified dynamics is crucial for developing the theory for the nodewise regressions underlying the desparsified lasso. The NED-order $m$ and sequence of size $-d$ play a key role in the development of the limit results in which these parameters enter as the corresponding asymptotic rates of related quantities. Obviously, there is a trade-off between the thickness of the tails allowed for and the moment of dependence - which is measures via the mixing rate.  

\begin{example}[ARCL Model with GARCH errors, see \cite{adamek2023lasso}]
Consider the autoregressive distributed lag (ARDL) model with GARCH errors such that
\begin{align}
y_t &= \sum_{i=1}^p \rho_i y_{t-i} + \sum_{j=0}^q \boldsymbol{\theta}_j^{\prime} \boldsymbol{w}_{t-j} + u_t =  \boldsymbol{x}_t^{\prime} \boldsymbol{\beta}^0 + u_t  
\\
u_t &= \sqrt{h}_t \varepsilon_t, \ \ \ \ \varepsilon_t \sim \mathcal{N} (0,1), 
\\
h_t &= \beta_0 + \beta_1 h_{t-1} + \beta_2 u_{t-1}^2, 
\end{align}
where the roots of the lag polynomial $\rho(z) = 1 - \sum{j=1}^p \rho_j z^j$ are outside the unit circle. Then, $u_t$ is a strictly stationary geometrically $\beta-$mixing process. Moreover, we assume that the vector of exogenous variables $\boldsymbol{w}_t$ is stationary and geometrically $\beta-$mixing as well with $2 \bar{m}$ finite moments.    

Given the invertability of the lag polynomial, we may then write 
\begin{align}
y_t = \rho^{-1} (L) v_t, \ \ \ \ v_t = \sum_{j=0}^q \boldsymbol{\theta}_j^{\prime} \boldsymbol{w}_{t-j} + u_t,     
\end{align}
and the inverse lag polynomial $\rho^{-1} (z)$ has geometrically decaying coefficients. Then, it follows directly that $y_t$ is NED on $v_t$, where $v_t$ is strong mixing of size $- \infty$ as its components are geometrically $\beta-$mixing, and the sum inherits the mixing properties.
\end{example}

\begin{example}[Equation-by-Equation VAR, see  \cite{adamek2023lasso}] Consider the VAR model below
\begin{align}
\boldsymbol{y}_t = \sum_{j=1}^p \boldsymbol{\Phi}_j \boldsymbol{y}_{t-j} + \boldsymbol{u}_t,     
\end{align}
where $\boldsymbol{y}_t$ is a $K \times 1$ vector of dependent variables and $\mathbb{E} | u_t |^{2 \bar{m} } \leq C$ and the $K \times K$ matrices $\boldsymbol{\Phi}_i$ satisfy appropriate stationarity and summability conditions. The equivalent equation-by-equation representation 
\begin{align}
y_{k,t} = \sum_{j=1}^p \big[ \Phi_{k,1,j},..., \Phi_{k,K,j} \big] \boldsymbol{y}_{t-j} + u_{k,t} = \big[ \boldsymbol{y}_{t-1}^{\prime},..., \boldsymbol{y}_{t-p}^{\prime} \big] \boldsymbol{\beta}_k + u_{k,t}, \ \ \ k \in \left\{ 1,..., K \right\}.    
\end{align}
Therefore, assuming a well-specified model with $\mathbb{E} \big[ \boldsymbol{u}_t | \boldsymbol{y}_{t-1},..., \boldsymbol{y}_{t-p} \big] = \boldsymbol{0}$, then the NED conditions are then satisfied trivially.     
\end{example}

\newpage

Notice that the above examples, provide cases in which although the classical martingale difference sequence framework fails to apply, the more general NED framework extends the applicability of certain time series processes.

\begin{example}[Misspecified AR Model, see  \cite{adamek2023lasso}] Consider an autoregressive AR model of order 2 as
\begin{align}
y_t = \rho_1 y_{t-1} +  \rho_2 y_{t-2} + v_t, \ \ \ v_t \sim \mathcal{N}(0,1).     
\end{align}
where the roots of the characteristic polynomial $( 1 - \rho_1 L - \rho_2 L^2 )$ are outside the unit circle. Moreover, define the misspecified model $y_t = \tilde{\rho} y_{t-1} + u_t$, where
\begin{align}
\tilde{\rho} = \underset{ \rho }{ \mathsf{arg min} } \ \mathbb{E} \left[ \left( y_t - \rho y_{t-1}    \right)^2 \right] = \frac{ \mathbb{E} \big[ y_t y_{t-1} \big]  }{ \mathbb{E} \left[ y_{t-1}^2 \right] } = \frac{ \rho_1 }{ 1 - \rho_2 }    
\end{align}
and $u_t$ is autocorrelated. 

Therefore, an \textit{m.d.s} assumption would be inappropriate in this case, since
\begin{align}
\mathbb{E} \big[ u_t | \sigma \left\{ y_{t-1}, y_{t-2},.... \right\} \big] = \mathbb{E} \big[ y_t = \tilde{\rho} y_{t-1} + u_t | \sigma \left\{ y_{t-1}, y_{t-2},.... \right\} \big] = - \frac{ \rho_1 \rho_2 }{ 1 - \rho_2 } y_{t-1} + \rho_2 y_{t-2} \neq 0.  
\end{align}
On the other hand, it can be shown that $\big( y_{t-1}, u_t \big)^{\prime}$ satisfies the NED-condition by considering the moving average representation of $y_t$ and thus by extension, of $u_t = y_t - \tilde{\rho} y_{t-1}$. In addition, since the coefficients are geometrically decaying, $u_t$ is thus clearly NED on $v_t$ and therefore the NED-condition is satisfied. 
\end{example}

Next, we focus on the related assumptions/conditions on the sparsity properties of the parameter vector. Specifically, the key condition to apply the lasso successfully is that the parameter vector $\boldsymbol{\beta}_0$ is (at least approximately) sparse. This is given by the following assumption.

\begin{assumption}
For some $0 \leq r \leq 1$ and sparsity level $s_r$, define the $N-$dimensional sparse compact parameter space such that 
\begin{align}
\mathcal{B}_N(r, s_r) := \big\{ \boldsymbol{\beta} \in \mathbb{R}^N : \norm{ \boldsymbol{\beta} }_r^r \leq s_r,  \norm{ \boldsymbol{\beta}}_{\infty} \leq C, \exists \ C < \infty \big\},    
\end{align}
and assume that $\boldsymbol{\beta}^0 \in \mathcal{B}_N(r, s_r) $
\end{assumption}

\begin{example}[Infinite Order AR,see  \cite{adamek2023lasso}] Consider an infinite order autoregressive model as below
\begin{align}
y_t = \sum_{j=1}^{\infty} \rho_j y_{t-j} + \varepsilon_t,     
\end{align}
where $\varepsilon_t$ is a stationary \textit{m.d.s} with sufficient moments existing, and the lag polynomial $\left( 1 - \sum_{j=1}^{\infty} \rho_j L^j \right)$ is invertible and satisfies the summability condition $\sum_{j=1}^{\infty} j^{\alpha} | \rho_j | < \infty$ for some $\alpha \geq 0$.
\end{example}

\newpage 

\subsubsection{Inference on low-dimensional parameters}

In this section we establish the uniform asymptotic normality of the desparsified lasso focusing on low-dimensional parameters of interest (see,  \cite{adamek2023lasso}). Specifically, we consider testing $P$ joint hypotheses of the form $\boldsymbol{R}_N \boldsymbol{\beta}^0 = \boldsymbol{q}$, via a Wald statistic, where $\boldsymbol{R}_N$ is an appropriate $P \times N$ matrix whose non-zero columns are indexed by the set 
\begin{align}
\mathcal{H}:= \left\{ j : \sum_{p=1}^P | r_{N,p,j} | > 0 \right\}    
\end{align}
with cardinality $h:= |H|$. Notice that we can allow for $h$ to increase in $N$ (and therefore $T$). Focusing on inference on finite set of parameters such that we can apply a standard central limit theorem. Therefore, given our time series setting, the long-run covariance matrix is 
\begin{align}
\boldsymbol{\Omega}_{N,T} = \mathbb{E} \left[ \frac{1}{T} \left( \sum_{t=1}^T \boldsymbol{w}_t \right) \left( \sum_{t=1}^T \boldsymbol{w}_t^{\prime} \right) \right],    
\end{align}
where $\boldsymbol{w}_t = \big( v_{1,t} u_t,..., v_{N,t} u_t  \big)^{\prime}$, enters the asymptotic distribution. Then, $\boldsymbol{\Omega}_{N,T}$ can equivalently be written 
\begin{align}
\boldsymbol{\Omega}_{N,T} &\equiv \boldsymbol{\Xi}(0) + \sum_{\ell=1}^{T-1} \big[ \boldsymbol{\Xi}(\ell) +  \boldsymbol{\Xi}^{\prime}(\ell) \big],    
\\
\boldsymbol{\Xi}(\ell) &\equiv \frac{1}{T} \sum_{t=\ell+1}^T \mathbb{E} \left[ \boldsymbol{w}_t \boldsymbol{w}_{t-\ell}^{\prime} \right].    
\end{align}
In order to estimate the asymptotic variance $\boldsymbol{\Psi}$, we suggest to estimate $\boldsymbol{\Omega}_{N,T}$ with the long-run variance kernel estimator such that 
\begin{align}
\boldsymbol{\Omega} &\equiv \widehat{\boldsymbol{\Xi}}(0) + \sum_{\ell=1}^{ Q_T-1 }  K \left(  \frac{\ell}{ Q_T } \right)  \big[ \boldsymbol{\Xi}(\ell) +  \boldsymbol{\Xi}^{\prime}(\ell) \big],    
\\
\boldsymbol{\Xi}(\ell) &= \frac{1}{T - \ell} \sum_{t = \ell + 1}^T \hat{\boldsymbol{w}}_t \hat{\boldsymbol{w}}_t^{\prime}    
\end{align}
with $\hat{w}_{j,t} = \hat{v}_{j,t} \hat{u}_t$, the kernel $K(.)$ can be taken as the Bartlett Kernel $K ( \frac{\ell}{ Q_T} ) = \left( 1 - \frac{\ell}{Q_T}  \right)$.

\subsubsection{Error Bound and the Consistency for the Lasso}

In this section, we consider a new error bound for the lasso in a high-dimensional time series model. Thus, the Lasso estimator of the parameter vector $\boldsymbol{\beta}^0$ in the model (see,  \cite{adamek2023lasso}) 
\begin{align}
\hat{ \boldsymbol{\beta} } := \underset{ \beta \in \mathbb{R}^N }{ \mathsf{arg min}  }  \ \left\{ \frac{1}{T} \norm{ \boldsymbol{y} - \boldsymbol{X} \boldsymbol{\beta} }_2^2 + 2 \lambda \norm{ \boldsymbol{\beta} }_1 \right\}.  
\end{align}
where $\boldsymbol{y} = \left( y_1,..., y_T \right)^{\prime}$ is the $T \times 1$ response vector, $\boldsymbol{X} = \big( \boldsymbol{x}_1,...., \boldsymbol{x}_T    \big)^{\prime}$ is the $( T \times N )$ design matrix and $\lambda > 0$ a tuning parameter. In other words, the above optimization problem adds a penalty term to the least squares objective to penalize parameters that are different from zero. Thus, when deriving this error bound, we typically require that $\lambda$ is chosen sufficiently large to exceed the empirical process such that $\underset{ j }{ \mathsf{max} } \ \left| \frac{1}{T} \sum_{t=1}^T x_{j,t} u_t \right|$, with high probability. Moreover, we define the following set such that 
\begin{align}
\mathcal{E}_T (z) := \left\{ \underset{ j \leq N, \ell \leq T }{ \mathsf{max} } \ \left| \sum_{t=1}^{\ell}  u_t x_{j,t} \leq z \right| \right\}     
\end{align}
and establish the conditions under which $\mathbb{P} \big( \mathcal{E}_T \left( \frac{T \lambda }{4} \right) \big) \to 1$. 

\subsubsection{Uniformly Valid Inference via the Disparsified Lasso}

Following the framework of \cite{adamek2023lasso}, we consider an application of the disparsified Lasso.  
\begin{example}[Sparse Factor Models]
Consider the factor model 
\begin{align}
y_t &=  \boldsymbol{\beta}^{0 \prime} \boldsymbol{x}_t + u_t, \ \ \ u_t \sim \mathcal{N} (0,1)
\\
\boldsymbol{x}_t &= \boldsymbol{\Lambda} \boldsymbol{f}_t + \boldsymbol{v}_t, \ \ \ \boldsymbol{v}_t \sim \mathcal{N} ( \boldsymbol{0}, \boldsymbol{\Sigma}_v ), \ \ \ \boldsymbol{f}_t \sim ( \boldsymbol{0}, \boldsymbol{\Sigma}_f ). 
\end{align}
Moreover, we assume that $\boldsymbol{\Sigma} = \boldsymbol{\Lambda} \boldsymbol{\Sigma}_f \boldsymbol{\Lambda}^{\prime} + \boldsymbol{\Sigma}_v$.
\end{example}

\begin{example}[Sparse VAR(1)] Consider a stationary VAR(1) model for $\boldsymbol{z}_t = \left( y_t, \boldsymbol{x}_t^{\prime} \right)^{\prime}$
\begin{align}
\boldsymbol{z}_t = \boldsymbol{\Phi} \boldsymbol{z}_{t-1} + \boldsymbol{u}_t, \ \ \ \mathbb{E} \left[ \boldsymbol{u}_t  \boldsymbol{u}_t^{\prime} \right] := \boldsymbol{\Omega}, \ \ \mathbb{E} \left[ \boldsymbol{u}_t  \boldsymbol{u}_t^{\prime} \right] = \boldsymbol{0}, \ \ \forall \ \ell \neq 0.       
\end{align}
Then, the regression $y_t = \boldsymbol{\phi}_1 \boldsymbol{z}_{t-1} + u_{1,t}$, where $\boldsymbol{\phi}_j$ is the $j-$th row of $\boldsymbol{\Phi}$, i.e., the first line of the VAR. 
\end{example}

\begin{example}
Consider the population nodewise regressions defined by the linear projections below 
\begin{align}
x_{j,t} = \boldsymbol{x}^{\prime}_{-j,t} \boldsymbol{\gamma}_j^0 + v_{j,t}, \ \ \ \ \boldsymbol{\gamma}_j^0 := \underset{ \boldsymbol{\gamma}  }{ \mathsf{arg min} } \ \mathbb{E} \left[ \frac{1}{T} \sum_{t=1}^T \left( x_{j,t} - \boldsymbol{x}_{-j,t}^{\prime} \boldsymbol{\gamma}     \right)^2 \right], \ \ \ \ j = 1,....,N.    
\end{align}
with $\tau_j^2 := \frac{1}{T} \sum_{t=1}^T \mathbb{E} \left[ v_{j,t}^2 \right]$. Specifically, if we let $\boldsymbol{\Phi} \phi \boldsymbol{I}$, with $| \phi | < 1$, and let $\boldsymbol{\Omega}$ have a Toeplitz structure, which implies that the elements of the matrix are defined such that $\omega_{i,j} = \rho^{ | i - j| }, | \rho | < 1$. Therefore, the nodewise regression parameter $\boldsymbol{\gamma}_j^0$ is only weakly sparse, in that it contains no zeroes, but its entries follow a geometrically decaying pattern\footnote{Notice that in high dimensional settings relevant conditions that capture features such as sub-Gaussianity, heavy tailed observations with slowy decaying temporal dependence are essential (see, \cite{baek2021local} and \cite{guillaumin2022debiased}. Dr.  Adam Sykulski gave a seminar with title: "The Debiased Spatial Whittle Likelihood", at the S3RI Departmental Seminar Series at the University of Southampton on the 5th of May 2022.}, meaning that $\underset{ j }{ \mathsf{max} } \norm{ \boldsymbol{\gamma}_j^0 }_r^r \leq C$.     

\end{example}

\newpage 

\section{High Dimensional Feature Selection Methods}

High dimensional statistical problems address the challenge of conducting robust estimation and inference when the number of unknown parameters, $p$, is much larger than the number of observations, $n$. The literature on high-dimensional feature selection methods include the conditional sure independence screening\footnote{Note that the covariate screening approach is widely used in the applied statistics literature, see \cite{aschard2017covariate}.} approach as discussed by \cite{fan2008sure}, \cite{wang2009forward}, \cite{ke2014covariance} and \cite{barut2016conditional} as well as \cite{he2013quantile} and \cite{kong2019screening} in the case of a conditional quantile functional form is employed. The covariate screening approach as a tool for variable selection works well in ultra high-dimensional linear regression models with stationary covariates, without requiring distributional restrictions. Statistical estimation methodologies in high-dimensional settings include \cite{belloni2011}, \cite{belloni2014inference}, \cite{wang2012quantile}, \cite{chernozhukov2015post}, \cite{fan2017estimation} as well as \cite{sun2015sprem} and \cite{wang2016high} for the cases of projection screening. 

A fundamental research question remains in most cases when considering modelling a high-dimensional vector of regressors is determining the statistical significance of either individual covariates or the model selection mechanism. \cite{lockhart2014significance} propose a test for detemining statistical significance with the Lasso, while \cite{chudik2018one} propose a methodology for selecting covariates in a high-dimensional environment using the multiple testing approach. Furthermore, \cite{battey2018distributed} considers hypothesis testing and parameter estimation in the context of the divide-and-conquer algorithm. In particular, in a unified likelihood-based framework, we propose new test statistics and point estimators obtained by aggregating various statistics from $k$ subsamples of size $n / k$, where $n$ is the sample size. 

Specifically, in low dimensional and sparse high dimensional settings, \cite{battey2018distributed} examine how large $k$ can be, as $n$ grows large, such as the loss of efficiency due to the divide-and-conquer algorithm is negligible. In other words, the resulting estimators have the same inferential efficiencies and estimation rates as an oracle with access to the full sample. While hypothesis testing in a low dimensional context is straightforward, in the sparse high dimensional setting, nuisance parameters introduce a nonnegligible bias, causing classical low dimensional theory to break down. Therefore, in their high dimensional Wald construction, the phenomenon is remedied through debiasing of the estimator, which gives rise to a test statistic with tractable limiting distribution. Thus, they find that the theoretical upper bound on the number of subsamples guaranteeing the same inferential or estimation efficiency as the whole-sample procedure is $k = o( ( s \mathsf{log} d )^{-1} \sqrt{n} )$ in the linear model, where $s$ is the sparsity of the parameter vector.      

\begin{lemma}[\cite{battey2018distributed}]
Assume that $\Sigma = \mathbb{E} \big( \boldsymbol{X}_i \boldsymbol{X}_i^{\top} \big)$ satisfies $C_{\mathsf{min}} < \lambda_{ \mathsf{min} } ( \Sigma ) \leq \lambda_{ \mathsf{min} } ( \Sigma ) \leq C_{\mathsf{max}}$ as well as $\norm{ \Sigma^{-1/2} \boldsymbol{X}_1 }_{ \psi_2 } = \kappa$, then it holds that
\begin{align}
\mathbb{P} \left( \underset{ j = 1,..., k }{ \mathsf{max} } \norm{ M^{(j)} \widehat{\Sigma}^{(j)} - I }_{ \mathsf{max} } \leq \alpha \sqrt{ \frac{ \mathsf{log} d }{n} }  \right) \geq 1 - 2 kd^{ - c_2 }, \ \ \ c_2 = \frac{ \alpha^2 C_{\mathsf{min} } }{ 24 e^2 \kappa^4 C_{\mathsf{max} } } - 2.
\end{align}
\end{lemma}

\newpage

On the other hand, in time series regression models incorporating high dimensionality features requires to develop further tools such as Gaussian approximations (e.g., see \cite{chernozhukov2014gaussian} and \cite{zhang2017gaussian}) and consistent model selection techniques robust to the presence of heteroscedasticity (see, \cite{halunga2017heteroskedasticity}) and nonstationarity. In particular in the time series econometrics literature many open problems remain such as robust model selection methodologies for ultra-high dimensional environments with heterogenous data. We return back to Vector Autoregression Processes, although not necessarily a high-dimensional VAR$(p)$ process. In practise, we consider the case where the dimension of the time-series vector is less than the sample size $n$. Formally, a $p-$dimensional vector-valued stationary time series $X_t = (  X_{1t},..., X_{pt} )$, with $t=1,...,n$, can be modelled using a VAR representation of lag $d$ with serially uncorrelated Gaussian errors, which takes the following form
\begin{align}
X_t = A_1 X_{t-1} + ... + A_d X_{t-d} + e_t, \ \ \ e_t \sim \mathcal{N} ( \boldsymbol{0}, \boldsymbol{\Sigma}_e )
\end{align}
where $A_1,..., A_d$ are (p x p) matrices and et is a p-dimensional vector of possibly correlated innovation shocks. Therefore, the main objective in VAR models is to estimate the transition matrices $A_1,..., A_d$, together with the order of the model $d$, based on time series realizations $( X_0, X_1,..., X_n )$. Then, the structure of the transition matrices provides insights into the complex temporal relationships amongst the p time series and the particular representation provides a way to apply forecasting techniques (see, \cite{michailidis2013autoregressive},  \cite{basu2015regularized} and \cite{basu2021graphical}).

\subsection{Ultra-high dimensionality under dependence}

According to \cite{yousuf2018variable}, it is possible to achieve the sure screening property in the ultrahigh dimensional setting with dependent errors and covariates. However, in order to do that we need to make stronger assumptions on the moments of both the error and covariate processes. Specifically, if the error and covariate processes are assumed to follow a stronger moment condition, such as, $\Delta_{0,q} (\varepsilon) < \infty$ and $\Phi_{0,q} ( \boldsymbol{x} ) < \infty$ for arbitrary $q > 0$, we can then achieve a much larger range of $p_n$ which will cover the ultra-high dimensional case. Specifically, we need a condition that implies that the tails of the covariate and error processes are exponentially tight. A wide range of processes satisfy the above condition. 

Suppose that $\varepsilon_i$ is a linear process such that 
\begin{align*}
\varepsilon_i = \sum_{j=0}^{\infty} f_j e_{i-j}, \ \ \ e_i \sim \textit{i.i.d},
\end{align*}
with $\sum_{\ell = 0}^{\infty} | f_{\ell} | < \infty$ then it holds that 
\begin{align}
\Delta_{0,q} (\varepsilon) ( \epsilon_{\ell} ) = \norm{ e_0 - e_0^* }_q \sum_{\ell = 0}^{\infty} | f_{\ell} |     
\end{align}
Thus, if we assume that $e_0$ is sub-Gaussian, then $\tilde{\alpha}_{\epsilon} = \frac{1}{2}$ since $\norm{ e_0 }_{ q } = \mathcal{O}_p ( \sqrt{q} )$ then, when $e_i$ is sub-exponential then it holds that  $\tilde{\alpha}_{\epsilon} = 1$.

\newpage

More generally, for $e_i := \sum_{j=0}^{\infty} f_j e_{i-j}^p$ and thus if $e_i$ is sub-exponential, we have that $\tilde{\alpha}_{\epsilon} = p$, where $p$ is the number of components of the high-dimensional vector. The high dimensionality as well as the non-parametric aspect can slow down the rate of convergence and thus especially the variance of the estimator can have slower convergence to the corresponding variance of the process. 

Thus, for any fixed $q$, we are not placing additional assumptions on the temporal decay rate of the covariate error processes besides requiring $\Delta_{0,q} (\varepsilon), \Phi_{0,q} (\varepsilon) < \infty$. On the other hand, the ultrahigh dimensional setting requires geometrically decaying strong mixing coefficients, in addition to requiring sub-exponential tails for the response. Consider for example the case that $\varepsilon_i = \sum_{j=0}^{\infty} f_j e_{i-j}$, geometrically decaying strong mixing coefficients would require the coefficients, $f_j$, to decay geometrically.  

\medskip
  
\begin{theorem}[Theorem 2 in \cite{yousuf2018variable}]
Define with $\alpha = \frac{2}{1 + 4 \tilde{\alpha}_x }$, then it holds that for any $c_2 > 0$
\begin{align}
        \mathbb{P} \left(  \underset{ j \leq p_n }{ \mathsf{max} } \left| \hat{\rho}_j - \rho_j \right| > c_2 n^{-k} \right) 
        \leq 
        \mathcal{O}_p \left( s_n p_n \mathsf{exp} \left\{ - \frac{ n^{1/2 - \kappa} }{  \nu_x^2 s_n } \right\}^{ \tilde{\alpha} } \right) + \mathcal{O}_p \left( p_n \mathsf{exp} \left\{ - \frac{ n^{1/2 - \kappa} }{  \nu_x \nu_{\epsilon} } \right\}^{ \tilde{\alpha}^{\prime} } \right)
\end{align}
\end{theorem}

\subsubsection{SIS with dependent observations}

Sure Independence Screening, is a method of variable screening based on ranking the magnitudes of the $p_n$ marginal regression estimates. Under appropriate conditions, this simple procedure is shown to possesses the sure screening property. Define with 
\begin{align}
\hat{\boldsymbol{\rho}} = \left( \hat{\rho}_1,..., \hat{\rho}_{p_n} \right), \ \ \text{where} \ \ \hat{\rho}_j = \left( \sum_{t=1}^n X_{tj}^2 \right)^{-1} \left( \sum_{t=1}^n X_{tj} Y_t \right).    
\end{align}
Therefore, $\hat{\rho}_j$ is the OLS estimate of the linear projection of $Y_t$ onto $X_{tj}$. Let
\begin{align}
\mathcal{M}_{*} = \big\{ 1 \leq i \leq p_n : \beta_i \neq 0 \big\}    
\end{align}
and let $| \mathcal{M}_{*} | = s_n << n$ be the size of the true sparse model. Moreover, we sort the elements of $\hat{\boldsymbol{\rho}}$ by their magnitudes. Thus, for any given $\gamma_n$, define a sub-model such that 
\begin{align}
\hat{\mathcal{M}}_{\gamma_n} = \big\{ 1 \leq i \leq p_n: \left| \hat{\rho}_i \right| \geq \gamma_n \big\}     
\end{align}
and let $\left|  \hat{\mathcal{M}}_{\gamma_n} \right| = d_n$ be the size of the selected model. Furthermore, notice that the screening property states that for an appropriate choice of $\gamma_n$, we have that 
\begin{align*}
\mathbb{P} \left( \mathcal{M}_{*} \subset \hat{\mathcal{M}}_{\gamma_n}  \right) \to 1.
\end{align*}

\newpage

\subsubsection{Experimental design examples}

\paragraph{Experimental Design I: Uncorrelated Features}

Consider the model for the covariate process we have 
\begin{align}
\boldsymbol{x}_t = \boldsymbol{A}_1 \boldsymbol{x}_{t-1} + \boldsymbol{\eta}_t    
\end{align}
where $A_1 = \mathsf{diag} \left\{ \gamma  \right\}$ and $\gamma \sim \mathsf{Unif} [ 0.4, 0.6 ]$. Moreover, we set with $\boldsymbol{\eta}_t \sim \mathcal{N} (\boldsymbol{0},\boldsymbol{\Sigma}_{\eta})$. In particular, when we are dealing with uncorrelated predictors we set $\Sigma_{\eta} = I_{ p_n }$ which represents the no correlation case. 

\medskip

\paragraph{Experimental Design II: Correlated Features}

One can compare the performance of SIS and GLSS in the case of correlated predictors. The covariate process is generated  with $A_1 = \mathsf{diag} \left\{  0.4^{ | i - j | + 1 }  \right\}_{ i,j \leq p_n }$ such that $\boldsymbol{\eta}_t \sim \mathcal{N} ( 0, \Sigma_{\eta} )$ and $\boldsymbol{\eta}_t \sim t_5 ( 0, \boldsymbol{V} )$ , with
\begin{align}
\Sigma_{\eta} = \left\{ 0.3^{ | i - j | } \right\}    
\end{align}

\begin{proof}
Recall that it holds that $\sum_{k=1}^{p_n} \mathbf{1} \left\{ \left| \beta_k \right| > 0 \right\} = s_n$. Hence, we obtain the following expression 
\begin{align}
\mathbb{P} \left( \left| S_1 - \mathbb{E} (S_1) \right| > \frac{ c_2 n^{ - \kappa }  }{2} \right)  
\leq 
\ \sum_{ k \in M_{*} } \mathbb{P} \left( \left|  \frac{ X_{tj} \left( X_{tk} \beta_k \right) }{n} - \beta_k \mathbb{E} \left[ X_{tj} X_{tk} \right]  \right| > \frac{ c_2 n^{- \kappa } }{ 2 s_n }  \right)
\end{align}
Moreover, it holds that $\norm{ X_{ij} }_r \leq \Delta_{0,r} (\boldsymbol{X}_j) \leq \Phi_{0,r}(\boldsymbol{x})$. Using this we compute the cumulative functional dependence measure of $X_{tk} X_{tj}$ such that 
\begin{align*}
\sum_{t=m}^{+ \infty} \norm{ X_{tj} X_{tk} - X_{tj}^{*} X_{tk}^{*}  }_{ r / 2 } 
&\leq
\sum_{t=m}^{ \infty } \big( \norm{ X_{tj} }_r \norm{ X_{tk} - X_{tk}^{*} }_r \norm{ X_{tj} -  X_{tj}^{*} }_r \big)
\\
&\leq 2 \Phi_{0,r} (x) \Phi_{m,r}(x) = \mathcal{O}_p \left( m^{- \alpha_x } \right). 
\end{align*}
Therefore, we obtain that 
\begin{align}
\underset{ m }{ \mathsf{sup} } \ ( m + 1 )^{\alpha_x} \sum_{t=m}^{\infty} \norm{ X_{tj} }_r \norm{ X_{tk} - X_{tk}^{*} }_{r/2} \leq 2 K_{x,r}^2.      
\end{align}
Using the above results we get the following probability bound
\begin{align}
\mathbb{P} \left( \left| S_1 - \mathbb{E} \left[ S_1 \right] \right| > \frac{c_2 n^{- \kappa} }{2} \right) \leq C s_n \times \left\{ \frac{ n^{\omega} K_{x,r}^r }{ ( n / s_n )^{ r / 2 - rk/2 } } + \mathsf{exp} \left( - \frac{ n^{ 1 - 2 \kappa } }{ s_n^2 K_{x,r}^4 } \right) \right\}    
\end{align}
Hence, by our choice of $\gamma_n$, we obtain that $\mathbb{P} \left( \mathcal{M}_{*} \subset \hat{\mathcal{M}}_{\gamma_n} \right)  > \mathbb{P} \left( \mathcal{A}_n \right)$. In other words, this allow us to obtain a probability bound with respect to the spectral radius of the corresponding covariance matrix 
\begin{align}
\sum_{k=1}^{p_n} \rho_k^2 = \mathcal{O}_p \left( \lambda_{\mathsf{max}} (\Sigma) \right).     
\end{align}

\newpage

Then, we extend the above result to the corresponding set defined by $\mathcal{B}_n$ such that 
\begin{align}
\mathcal{B}_n := \left\{ \underset{ k \leq p_n }{ \mathsf{max} } \left| \hat{\rho}_k - \rho_k \right| \leq c_4 n^{ - \kappa } \right\}    
\end{align}
Therefore, we obtain a measurability condition for the number of $\left\{ k: \left| \hat{\rho}_k \right| > 2 c_4 n^{-4} \right\}$ cannot exceed the number of $\left\{ k: \left| \rho_k \right| > c_4 n^{-k} \right\}$ which is bounded by $\mathcal{O}_p \left( n^{2k} \lambda_{\mathsf{max}} (\Sigma) \right)$. Thus, by choosing $c_4 = c_3 / 2$ we can obtain the correct probability bound such that 
\begin{align}
\mathbb{P} \left( \left| \hat{\mathcal{M}}_{\gamma_n} \right| < \mathcal{O}_p  \left( n^{2k} \lambda_{\mathsf{max}} (\Sigma) \right)  \right)  > \mathbb{P} \left( \mathcal{B}_n \right) 
\end{align}
We are interested to bound the following probability event 
\begin{align}
\mathcal{A}_1 := \mathbb{P} \left\{  \left| \hat{\gamma}_{i,k} - \gamma_{i,k} \right| > \frac{ c n^{-k} }{ \ell_n } \right\}  
\end{align}
where a probability bound for the event $\mathcal{A}_1$, can be obtained as below:
\begin{align*}
\mathcal{A}_1 
&\leq  
\mathbb{P} \left( \left| \frac{1}{n} \sum_{t=1}^{ n - |i|} \varepsilon_{t,k} \varepsilon_{t + |i|,k} - \mathbb{E} \left[ \frac{1}{n} \sum_{t=1}^{n - |i|}   \varepsilon_{t,k} \varepsilon_{t + |i|,k} \right] \right|  \right)   
\\
&+
\mathbb{P} \left( \left| \mathbb{E} \left[ \frac{1}{n} \sum_{t=1}^{ n - |i|}  \varepsilon_{t,k} \varepsilon_{t + |i|,k} \right] - \gamma_{i,k}  \right| > c n^{-k} / 4 \ell_n \right)
\\
&+
\mathbb{P} \left(  \left| \frac{1}{n} \sum_{t=1}^{n - |i|} \varepsilon_{t + |i|,k} X_{t + |i|, k} \left( \frac{ \sum_{j=1}^n X_{jk} \varepsilon_{j,k} / n  }{  \sum_{j=1}^n X_{jk}^2 / n  } \right) \right|  > c n^{-k} / 4 \ell_n  \right)
\\
&+
\mathbb{P} \left(  \left| \frac{1}{n} \sum_{t=1}^{n - |i|} X_{tk} X_{t + |i|, k} \left( \frac{ \sum_{j=1}^n X_{jk} \varepsilon_{j,k} / n  }{  \sum_{j=1}^n X_{jk}^2 / n  } \right)^2 \right|  > c n^{-k} / 4 \ell_n  \right)
\end{align*}
Furthermore, we can determine the bias of the term as below
\begin{align}
\left| \mathbb{E} \left[ \sum_{t=1}^{n - |i|} \frac{ \varepsilon_{t,k} \varepsilon_{t + |i|,k } }{n} \right]  - \gamma_{i,k} \right| \leq \frac{ i \gamma_{i,k} }{n}.  
\end{align}
Moreover, it holds that 
\begin{align*}
\mathbb{P} &\left( \left| \frac{1}{n} \sum_{t=1}^{ n - |i| } \varepsilon_{t,k} X_{t+|i|,k} > M           \right| \right) 
\leq 
\mathbb{P} \left( \left| \frac{1}{n} \sum_{t=1}^{n - |i|} \varepsilon_{t,k} X_{t+|i|,k} - \mathbb{E} \left[ \frac{1}{n} \sum_{t=1}^{n - |i|} \varepsilon_{t,k} X_{t+|i|,k}  \right]  \right| > M - \left| \mathbb{E} \left[ \frac{1}{n} \sum_{t=1}^{n - |i|} \varepsilon_{t,k} X_{t+|i|,k}  \right] \right|  \right)   
\end{align*}
\begin{align}
M > \underset{ k \leq p_n }{ \mathsf{max} } \ \underset{ i \leq \ell_n }{ \mathsf{max} } \ 2 \left| \mathbb{E} \left( \varepsilon_{t,k} X_{t+|i|,k} \right) \right| + \varepsilon, \ \ \text{for some} \ \varepsilon > 0.
\end{align}
\end{proof}

\newpage

\subsection{Uniform Asymptotic Inference and Model Selection}

This is another important topic which we cannot cover here extensively (see, \cite{tibshirani2018uniform}). 

\subsubsection{Uniform-in-Submodel Bounds}

Following the framework proposed by \cite{kuchibhotla2021uniform}, suppose that $\mathcal{M} = \left\{ M_1,...,M_L   \right\}$ denotes a collection of submodels, where $M_j$ represents a subset of covariates for $1 \leq j \leq L$. Also, denote with $\widehat{\beta}_{M_j}$ to represent the least-squares estimator for the linear regression of the response on the covariates of the model $M_j$. Then, by simultaneous consistency, we mean the existence of target vectors $\left\{ \beta_{M_j} : 1 \leq j \leq L \right\}$ such that the following result holds
\begin{align}
\label{asy}
\underset{ M \in \mathcal{M} }{ \text{sup} } \norm{ \widehat{\beta}_{M} - \widehat{\beta} } = o_p(1), \ \ \text{as} \ \ n \to \infty, \ \ \ \ \text{for some} \ \norm{ . }.
\end{align}
More specifically, if $\widehat{M}$ is a selected model, then one can perform inference of $\beta_{ \widehat{M} }$ by estimating the distribution of $\widehat{ \beta }_{ \widehat{M} }$. According to \cite{kuchibhotla2021uniform}, in practise, even though various model-selection criteria like $C_p$ (AIC), (BIC) and lasso have been recommended for covariate selection in linear regression, developing more general asymptotic results to the \textit{asymptotic uniform linear representation} in the special case of the least-squares linear regression estimator are still an ongoing research field. The framework of variable selection and linear regression is often considered in the context of  high-dimensional linear regressions. Although, the procedure uses only a reduced set of variables in the final regression, it uses all variables in the proceding covariate selection stage of estimation. 

Suppose that $\widehat{M} \in \mathcal{M}$ is the final selected submodel, where $\mathcal{M}$ is some finite and countable collection of models, with $\widehat{ \beta }_{ \widehat{M} }$ the corresponding least-squares estimator. The estimator $\widehat{ \beta }_{ \widehat{M} }$ is known as the post-regularization estimator in the high-dimensional statistics literature when $\widehat{M}$ is obtained from a regularized least-squares procedure. A relevant question regarding the asymptotic behaviour of the particular estimator is "\textit{What does $\widehat{ \beta }_{ \widehat{M} }$ estimate consistently?}".  These considerations are relevant to the literature of \textit{variable and model selection}, which includes: methodologies for determining the statistical significance of predictors in settings with many covariates (see, \cite{chudik2018one}) as well as the uniform inference post-model selection which includes among others the studies of \cite{belloni2016post} and \cite{farrell2015robust}).

Here we focus on the framework proposed by \cite{kuchibhotla2021uniform} and present key results. Suppose that for any $M \subset \left\{ 1,2,...,p \right\}$ the ordinary least-squares (OLS) empirical risk (or objective) function is
\begin{align}
\hat{R}_n \left( \theta ; M \right) := \frac{1}{n} \sum_{i=1}^n \left\{ Y_i - X_i^{\top} \left( M \right) \theta \right\}^2, \ \ \text{for} \ \ \theta \in \mathbb{R}^{|M|}. 
\end{align} 
Then, by expanding the square function, it is clear that 
\begin{align}
\hat{R}_n \left( \theta ; M \right) = \frac{1}{n} \sum_{i=1}^n Y_i^2 - \frac{2}{n} \sum_{i=1}^n Y_i X_i^{\top} (M) \theta + \theta^{\top} \left(  \frac{1}{n} \sum_{i=1}^n X_i(M) X_i^{\top} (M) \right) \theta. 
\end{align}

\newpage

Furthermore, define the following quantities
\begin{align}
\hat{ \Sigma }_n := \frac{1}{n} \sum_{i=1}^n X_i X_i^{\top} \in \mathbb{R}^{p \times p}, \ \ \ \text{and} \ \ \ \hat{ \Gamma}_n :=  \frac{1}{n} \sum_{i=1}^n X_i Y_i  \in \mathbb{R}^{p \times p} .   
\end{align}
Then, the least-squares linear regression estimator $\hat{\beta}_{n,M}$ is defined as below
\begin{align}
\hat{\beta}_{n,M} 
= 
\underset{ \theta \in \mathbb{R}^{|M|} }{ \text{arg min} } \ \hat{R}_n \left( \theta ; M \right) 
= 
\underset{ \theta \in \mathbb{R}^{|M|} }{ \text{arg min} } \left\{  \theta^{\top} \hat{ \Sigma }_n (M) \theta - 2 \theta^{\top} \hat{ \Gamma}_n (M)  \right\}. 
\end{align} 
The notation arg min$_{ \theta } f( \theta )$ denotes the minimized of $f( \theta )$. Based on the quadratic expansion of the empirical objective function $\hat{R}_n \left( \theta ; M \right)$, the estimator $\widehat{\beta}_{n,M}$ is given by the closed form expression below
\begin{align}
\widehat{\beta}_{n,M} = \left[ \hat{ \Sigma }_n  (M) \right]^{-1} \hat{ \Gamma}_n (M), 
\end{align}
assuming nonsingularity of $\hat{ \Sigma }_n  (M)$. Therefore, it is clear that $\widehat{\beta}_{n,M}$ is a smooth (nonlinear) function of two averages  $\hat{ \Sigma }_n  (M)$ and $\hat{ \Gamma}_n(M)$. Then, assuming that the random vectors $( X_i, Y_i )$ are \textit{i.i.d} with finite fourth moments, it follows that $\hat{ \Sigma }_n  (M)$  and $\hat{ \Gamma}_n(M)$ converge in probability to their expectations. Thus, we can define the expected value in terms of matrices and vectors as below 
\begin{align}
\Sigma_n := \frac{1}{n} \sum_{i=1}^n \mathbb{E} \left[ X_i X_i^{\top} \right] \in \mathbb{R}^{p \times p}, \ \ \ \text{and} \ \ \ \Gamma_n :=  \frac{1}{n} \sum_{i=1}^n  \mathbb{E} \left[ X_i Y_i \right]  \in \mathbb{R}^{p \times p}.   
\end{align}
Following the notation above $\widehat{\beta}_{n,M} = \left[ \hat{ \Sigma }_n  (M) \right]^{-1} \hat{ \Gamma}_n (M)$, and if $\big( \hat{ \Sigma }_n - \hat{ \Sigma }_n , \hat{ \Gamma}_n - \Gamma_n  \big) \overset{ p }{ \to } 0$ as $n \to \infty$, using a Slutsky-type argument, it follows that $
\left( \widehat{\beta}_{n,M}  - \beta_{n,M} \right) \overset{ p }{ \to } 0$, as $n \to \infty$, where $\beta_{n,M}$ is:
\begin{align}
\beta_{n,M} := \left[ \Sigma_n(M) \right]^{-1} \Gamma_n(M) \equiv 
\underset{ \theta \in \mathbb{R}^{|M|} }{ \text{arg min} } \left\{  \theta^{\top} \Sigma_n (M) \theta - 2 \theta^{\top} \Gamma_n (M) \right\}. 
\end{align}
Thus, the convergence statement above only concerns a single submodel $M$ and is not uniform over $M$. By uniform-in-submodel $\norm{.}_2-$norm consistency of $\hat{ \beta }_{n,M}$ to $\beta_{n,M}$, for $M \in \mathcal{M} (k)$, we mean that 
\begin{align}
\underset{ M \in \mathcal{M}(k) }{ \text{sup} } \norm{  \widehat{\beta}_{ n,M } - \beta_{ n, M }  } = o_p(1) \ \ \text{as} \ \ n \to \infty.
\end{align} 
Thus, converges of $\widehat{\beta}_{ n,M } $ to $ \beta_{ n, M }$ only requires convergence of $\widehat{\Sigma}_n(M)$ to $\Sigma_n(M)$ and $\widehat{\Gamma}_n(M)$ to $\Gamma_n(M)$.

\medskip

\begin{remark}
A relevant stream of literature to uniform inference for submodel selection, is the testing methodologies which consider the nested property between a comparison of model as a mechanism for specification testing. In particular, \cite{hagemann2012simple}  considers a non-nested specification test (see, also  \cite{mackinnon1983model}) while \cite{rinaldo2019bootstrapping} develops a framework for model selection based on sample-splitting for high dimensional assumption lean inference purposes. 
\end{remark}

\newpage

\subsection{Non-nested Regressions and Variable Selection Specification Testing}

Following the framework proposed by \cite{hagemann2012simple}, suppose that we observe covariates 
\begin{align*}
\left\{ \big( x_{i,1}^{\top},..., x_{i,M}^{\top} \big)^{\top} \in \mathbb{R}^{d_1 + ... + d_M } : i = 1,...,n \right\}
\end{align*}
that give rise to $M \geq 2$ different possible linear regression models for $y := ( y_1, ..., y_n )^{\top} \in \mathbb{R}^n$, such that 
\begin{align}
y = X_m \beta_m + u_m, \ \ m \in \mathcal{M} := \left\{ 1,..., M  \right\}, 
\end{align} 
where $X_m := \big( x_{1,m},..., x_{n,m} \big)^{\top} \in \mathbb{R}^{ n \times d_m }$ is the design matrix of model $m$. 

Then, the matrices $X_1,...., X_M$ are assumed to be non-nested, which implies that for any two matrices with index $m \neq \ell \in \mathcal{M}$, no matrix can be obtained by another by a linear transformation. Furthermore, the particular assumption does not rule out the possibility that some of the columns of $X_m$ and $X_{\ell}$ are identical or that they may be nonlinear transformations of another. In particular, suppose that there is an unobserved design matrix $X_{ m^{*} } := \big( x_{ 1, m^{*} },..., x_{n, m^{*} }     \big)^{*}$ associated with the correct model $m^{*}$. Suppose that $\mathcal{F} := \big\{ X_m :  m \in \mathcal{M} \cup \left\{ m^{*} \right\} \big\}$ has the following properties: 

\medskip

\begin{assumption}
Consider that $\left\{ y_i, ( x_{i,m} )^{\top}_{ m \in \mathcal{M} \cup \left\{ m^{*} \right\}} : i \geq 1 \right\}$ is a sequence of $\textit{i.i.d}$ random vectors. We have that $\mathbb{E} | x_{i,m} |^4 < \infty$ for all $m \in \mathcal{M} \cup \left\{ m^{*} \right\}$, where the number of elements of $\mathcal{M}$ (number of covariates) does not depend on $n$ (sample size). For all $m \in \mathcal{M}$, the matrices $\mathbb{E} \left( x_{i,m} x_{i,m}^{\top} \right)$ are positive definite.    
\end{assumption}
Next, we impose a condition for the existence of a correct model specification. 

\begin{assumption}
Model $m^{*}$ satisfies $\mathbb{E} \big( y | \mathcal{F}  \big) = X_{ m^{*} } \beta_{ m^{*} }$. Define with $u_{i, m^{*} } := y_i - x_{i,m}^{\top} \beta_{ m^{*} }$ for all $i \geq 1$ and $\mathbb{E} u^{4}_{i, m^{*} } < \infty$.
\end{assumption}
The proposed $J$ test presumes that for some predetermined $m \in \mathcal{M}$, the null hypothesis is formulated such as $m = m^{*}$ against the alternative hypothesis $ m \neq m^{*}$ in the presence of non-nested alternatives $\ell \in \mathcal{M} \backslash \left\{ m \right\}$. To do this, we artificially nesting the models via an additional parameter vector $\alpha_{m} := \left( \alpha_{\ell, m} \right)_{\ell \in \mathcal{M} \backslash \left\{ m \right\} } \in \mathbb{R}^{M-1}$ such that 
\begin{align}
y &= X_m b_m + \sum_{ \ell \in \mathcal{M} \backslash \left\{ m \right\} } \alpha_{\ell, m}  X_{\ell} \beta_{\ell} + u, 
\\
b_m &:= \left( 1 - \sum_{ \ell \in \mathcal{M} \backslash \left\{ m \right\} } \alpha_{\ell, m} \right) \beta_m. 
\end{align}

Therefore, since the vectors $\left(  \alpha_{\ell, m}, \beta_{\ell} \right)_{ \ell \in \mathcal{M} \backslash \left\{ m \right\} }$ of the nesting model may not be identified, we can replace the $\beta_{\ell}$ by the OLS estimates such that $\hat{\beta}_{\ell} = \left( X_{\ell}^{\top} X_{\ell} \right)^{-1} X_{\ell}^{\top} y$. (see, \cite{hagemann2012simple}).

\newpage

\begin{remark}
In particular, the minimum J-test proposed by \cite{hagemann2012simple} does not require the correct model to be among the considered specifications and avoids ambiguous test outcomes. More specifically, the MJ test determines with asymptotically correct size if the correct model is among the specifications under consideration. Furthermore, if the correct model is present, it is chosen with probability approaching one as the sample size becomes large.  
\end{remark}
After redefining the error term appropriately, this yields that 
\begin{align}
y = X_m b_m + \sum_{ \ell \in \mathcal{M} \backslash \left\{ m \right\} } \alpha_{\ell,m} X_{\ell} \hat{\beta}_{\ell} + u,  
\end{align}
Therefore, a Wald test for testing the null hypothesis that $\alpha_m = 0$ is a $J$test for the validity of model $m$ in the presence of the alternatives $\mathcal{M} \backslash \left\{ m \right\}$. Thus, to construct the test statistic, let 
\begin{align}
\lambda_{n,m} &:= n^{-1/2} \left( y^{\top} P_{\ell} M_m y \right)_{ \ell \in \mathcal{M} \backslash \left\{ m \right\}  }    
\\
\hat{\Sigma}_{n,m} &:= n^{-1} \left(  y^{\top} P_{\ell} M_m \hat{\Omega}  M_m P_{\ell^{\prime} } y \right)_{ \ell, \ell^{\prime} \in \mathcal{M} \in  \mathcal{M} \backslash \left\{ m \right\}  } 
\\
P_m &:= X_m \left( X_m^{\top} X_m \right)^{-1} X_m^{\top} \ \ \ \text{and} \ \ \ M_m:= I_n - \left( P_m \right)_{ \textcolor{red}{ n \times n }  }    
\end{align}
are the usual projection matrices and $\hat{\Omega}_{n,m}$ is an estimate of the conditional expectation of $\mathbb{E} \left( u_m^{*}  u_m^{*\top} | \mathcal{F} \right)$. The $J$ test statistic for model $m$ is then given by   
\begin{align}
J_{n,m} := \lambda_{n,m}^{\top} \hat{\Sigma}_{n,m}^{-1}  \lambda_{n,m}      
\end{align}
In other words, the hypothesis that the model $m$ is the true model is rejected for large values of $J_{n,m}$. The following asymptotic results holds for obtaining the empirical size and power of the test statistic. 

\begin{lemma}[\cite{hagemann2012simple}]
Suppose that Assumptions are satisfied. Then, it holds that 
\begin{itemize}

\item[(i).] If $m^{*} \in \mathcal{M}$, then $J_{n,m^{*}} \Rightarrow \chi^2_{M-1}$.

\item[(ii).] For every $m \in \mathcal{M} \backslash \left\{ m \right\}$ and every $B \in \mathbb{R}$, we have that $\mathsf{lim}_{ n \to \infty } \mathbb{P} \big( J_{n,m} > B  \big) = 1$. 
\end{itemize}
\end{lemma}

\begin{lemma}[\cite{hagemann2012simple}]
Suppose that the conditions of the Theorem hold. If $m^{*} \in \mathcal{M}$, then
\begin{align}
\underset{ n \to \infty }{ \mathsf{plim} } \ \underset{ x \in \mathbb{R} }{ \mathsf{sup} } \ \left|  \mathbb{P}^{*} \big( MJ_n^{*} \leq x  \big) -  \mathbb{P}^{*} \big( MJ_n^{*} \leq x  \big) \right| = 0.   
\end{align}
\end{lemma}

\begin{proof}
Let $\hat{m} := \mathsf{argmin} \mathcal{J}_n^{*}$. We first show that if $m^{*} \in \mathcal{M}$, then $\hat{m}^{*}$ approximates $m^{*}$. 

Without loss of generality, fix any $0 < \epsilon < 1$, then
\begin{align}
\mathbb{P} \big(  \boldsymbol{1} \left\{ \exists \ m \in \mathcal{M} \backslash \left\{ \hat{m} \right\} : J_{n,m} \leq B \right\} > \epsilon  \big)   
\leq 
\mathbb{P} \big(  \exists \ m \in \mathcal{M} \backslash \left\{ \hat{m} \right\} : J_{n,m} \leq B  \big)  + \mathbb{P} \big(  \hat{m} \neq m^{*} \big),
\end{align}
which converges to zero as $n \to \infty$.

\newpage

In summary, it can be proved that 
\begin{align}
\underset{ n \to \infty }{ \mathsf{plim} } \ \mathbb{P}^{*} \big( J^{*}_{n, \hat{m}^{*} } \leq x \big) \equiv K(x),      
\end{align}
where $K(x)$ is the distribution function of a $\chi^2_{M-1}$ variable. 
\end{proof}
\begin{remark}
Notice that although the methodology proposed by \cite{hagemann2012simple} considers testing for correct model specification it differs from the literature of conditional subvector testing. Specifically, the subvector testing method considers the roots of the following characteristic polynomial
\begin{align}
\left|  \hat{\kappa} I_p - n^{-1} \widehat{G}_n^{-1/2} \big( \bar{Y}_0, W \big)^{\prime} Z \widehat{H}_n Z^{\prime} \big( \bar{Y}_0, W \big) \widehat{G}_n^{-1/2}  \right| = 0.    
\end{align}
\end{remark}

\medskip

\begin{remark}
A different stream of literature considers the classification of non-overlapping and overlapping models and discuss how the relationship between candidate models affects the asymptotic distributions of the test statistics. The particular stream of literature was initiated by \cite{vuong1989likelihood}. This approach considers measuring the distance from the model under investigation to the true distribution, which requires to solve the minimization problem $\mathsf{inf}_{ P \in \mathcal{P} } d (P, \mu)$. 
\end{remark}

\subsection{Divide and Conquer Variable Selection Algorithm}

Further algorithmic procedures for model selection purposes in high dimensional linear regression models includes the framework of \cite{meinshausen2010stability} propose a variable selection methodology where the algorithm repeatedly employs a subset selection and apply the backward algorithm. Then, the final stage provides all the significantly selected covariates with a pre-specified control rate. Therefore, the particular \textit{stability selection} procedure is designed to address the issue of choosing the amount of regularization such that a certain familywise type I error rate in multiple testing can be conservatively controlled for finite sample size. The issue of correlated covariates can affect the behaviour of stability selection for high correlated designs. Specifically, the stability selection algorithm puts a large emphasis on avoiding false positive selections and, as a consequence, might miss important variables if they are highly correlated with irrelevant variables. Moreover, \cite{battey2018distributed} (see, also \cite{fan2015power}) consider the so-called "divide and conquer" variable selection algorithm as we briefly explain below. 

On each subset $\mathcal{D}_j$, we compute the debiased estimator of $\boldsymbol{\beta}^{*}$ such that
\begin{align}
\widehat{ \boldsymbol{\beta} }^d =  \widehat{ \boldsymbol{\beta} }_{ Lasso }^d ( \mathcal{D}_j ) + \frac{1}{ n_k } M^{ (j) } \left( X^{(j)} \right)^{\top} \left( Y^{(j)} - X^{(j)} \widehat{ \boldsymbol{\beta} }_{ Lasso }^d ( \mathcal{D}_j )  \right), 
\end{align}
where $d$ is used to indicate the debiased version of the estimator, such that $M^{(j)} = \big(  \boldsymbol{m}_1^{(j)},..., \boldsymbol{m}_d^{(j)} \big)^{\top}$ and $\boldsymbol{m}_v$ is the solution of 
\begin{align}
\boldsymbol{m}_v^{(j)} = \underset{ \boldsymbol{m}  }{ \mathsf{argmin} } \ \boldsymbol{m}^{\top} \widehat{\boldsymbol{\Sigma}}^{(j)} \boldsymbol{m}  \ \ \ \text{s.t} \ \ \ \norm{ \widehat{\boldsymbol{\Sigma}}^{(j)} \boldsymbol{m} - \boldsymbol{e}_v }_{ \infty } \leq \vartheta_1, \norm{ \widehat{\boldsymbol{\Sigma}}^{(j)} \boldsymbol{m} } \leq \vartheta_2.  
\end{align}

\newpage

\begin{theorem}
Suppose that $\mathbb{E} \big[ \varepsilon_1^4 \big] < \infty$ and choose $\vartheta_1, \vartheta_2$ and $k$ such that $\theta_1 \approx \sqrt{ k \mathsf{log} d / n }, \theta_2 n^{-1/2} = o(1)$ and $k = o \big( \sqrt{n} ( s \mathsf{log} d )^{-1} \big)$. For any $v \in \left\{ 1,..., d \right\}$, 
\begin{align}
\frac{ \sqrt{n} }{ k } \sum_{j=1}^k \frac{ \widehat{\beta}^d_v ( \mathcal{D}_j ) - \beta_v^{*} }{  \widehat{Q}_v^{(j)} } \Rightarrow \mathcal{N} ( 0, \sigma^2 ), 
\ \ \ \text{where} \ \ \ 
\widehat{Q}_v = \left( \boldsymbol{m}_v^{(j) \top} \widehat{\boldsymbol{\Sigma}}^{(j)} \boldsymbol{m}_v^{(j)}  \right)^{1/2}
\end{align}
\end{theorem}
The above procedure implies a divide-and-conquer Wald statistic of the following form 
\begin{align}
\bar{S}_n 
= 
\frac{ \sqrt{n} }{ k } \sum_{j=1}^k \frac{ \widehat{\beta}^{d_v} ( \mathcal{D}_j ) - \beta_v^H }{  \bar{\sigma} \left(  \boldsymbol{m}_v^{(j \top) } \widehat{\boldsymbol{\Sigma}}^{(j)}  \boldsymbol{m}_v^{(j)} \right)^{1/2} }.      
\end{align}
for $\beta_v^{*}$, where $\bar{\sigma}$ is an estimator for $\sigma$ based on the $k$ subsamples. Consider the desparsified estimator for sub-sample $\mathcal{D}_j$ is given by 
\begin{align}
\widehat{\boldsymbol{\beta}}^d ( \mathcal{D}_j ) = \widehat{\boldsymbol{\beta}}^{\lambda} ( \mathcal{D}_j )
 - 
\widehat{\Theta}^{(j)} \nabla \ell_{ n_k }^{(j)} \left( \boldsymbol{\beta}^{\lambda} ( \mathcal{D}_j ) \right),     
\end{align}
where $\widehat{\Theta}^{(j)}$ is a regularized inverse of the Hessian matrix of second-order derivatives of $\ell_{ n_k }^{ (j) } ( \boldsymbol{\beta} )$ evaluated at $\boldsymbol{\beta}^{ \lambda } ( \mathcal{D}_j )$, denote by 
\begin{align}
\widehat{J}^{ (j) } = \nabla^2  \ell_{ n_k }^{(j)} \left( \widehat{\boldsymbol{\beta} }^{\lambda} ( \mathcal{D}_j ) \right)      
\end{align}
The proposed approach for estimating $\widehat{J}^{ (j) }$, reduces to the empirical covariance of the design matrix in the case of the linear model. The, the aggregated debiased estimator over the $k$ subsamples is 
\begin{align}
\bar{ \boldsymbol{\beta} }^d := \frac{1}{k} \sum_{ j = 1 }^k \widehat{ \boldsymbol{\beta} }^d \left( \mathcal{D}_j \right).     
\end{align}
Thus, to approximate the nodewise Lasso requires to approximately invert $\widehat{J}^{(j)}$ via $L_1 -$regularization. The basic idea is to find the regularized invert row via a penalized $L_1-$regression, which is the same as regressing the variable $X_{ \nu }$ on $\boldsymbol{X}_{ - \nu }$ but expressed in the sample covariance form. In particular for each row $\nu \in \left\{ 1,..., d \right\}$, consider the optimization as below: 
\begin{align}
\widehat{ \boldsymbol{\kappa} }_{ \nu } \left( \mathcal{D}_j \right) = \underset{  \boldsymbol{\kappa}  \in \mathbb{R}^{d-1} }{ \mathsf{argmin} } \left(  \widehat{J}_{ \nu \nu }^{(j)} - 2 \widehat{J}_{ \nu, - \nu }^{(j)} \boldsymbol{\kappa}  + \boldsymbol{\kappa}^{ \top } \widehat{J}_{ - \nu, - \nu }^{(j)} \boldsymbol{\kappa} + 2 \lambda_{ \nu } \norm{  \boldsymbol{\kappa} }_1 \right), 
\end{align}
where $\widehat{J}_{ \nu, - \nu }^{(j)}$ denotes the $\nu-$th row of $\widehat{J}^{(j)}$ without the $( \nu, \nu )-$th diagonal element and $\widehat{J}_{ -\nu, - \nu }^{(j)}$ is the principal submatrix without the $\nu-$th row and $\nu-$th column.

\begin{remark}
Notice that \cite{meinshausen2006high} established a link between the nodewise regression and the optimal linear prediction of excess asset returns under the assumption that the returns are normally distributed. Relevant studies include \cite{callot2021nodewise} and \cite{caner2023sharpe}.  
\end{remark}

\newpage

Introduce the following matrix 
\begin{align}
\widehat{C} := 
\begin{pmatrix}
1 & - \widehat{\kappa}_{1,2} ( \mathcal{D}_j ) & \hdots & - \widehat{\kappa}_{1,d} ( \mathcal{D}_j )
\\
- \widehat{\kappa}_{2,1} ( \mathcal{D}_j )  &  1  & \hdots & - \widehat{\kappa}_{1,d} ( \mathcal{D}_j )
\\
\vdots & \vdots & \ddots & \vdots
\\
- \widehat{\kappa}_{d,1} ( \mathcal{D}_j )  & - \widehat{\kappa}_{d,2} ( \mathcal{D}_j ) & \hdots & 1
\end{pmatrix}.
\end{align}
Denote with 
\begin{align}
\widehat{\Xi}^{(j)} = \mathsf{diag} \big(  \mathcal{\tau}_1 ( \mathcal{D}_j ),...,  \mathcal{\tau}_d ( \mathcal{D}_j ) \big), \ \ \ \text{where} \ \ \  \widehat{\mathcal{\tau} }_{ \nu} \left( \mathcal{D}_j \right)^2 \equiv \widehat{J}_{ \nu \nu }^{ (j) } -  \widehat{J}_{ \nu, - \nu }^{ (j) } \widehat{\boldsymbol{\kappa} }_{\nu} ( \mathcal{D}_j ). 
\end{align}
Moreover, the estimation of the matrix $\widehat{ \Theta }^{(j)}$ is given by 
\begin{align}
\widehat{ \Theta }^{(j)} := \left(  \widehat{\Xi}^{(j)} \right)^{-2} \widehat{C}^{ (j) },      
\end{align}
Therefore, we consider establishing the limit distribution of the following term
\begin{align}
\bar{S}_n := \frac{ \sqrt{n}  }{ k } \sum_{j=1}^k \frac{ \beta_{\nu}^d ( \mathcal{D}_j ) - \beta_{\nu}^H }{  \sqrt{ \Theta^{*}_{\nu \nu} }   }    
\end{align}
for any $\nu \in \left\{ 1,..., d \right\}$ under the null hypothesis that $H_0: \beta_{\nu} = \beta_{ \nu }^H$.  

\begin{remark}
Therefore, the above statistical procedure provides the basis for the statistical testing based on \textit{divide and conquer}. On the other hand, the particular methodology is only implemented for estimation and inference purposes in high dimensional sparse models. In the sparse high dimensional setting, nuisance parameters introduce a nonnegligible bias, causing classical low dimensional theory to break down. On the other hand, in their high dimensional Wald construction, the phenomenon is remedied through a debiasing of the estimator, which gives rise to a test statistic with tractable limiting distribution.     
\end{remark}

\paragraph{Open Problems}

An open problem in econometrics and statistics when considering model comparison and selection methodologies remains the aspect of robust statistical inference for nonnested environments. More specifically, various studies in the literature have considered methodologies for establishing the consistency of the Bayes factor in these settings. A particular interesting solution is the approach of converting a nonnested problem to a nested statistical problem which can provide a pseudo-distance between the base model and the full model. Another approach is to have on overlapping window where the comparisons between the nested models is made (see, \cite{berger1996intrinsic}, \cite{berger1999default}). Moreover, the framework proposed by \cite{vuong1989likelihood} relates the probabilistic model selection approach to the classical nested-hypothesis testing situation. Lastly, in the next section we discuss an application from the statistics literature on sample-splitting to distinguish it from subsampling. Although the particular implementation corresponds to \textit{i.i.d} data, an extension to dependent data (e.g., time series data with a weak form of dependence) could be an interesting fruitful avenue for future research.

\newpage

\subsection{Sample-Splitting and Variable Importance Algorithm}

In this section, we follow and discuss the framework proposed by \cite{rinaldo2019bootstrapping}. Specifically, consider a distribution-free regression framework, where the pair $Z = (X,Y) \in \mathbb{R}^d \times \mathbb{R}$ of $d-$dimensional covariates and response variables has an unknown distribution belonging to a large non-parametric class of $\mathcal{Q}$. Then, using minimal assumptions (assumption-lean inference) on the regression function  
\begin{align*}
x \in \mathbb{R}^d \mapsto \mu (x) := \mathbb{E} \big[ Y | X = x \big]
\end{align*}
where $\mu(x)$ describes the relationship between the vector of covariates and the expected value of the response variable. We observe the set of data observations $\mathcal{D}_n = \left( Z_1,...,Z_n  \right)$ where $Z_i = ( X_i, Y_i ) \in \mathbb{R}^{d+1}$, where $i = 1,...,n$ and the class $\mathcal{Q} = \mathcal{Q}_n$, which may depend on the sample size. Then, we apply to the data a procedure $w_n$, which returns both a subset of the covariates and an estimator of the regression function over the selected covariates. Formally, we have that 
\begin{align}
\mathcal{D}_n \mapsto w_n \left( \mathcal{D}_n \right) = \left( \widetilde{S}, \widetilde{\mu}_{\widetilde{S}} \right), 
\end{align}
where $\widetilde{S}$ is the selected model, is a random, nonempty subset of $\left\{ 1,...,d \right\}$ and $\widehat{\mu}_{\widetilde{S}}$ is an estimator of the regression function $x \in \mathbb{R}^d \mapsto \mathbb{E} \left[ Y | X_{\widetilde{S}} = x_{\widetilde{S}} \right]$ restricted to $\widehat{S}$, where $(X,Y) \sim P$ independent of $\mathcal{D}_n$ and, for a vector $x = \left( x(1),..., x(d)\right) \in \mathbb{R}^d$, we set $x_{\widehat{S}} = \left( x(j), j \in \widehat{S} \right)$. The only assumption that is imposed on $w_n$ is that the maximum size of the selected model is controlled by the experimenter, that is, $1 \leq | \widehat{S} | \leq k$, for a predefined positive integer $k \leq d$, where $k$ and $d$ can both increase with the sample size. Therefore, \cite{rinaldo2019bootstrapping} propose a methodology for defining $\widehat{S}$ such that it contains any optimal model. In particular, their framework allows for arbitrary procedures including sparse variable selection, and stepwise-forward regression. Thus, the goal of the framework proposed by \cite{rinaldo2019bootstrapping}, is to provide statistical guarantees for various measures of variable importance applied to the covariates in $\widehat{S}$, uniformly over the choice of $w_n$ and over all the distributions $P \in \mathcal{Q}_n$. Then confidence sets are considered for four random parameters taking values in $\mathbb{R}^{\widehat{S}}$, each providing a different assessment of the level of statistical significance of the variables in $\widehat{S}$ from a purely predictive standpoint. All the random parameters under consideration are function of the data generating distribution $P$, of the sample $\mathcal{D}_n$ and its size $n$, as well as of the mechanism for model selection and the estimation procedure of $w_n$. 

\subsubsection{The projection parameter $\beta_{\widehat{S}}$} 

Consider a linear estimator of the form $x \mapsto \widehat{\mu}_{\widehat{S}} (x) = \widehat{\beta}^{\top}_{\widehat{S}} x_{\widehat{S}}$, where $\widehat{\beta}_{\widehat{S}}$ is any estimator of the linear regression coefficients for the selected predictors based on OLS. Then, the linear projection parameter $\widehat{\beta}_{\widehat{S}}$ is defined to be the vector of coefficients of the best linear predictor of $Y$ using $X_{\widehat{S}}$, given by 
\begin{align}
\widehat{\beta}_{\widehat{S}} = \underset{ \beta \in \mathbb{R}^{ \widehat{S}} }{ \mathsf{arg \ min} } \ \mathbb{E}_{X,Y} \left( Y - \beta^{\top} X_{ \widehat{S} } \right)^2, 
\end{align}

\newpage

where $\mathbb{E}_{X,Y}$ denote the expectation with respect to the joint distribution of the random variables $\left( X,Y \right)$. 

In other words, based on the framework proposed by \cite{rinaldo2019bootstrapping}, these \textit{projection parameters} correspond to the fact that $X^{\top} \beta_{\widehat{S}}$ is the $L_2$ projection of $Y$ into the linear space of all random variables that can be obtained as linear functions of $X_{\widehat{S}}$. The projection parameter is well-defined even though the true regression function $\mu$ is not linear. Indeed, it is immediate that 
\begin{align}
\beta_{\widehat{S} } = \Sigma^{-1}_{\widehat{S} } Q_{\widehat{S}}, \ \ \ \text{with} \ \ \ Q_{\widehat{S}} := \left( Q_{\widehat{S}}, j \in \widehat{S} \right)
\end{align} 
where 
\begin{align*}
Q_{\widehat{S}} = \mathbb{E}_{X,Y} \left[ Y X_{\widehat{S}} (j) | \mathcal{D}_n \right] \ \ \ \text{and} \ \ \ \Sigma_{\widehat{S}} = \mathbb{E} \left[ X_{\widehat{S}} X_{\widehat{S}}^{\top} | \mathcal{D}_n \right]
\end{align*}

\subsubsection{The LOCO parameters $\gamma_{\widehat{S}}$ and $\phi_{\widehat{S}}$} 

A commonly used measure of the importance of the selected covariates is $\beta_{\widehat{S}}$, but there are of course other ways to quantify variable significance. We consider two parameters of variable importance, which we refer to as \textit{Leave Out COvariate Inference}-or \textit{LOCO}-parameters (see, \cite{rinaldo2019bootstrapping}). 

The first LOCO parameters is $\gamma_{\widehat{S}} = \left( \gamma_{\widehat{S}} (j) : j \in \widehat{S} \right)$, where 
\begin{align}
\gamma_{\widehat{S}} (j) = \mathbb{E} \left[ \big| Y - \widehat{\beta}^{\top}_{\widehat{S}(j)} X_{\widehat{S}(j)} \big|  - \big| Y - \widehat{\beta}^{\top}_{\widehat{S}} X_{\widehat{S}} \big|  \bigg| \mathcal{D}_n \right].
\end{align}  

In the last expression, $\widehat{\beta}_{\widehat{S}}$ is any estimator of the projection parameter $\beta_{\widehat{S}}$ and $\widehat{S}(j)$ and $\widehat{\beta}_{\widehat{S}(j)}$ are obtained by rerunning the model selection and estimation procedure after removing the $j-$th predictor. In other words, for each $j \in \widehat{S}, \widehat{S}(j)$ is a subset of size at most $k$ of $\left\{ 1,...,d \right\} \ \left\{j\right\}$. Notice that the selected model can be different when the $j-$th covariate is held out from the data, so that the intersection between $\widehat{S}(j)$ and $\widehat{S}$ can be smaller than $k-1$. The interpretation of $\gamma_{\widehat{S}}(j)$ is simple: it is the increase in prediction error by not including the $j-$th predictor in the model. For instance, it is easy to extend the definition of this parameter by leaving out several variables from $\widehat{S}$ at once without additional conceptual difficulties. 

Moreover, the parameter $\gamma_{\widehat{S}}$ has advantages over the projection parameter $\beta_{\widehat{S}}$ and this is because it refers directly to prediction error and as also demonstrate the accuracy of the Normal approximation and the bootstrap is much higher. The second type of LOCO parameters that we consider are the median LOCO parameters, $\phi_{\widehat{S}} = \left( \phi_{\widehat{S}}(j), j \in \widehat{S} \right)$ with  
\begin{align}
\phi_{\widehat{S}} (j) = \mathsf{median} \left[ \big| Y - \widehat{\beta}^{\top}_{\widehat{S}(j)} X_{\widehat{S}(j)} \big|  - \big| Y - \widehat{\beta}^{\top}_{\widehat{S}} X_{\widehat{S}} \big|  \bigg| \mathcal{D}_n \right].
\end{align} 
As with $\gamma_{\widehat{S}}$, we may leave out multiple predictors at the same time.

\newpage

\subsubsection{The prediction parameter $\rho_{\widehat{S}}$} Another interesting measure of variable importance is an omnibus parameter that measures how well the selected model will predict future observations. To this end, we define the future prediction error as
\begin{align}
\rho_{\widehat{S}} = \mathbb{E} \left[ \big| Y - \widehat{\beta}^{\top}_{\widehat{S}} X_{\widehat{S} } \big| \bigg| \mathcal{D}_n    \right],
\end{align}
where $\widehat{\beta}_{\widehat{S}}$ is computed based on $\mathcal{D}_n$. In particular, the main idea of the approach proposed by \cite{rinaldo2019bootstrapping}
relies on sample splitting: assuming for notational convenience that the sample size is $2n$, we randomly split the data $\mathcal{D}_{2n}$ into two halves, $\mathcal{D}_{1,n}$ and $\mathcal{D}_{2,n}$. Next, we run the model selection and estimation procedure $w_n$ on $\mathcal{D}_{1,n}$, obtaining both $\widehat{S}$ and $\widehat{\mu}_{\widehat{S}}$. We then use the second half of the sample $\mathcal{D}_{2,n}$ to construct an estimator $\widehat{\theta}_{\widehat{S}}$ and a confidence set $\widehat{C}_{\widehat{S}}$ for $\theta_{\widehat{S}}$ satisfying the following properties:

\begin{itemize}

\item \textbf{Concentration:}
\begin{align}
\underset{ n \to \infty }{ \mathsf{lim \ sup}  } \underset{ w_n \in \mathcal{W}_n }{ \mathsf{sup}  }  \underset{ P \in \mathcal{Q}_n }{ \text{sup}} \mathbb{P} \left( \big\| \widehat{\theta}_{ \widehat{S}} - \theta_{ \widehat{S}} \big\| > r_n \right) \to 0, 
\end{align}

\item \textbf{Coverage Validity:} 
\begin{align}
\underset{ n \to \infty }{ \mathsf{lim \ inf}  } \underset{ w_n \in \mathcal{W}_n }{ \mathsf{inf}} \underset{ P \in \mathcal{Q}_n }{ \mathsf{inf}} \mathbb{P} \left(   \theta_{ \widehat{S}}  \in \widehat{C}_{ \widehat{S}}  \right) \geq 1 - \alpha, 
\end{align}

\item \textbf{Accuracy:}
\begin{align}
\underset{ n \to \infty }{ \mathsf{lim \ sup}  } \underset{ w_n \in \mathcal{W}_n }{ \mathsf{sup}} \underset{ P \in \mathcal{Q}_n }{ \mathsf{sup}} \mathbb{P} \left(   \nu \left( \widehat{C}_{ \widehat{S}} \right) > \epsilon_n \right) \to 0, 
\end{align}
where $\alpha \in (0,1)$ is a prespecified level of significance, $\mathcal{W}_n$ is the set of all the model selection and estimation procedures on samples of size $n$, $r_n$ and $\epsilon_n$ both vanish as $n \to \infty$ and $\nu$ is the volume (Lebsegue measure) of the set. The probability statements above take into account both the randomness in the sample $\mathcal{D}_n$ and the randomness associated to splitting it into halves. 

\end{itemize}

\begin{remark}
The property that the coverage of $\widetilde{C}_{\widetilde{S}}$ is guaranteed uniformly over the entire class $\mathcal{Q}_n$ is known as (asymptotic) honesty. Moreover, the confidence sets are for the random parameters (based on half the data)  but the uniform coverage , accuracy and concentration guarantee to hold with respect to the distribution of the entire sample and the randomness associated to the splitting of the full sample.   
\end{remark}

\begin{remark}
Since we are particularly interested to study suitable variable selection methodologies for time series regression models, notice that usually the subsampling approach in such modelling settings requires to consider the dependence structure of the data. For example, the greater the number of subsamples, the larger the bias due to the fact that consistent estimation of long-memory parameters requires larger samples. Determining the length of subsamples under the presence of long memory in time series data is another aspect of concern. 
\end{remark}

\newpage

\subsubsection{Confidence sets for the projection parameters: The bootstrap}

Following \cite{rinaldo2019bootstrapping}, the confidence set based on the Normal approximation require the evaluation of both the matrix $\widehat{\Gamma}_{\widehat{S}}$ and the quantile $\widehat{t}_{\alpha}$ which may be computationally inconvenient. The authors show that the paired bootstrap can be deployed to construct analogous confidence sets, centered at $\widehat{\beta}_{\widehat{S}}$, without knowledge of $\widehat{\Gamma}_{\widehat{S}}$. Thus, the bootstrap distribution corresponds to the empirical probability measure associated to the subsample $\mathcal{D}_{2,n}$ and conditionally on $\mathcal{D}_{1,n}$ and the outcome of the sample splitting procedure. Denote $\widehat{\beta}_{\widehat{S}}^{*}$ the estimator of the projection parameters $\beta_{\widehat{S}}$ arising from \textit{i.i.d} sample of size $n$ drawn from the bootstrap distribution. For a given $\alpha \in (0,1)$, let $\widehat{t}^{*}_{\alpha}$ be the smallest positive number 
\begin{align}
\mathbb{P} \left( \sqrt{n} \left\| \widehat{\beta}_{\widehat{S}}^{*} - \widehat{\beta}_{\widehat{S}} \right\| \leq \widehat{t}^{*}_{\alpha} \big|  \mathcal{D}_{2,n} \right) \geq 1 - \alpha.
\end{align}  
Next, let $\left( \widetilde{t}_j^{*}, j \in \widehat{S} \right)$ be such that 
\begin{align}
\mathbb{P} \left( \sqrt{n} \left| \widehat{\beta}_{\widehat{S}}^{*}(j) - \widehat{\beta}_{\widehat{S}}(j) \right| \leq \widetilde{t}^{*}_{j}, \ \forall \ j \big|  \mathcal{D}_{2,n} \right) \geq 1 - \alpha.
\end{align}
By the union bound, each $\widetilde{t}^{*}_j$ can be chosen to be the largest positive number such that 
\begin{align}
\mathbb{P} \left( \sqrt{n} \left| \widehat{\beta}_{\widehat{S}}^{*}(j) - \widehat{\beta}_{\widehat{S}}(j) \right| > \widetilde{t}^{*}_{j}, \ \forall \ j \big|  \mathcal{D}_{2,n} \right) \leq \frac{\alpha}{k}.
\end{align}
Consider the following two bootstrap confidence sets:
\begin{align}
\widehat{C}^{*}_{ \widehat{S} } 
&=
\left\{ \beta \in \mathbb{R}^{\widehat{S}} : \left| \beta - \widehat{\beta}_{\widehat{S}} \right| \leq  \frac{\widehat{t}_{\alpha}^{*}}{\sqrt{n} } \right\}, 
\\
\widetilde{C}^{*}_{ \widehat{S} }
&=
\left\{ \beta \in \mathbb{R}^{\widehat{S}} : \left| \beta(j) - \widehat{\beta}_{\widehat{S}}(j) \right| \leq  \frac{\widetilde{t}_{j}^{*}}{\sqrt{n} }, \ \forall \ j \in \widehat{S} \right\}.
\end{align}
It is immediate that $\widehat{C}^{*}_{ \widehat{S} }$ and $\widetilde{C}^{*}_{ \widehat{S}}$ are just the bootstrap equivalent of the confidence sets above. Therefore, the case of sparse fitting where $k = \mathcal{O}(1)$ implies that the size of the selected model is not allowed to increase with $n$. The standard central limit theorem shows that
\begin{align}
\sqrt{n} \left( \widehat{\beta}  - \beta \right) \to \mathcal{N} \left( 0, \Gamma \right), \ \ \Gamma = \Sigma^{-1} \mathbb{E} \big[ \left( Y - \beta^{\top} X \right)^2 \big] \Sigma^{-1}. 
\end{align}
Note that $\Gamma$ can be consistently estimated by the sandwich estimator $\widehat{\Gamma} = \widehat{\Sigma}^{-1}  A  \widehat{\Sigma}^{-1}$, where $A = n^{-1} \mathbb{X}^{\top} R \mathbb{X}$, $\mathbb{X}_{ij} = X_i(j)$, $R$ is the $k \times k$ diagonal matrix with $R_{ii} = \left( Y_i - X_i^{\top} \widehat{\beta} \right)^2$. By Slutsky's theorem, valid asymptotic confidence sets can be based on the Normal distribution with $\widehat{\Gamma}$ in place of $\Gamma$.

However, as \cite{rinaldo2019bootstrapping} points out there is a clear prediction/accuracy trade-off when employing the sample splitting approach for variable selection. In other words, the selected model may be less accurate because only part of the data are used to select the model. Thus, although splitting creates gains in accuracy and robustness for inference it with some loss of prediction accuracy.

\newpage

\paragraph{Open Problems} An interesting application would be to consider the feasibility of using sample splitting as an estimation and inference approach in time series regression models. For instance as in the case, when  quantile regression models are employed for modelling risk measures such as the VaR and the CoVaR. In that case, sample splitting can be seen as a methodology for out-of-sample forecasting of the unknown risk quantities. Thus, developing  robust methodologies for accommodating not only the uncertainty induced by the selection of the optimal size of the out-of-sample period but also to be able to capture correctly the persistence properties of regressors included in the model. Another example, consider the block bootstrap which is commonly  which as a resampling method to preserve the dependence structure in predictive regression models. Thus, in the case of time-varying predictive regression models, essentially one models the effect of time-varying persistence using a rolling window. Combining these sets of information in a meaningful way is crucial in understanding how the persistence properties of predictors affect the asymptotic efficiency of statistics such as the bootstrap confidence set or other methods of variable selection in high-dimensional settings with dependent data.  

\subsection{Multiple Testing Procedure and Variable Selection}

In this section, we consider an alternative methodology proposed in the statistical literature especially, since the pioneered work of Abraham Wald on \textit{Sequential tests of statistical hypotheses}. Towards the direction of the multiple testing approach applied to high dimensional regression models for variable selection and statistical inference purposes, a relevant framework is proposed by \cite{chudik2018one}. Therefore, in this section we discuss their proposed statistical methodology, so-called \textit{One Covariate at a Time Multiple Testing} (OCMT) (see, also \cite{zhang1993model}). Specifically, the particular procedure is computationally simple and fast even for extremely large data sets (see, also \cite{romano2005exact}). 

Thus, when a Lasso penalization is allowed then the estimator of $\boldsymbol{\beta}$ is obtained via  
\begin{align}
\hat{\boldsymbol{\beta}} = \underset{ \boldsymbol{\beta} \in \mathbb{R}^p }{ \mathsf{arg \ min} } \sum_{t=1}^T \left\{ \big( y_t - x_{nt}^{\prime} \boldsymbol{\beta} \big)^2 + P_{\lambda}(\boldsymbol{\beta}) \right\}.
\end{align}
A certain degree of sparsity is required in order to apply standard penalized linear models in high-dimensional settings. However, the methodology does not require that the regressor vector $\boldsymbol{x}_{nt}$ to have a sparse covariance matrix, and therefore it is still applicable even if the covariance matrix of the noise variables, is not sparse. The particular algorithm can be thought as a model selection device since the main idea of this procedure is to test the statistical significance of the net contribution of all $n$ available potential covariates in explaining $y_t$ individually, while accounting for the multiple testing nature of the problem under consideration. However, the OCMT procedure is not sequential and selects in a single step all covariates whose $t-$ratios exceed a given threshold. Moreover, using the OCMT procedure post-selection is only applied when there are still covariates whose net contribution to $y_t$ is zero, despite the fact that they belong to the true model for $y_t$. From the statistics perspective \textit{sequential model selection} is discussed by \cite{fithian2015selective}. These approaches construct regression models by selecting variables from active sets, based on a sequence of sets.

\newpage

\subsubsection{Rejection Principle of Familywise Error Control}

From the statistical theory perspective the $\textit{partitioning principle}$ is a powerful tool in multiple decision theory. The following theorem explains the concept of \textit{familywise error rejection}, (FWER). 

\medskip

\begin{theorem}[\cite{finner2002partitioning}]
Let $\alpha, \gamma \in (0,1)$. If the null hypothesis $H_{0,j} : \beta_j = 0$ gets rejected whenever $Q_j ( \gamma ) \leq \alpha$, then the FWER is asymptotically controlled at level $\alpha$, that is, 
\begin{align}
\underset{ n \to \infty }{ \mathsf{lim \ sup} } \ \mathbb{P} \bigg( \underset{ j \in \mathbb{N} }{ \mathsf{min} } \ Q_j ( \gamma ) \leq \alpha \bigg) \leq \alpha.    
\end{align}
\end{theorem}
Closed testing and \textit{partitioning} are recognized as fundamental principles of \textit{familywise error control}, (FWE). In practice, various multiple testing procedures that control the FWE are often sequential, in the sence that rejection of some of the hypotheses may make rejection of the remaining hypothesis easier. Sequential rejective procedures in the literature include that of \cite{romano2005exact} (see, also \cite{politis1999subsampling}). Moreover, the paper of \cite{goeman2010sequential} presents a unified approach to the class of sequentially rejective multiple testing procedures, emphasizing the sequential aspect. Thus, the authors consider a general sequentially rejective procedure as a sequence of single-step methods, determined by a rule for setting the rejection regions for each null hypothesis based on the current collection of unrejected regions for each null hypothesis based on the current collection of unrejected null hypotheses. Thus, the sequential rejection principle, implies that a single-step familywise error controlling procedure is turned into a sequential one which is a general principle of familywise error control. 

\medskip

\begin{theorem}[Sequential rejection principle, \cite{goeman2010sequential}]
\

Suppose that for every $\mathcal{R} \subseteq \mathcal{S} \subset \mathcal{H}$, almost surely, $\mathcal{N} (\mathcal{R}) \subseteq (\mathcal{S}) \cup \mathcal{S}$ and that for every $M \in \mathbb{M}$, 
\begin{align}
\mathbb{P} \big(  \mathcal{N} ( \mathcal{F}(M) ) \subseteq\mathcal{F}(M) \big) \geq 1 - \alpha.    
\end{align}
Then, for every $M \in \mathbb{M}$, $\mathbb{P} \big( \mathcal{R}_{\infty} \subseteq \mathcal{F}(M) \big) \geq 1 - \alpha$.
\end{theorem}

\begin{remark}
The first condition - \textit{monotonicity condition}, guarantees that no false rejection in the critical case (during the single-step), which implies no false rejection in situations with fewer rejections than in the critical case so that type I error control in the critical case is sufficient for overall FWE control of the sequential procedure. The second condition above - \textit{single-step condition}, guarantees FWE control when we have rejected all false null hypotheses and none of the true ones. 
A general admissibility criterion is the case of \textit{restricted combinations}\footnote{A standard example concerns testing pairwise equality of means in a one-way ANOVA model: if any single null hypothesis is false, it is not possible that all other null hypotheses are simultaneously true (see, also \cite{vesely2021permutation}).} can be also constructed. Restricted combinations occur if, for some $\mathcal{R} \subseteq \mathcal{H}$, there is no model $M \in \mathbb{M}$ such that $\mathcal{R} = \mathcal{F} (M)$ (see, \cite{goeman2010sequential}).
\end{remark}

\newpage

\begin{itemize}
   
\item \textbf{Resampling-based multiple testing procedures} use resampling techniques to let the multiple testing procedure estimate or accommodate the actual dependence structure between the test statistics. For example, resampling of a test statistic under the complete null hypothesis, using permutations or the bootstrap, can give consistent estimates of the desired quantiles. However, regardless of the underlying assumptions, consistent estimation of the quantiles of $\mathsf{max}_{ H \in \mathcal{T}(M) } S_{H}$ only guarantees control of the familywise error in an asymptotic sense. Resampling-based methods with exact FWE control, imply that one can obtain control of the FWE by generalizing the treatment of permutation testing to a multiple testing procedure. In order to define a resampling-based sequentially rejective multiple testing procedure with exact familywise error control, we choose a set $\boldsymbol{\pi} = \left\{ \pi_1,..., \pi_r \right\}$ of $r$ functions that we shall refer to as null-invariant transformations. 

\item \textbf{Graph-based procedures.} Another application of the \textit{sequential rejection principle}, which is of interest in the statistics, econometrics and finance literature, is the development of multiple testing\footnote{Notice that for multiple testing procedures researchers are also interested in reporting multiplicity-adjusted p-values. Such multiplicity-adjusted p-values are defined for each null hypothesis as the smallest $\alpha-$level that allows rejection of that hypothesis (see also the method proposed by \cite{mckeague2015adaptive} and \cite{huang2019marginal}).} procedures for graph-structured hypotheses. Specific procedures for controlling the familywise error for graph-structured hypotheses have been proposed by several authors. 

\end{itemize}

\begin{example}[Neighborhood Selection with the Lasso,  \cite{meinshausen2006high}]
\

Consider the $p-$dimensional multivariate normal distributed random variable $X = \left( X_1,..., X_p \right) \sim \mathcal{N} \left( \mu, \Sigma \right)$. This includes the Gaussian linear models where $X_1$ is the response variable and $\left\{ X_k, 2 \leq k \leq p \right\}$ are the predictor variables. The conditional independence structure of the distribution can be represented by a graphical model $\mathcal{G} = ( \mathcal{V}, \mathcal{E} )$, where  $\mathcal{V} = \left\{ 1,...,p \right\}$ is the set of nodes and $\mathcal{E}$ is the set of edges in $\mathcal{V} \times \mathcal{V}$. A pair $( a, b )$ is contained in the edge set $\mathcal{E}$ if and only if $X_a$ is conditionally dependent on $X_b$, given all remaining variables $X_{  \mathcal{V} \backslash \left\{ a, b \right\}}$. Every pair of variables not contained in the edge set is conditionally independent, given all remaining variables, and corresponds to a zero entry in the inverse covariance matrix. When predicting a variable $X_a$ with all remaining variables $\left\{ X_k ; k \in \Gamma(n) \backslash \left\{ a \right\} \right\}$, the vanishing Lasso coefficient estimates identify asymptotically the neighborhood of a node $a$ in the graph. Let the $n \times p(n)-$dimensional matrix $\mathbf{X}$ contain $n$ independent observations of $X$, so that the columns $\mathbf{X}_a$ correspond for all $a \in \Gamma(n)$ to the vector of $n$ independent observations of $X_a$. Let $\langle ., . \rangle$ be the usual inner product on $\mathbb{R}^n$ and $\norm{ \ . \ }_2$ the corresponding norm. The Lasso estimate $\hat{\theta}^{a, \lambda}$ of $\theta^a$ is given by 
\begin{align}
\hat{\theta}^{a, \lambda} = \underset{ \theta : \theta_a = 0  }{ \mathsf{arg \ min} } \left( n^{-1} \norm{ \mathbf{X}_a - \mathbf{X} \theta }_2^2  + \lambda \norm{ \theta }_1  \right), \ \ \norm{ \theta }_1 = \sum_{ b \in \Gamma(n) } | \theta_b |,
\end{align}
where $\ell_1-$norm of the coefficient vector.   
Furthermore, \cite{salgueiro2005power} and \cite{salgueiro2006power} propose a formal statistical framework for edge exclusion in graphical Gaussian models. Moreover, statistical theory for constructing relevant matrix moments functions is presented by  \cite{roverato1998isserlis}. Recently, \cite{fan2020projection} propose a conditional dependence measure in high-dimensional \textit{undirected graphical models} (UGM). In particular, the UGM approach examines the internal conditional dependency structure of a multivariate random vector.
\end{example}

\newpage

\subsubsection{Subgroup Selection Methodology}

Lastly, we briefly discuss the subgroup selection methodology proposed by \cite{reeve2021optimal}\footnote{Dr. Henry Reeve gave a seminar with title: "Subgroup Selection in nonparametric regimes", at the S3RI Departmental Seminar Series at the University of Southampton on the 10th of November 2022.}. In  \textit{subgroup selection}, the main objective of the statistician is to leverage the underline structure dynamics in data to identify a subset $\hat{\mathcal{A}}$ of the population to treat. Suppose that we have a distribution $P$ on a covariate-response pairs $(X,Y)$ in $\mathbb{R}^d \times \mathbb{R}$. Let $\mu := \mu_P$ denote the marginal distribution of the covariate $X \in \mathbb{R}^d$. Let $\eta : \mathbb{R}^d \to \mathbb{R}$ be the regression function defined by $\eta(x) := \mathbb{E} \big[ Y | X = x \big]$ for some $x \in \mathbb{R}^d$.  Then, one would like to select a subgroup $A \subset \mathbb{R}^d$ such that $\eta$ is above a user-specified threshold $\tau \in \mathbb{R}$ on $A$. Hence, we are interested in subsets of the $\tau-$super level set such that the following set holds
\begin{align}
\mathcal{X}_{\tau} (\eta) := \big\{ x \in \mathbb{R}^d: \eta(x) \geq \tau \big\}.    
\end{align}
In other words, the statistical mechanism of \cite{reeve2021optimal} implies that the user chooses a \textit{data-dependent subgroup} $\hat{A} \equiv \hat{A} ( \mathcal{A} )$, which is a random subset of $\mathbb{R}^d$, taking values in $\mathcal{A}$. Specifically, the practitioner has access to a sample $\mathcal{D} := \big\{ (X_1,Y_1),..., (X_n,Y_n) \big\} \overset{ \textit{i.i.d} }{ \sim } P$. Therefore, from the objective function perspective the aim is to select a data-dependent subgroup $\hat{A} ( \mathcal{D} ) \subset \mathcal{X}_{\tau} (\eta)$, with high-probability where $\mathcal{X}_{\tau} (\eta) := \big\{ x \in \mathbb{R}^d: \eta(x) \geq x \big\}$.

\medskip

\begin{proposition}[Type 1 error guarantee,\cite{reeve2021optimal}]
Let $\mathcal{P}$ be a family of distributions $P$ on pairs $(X,Y)$ and choose a significance level $\alpha \in (0,1)$. We say that the data-dependent subgroup $\hat{A}$ controls Type 1 error at the level $\alpha$ over the class $\mathcal{P}$ if
\begin{align}
\underset{ P \in \mathcal{P} }{ \mathsf{inf} } \ \mathbb{P}_P \left(  \hat{A}(\mathcal{D}) \subseteq \mathcal{X}_{\tau} (\eta) \right) \geq 1 - \alpha.   
\end{align}
\end{proposition}
Therefore, our objective is to choose $\hat{A} \equiv \hat{A} ( \mathcal{D} )$ which minimizes regret $R_{\tau} ( \hat{A} )$, subject to 
\begin{align*}
\underset{ P \in \mathcal{P} }{ \mathsf{inf} } \ \mathbb{P}_P \left(  \hat{A}(\mathcal{D}) \subseteq \mathcal{X}_{\tau} (\eta) \right) \geq 1 - \alpha.       
\end{align*}

\begin{definition}[Holder Class, \cite{reeve2021optimal}]
Given that $( \beta, C_S ) \in (0,1] \times [1,\infty)$, we let $\mathcal{P}_{Hol, \tau} ( \beta, C_S )$ denote the class of all distributions $P$ on $\mathbb{R}^d \times [0,1]$ with marginal $\mu$ on $\mathbb{R}^d$ such that the regression function $\eta$ is $( \beta, C_S )-$Holder on $\mathcal{X}_{\tau} ( \eta) \cap  \mathsf{supp} ( \mu )$ in the sence that 
\begin{align}
\left| \eta( x^{\prime} ) - \eta (x) \right| \leq C_X . \norm{ x^{\prime} - x }_{\infty}^{\beta},     
\end{align}
for all $x, x^{\prime} \mathcal{X}_{\tau} ( \eta) \cap  \mathsf{supp} ( \mu )$.
\end{definition}
Further details on the implementation of this methodology as well as related statistical theory and numerical illustrations are presented in \cite{reeve2021optimal}. More recently the isotonic subgroup selection approach is proposed by \cite{muller2023isotonic}.

\newpage 

\begin{remark}
Notice that according to \cite{reeve2021optimal}, the subgroup selection approach requires from the statistician to use both FWE control as well as the selection step. In other words, the subgroup selection approach corresponds to a data-dependent selection methodology, such that $\hat{\mathcal{A}} = \hat{\mathcal{A}} (D)$. Therefore, a special class of functions needs to be considered which allows to test for local null hypothesis. 
\end{remark}

Thus, the statistical problem corresponds to minimizing the regret that controls the type I error. Although one in practice needs to check the exact set that corresponds to the minimized regret. Therefore, the following expression holds: 
\begin{align}
\underset{ p \in P }{ \mathsf{inf} } \ \mathbb{P}_{P} \left(  \hat{\mathcal{A}} (D) \subset \mathcal{X}_{\uptau} (n) \right) \geq 1 - \alpha.
\end{align}
Then, the Type I error guarantee implies that the statistical subgroup selection procedure is constructed by sequential partial ordering of multiple covariates. Then, the statistician chooses only the one null hypothesis with the highest partial ordering. In other words, using a graph with acyclic structure allows to expand the multiple testing procedures (the null hypothesis of all ancestors is true). Thus, if the null at that node is false, then the null hypothesis associated with all ancestors is also false.

\newpage

\subsection{Model Selection in Cointegrating Regressions}

According to \cite{mendes2011model}, usually under the presence of cointegration in the VAR system, which implies that there are $(k-m)$ linear combinations of $X_t$ which are $I(0)$, a bias analysis is necessary to obtain a robust econometric estimation methodology. To do this, we consider the $m$ common stochastic trends and at most $( k - m )$ cointegrating relations amongst the $I(1)$ components are represented by the same generating process as $\Delta X_{t-1} = ( A - I_k ) X_{t-1} + \varepsilon_t$, where $\epsilon_t \sim \mathcal{N} ( \boldsymbol{0}, \boldsymbol{\Omega} )$. In particular, these $m$ components are represented by $m$ unit roots, and the $(k-m)$ stable roots of the stochastic difference equation represent the $I(0)$ components of $X_t$ as well as cointegrating relations between the $I(1)$ variates of $X_t$ (see, \cite{abadir1999influence}). 

\begin{example}
Consider the following cointegration regression model given by the expression
\begin{align}
 y_t = \alpha_0 + \beta_0^{\prime} x_t + \gamma_0^{\prime} z_t + u_t   
\end{align}
where $\beta_0$ is an $n_1 \times 1$ vector of parameters and $\gamma_0$ is an $n_2 \times 1$ vector of parameters. Furthermore, the process $\left\{ x_t \right\}_{t=1}^{\infty}$ satisfies the following integrated stochastic process
\begin{align}
x_t = x_{t-1} + v_t,     
\end{align}
Therefore, the main aim shall be to control the number of $I(0)$ variables in the model and we assume that $n \equiv n_1 + n_2$ to possibly be grater than $T$, but only  a fraction of those coefficients are in fact nonzero. Furthermore, we assume without loss of generality that each coefficient vectors can be partitioned into zero and non-zero coefficients, such that $\beta_0 = \big( \beta_0 (1)^{\prime}, \beta_0 (2)^{\prime} \big)^{\prime}$ and $\gamma_0 = \big( \gamma_0 (1)^{\prime}, \gamma_0 (2)^{\prime} \big)^{\prime}$, with all non-zero coefficients stacked first, where $\beta_0(1)$ is $q_1 \times 1$ and $\gamma_0(1)$ is $q_2 \times 1$. In particular, we assume that the number of non-zero coefficients, measured by $q_1$, is fixed (does not depend on $T$), while the number of zero coefficients, measured by $q_2$, may depend on $T$, also set $q = ( q_1 + q_2 )$.   

Then, the Adaptive Lasso estimate is given by  (see, also \cite{medeiros2017adaptive})
\begin{align}
\left( \hat{\beta}, \hat{\gamma} \right) = \underset{ \beta, \gamma }{ \mathsf{arg min}  } \norm{ Y - X \beta - Z \gamma }_2^2 + \lambda_1 \sum_{j=1}^n \lambda_{1j} | \beta_j | + \lambda_2 \sum_{j=1}^n \lambda_{2j} | \gamma_j |,    
\end{align}
\end{example}

\begin{lemma}[KKT Conditions] The solutions $\hat{\beta} = \big( \hat{\beta}(1)^{\prime}, \hat{\beta}(2)^{\prime}  \big)$ and $\hat{\gamma} = \big( \hat{\gamma}(1)^{\prime}, \hat{\gamma}(2)^{\prime} \big)$ to the minimization problem above exists if: 
\begin{align}
\frac{ \partial \norm{ Y - X \beta - Z \gamma }_2^2 }{ \partial \beta_j }  \bigg|_{ \beta_j(1) = \hat{\beta}_j(1) }
&= \mathsf{sgn} \left( \hat{\beta}_j(1) \right) \lambda_1 \lambda_{1j} 
\\
\frac{ \partial \norm{ Y - X \beta - Z \gamma }_2^2 }{ \partial \beta_j }  \bigg|_{ \gamma_j(1) = \hat{\gamma}_j(1) }
&= \mathsf{sgn} \left( \hat{\gamma}_j(1) \right) \lambda_2 \lambda_{2j} 
\end{align}
\end{lemma}

\newpage 

\subsubsection{Model Selection Consistency and Oracle Property}

We focus on showing that under certain conditions on $n, p$ and $\lambda'$s the Adaptive Lasso selects the correct subset of variables -  \textit{sign consistency} and it has the \textit{oracle property}, meaning that our estimate has the same asymptotic distribution of the OLS as if we knew a priori the selected model and at optimal rate.   
\begin{lemma}
Let
\begin{align}
\Omega_{\infty} &= 
\begin{pmatrix}
\Omega_{X, \infty} & \boldsymbol{0}
\\
\boldsymbol{0}^{\prime} & \Omega_{X, \infty} 
\end{pmatrix}
\\
\Omega_{X, \infty} &= \int_0^1 B_X(1) B_X(1)^{\prime} (r) dr \ \ \ \text{and} \ \ \ \Omega_{Z, \infty} = \Sigma_{Z(1)^2},    
\end{align}
where for any $0 \leq r \leq 1$, $B_{X(1)}(r) = \underset{ T \to \infty }{ \mathsf{lim} } \sum_{t=1}^{ \floor{rT} } \boldsymbol{v}_t(1)$. Similarly, split the matrix $\Omega_{11}$ into the partition
\begin{align}
\Omega_{11} =
\begin{pmatrix}
\Omega_{X(1)^2} & \Omega_{Z(1) X(1)}
\\
\Omega^{\prime}_{Z(1) X(1)} & \Omega_{Z(1)^2}
\end{pmatrix}
= 
\begin{pmatrix}
\frac{1}{T^2} X(1)^{\prime} X(1) & \frac{1}{T^{3/2} } Z(1)^{\prime} X(1)
\\
\frac{1}{T^{3/2}} X(1)^{\prime} Z(1) &  \frac{1}{T} Z(1)^{\prime} Z(1)
\end{pmatrix}.
\end{align}
\end{lemma}
Hence, 
\begin{align*}
\mathbb{P} \big( \mathcal{A}^c_T (X) \big) 
&= 
\mathbb{P} \left( \left\{  \left[ \big| T^{-1} \Omega_{X, \infty}^{-1} X(1)^{\prime} U \big| \right]_j > T | \beta_{0j} | \right\}, j = 1,...,q_1 \right) + o_p(1), 
\\
&\leq 
\sum_{j=1}^{q_1} \mathbb{P} \left( \left[ \big| T^{-1} \Omega_{X, \infty}^{-1} X(1)^{\prime} U \big| \right]_j > T | \beta_{0j} | \right) + o_p(1),  
\\
&\leq 
\frac{ q_1 }{ T^2 \beta^2_{*} } \ \underset{ 1 \leq j \leq q_1 }{ \mathsf{max} } \ \mathbb{E} \left( \left[  T^{-1} \big| \Omega_{X, \infty}^{-1} X(1)^{\prime} U \big| \right]^2_j  \right)
\to 0. 
\end{align*}
Further aspects of consideration in the more recent literature include the development of robust frameworks for time series regressions within a high-dimensional environment for the purpose of estimation, inference and forecasting (see, \cite{gupta2019robust}, \cite{baillie2022robust}). Moreover, the development of a framework that allows for the use of an ultra-high dimensional environment in nonstationary time series models can be useful when considering statistical properties such as model selection consistency and hypothesis testing accuracy. An econometric framework for high-dimensional quantile predictive regressions is proposed by \cite{fan2023predictive}. Various studies consider  methods for variable screening in high dimensional linear models, however less attention is paid on how these methodologies perform in the case of nonstationary time series models, especially when for the conditional quantile functional form. Relevant research aspects include the development of a statistical mechanism for dimension reduction or screening methodology in the  presence of a high-dimensional vector of possibly nonstationary predictors, which can improve forecasting performance (see, \cite{pitarakis2023novel} and \cite{gonzalo2023out}). Moreover, choosing variables from irrelevant cointegrating relations is another important issue (see, \cite{khalaf2020simulation} and \cite{richard2023model}).

\newpage

\subsubsection{Model Selection and Rank of a Matrix}

In many econometric and statistic applications knowing the rank of a matrix is crucial to ensure robust inference and testing. For example, from the financial economics literature one can test the implications of Arbitrage Pricing Theory by testing the corresponding rank restrictions. Moreover, the identification of parameters in econometric models depends on conditions and restrictions regarding the rank of related moment matrices such as the Jacobian matrix.

Specifically, various tests for the rank of a matrix can be found in the literature. In particular, these testing methodologies are constructed under the assumption that the unrestricted matrix estimator has a kronecker covariance matrix, however this approach can be sensitive to the ordering of the variables (see,  \cite{kleibergen2006generalized}). Furthermore, \cite{cragg1997inferring} propose suitable testing procedures for determining the rank of a matrix. In particular, the authors examine the use of model selection criterion and sequential hypothesis testing methods to estimate the rank consistently. On the other hand, \cite{kleibergen2006generalized} propose a novel rank statistic which uses a $\sqrt{n}-$consistent estimator of the unrestricted matrix which does not have to have a kronecker covariance matrix. To do this, the authors decompose the estimator of the unrestricted matrix using the spectral vector decomposition. Testing for the rank of the matrix can be also used in cointegration testing for non-stationary time series models. Specifically, in this case the limiting distribution is found to be functional of Brownian motions and is equal to the asymptotic distribution of the \cite{johansen1991estimation} trace test. 

\medskip

\begin{example}
Consider the following ECM model
\begin{align}
\Delta x_t = \Pi x_{t-1} + \sum_{i=1}^{k-1} \Psi_i \Delta x_{t-i} + \Phi d_t + \epsilon_t 
\end{align}
where $\Pi = \alpha \beta^{\prime}$ and $\beta$ is the $p \times r$ matrix with the cointegration vectors and $\alpha$ is the $p \times r$ matrix of adjustment coefficients. The number of long-rank relations is equal to the rank of $\Pi$, which is called the cointegration rank. Both $\alpha$ and $\beta$ are full rank matrices. Moreover, the symmetric $p \times p$, $\Psi_i$ matrix governs the short-term dynamics of the system. 
\end{example}

\newpage

\section{Statistical Learning Methods in Time Series Analysis}

\subsection{Motivation}

In this section we discuss some key applications of statistical learning methodologies in time series analysis. In particular the use of Neural Networks is now widely spread in the finance, economics and actuarial statistics fields. Some applications worth mentioning  include: 

\begin{itemize}

\item[\textit{(i)}] asset pricing modelling (see, \cite{feng2018deep}, \cite{guijarro2021deep}, \cite{fan2022structural}, \cite{chen2023deep} and \cite{caner2023deep}). Relevant applications of the Lasso shrinkage in asset pricing and finance theory include the studies of \cite{feng2020taming} and \cite{chinco2019sparse}.

\item[\textit{(ii)}] heterogeneity in mortality modelling(see, \cite{pitacco2019heterogeneity}). A relevant question for the latter is: "\textit{What is the link between heterogeneity of unobservable factors and mortality deceleration?}".

\item[\textit{(iii)}] credit risk and correlated defaults modelling in financial markets (see, \cite{angelini2008neural} and \cite{bhatore2020machine}). 

\end{itemize}
Furthermore, on the aspect of modelling heterogeneity in mortality, more specifically the statistician can employ a functional form represented by $f$ ("biometric" function) to represent the age pattern of mortality. Then, the following representation follows:
\begin{align}
f = w_1 \cdot f^{(1)} + w_2 \cdot f^{(2)} + ... + w_N \cdot f^{(N)} \ \ \ \text{and} \ \ \ \mu_x = \frac{ a e^{ \beta x } }{ \delta e^{ \beta x } + 1  },
\end{align}
where $\big\{ f^{(1)}, f^{(2)},..., f^{(N)}     \big\}$ represent risk factors (such as health status, occupation etc.) in order to capture the non-linear effects of mortality over time. In particular, these risk factors are considered as contributing factors that worsens the mortality level or probability of an individual that can be incorporated into a rating system. Therefore, splitting individuals with similar risk into groups allows to gain more information about uncertainty in heterogeneity\footnote{Professor Ermanno Pitacco gave a seminar with title: "Heterogeneity in mortality: A survey with an actuarial focus" at the S3RI Departmental Seminar Series at the University of Southampton on 24 of September 2018.} (which is a well-known principle in the statistical literature of credibility theory and risk premium models). Moreover, according to \cite{pitacco2019heterogeneity} a rating system can be constructed based on the form $q_{t + x}^{spec} = q_{t+h} \cdot \left( 1 + \sum_{j=1}^r \rho^{(j)} \right)$, in which case peak mortality is achieved after entering disable stage, such that $q_{t + x}^{spec} = q_{t+x} \cdot \Delta ( x + h; \alpha, \gamma )$. Then, the unobservable risk factors can be captured via the form $\mu_{x}^{spec} = \Phi \left( \mu_{x+t}; z_{x,t} \right)$; known as random heterogeneity.

The aforementioned examples provide some brief illustrations of the various applications from the econometrics and statistics literature that can motivate the investigation of relevant research questions both from the theoretical as well as the empirical perspective when we consider statistical learning methods for time series analysis purposes. Further resources related to deep learning theory and applications can be found in  \cite{goodfellow2016deep} (see, also \cite{lecun2015deep}).


\newpage 

\subsection{Non-Asymptotic Probability Theory}

One of the main goals of learning theory is the development of stability bounds of algorithmic procedures. Following the framework of \cite{mohri2010stability} we employ the following definitions. 

\medskip

\begin{definition}
A learning algorithm is said to be (uniformly) $\widehat{\beta}-$stable if the hypotheses it returns for any two training samples $S$ and $S^{\prime}$ that differ by removing a single point satisfy
\begin{align}
\forall \ z \in X \times Y, \ \ \ \left| c( h_S, z ) -  c( h_S^{\prime}, z ) \right| \leq \widehat{\beta}.     
\end{align}
\end{definition}

\begin{remark}
Notice that a $\widehat{\beta}-$stable algorithm is also stable with respect to replacing a single point. Let $S$ and $S_i$ be two sequences differing in the $i-$th coordinate, and $S_{|i}$ be equivalent to $S$ and $S_i$ but with the $i-$th point removed. then, for a $\widehat{\beta}-$stable algorithm we have that
\begin{align}
\left| c( h_S, z ) -  c( h_S^{\prime}, z ) \right| \leq 2 \widehat{\beta}.     
\end{align}
\end{remark}
Specifically, the use of stability allow us to derive generalization bounds and an exponential concentration bounds of the following form 
\begin{align}
\mathbb{P} \left( \left| \Phi - \mathbb{E} \left[ \Phi \right] \right| \geq \epsilon \right) \leq \mathsf{exp} \left\{ - \frac{m \epsilon^2 }{ \tau^2 } \right\},    
\end{align}
where the probability is over a sample of size $m$ and where $\frac{\tau}{m}$ is the Lipschitz parameter of $\Phi$, with $\tau$ a function of $m$.  In the first section below, we consider the problem of shallow-vs-deep expressiveness from the perspective of approximation theory and general spaces of functions having derivatives up to certain order (Sobolev-type spaces). Specifically, in this framework the problem of expressiveness is very well studied in the case of shallow networks with a single hidden layer, where it is known, in particular, that to approximate a $C^n-$function on a $d-$dimensional set with infinitesimal error $\epsilon$ one needs a network of size about $\epsilon^{ - d / n}$, assuming a smooth activation function.  

\begin{definition}[Lipschitz functions] Consider the class of Lipschitz functions such that
\begin{align*}
\mathcal{F}_L := \big\{ \mathsf{g}: [0,1] \to \mathbb{R} | \mathsf{g}(0) = 0 \ \text{and} \ \left| \mathsf{g}(x) - \mathsf{g^{\prime}} \right| \leq L | x - x^{\prime} | \ \forall \ x, x^{\prime} \in [0,1] \big\}.     
\end{align*}
where $L > 0$ is a fixed constant, and all of the functions in the class obey the Lipschitz bound uniformly  condition over all of $[0,1]$.     
\end{definition}

\begin{remark}
Learning theory and non-asymptotic probability theory is useful for understanding the local behaviour of statistical learners for a class of functions based on regularity conditions.    Further applications include the aspect of robustness and generalization for metric learning\footnote{Dr. Xiaochen Yang gave a seminar with title: "Towards better robustness and generalisation of metric learning methods", at the S3RI Departmental Seminar Series at the University of Southampton on the 3rd of November 2022.} (see, \cite{yang2022toward}). 
\end{remark}

\newpage

\subsection{Shallow Neural Network Estimate learned by Gradient Descent}

\begin{definition} A shallow neural network with one output is a function $f : \mathbb{R}^d \to \mathbb{R}$ of the form
\begin{align}
f( \boldsymbol{x} ) = \sum_{j=1}^m c_j \sigma \left( \boldsymbol{w}_j^{\top} \boldsymbol{x} + v_j   \right), \ \ \ \boldsymbol{w}_j \in \mathbb{R}^d, \ v_j, c_j \in \mathbb{R},    
\end{align}
where $\sigma: \mathbb{R} \to \mathbb{R}$ is the activation function (see, \cite{braun2019rate} and \cite{braun2021smoking}). 
\end{definition}
Consider functions of the form 
\begin{align}
f(x) = \sigma \left( \sum_{j=1}^d w_j . x^{ (j)} + w_0 \right), \ \ \text{where} \ \ x = \left( x^{(1)},..., x^{(d)} \right)^{\top} \in \mathbb{R}^d.    
\end{align}
where $w_0,...,w_d \in \mathbb{R}$ the weights of the neuron and $\sigma: \mathbb{R} \to \mathbb{R}$ the activation function. \textit{Shallow Neural Networks} have only one hidden layer where a simple linear combination of neurons is used to define a function $f: \mathbb{R}^d \to \mathbb{R}$ by
\begin{align}
f(x) = \sum_{k=1}^K \alpha_K . \sigma \left( \sum_{j=1}^d \beta_{k,j} . x^{ (j)} + \beta_{k,0} \right) + \alpha_0. 
\end{align}
Denote the weight between neuron $j$ in layer $(s-1)$ and neuron $i$ in layer $s$ by $w_{i,j}^{(s)}$. This leads to the following recursive definition of a neural network with $\mathcal{L}$ layers and $k_s$ neurons in layer $s \in \left\{ 1,..., \mathcal{L}   \right\}$ such that the following representation applies:
\begin{align}
f(x) = \sum_{i=1}^{ k_{\mathcal{L}} } w_{1,i}^{ ( \mathcal{L} ) } f_i^{ ( \mathcal{L} ) } (x) + w_{1,0}^{ ( \mathcal{L} ) }     
\end{align}
for some $w_{1,0}^{ ( \mathcal{L} ) },..., w_{1, k_L }^{ ( \mathcal{L} ) }$ and for $f_i^{ ( \mathcal{L} ) }$'s recursively defined by 
\begin{align}
f_i^{(s)}(x) = \sigma \left( \sum_{j=1}^{ k_s - 1 } w_{i,j}^{ (s-1)} f_j^{(s-1)}(x)  + w_{i,0}^{ (s-1)} \right). 
\end{align}
Main results in the literature show that NNs can achieve dimension reducion provided the regression function is a composition of sums of functions, where the input dimension of each of the functions is at most $d^{*} < d$. Denote by $f_{\mathsf{net},w}$ the neural network with weight vector $\boldsymbol{w} = \left( w_{j,k}^{(s)} \right)_{ s= 0,.., \mathcal{L}, j=1,..., k_{s+1}, k = 0,..., k_s }$ and set the following function 
\begin{align}
F( \boldsymbol{w} ) = \frac{1}{n} \sum_{i=1}^n \left| Y_i -  f_{\mathsf{net},w} \right|^2, \ \ \boldsymbol{w}(0) = \boldsymbol{v}, 
\end{align}
for some randomly chosen initial vector $\boldsymbol{v}$.

\newpage 

Then the optimization problem can be written as below: 
\begin{align}
\boldsymbol{w} ( t + 1) = \boldsymbol{w}(t) -  \lambda_n . \nabla_{\boldsymbol{w}} F( \boldsymbol{w}(t) ), \ \ \text{for} \ t \in \left\{ 0,...t_n - 1 \right\}.       
\end{align}

\medskip

\begin{example}
We approximate $m$ by networks with one hidden layer and $K.r$ neurons in this hidden layer: 
\begin{align}
f_{ \mathsf{net}, ( \boldsymbol{a}, \boldsymbol{b} ) } (x) = \sum_{k=1}^{ K. r } \alpha_k. \sigma \left( \sum_{j=1}^d b_{k,j} . x^{(j)} + b_{k,0}    \right) + \alpha_0.  
\end{align}
where $K.r \in \mathbb{N}$ is the number of neurons and $\sigma: \mathbb{R} \to \mathbb{R}$ is the activation function. 

Furthermore, the unknown optimal vector of weights is obtained by employing the gradient descent algorithm. More precisely, we minimize the penalized empirical $L_2$ risk as
\begin{align}
F( \boldsymbol{a}, \boldsymbol{b} ) = \frac{1}{n} \sum_{i=1}^n \left| f_{\mathsf{net}, ( \boldsymbol{a}, \boldsymbol{b} ) } (X_i) - Y_i \right|^2 + \frac{c_1}{n} . \sum_{k=0}^{ K.r } a_k^2.   
\end{align}
by choosing the appropriate starting value $\left( \boldsymbol{a}^{(0)}, \boldsymbol{b}^{(0)} \right)$ and by setting 
\begin{align}
\begin{pmatrix}
\boldsymbol{a}^{(t+1)}
\\
\boldsymbol{b}^{(t+1)}
\end{pmatrix}
= 
\begin{pmatrix}
\boldsymbol{a}^{(t)}
\\
\boldsymbol{b}^{(t)}
\end{pmatrix}
- \lambda_n . \left( \nabla_{ ( \boldsymbol{a}, \boldsymbol{b} ) } F      \right) \left( \boldsymbol{a}^{(t)},  \boldsymbol{b}^{(t)} \right)
\end{align}
for some $\lambda_n > 0$ chosen below and $t \in \left\{ 0,1,..., t_n - 1   \right\}$. 
\end{example}

\begin{theorem}
Let $n \geq 1$, let $A \geq 1$ and let $\left\{ ( X_1,Y_1 ),...,  ( X_n,Y_n ) \right\}$ be \textit{i.i.d} random variables with values in $[ -A, A ]^d \times \mathbb{R}$. Set $m (x) = \mathbb{E} \left[ Y | X = x    \right]$ and assume that $( X, Y )$ satisfies 
\begin{align}
\mathbb{E} \left( e^{ c_2 . |Y|^2 } \right) < + \infty    
\end{align}
for some constants $c_2 > 0$, and that $m$ satisfies 
\begin{align}
m(x) = \sum_{s=1}^r g_s \left( \boldsymbol{c}_s^{\top} x \right), \ \ x \in \mathbb{R}^d.    
\end{align}
for some $r \in \mathbb{N}, \boldsymbol{c}_s \in [-1,1]^d$, where $\norm{ \boldsymbol{c}_s  } = 1$ and $g_s: \mathbb{R} \to \mathbb{R}$ for $s \in \left\{ 1,..., r \right\}$. 
\end{theorem}

\medskip

\begin{remark}
From the statistical perspective, relevant properties and convergence rates for gradient descent algorithms are given in the studies of \cite{braun2019rate}, \cite{shao2022berry} (see, also \cite{toulis2017asymptotic}) as well as by \cite{schmidt2020nonparametric}. Furthermore, in recent years the use of ANN and DNN for econometric applications has seen growing attention (see, \cite{white1990connectionist} \cite{kuan1994artificial},  \cite{farrell2021deep}).
\end{remark}

\newpage

\subsubsection{Learning of linear penalized least squares estimates by gradient descent}

Let $( x_1,y_1 ),..., ( x_n,y_n ) \in \mathbb{R}^d \times \mathbb{R}$, let $K \in \mathbb{N}$ and let $B_1,..., B_K : \mathbb{R}^d \to \mathbb{R}$, with $c_1 > 0$. We consider the problem to minimize as below
\begin{align}
F ( \boldsymbol{a} ) = \frac{1}{n} \sum_{i=1}^n \left| \sum_{k=1}^K a_k . B_k ( x_i ) - y_i \right|^2 + \frac{c_1}{ n } . \norm{ \boldsymbol{a} }^2  
\end{align}
where $\boldsymbol{a} = \left( a_1,..., a_K \right)^{\top}$ \ \ \ \text{and} \ \ \ $\norm{ \boldsymbol{a} }^2 = \sum_{j=1}^K a_j^2$, by gradient descent.

To obtain the solution of the optimization problem, we choose $\boldsymbol{a}^{(0)} \in \mathbb{R}^K$ and set with 
\begin{align}
\boldsymbol{a}^{(t+1)} = \boldsymbol{a}^{(t)} - \lambda_n . \left( \nabla_{ \boldsymbol{a} }  F \right) \left( \boldsymbol{a}^{(t)} \right)   
\end{align}
for some chosen $\lambda_n > 0$. 

Consider that a Lipschitz continuity condition holds
\begin{align}
\norm{ \left( \nabla_{ \boldsymbol{a} } F \right) ( \alpha_1 ) - \left( \nabla_{ \boldsymbol{a} } F \right) ( \alpha_2 ) } \leq L_n . \norm{ \boldsymbol{a}_1 - \boldsymbol{a}_2 } \ \ \text{with} \ \ \ ( \boldsymbol{a}_1, \boldsymbol{a}_2 ) \in \mathbb{R}^K.     
\end{align}
Then, we have that 
\begin{align}
F \left( \boldsymbol{a}^{(t+1)} - F \left( \boldsymbol{a}^{(t)} \right) \right) \leq - \frac{1}{ 2 . L_n } . \norm{ \left( \nabla_{ \boldsymbol{a} } F \right) \left( \boldsymbol{a}^{(t)}  \right) }^2.     
\end{align}

\begin{lemma}
Let $\sigma$ be the logistic squasher. Let $\bar{\boldsymbol{c}} \in [-1,1]^d$ with $\norm{ \bar{\boldsymbol{c}} } = 1$ and let $g : \mathbb{R} \to \mathbb{R}$ be $( p, C )-$smooth for some $p \in (0,1]$ and $C > 0$. Let $\rho_n > 0$, $K \in \mathbb{N}$ and choose $b_1, b_2,..., b_K \in \mathbb{R}$ such that $b_1 < b_2 < ... < b_K$ and 
\begin{align}
b_1 \leq - A . \sqrt{d} \ \ \ \text{and} \ \ \  b_K \geq A . \sqrt{d} - \frac{ 4.A.\sqrt{d} }{ K - 1 } 
\end{align}
such that
\begin{align}
\frac{ A . \sqrt{d}  }{ (n+1).(K-1) } \leq | b_{k+1} - b_k | \leq \frac{ 4.A. \sqrt{d} }{ K - 1 }, \ \ k \in \left\{ 1,..., K - 1 \right\}.     
\end{align}
Let $a_0 = g( b_1 )$ and $a_k = g(b_k) - g( b_{k-1} )$ with $k \in \left\{ 1,..., K \right\}$ where $K$ is the number of layers. Then, we have that
\begin{align*}
\underset{ x \in [-A,A]^d  }{ \mathsf{sup}  } \ 
& \left| a_0 + \sum_{k=1}^K a_k . \sigma \big( \rho_n . \left(  \bar{\boldsymbol{c}}^{\top} x - b_k \right) \big) - g \left(  \bar{\boldsymbol{c}}^{\top} x \right) \right|   
\\
&\leq
\frac{ 3. ( 4.A.\sqrt{d} )^p.C }{ (K-1)^p } + C.( 4.A. \sqrt{d} )^p . (K-1)^{1-p} . e^{ - \frac{ \rho_n. (A.\sqrt{d}) }{ (n+1).(K-1) } }.
\end{align*}
\end{lemma}

\newpage 

\begin{proof}
Notice that we assume that the logistic activation function $\sigma(x) = \frac{1}{ 1 + e^{-x} }$ is approximated by the step function $\boldsymbol{1} \left\{  [ 0, \infty) \right\} (x)$. In other words, for any $x \in \mathbb{R}$ we have that 
\begin{align}
\left| \sigma(x) - \boldsymbol{1} \left\{  [0, \infty) \right\} (x) \right| \leq e^{ -|x| }. 
\end{align}
We have that 
\begin{align*}
&
\left| a_0 + \sum_{k=1}^K a_k . \sigma \big( \rho_n . \left(  \bar{\boldsymbol{c}}^{\top} x - b_k \right) \big) - g \left(  \bar{\boldsymbol{c}}^{\top} x \right) \right| 
\\
&\leq
\left| a_0 + \sum_{k=1}^K a_k . \sigma \big( \rho_n . \left(  \bar{\boldsymbol{c}}^{\top} x - b_k \right) \big) - \sum_{k=1}^K a_k. \boldsymbol{1}_{ [ b_k, \infty) } \left(  \bar{\boldsymbol{c}}^{\top} x \right)   \right| 
\\
&\ \ \ \ \ \ + 
\left| a_0 + \sum_{k=1}^K a_k. \boldsymbol{1}_{ [ b_k, \infty) } \left(  \bar{\boldsymbol{c}}^{\top} x \right) - g \left(  \bar{\boldsymbol{c}}^{\top} x \right) \right| 
\end{align*}
Moreover, for each $b_j \leq \bar{\boldsymbol{c}}^{\top} x < b_{j+1}$, where $j \left\{ 1,..., K - 1 \right\}$, we can conclude that the definition of $a_k$, from the $( p, C )-$smoothness of $g$ and from our choice of the $b_k$ we obtain 
\begin{align*}
&\left| a_0 + \sum_{k=1}^K a_k . \boldsymbol{1}_{ [ b_k, \infty) } \left(  \bar{\boldsymbol{c}}^{\top} x \right) - g \left(  \bar{\boldsymbol{c}}^{\top} x \right) \right|     
\\
&=
\left| a_0 + \sum_{k=1}^j a_k  - g \left(  \bar{\boldsymbol{c}}^{\top} x \right) \right|    
=
\left| g(b_j) - g \left( \bar{\boldsymbol{c}}^{\top} x \right) \right|
\\
&\leq 
C. \left| b_j - \bar{\boldsymbol{c}}^{\top} x \right|^p \leq C . \left| b_{j+1} - b_j \right|^p \leq \frac{ C. \left( 4.A.\sqrt{d} \right)^p }{  (K-1)^p }.
\end{align*}
Therefore, we have shown that 
\begin{align}
\underset{ x \in [-A,A]^d  }{ \mathsf{sup}  } \ 
\left| a_0 + \sum_{k=1}^K a_k . \boldsymbol{1}_{ [ b_k, \infty) } \left(  \bar{\boldsymbol{c}}^{\top} x \right) - g \left(  \bar{\boldsymbol{c}}^{\top} x \right) \right| \leq \frac{ C. \left( 4.A.\sqrt{d} \right)^p }{  (K-1)^p }.     
\end{align}
We complete the proof by showing that 
\begin{align*}
\underset{ x \in [-A,A]^d  }{ \mathsf{sup}  } \ 
& \left| a_0 + \sum_{k=1}^K a_k . \sigma \big( \rho_n . \left(  \bar{\boldsymbol{c}}^{\top} x - b_k \right) \big) - \sum_{k=1}^K a_k . \boldsymbol{1}_{ [ b_k, \infty) } \left(  \bar{\boldsymbol{c}}^{\top} x \right) \right| 
\\
&\leq
\frac{ 2. \left( 4.A.\sqrt{d} \right)^p . C }{ (K-1)^p } + C. \left( 4.A.\sqrt{d} \right)^p . (K-1)^{1-p}. e^{ - \frac{ \rho_n. (A.\sqrt{d}) }{ (n+1).(K-1) } }. 
\end{align*}

\end{proof}

\newpage

\begin{lemma}[\cite{braun2019rate}]
Let $\sigma$ be the logistic activation function. Define $F$ to be the ridge regression and set with 
\begin{align}
\bar{\boldsymbol{b}} = \boldsymbol{b} - \lambda_n . \left( \nabla_{ \boldsymbol{b} } F \right) ( \boldsymbol{a}, \boldsymbol{b} )    
\end{align}
for some $\lambda_n > 0$, where 
\begin{align*}
\boldsymbol{a} = \big( a_1,..., a_K \big)^{\top} \in \mathbb{R}^K \ \ \text{and} \ \ \boldsymbol{b} = \big( b_{1,0}, b_{1,1},..., b_{1,d},..., b_{K,0}, b_{K,1},..., b_{K,d} \big)^{\top} \in \mathbb{R}^{ K. (d+1) }.    
\end{align*}
Then, we have that for any $k \in \left\{ 1,..., K \right\}$ and any $j \in \left\{ 0,..., d \right\}$ such that: 
\begin{align*}
\left| \bar{b}_{k,j} - b_{k,j} \right| \leq \lambda_n . 2. \sqrt{ F( \boldsymbol{a}, \boldsymbol{b} ) }. \mathsf{max} \left\{ 1, \mathsf{max} \left\{ \left| x_i^{(\ell)} \right| \right\} \right\}. \mathsf{exp} \left( - \underset{ i = 1,..., n }{ \mathsf{min} } \left\{ \sum_{j=1}^d b_{k,j} . x_i^{(j)} + b_{k,0} \right\}  \right).
\end{align*}
\end{lemma}
Consider the Cauchy-Schawrz inequality and since the activation function $\sigma$ is Lipschitz continuous 
\begin{align*}
\frac{1}{n} &\sum_{i=1}^n \big[ f_{\mathsf{net}, ( \bar{\boldsymbol{a}}, \boldsymbol{b}^{(t)} ) } (x_i) - f_{\mathsf{net}, ( \bar{\boldsymbol{a}}, \boldsymbol{b}^{(0)} ) } (x_i) \big]^2    
\\
&=
\frac{1}{n} \sum_{i=1}^n \left\{ \sum_{k=1}^K \bar{a}_k . \left[ \sigma \left(  \sum_{j=1}^d b_{k,j}^{(t)} . x_i^{(j)} + b_{k,0}^{ \textcolor{red}{(t) } } \right) -  \sigma \left(  \sum_{j=1}^d b_{k,j}^{(t)} . x_i^{(j)} + b_{k,0}^{ \textcolor{red}{(0)} } \right)  \right] \right\}^2
\\
&\leq
\sum_{k=1}^K \bar{a}_k^2 . \mathsf{max} \left\{ 1, \underset{ i,j }{ \mathsf{max} } | x_i^{(j)} |^2 \right\}. (d+1). \sum_{k=1}^K \sum_{j=0}^d \left| b_{k,j}^{ \textcolor{red}{(t) } } -  b_{k,j}^{ \textcolor{red}{(0) } } \right|
\end{align*}

\begin{remark}
Notice that \textit{dimensionality reduction} is crucial concept that commonly discussed in the neural network literature from the perspective of layers and related algorithms via the implementation of feedforward neural networks techniques. Alternative approaches include the implementation of taylor expansions which is more commonly used as a methodology for estimating econometric models (e.g. see \cite{olmo2022nonparametric}) as well as the method of sieves and sieve estimators (see, \cite{chen2007large}). Specifically, focusing on comparing these two estimation methodologies can provide some intuition on the main implications of the choice of the \textit{contraction mapping} (see, \cite{keeler1969theorem}, \cite{reich1971some}) to the underline asymptotic theory. In particular, the framework of \cite{olmo2022nonparametric} corresponds to a nonparametric linear regression model with a lagged regressor. Due to the fact that the functional form is estimated using a taylor expansion with partioning\footnote{On the large sample properties of partitioning-based series estimators see, \cite{cattaneo2020large}.}, implies that the corresponding stochastic approximation has discontinuous increments. Thus, the asymptotic behaviour of estimators and test statistics consists of asymptotic functionals that correspond to the supremum of Gaussian processes (e.g., see \cite{beder1987sieve}) which might not even satisfy regularity conditions such as tightness. However, this is problematic for several reasons and especially when the interest of the econometrician is the modeling of nonstationarity in regressors but the partitioning and estimation methodology does not correspond to the relevant stochastic approximations. A good understanding of the principles of \textit{contraction mappings} is crucial.  
\end{remark}

\newpage

\subsection{Deep Neural Network Estimate learned by Gradient Descent}

Based on the framework proposed by \cite{shen2021deep}, we present the following results which are useful to investigate the properties of Deep Neural Networks learned by gradient descent with respect to their dimensionality and complexity. In particular, the concept of \textit{pseudo-dimension} is considered as a measure of complexity (see, \cite{mohri2008rademacher}).
Thus, the framework of \cite{shen2021deep} considers the implementation of deep neural network for a high-dimensional quantile regression. An application to forecasting in time series is presented by \cite{chronopoulos2023forecasting}. A key ingredient to derive excess risk bounds for the deep quantile regression framework is to consider the properties of the functional classes within which identification and estimation holds. 

\begin{assumption}[\cite{shen2021deep}]
\label{ass1}
To derive excess risk bounds the following conditions hold: 

\begin{itemize}

\item[(i)] The conditional $\tau-$th quantile of $\eta$ given $X = x$ is $0$ and $\mathbb{E} \big( | \eta | | X = x \big) < \infty$ for almost every $x \in X$. 

\item[(ii)] The support of covariates $\mathcal{X}$ is a bounded compact set in $\mathbb{R}^d$, and without loss of generality $\mathcal{X} = [ 0,1]^d$.

\item[(iii)] The response variable $Y$ has a finite $p-$th moment for some $p > 1$, that is, there exists a finite constant $M > 0$ such that $\mathbb{E} | Y |^p \leq M$.  
    
\end{itemize}
\end{assumption}

\begin{lemma}[Lemma 2 in \cite{shen2021deep}]
Consider the $d-$variate nonparametric regression model with an unknown regression function $f_0$. Let $\mathcal{F} = \mathcal{F}_{ \mathcal{D}, \mathcal{W}, \mathcal{U}, \mathcal{S}, \mathcal{B} }$ be a Holder class of feed-forward neural networks with a continuous piecewise-linear activation function of finite elements ("pieces") and 
\begin{align}
\hat{f}_{\phi} \in \underset{ f \in \mathcal{F}_{\phi}  }{ \mathsf{argmin} } \ R_{ n \tau } (f)
\end{align}
be the empirical risk minimizer over $\mathcal{F}_{\phi}$. 
Assume that Assumption \ref{ass1} above holds and that $\norm{ f_0 }_{\infty} \leq \mathcal{B}$ for $\mathcal{B} \geq 1$. Then, for $2n \geq P \mathsf{dim} ( \mathcal{F}_{\phi} )$ and any $\tau \in (0,1)$, it holds that 
\begin{align}
\underset{ f \in \mathcal{F}_{\phi}   }{ \mathsf{sup} } \ \big| \mathcal{R}_{\tau} (f) - \mathcal{R}_{n\tau} (f) \big| \leq c_0  \frac{ \mathsf{max} \left\{ \tau, 1 - \tau \right\}  \mathcal{B} }{ n^{  1- 1 / p}  } \mathsf{log} \big( \mathcal{N}_{2n} \left( n^{-1}, \norm{.}_{\infty}, \mathcal{F}_{\phi} \right),  
\end{align}
Notice that this quantity is considered as an upper bound of the empirical risk minimizer, where $c_0 > 0$ is a constant independent (i.e., not depending upon) of $n, d, \tau$ and the remaining components. Moreover, its expected value (here an aspect of interest is to determine the underline distribution theory which relates these quantities w.r.t to marginal and joint distributions) has the following upper bound 
\begin{align}
\mathbb{E} \big[ \mathcal{R}_{\tau} ( \hat{f}_{\phi} ) - \mathcal{R}_{n\tau} ( f_{0} ) \big] \leq C_0  \frac{ \mathsf{max} \left\{ \tau, 1 - \tau \right\}  \mathcal{B} \mathcal{B} \mathcal{B} \mathsf{log} (S)  \mathsf{log} (n) }{ n^{  1- 1 / p}  }  + 2 \underset{ f \in \mathcal{F}_{\phi}   }{ \mathsf{inf} } \ \big| \mathcal{R}_{\tau} (f) - \mathcal{R}_{\tau} (f_0) \big| 
\end{align}
\end{lemma}
Notice that the denominator can be improved to $n$ if the response $Y$ is assumed to be sub-exponentially distributed, that is, there exists a constant $\sigma_Y > 0$ such that $\mathbb{E} \big[ \mathsf{exp} \big( \sigma_Y | Y | \big) \big]$.

\newpage 

Moreover, further interesting aspects here include the technical results with respect to the DNN architecture (e.g., composition), the relation of continuous functionals with the underline quantile regression model and the number of variables as well as linearity of the model and how it relates with the dimensionality of the problem. In particular, one can use numerical experiments to evaluate  the performance of the standard kernel-based method for quantile regressions against the DQR estimation approach.
\begin{remark}A statistical procedure should be providing such evidence (such as a formal specification test approach). On the other hand, starting with some comparisons of relative efficiency and a bias analysis can provide some insights. Various open problems remain in the literature of DNN especially when considering the correct model specification. However, such a research endeavour will require to define what DNN estimable means. Then this could allow us to construct a specification test for a testing that a function under the null hypothesis is correctly DNN estimable against the alternative of a nonparametric kernel-estimation approach. To do this we need to define what a distance function is between a DNN estimable functional form against a nonparametric kernel-based estimation approach. Relevant literature where these aspects are discussed are the studies that present  statistical frameworks for consistent specification testing (see, \cite{stinchcombe1998consistent} and \cite{white1996estimation}). Nevertheless, the key point is consider a suitable functional class within which a statistical test can be constructed for evaluating differences between a DNN estimable functional form and a nonparametrically fitted functional form.   
\end{remark}

\begin{definition}[\cite{shen2021deep}]
For a class $\mathcal{F}$ of functions: $\mathcal{X} \to \mathbb{R}$, its pseudo-dimension, denoted by $\mathsf{Pdim} (\mathcal{F})$, is defined to be the largest integer $m$ for which there exists $\big( x_1,..., x_m, y_1,..., y_m \big) \in \mathcal{X}^m \mathbb{R}^m$ such that for any $( b_1,..., b_m ) \in \left\{ 0, 1 \right\}^m$ there exists $f \in \mathcal{F}$ such that $\forall i : f(x_i) > y_i \iff b_i = 1$.        
\end{definition}

\begin{remark}
For a class of real-valued functions generated by neural networks, pseudo dimension is a natural measure of its complexity. In particular, if $\mathcal{F}$ is the class of functions generated by a neural network with a fixed architecture and fixed activation functions, we have that $P \mathsf{dim} ( \mathcal{F} ) = VC \mathsf{dim} ( \mathcal{F} )$. Thus, a usual assumption under these settings is to require the sample size $n$ to be greater than the pseudo dimension of the class of neural networks considered. Furthermore, the smoothness index works as a covering number to measure the complexity or dependence structure of the DNN. For instance, a function of the form $f_0^d = \big( h_q \circ ... \circ h_0 \big)$ is a recursive function which implies that for each composition of functions, its argument is converted into one-dimensional functional form. In other words, the main intuition here is that by assuming a compositional structure of high-dimensional quantile regressions, this provides an efficient solution to the "curse of dimensionality" problem, since in each composition applied, there is a reduction in the dimensions of the statistical problem.  
\end{remark}

\begin{lemma}[Lemma 1 in \cite{shen2021deep}] For any random sample $S = \left\{ (X_i, Y_i)_{i=1}^n \right\}$, the excess risk of the DQR estimator $\hat{f}_n$ satisfies
\begin{align}
\mathcal{R}_{\tau} ( \hat{f}_n ) - \mathcal{R}_{\tau} ( f_0 ) 
\leq 
2 \underset{ f \in \mathcal{F}_n }{ \mathsf{sup} } \ \big| \mathcal{R}_{\tau} ( f ) - \mathcal{R}_{n \tau} ( f )  \big| + \underset{ f \in \mathcal{F}_n }{ \mathsf{inf} } \ \big(  \mathcal{R}_{\tau} ( f ) - \mathcal{R}_{\tau} ( f_0 ) \big),
\end{align}
where $\mathcal{R}_{n\tau}$ is defined above. 
\end{lemma}

\newpage

\begin{remark}
The excess risk of the DQR estimator is bounded above by the sum of two terms: the stochastic error $2 \underset{ f \in \mathcal{F}_n }{ \mathsf{sup} } \ \big| \mathcal{R}_{\tau} ( f ) - \mathcal{R}_{n \tau} ( f )  \big|$  and the approximation error $\underset{ f \in \mathcal{F}_n }{ \mathsf{inf} } \ \big(  \mathcal{R}_{\tau} ( f ) - \mathcal{R}_{\tau} ( f_0 ) \big)$. 
Furthermore, an interesting conjecture from the above result is that the upper bound no longer depends on the DQR estimators itself, but the function class $\mathcal{F}_n$, the loss function $\rho_{\tau}$ and the random sample $S$. Then, the stochastic error $2 \underset{ f \in \mathcal{F}_n }{ \mathsf{sup} } \ \big| \mathcal{R}_{\tau} ( f ) - \mathcal{R}_{n \tau} ( f )  \big|$ can be analyzed using the empirical process theory. A key ingredient is to calculate the complexity measure $\mathcal{F}_n$ in terms of its covering (bracketing) number (see, \cite{vaart2000empirical}, \cite{van2000asymptotic}, \cite{massart2007concentration} and \cite{wainwright2019high}).     
\end{remark}
Furthermore, for the excess risk of the DQR estimator and the error bounds for the models, based on appropriately specified network parameters (depth, width and size of the network), we have the following upper bound for the excess risk,
\begin{align}
\mathbb{E} \big[ \mathcal{R}_{\tau} ( \hat{f}_{\phi} ) - \mathcal{R}_{\tau} (f_0) \big] \leq C_0 C_{ d,d^{*} } \big( \mathsf{log} n \big)^2 n^{ - \left( 1 - \frac{1}{p} \right)  \frac{2 \alpha^{*} }{ 2 \alpha^{*} + t^{*} } }    
\end{align}
\begin{itemize}

\item $C_0$ is a constant only depending on the model parameters such as the smoothness index (smoothness parameter) of the underlying conditional quantile function. 
    
\item $C_{d,d^{*}}$ is the prefactor depending on $d$, the dimension of the predictor variable in the model. Although a question remains whether this parameter has a fixed functional form when considering each lower dimension term. 

\item $d^{*}$ is determined by the dimensions of the component functions in the composite function. 

\end{itemize}

Then, the convergence rate which is part of the error bound, $n^{ - \left( 1 - \frac{1}{p} \right)  \frac{2 \alpha^{*} }{ 2 \alpha^{*} + t^{*} } }$, is determined by the number of moments $p$ of the response $Y$, the smoothness index of the composite function $\alpha^{*}$ and the intrinsic dimension of the model $t^{*}$. Other relevant aspects include whether imposing further distributional assumptions (such as heavy-tailed errors versus Gaussian errors) can affect the limit theory of estimators as well as its min-max optimality properties. Next, we employ the risk function (i.e., the expectation of the loss function) such that $R_{\tau} (f) := \mathbb{E} \big[ \rho_{\tau} ( Y - f(x)) \big]$ as well as the empirical risk, which is the minimizer of the empirical risk function (i.e., it allows to restrict the functional class). For the proof of the above lemma (Lemma 1 in \cite{shen2021deep}) the following two steps are necessary and sufficient:

\begin{itemize}

\item[Step 1.] Prediction error decomposition (see, \cite{shen2021deep}).

The "best in class" estimator $f_{\phi}^{*}$ is defined as the estimation in the function class $\mathcal{F}$ with minimal $L$ risk such that:
\begin{align}
f_{\phi}^{*} = \underset{ f \in \mathcal{F}_{\phi} }{ \mathsf{arg max} } \ \mathcal{R}_{\tau}(f).   
\end{align}
Then, the approximation error of $f_{\phi}^{*}$, $\mathcal{R}_{\tau}( f_{\phi}^{*} ) - \mathcal{R}_{\tau}(f_0)$ only depends on the function class and the distribution of data. Moreover, by the definition of the empirical risk minimizer, the following inequality result holds 

\newpage

\begin{align}
\mathbb{E}_S \left[ \frac{1}{n} \sum_{i=1}^n \mathsf{g} \big( \hat{f}_{\phi}, Z_i \big) \right] \leq \mathbb{E}_S \left[ \frac{1}{n} \sum_{i=1}^n \mathsf{g} \big( f_{\phi}^{*}, Z_i \big) \right].   
\end{align}
Therefore, it can be proved that the predictor error is upper bounded by the sum of an expectation of a stochastic term and an approximation error. 

\item[Step 2.] Bounding the stochastic term (see, \cite{shen2021deep}). 

To bound the stochastic term, we obtain an upper bound of the expression which includes the stochastic term, using a truncation argument and the classical chaining technique for empirical processes.Denote with $G( f, Z_i ) := \mathbb{E}_{ S^{\prime} } \big[ \mathsf{g} ( f, Z_i^{\prime} ) \big] - 2 \mathsf{g} ( f, Z_i )$ for any function such that $f \in \mathcal{F}_{\phi}$. Given a $\delta-$uniform covering of $\mathcal{F}_{\phi}$, we denote the centers of the balls by $f_j, j = 1,2,...\mathcal{N}_{2n}$, where $\mathcal{N}_{2n} = \mathcal{N}_{2n} \big( \delta, \norm{.}_{\infty}, \mathcal{F}_{\phi} \big)$ is the uniform covering number with radius $\delta ( \delta < \mathcal{B} )$ under the norm $\norm{.}_{\infty}$. By the definition of covering, there exists a (random) $j^{*}$ element such that 
\begin{align}
\norm{ \hat{f}_{\phi}(x) - f_{  j^{*} }(x) }_{ \infty } \leq \delta \ \ \text{on} \ \ x = \big( X_1,..., X_n, X_1^{\prime},..., X_n^{\prime} \big) \in \mathcal{X}^{2n},    
\end{align}
Recall that we have $\mathsf{g} \big( f, Z_i \big) = \big[ \rho_{\tau} \big( f(X_i) - Y_i \big) - \rho_{\tau} \big( f_0(X_i) - Y_i \big) \big]$. Denote with $\lambda_{\tau} = \mathsf{max} \left\{ \tau, 1 - \tau \right\}$, then by the Lipschitz property of $\rho_{\tau}$, for $a, b \in \mathbb{R}$ it holds that 
\begin{align}
\big| \rho_{\tau} (a) - \rho_{\tau} (b) \big| \leq \mathsf{max} \left\{ \tau, 1 - \tau \right\} |a - b| = \lambda_{\tau} |a-b|,    
\end{align}
Therefore, for $i = 1,...,n$ we have that 
\begin{align}
\left|  \mathsf{g} \big( \hat{f}_{\phi}(x) , Z_i \big) -  \mathsf{g} \big( f_{  j^{*}}(x) , Z_i \big) \right|  &\leq \lambda_{\tau} \delta, \
\ 
\left| \mathbb{E}_{ S^{\prime} } \left[ \mathsf{g} \big( \hat{f}_{\phi}(x) , Z_i^{\prime} \big) \right] -  \mathbb{E}_{ S^{\prime} } \left[ \mathsf{g} \big( f_{  j^{*}}(x) , Z_i^{\prime} \big) \right]  \right|  &\leq \lambda_{\tau} \delta
\end{align}
Then, it holds that 
\begin{align}
\mathbb{E}_{ S } \left[  \frac{1}{n} \sum_{i=1}^n  \mathsf{g} \big( \hat{f}_{\phi}(x) , Z_i^{\prime} \big) \right]  \leq   \frac{1}{n} \sum_{i=1}^n  \mathbb{E}_{ S } \left[ \mathsf{g} \big( f_{  j^{*}}(x) , Z_i^{\prime} \big) \right]    + \lambda_{\tau} \delta 
\end{align}
Moreover, denote with $\beta_n \geq \mathcal{B} \geq 1$ to be a positive number who may depend on the sample size $n$. Denote with $T_{ \beta_n }$ as the truncation operator at level $\beta_n$, such that it holds that for any $Y \in \mathbb{R}$, $T_{ \beta_n } Y = Y$, if $|Y| \leq \beta_n$ and $T_{ \beta_n } Y = \beta_n . \mathsf{sign} (Y)$ otherwise. Define the function $f_{ \beta_n }^{*} : \mathcal{X} \to \mathbb{R}$ pointwisely by 
\begin{align}
f_{ \beta_n }^{*} = \mathsf{arg} \underset{ f(x) : \norm{ f }_{\infty} \leq \beta_n }{ \mathsf{min}  } \ \mathbb{E} \big[ \rho_{\tau} \big( f(X) - T_{ \beta_n } Y \big) | X = x \big],   
\end{align}
for each $x \in \mathcal{X}$. Moreover, recall that $\norm{ f^{*} }_{\infty} \leq \mathcal{B} \leq \beta_n$ and 
\begin{align}
f_0 (x) := \mathsf{arg} \underset{ f(x) : \norm{ f }_{\infty} \leq \beta_n }{ \mathsf{min}  } \ \mathbb{E} \big[ \rho_{\tau} \big( f(X) - Y \big) | X = x \big].    
\end{align}

\end{itemize}

\newpage

\subsection{Deep Neural Network Architecture Approximations}

Related studies on the theoretical aspects of network achitecture include \cite{farrell2021deep}, \cite{zeng2021deep} and \cite{dung2021deep}.  The approximation by deep ReLU neural networks of functions having a mixed smoothness is related to the high-dimensional sparse-grid approach which was introduced by Zenger for numerical solving partial differential equations (\cite{dung2021deep}). Next, we consider the mathematical analysis for deriving error bounds of DNNs (see also \cite{yarotsky2017error} and \cite{yarotsky2018optimal}). 

\subsubsection{Error Bounds of DNNs}

Consider the following $\mathsf{g}_s$ function with $2^{s-1}$ uniformly distributed "triangles". The key observation here is that the function $f(x) = x^2$ can be approximated by linear combinations of the functions $\mathsf{g}_s$. Specifically, $f_m$ is a piece-wise linear interpolation of $f$ with $2^m + 1$ uniformly distributed knots (breakpoints) $\frac{k}{2^m}$, $k \in \left\{ 0,..., 2^m \right\}$, such that
\begin{align}
f_n \left( \frac{k}{2^m } \right) = \left( \frac{k}{2^m } \right)^2, \ \ \ k \in \left\{ 0,..., 2^m \right\}    
\end{align}
Thus, the function $f_m$ approximates $f$ with the error $\epsilon_m = 2^{-2(m+1)}$. Applying the linear interpolation method from $f_{m-1}$ to $f_m$ amounts to adjusting it by a function proportional to a sawtooth function given by
\begin{align}
f_{m-1} (x) - f_m(x) = \frac{ \mathsf{g}_m(x) }{  2^{2m} }.     
\end{align}
Hence, it holds that 
\begin{align}
f_m(x) = x - \sum_{s=1}^m \frac{ \mathsf{g}_s(x) }{ 2^{2s} }.    
\end{align}
\begin{proposition}
Given $M > 0$ and $\epsilon \in (0,1)$, there is a ReLU network $\eta$ with two input units that implements a function $\widetilde{\times} : \mathbb{R}^2 \to \mathbb{R}$ such that  
\begin{itemize}
    \item[(a)] for any inputs $x,y$, if $|x| \leq M$ and $|y| \leq M$, then 
    \begin{align}
       \left| \widetilde{\times} (x,y) - xy \right| \leq \epsilon 
    \end{align}

    \item[(b)] if $x = 0$ or $y = 0$, then $\widetilde{\times} (x,y) = 0$.
\end{itemize}
\end{proposition}

\begin{remark}
Notice that the derivations of these error bounds relies on piecewise linear (or quadratic) approximations of the function $f_m$. Some relevant discussion is given by \cite{pottmann2000piecewise}. 
\end{remark}

\newpage

\begin{proof}

Let $\widetilde{f}_{sq,\delta}$ be the approximate squaring function such that $\widetilde{f}_{sq,c\delta} (0) = 0$ and 
\begin{align}
\left|  \widetilde{f}_{sq,\delta} (0) - x^2 \right| < \delta, \ \ \ \text{for} \ \ x \in [0,1].
\end{align}
Assume without loss of generality that $M \geq 1$ and set
\begin{align}
\widetilde{\times} (x,y) = \frac{ M^2 }{8} \left[  \widetilde{f}_{sq,\delta} \left( \frac{ | x + y| }{ 2M }  \right) - \widetilde{f}_{sq,\delta} \left( \frac{ | x | }{ 2M }  \right)  - \widetilde{f}_{sq,\delta} \left( \frac{ | y| }{ 2M }  \right) \right]    
\end{align}
where $\delta = \frac{8 \epsilon}{3 M^2}$.

\end{proof}

\medskip

\begin{theorem}
For any $d, n$ and $\epsilon \in (0,1)$, there is a ReLU network architecture that 
\begin{itemize}

\item[(i)] is capable of expressing any function from $F_{d,n}$ with error $\epsilon$;

\item[(ii)] has the depth at most $c ( ln(1/\epsilon) + 1 )$ and at most $c \epsilon^{-d/n} ( ln(1/\epsilon) + 1 )$ wights and computation units, with some constant $c = c(d,n)$.
    
\end{itemize}

\end{theorem}

\medskip

\begin{remark}
Notice that the main idea of the proof of the above theorem is the use of piecewise linear approximation of quadratic functions (see, \cite{pottmann2000piecewise}). For example, one can consider the squared distance of the two vertices $\boldsymbol{v}_i, \boldsymbol{v}_j$ in the Caley-Klein metric induced by $Q$ which is defined by 
\begin{align}
d^2 \left( \boldsymbol{v}_i, \boldsymbol{v}_j \right)  = \left( \boldsymbol{v}_i -  \boldsymbol{v}_j \right)^{\top} Q  \left( \boldsymbol{v}_i -  \boldsymbol{v}_j \right)
\end{align}
In other words, according to their Theorem 12, an $L^{\infty}-$optimal piece-wise linear approximant over a triangulation of $\mathbb{R}^2$ to a quadratic bivariate function $f$ whose quadratic form is indefinite is defined over a triangulation which is regular in the pseudo-Euclidean metric induced by $f$. Thus, the linear approximant interpolates the function values at the vertices of the triangulation. Discussion about the use of triangulation can be found in the studies of \cite{montanelli2019new}. In particular, the curse of dimensionality can be lessened by establishing a connection with sparse grids. 
\end{remark}

\begin{proof}
The first key part of the proof is to approximate $f$ by a sum-product combination $f_1$ of local Taylor polynomials and one-dimensional piecewise-linear functions. Moreover, we can also approximate $f_1$ using a neural network. 
Let $N$ be a positive integer. Consider a partition of unity formed by a grid of $( N + 1 )^d$ functions $\phi_m$ on the domain $[0,1]^d$ such that
\begin{align}
\sum_m \phi_m (\boldsymbol{x}) \equiv 1, \ \ \ \boldsymbol{x} \in [0,1]^d.  
\end{align}

\newpage

Here, we have that $\boldsymbol{m} = ( m_1,..., m_d ) \in \left\{ 0,1,..., N \right\}^d$, and the function $\phi_m$ is defined as the product below
\begin{align}
\phi_m (\boldsymbol{x}) = \prod_{ k = 1}^d \psi \left(  3N \left( x_k - \frac{m_k}{N}    \right) \right)    
\end{align}
where 
\begin{align}
\psi(x) = 
\begin{cases}
1, & |x| < 1,
\\
0, & 2 < |x|,
\\
2 - |x|, & 1 \leq |x| \leq 2.
\end{cases}
\end{align}
Moreover, it holds that
\begin{align}
\mathsf{supp} \ \phi_m \subset \left\{ \boldsymbol{x}: \left| x_k - \frac{m_k}{N} \right| < \frac{1}{N} \ \forall \ k \right\}.     
\end{align}
Then, for any $\boldsymbol{m} \in \left\{ 0,..., N \right\}^d$, consider the degree$-(n-1)$ Taylor polynomial for the function $f$ at $\boldsymbol{x} = \frac{ \boldsymbol{m} }{N}$ 
\begin{align}
P_{\boldsymbol{m}} = \sum_{ \boldsymbol{n}: |  \boldsymbol{n} | < n } \frac{D^n f}{ \boldsymbol{n} ! } \bigg|_{ \boldsymbol{x} = \frac{ \boldsymbol{m} }{N} }  \left( \boldsymbol{x} -  \frac{ \boldsymbol{m} }{N} \right)^{ \boldsymbol{n}  },   
\end{align}
with the usual notation 
\begin{align}
\boldsymbol{n} ! = \prod_{k=1}^d n_k! \ \ \ \text{and} \ \ \  \left( \boldsymbol{x} -  \frac{ \boldsymbol{m} }{N} \right)^{ \boldsymbol{n}  } = \prod_{k=1}^d \left( x_k - \frac{m_k}{ N }    \right)^{n_k}.   
\end{align}
Therefore, an approximation to $f$ by $f_1$ is given by 
\begin{align}
f_1 = \sum_{ \boldsymbol{m} \in \left\{ 0,..., N \right\}^d } \phi_{\boldsymbol{m} } P_{ \boldsymbol{m} }.    
\end{align}
Therefore, for the proof we bound the approximation error using the Taylor expansion of $f$ as
\begin{align*}
| f( \boldsymbol{x} ) - f_1 (\boldsymbol{x} ) | 
&= 
\left| \sum_{\boldsymbol{m}} \phi_{ \boldsymbol{m} } (\boldsymbol{x} ) \big( f(\boldsymbol{x} ) - P_{ \boldsymbol{m} } (\boldsymbol{x} ) \big) \right|
\\
&\leq
\sum_{ \left\{ \boldsymbol{m}: \left|  x_k - \frac{m_k}{N} \right| < \frac{1}{N} \ \forall \ k \right\} } \big|  f(\boldsymbol{x}) - P_{ \boldsymbol{m} }(\boldsymbol{x}) \big| 
\\
&\leq
2^d \ \underset{ \left\{ \boldsymbol{m}: \left|  x_k - \frac{m_k}{N} \right| < \frac{1}{N} \ \forall \ k \right\}  }{ \mathsf{max} } \ \big|  f(\boldsymbol{x}) - P_{ \boldsymbol{m} }(\boldsymbol{x}) \big| 
\\
&\leq
\frac{ 2^d d^n }{ n! } \left( \frac{1}{N} \right)^n \ \underset{ \boldsymbol{n}: |\boldsymbol{n}| = n }{ \mathsf{max} }  \ \underset{ \boldsymbol{x} \in [0,1]^d }{ \mathsf{ess} \ \mathsf{sup} } \ \left| D^{ \boldsymbol{n} } f ( \boldsymbol{x} ) \right|
\\
&\leq
\frac{ 2^d d^n }{ n! }  \left( \frac{1}{N} \right)^n. 
\end{align*}

\newpage 

We have that the coefficients of the polynomials $P_{\boldsymbol{m} }$ are uniformly bounded for all $f \in F_{d,n}:$
\begin{align}
\label{polynomials}
P_{ \boldsymbol{m} } ( \boldsymbol{x} ) = \sum_{ \boldsymbol{n} : | \boldsymbol{n}  | < n } a_{ \boldsymbol{m}, \boldsymbol{n} } \left(  \boldsymbol{x} - \frac{ \boldsymbol{m} }{ N } \right)^{ \boldsymbol{n} }, \ \ \ | a_{ \boldsymbol{m}, \boldsymbol{n} } | \leq 1.   
\end{align}
The main intuition of the above result is that we can construct a network architecture capable of approximating with uniform error $\frac{ \varepsilon }{2}$ any function with a similar form as the function $f_1$, assuming that $N$, has the form $N = \floor{ \left(  \frac{n!}{ 2^d d^n } \frac{\varepsilon}{2} \right)^{- \frac{1}{n} } }$ and the polynomials can be written as expression \eqref{polynomials}.

\medskip

Therefore, expanding $f_1$ based on the above conditions we obtain
\begin{align}
f_1 ( \boldsymbol{x} ) = \sum_{ \boldsymbol{m}: \left\{0,..., N \right\}^d } \  \sum_{ \boldsymbol{n} : | \boldsymbol{n} | < n } \ a_{ \boldsymbol{m}, \boldsymbol{n} } \phi_{ \boldsymbol{m} } ( \boldsymbol{x} ) \left(  \boldsymbol{x} - \frac{ \boldsymbol{m} }{ N } \right)^{ \boldsymbol{n} }.  
\end{align}

Thus, we consider the approximation of the product $\phi_{ \boldsymbol{m} } ( \boldsymbol{x} ) \left(  \boldsymbol{x} - \frac{ \boldsymbol{m} }{ N } \right)^{ \boldsymbol{n} }$. 
\end{proof}

\subsubsection{Topological Space for DNN Architecture}

\begin{assumption}[\cite{farrell2021deep}] Assume that $f_{*}$ lies in the Holder ball $\mathcal{W}^{\beta, \infty} \left( [-1,1]^d \right)$, with smoothness $\beta \in \mathbb{N}$ such that 
\begin{align}
f_{*} (x) \in \mathcal{W}^{\beta, \infty} \left( [-1,1]^d \right) := \left\{ f: \underset{ \boldsymbol{\alpha}, | \boldsymbol{\alpha} | \leq \beta }{ \mathsf{max} } \ \underset{ x \in [-1,1]^d }{ \mathsf{ess} \ \mathsf{sup} } \left| D^{\alpha} f(x) \right| \leq 1 \right\},    
\end{align}
where $\boldsymbol{\alpha} = \left( \alpha_1,....,  \alpha_d \right)$, $| \boldsymbol{\alpha} | = \alpha_1 + ... + \alpha_d$ and $D^{\alpha} f$, is the weak derivative.    
\end{assumption}

Then, we focus based on Assumptions 1 and 2 above, on deriving high-probability bounds. 

\medskip

\begin{assumption} Let $f_{*}$ lie in a class $\mathcal{F}$. For the feedforward network class $\mathcal{F}_{DNN}$, let the approximation error $\varepsilon_{DNN}$ be 
\begin{align}
\varepsilon_{DNN} := \underset{ f_{*} \in \mathcal{F} }{ \mathsf{sup} } \ \underset{ f \in \mathcal{F}_{DNN} }{ \mathsf{inf} } \ \norm{ f - f_{*} }_{\infty}.  
\end{align}    
\end{assumption}

\begin{remark}
In other words, many recent studies focus on the aspects of how DNN solve the problem of the curse of dimensionality. In other words, various studies consider the approximation for compositional functions by deep neural networks. Specifically, by assuming that the sparse structure in each composition layer, we can show that the total compositional function in high dimensional space owns a low dimensional property, in the sense that it can be approximated by DNN with a convergence rate only dependent on the intrinsic low dimension. 
\end{remark}

\newpage

\subsection{Statistical Inference with Stochastic Gradient Descent}

\subsubsection{Literature Review}

Notice that both for the estimation of \textit{Shallow} as well as \textit{Deep Neural Networks}, useful optimization algorithms under examination are the Gradient Descent and Stochastic Gradient Descent. Related studies that discuss the implementation and applications of these algorithms are presented by \cite{chen2020statistical}. Furthermore,  \cite{toulis2017asymptotic} study the asymptotic theory analysis and finite-sample properties of these estimators (see also \cite{tran2015stochastic}). Moreover, a selective overview of deep learning estimation methodologies are discussed in the study of \cite{fan2021selective} and \cite{braun2021smoking}.      

\begin{itemize}

\item The framework of \cite{zhang2022sieve} consider a sieve SGD algorithm which is a nonparametric estimation approach in the Sobolev ellipsoid space. Moreover, \cite{chen2014sieve} propose a framework for Sieve M estimation on irregular parameters. 

\item SGD procedures involve parameter updates that are implicitly defined. Implicit updates shrink standard SGD updates. The amount of shrinkage depends on the observed Fisher information matrix, which does not need to be explicitly computed.    

\item Theoretical analysis gives a full characterization of the asymptotic behaviour of both standard and implicit SGD-based estimators, including finite-sample error bounds. Moreover, \cite{shao2022berry} present estimation error bounds of M-estimators and SGD algorithms. 

\end{itemize}

Following \cite{chen2020statistical}, assume that we want to estimate the true parameter $\theta_0 \in \mathbb{R}^p$ of a distribution $f$ from $\textit{i.i.d}$ data points $( X_i, Y_i )$ such that conditional on covariate $X_i \in \mathbb{R}^p$ outcome $Y_i \in \mathbb{R}^d$ is distributed according to $f( Y_i; X_i, \theta_0 )$. Such statistical problems reduce to optimization. More specifically, suppose that our aim is to estimate the true parameter $\theta_0 \in \mathbb{R}^p$ of a distribution $f$ from \textit{i.i.d} data points $( X_i, Y_i )$ such that conditional on covariate $X_i \in \mathbb{R}^p$ outcome $Y_i \in \mathbb{R}^d$ is distributed according to $f( Y_i; X_i, \theta_0 )$. Such statistical problems reduce to optimization. 
\begin{align}
\theta_n^{sgd} =  \theta_{n-1}^{sgd} + \gamma_n C_n \nabla \mathsf{log} \left[ f \left( Y_n ; X_n, \theta_{n-1}^{sgd} \right) \right]    
\end{align}
where $\gamma_n$ is the learning rate sequence.

Consider the equivalent moment condition below
\begin{align}
\mathbb{E} \big[ \nabla \mathsf{log} f ( Y; X, \theta_0 ) | X \big] = 0,    
\end{align}
where the expectation corresponds to the true conditional distribution of outcome $Y$ given covariate $X$.

\newpage

\begin{remark}
From computational perspective, SGD is appealing because it avoids expensive matrix inversions and single data point $( X_n, Y_n )$ evaluations. Implicit SGD does not condition on the observed ordering of data points, but conditions on a random ordering instead. 
\begin{align}
\textcolor{red}{\boldsymbol{\theta}_n^{im}} =  \theta_{n-1}^{im} + \gamma_n C_n \nabla \mathsf{log} \left[ f \left( Y_n ; X_n, \textcolor{red}{\boldsymbol{\theta}_{n}^{im}} \right) \right],    
\end{align}
\end{remark}
Therefore, it holds that
\begin{align}
\textcolor{red}{\theta_n^{im}} =  \underset{ \theta }{ \mathsf{argmax} }  \left\{ - \frac{1}{2 \gamma_n } \norm{ \theta - \theta_{n-1}^{im} }^2  + \mathsf{log} \left[ f \left( Y_n ; X_n, \theta \right) \right] \right\},    
\end{align}
Recall the Fisher Information matrix is denoted $\hat{\mathcal{I}}_n ( \theta ) = - \nabla^2 \ell \left( X_n^{\prime} \theta; Y_n \right)$. Using a Taylor approximation of the gradient $\nabla \mathsf{log} \left[ f \left( Y_n ; X_n, \textcolor{red}{\theta_n^{im}}  \right) \right]$ yields 
\begin{align}
\Delta \textcolor{red}{\theta_n^{im}} \approx \left[ I + \gamma_n \hat{\mathcal{I}}_n ( \theta_0 ) \right]^{-1} \Delta \theta_n^{sgd},    
\end{align}
where $\Delta \textcolor{red}{\theta_n^{im}} = \textcolor{red}{\theta_n^{im}} - \theta_0$ and $\Delta \theta_n^{sgd} = \theta_n^{sgd} - \theta_0$.

\medskip

\begin{remark}
Notice that these Gradient Descent Algorithms have a required first step the initialization of the vector. Although the choice of this initial vector is just an approximation it provides a good starting point so that the algorithmic procedure converges to the optimal choice vector.     
\end{remark}

\subsubsection{Preliminary Theory}

Throughout, we use $\norm{x}_p$ to denote the $\ell_p-$norm of $x$, $\norm{X}$ the matrix operator norm of $X$ and $\norm{ X }_{\infty} = \mathsf{max}_{i,j} | X_{i,j} |$ the elementwise $\ell_{\infty}-$norm of $X$. For any sequences $\left\{ a_n \right\}$ and $\left\{ b_n \right\}$ of positive numbers, we write with $a_n \geq b_n$ and $a_n \leq b_n$. 

\medskip
 
Consider the SGD method where the iteration is given by (see, \cite{chen2021statistical})
\begin{align}
x_n = x_{n-1} \eta_n \nabla F ( x_{n-1} ) + \eta_n \xi_n,    
\end{align}
where $\xi_n := \nabla F( x_{n-1} ) - \nabla f ( x_{n-1}, \zeta_n )$. In particular, the above formulation decomposes the descent into two parts: $\nabla F( x_{n-1} )$ which represents the direction of population gradient which is the major driving force behind the convergence of SGD and $\xi_n$ is a martingale difference sequence under the above assumption. Furthermore, it holds that 
\begin{align}
\mathbb{E}_{n-1} \left[ \xi_n \right] := \nabla F ( x_{n-1} ) - \mathbb{E}_{n-1} \big[ \nabla f( x_{n-1}, \zeta_n ) \big] = 0.     
\end{align}

\newpage

Notice that $\mathbb{E}_n (.)$ denotes the conditional expectation $\mathbb{E}_n (. | \mathcal{F}_n )$, where $\mathcal{F}_n$ is the $\sigma-$algebra generated by $\left\{ \zeta_1,..., \zeta_n \right\}$. Let $\Delta_n := x_n - x^*$ be the error of the $n-$th iterate. An equivalent expression gives us
\begin{align}
\Delta_n = \Delta_{n-1} - \eta_n \nabla F( x_{n-1} ) + \eta_n \xi_n,    
\end{align}

\medskip

Given the SGD recursion and under suitable assumptions it can be shown that when the step size sequence $n_i = \eta i^{ -\alpha }$, for $i \in \left\{ 1,..., n \right\}$ with $\alpha \in (1/2, 1)$, we have that 
\begin{align}
\sqrt{n}. \bar{\Delta}_n \Rightarrow \mathcal{N} \big( 0, A^{-1} S A^{-1} \big)    
\end{align}
where $\bar{\Delta}_n = \frac{1}{n} \sum_{i=1}^n \Delta_i = \bar{x}_n - x^{*}$. 
\color{black}

\subsubsection{Assumptions and error bounds}

\begin{assumption}[Strong convexity and Lipschitz continuity of the gradiet]
Assume that the objective function $F(x)$ is continuously differentiable and strongly convex with parameter $\mu > 0$, that is, for any $x_1$ and $x_2$,
\begin{align}
F(x_2) \geq F(x_1) + \langle \nabla F(x_1), x_2 - x_1 \rangle + \frac{\mu}{2} \norm{ x_1 - x_2 }_2^2.    
\end{align}
\end{assumption}

Further assume that $\nabla^2 F( x^{*} )$ exists, and $\nabla F(x)$ is Lipschitz continuous with a constant $L_F$, that is for any $x_1$ and $x_2$ 
\begin{align}
\norm{ \nabla F(x_1) - \nabla F(x_2) }_2^2 \leq L_F \norm{ x_1 - x_2 }_2.    
\end{align}

\begin{remark}
Notice that strong convexity of $F(x)$ can be assumed to hold in order to derive the limiting distribution of averaged SGD. In fact, the strong convexity of $F(x)$ implies that $\mathsf{\lambda}_{min} (A) = \mathsf{\lambda}_{min} \big( \nabla^2 F ( x^* )  \big) \geq \mu$ is an important condition for parameter estimation and inference.   
\end{remark}

\begin{assumption}
The following conditions hold for the sequence $\xi_n = \nabla F( x_{n-1} ) - \nabla f( x_{n-1}, \zeta_n )$: 

\begin{itemize}

\item[1.] Assume that $f( x, \zeta )$ is continuously differentiable in $x$ for any $\zeta$ and $\norm{ \nabla f (x, \zeta) }_2$ is uniformly integrable for any $x$ so that $\mathbb{E}_{n-1} \xi_n = 0$. 

\item[2.] The conditional covariance of $\xi_n$ has an expansion around $x = x^{*}$ such that 
\begin{align}
\mathbb{E}_{n-1} \left[ \xi_n \xi_n^{\top} \right] = S + \Sigma ( \Delta_{n-1} )    
\end{align}

\end{itemize}
\end{assumption}

\newpage

\begin{lemma}
If there is a function $H( \zeta )$ with bounded fourth moment, such that the Hessian of $f( x, \zeta )$ is bounded by 
\begin{align}
\norm{ \nabla^2 f (x, \zeta) } \leq H( \zeta )    
\end{align}
for all $x$ and $\nabla f( x^{*}, \zeta )$ have a bounded fourth moment. 
\end{lemma}

\subsubsection{Estimators for asymptotic covariance}

Consider two consistent estimators, the plug-in estimator and the batch-means estimator (see, \cite{chen2021statistical}). Consider the sample estimate as below
\begin{align}
A_n := \frac{1}{n} \sum_{i=1}^n \nabla^2 f( x_{i-1}, \zeta_i ), \ \ \ \ S_n := \frac{1}{n} \sum_{i=1}^n  \nabla f( x_{i-1}, \zeta_i )  \nabla  f( x_{i-1}, \zeta_i )^{\top}  
\end{align}
Consider the following nodewise Lasso approach 
\begin{align}
\widehat{\gamma}^j = \underset{ \gamma^j \in \mathbb{R}^{d-1} }{ \mathsf{argmin} } \frac{1}{2n} \norm{ D_{.,j} -  D_{.,-j} \gamma^j }_2^2 + \lambda_j \norm{ \gamma^j }_1,  
\end{align}
where $D_{.,j}$ is the $j-$th column of the design matrix $D$ and $D_{.,-j}$ is the design submatrix without the $j-$th column (see also \cite{guo2022doubly}). Further, one can estimate $\Omega_{j,j}$ by
\begin{align}
\widehat{\tau}_j = \frac{1}{n} \left( D_{.,j} - D_{.,-j} \widehat{\gamma}^j \right)^{\top} D_{.,j}.   
\end{align}

\begin{remark}
Notice that $\widehat{\gamma}^j$ is the output of a stochastic gradient-based algorithm.      
\end{remark}

\subsubsection{Averaged SGD Algorithms}

Consider the statistical problem of searching for the minimum point $\theta_0$ of a smooth function $f( \theta )$ where $\theta \in \Theta \subset \mathbb{R}^d$. The stochastic GD method provides a direct way to solve the minimization problem. The algorithm is given as below. Let $\theta_0 \in \mathbb{R}^d$ be the initial value and for $n \geq 1$, we update $\theta_n$ using
\begin{align}
\theta_n = \theta_{n-1} - \gamma_n \big( \nabla f( \theta_{n-1} ) + \zeta_n \big)     
\end{align}
where $\theta_n = \frac{1}{n} \sum_{i=0}^{ n-1 } \theta_i$.
\begin{remark}
Notice that the convergence rate of $\mathbb{E} \norm{ \theta_n - \theta^* }^2$ and $\mathbb{E} \norm{ \bar{\theta}_n - \theta^* }^2$ is examined in the literature as well as the normality of $\sqrt{n} \big( \bar{\theta}_n - \theta^{*} \big)$ (see, \cite{shao2022berry}).      
\end{remark}

\begin{assumption}
There exists positive constants $c_2$ and $\beta$ such that $\forall \ \theta$ with $\norm{ \theta - \theta^{*} } \leq \beta$,
\begin{align}
\norm{  \nabla^2 f( \theta ) - \nabla^2 f( \theta^{*} ) } \leq c_2 \norm{ \theta - \theta^* }.    
\end{align}
\end{assumption}

\newpage 

\subsection{Sieve Estimation Applications}

Consider the knowledge transmission through neural networks as a recursive process that focuses on the diffusion of processes. Similar to neurobiological processes in the brain neural networks can activate certain neuros in the brain for transmission of understanding. 

Stochastic approximations recursively approximate the zeros of an unknown function $\Psi ( \theta )$, say $\theta^{*}$,  
\begin{align}
\hat{\theta}_{t+1} = \hat{\theta}_t  + a_t \psi ( z_t, \hat{\theta}_t ), \ \ \ t = 1,2,...  
\end{align}
where $a_t$ is a "learning rate" tending to zero and $\psi ( Z_t, \theta )$ is a measurement of $\Psi ( \theta )$ at time $t$, influenced by random variables $Z_t$. Furthermore, when $\Psi ( \theta ) \equiv \mathbb{E} \big[ \psi ( Z_t, \theta )  \big]$ this method yields a recursive implementation of the method of m-estimation of Huber (see,  \cite{kuan1994artificial}). The particular methodology can be used to estimate recursively the parameters of nonlinear regression models, such as those arising in neural network applications. Specifically, offline nonparametric estimation methods can usually be applied to the case of mixing processes using results for the methods of sieves. On the other hand, online nonparametric estimation methods require convergence to a global optimum of the underlying least squares problem, not just the local optimum.

\subsubsection{Asymptotic Properties of Sieve Estimation}

Another application in which constraint optimization is used is when estimating and simulating neural networks with sieves. In particular, based on a construction of the neural network sieve estimators, in each sieve space $\mathcal{F}_{r_n}$, there is a constraint on the $\ell_1$ norm for which $\sum_{i=0}^{ r_n } | \alpha_i | \leq V_n$. Therefore, finding the nearly optimal function in $\mathcal{F}_{r_n}$ for $\mathbb{Q}_n (f)$ is in fact a constrained optimization problem. A classical way to conduct this optimization is through introducing a Lagrange multiplier for each constraint. On the other hand, due to the difficulty in finding an explicit connection between the Lagrange multiplier and the upper bound in the inequality constraint we use instead the subgradient method (see, \cite{shen2023asymptotic}). The main idea is to update the parameter $\left\{  \alpha_0,...., \alpha_{r_n}   \right\}$ through 
\begin{align}
\alpha_i^{(k+1)} =  \alpha_i^{(k)} - \delta_k \mathsf{g}^{(k)}, \ \ \ i = 0,...,r_n,    
\end{align}
This allows us to investigate the asymptotic properties, including consistency, rate of convergence and asymptotic normality for neural network sieve estimators with one hidden layer. 

\medskip

Using the Generalized Dominated Convergence Theorem, we have that 
\begin{align}
\mathbb{E}^{*} \left[ \underset{ f \in \mathcal{F}_n  }{ \mathsf{sup} } \ \left| \frac{1}{n} \sum_{i=1}^n \epsilon_i \big( f( \boldsymbol{x}_i ) - f_0( \boldsymbol{x}_i ) \big) \right| \right] \leq 2 \mathbb{E}_{\epsilon}  \mathbb{E}_{\xi} \left[ \underset{ f \in \mathcal{F}_n  }{ \mathsf{sup} } \ \left| \frac{1}{n} \sum_{i=1}^n \xi \epsilon_i \big( f( \boldsymbol{x}_i ) - f_0( \boldsymbol{x}_i ) \big) \right| \right]  \to 0.
\end{align}

\newpage

\begin{example}
An application of sieve estimation in econometrics is presented in the study of \cite{su2016sieve}, where the authors propose a framework for sieve instrumental variable quantile regression estimation of functional cofficient models. More precisely, one can demonstrate that 
\begin{align}
\sqrt{n} \big(  \boldsymbol{\beta}_{2 \tau}  - \beta_{2 \tau}   \big) \overset{d}{\to} \mathcal{N}  \left( 0, \tau (1 - \tau) \underset{ K \to \infty }{ \mathsf{lim} } \mathbb{S}_2 \Omega_{B_{\tau}} \Psi_K \Omega_{B_{\tau}}^{\top} \mathbb{S}_2^{\top}  \right).
\end{align}
Consequently, one can conduct statistical inference on $\beta_{2 \tau}$ as usual by estimating the AVC matrix given above. Alternatively, one can apply the boostrap method to obtain standard errors and make inference. Furthermore, we can construct a specification test where we are interested in testing the null hypothesis
\begin{align}
\mathbb{H}_0: \delta_{1 \tau} ( U_i) \equiv \mathbb{S} \delta_{ \tau} ( U_i) = \delta_{1 \tau}.  
\end{align}
almost surely for some parameter $\delta_{1 \tau} \in \mathbb{R}^r$. Under $\mathbb{H}_0$, $r$ of the $( k_1 + k_2 )$ functional coefficients are constant, whereas under the alternative hypothesis $\mathbb{H}_1$, we have that at least one of the functional coefficients in $\delta_{1 \tau} (.)$ is not constant. Note that within the setting of nonstationary time series a relevant framework is presented by \cite{dong2021weighted}. 
\end{example}

\medskip

\subsubsection{Sieve Estimation for Panel Data}

A growing interest in the estimation of panel data models with cross-section dependence but most of the literature focuses on the linear specification of the regression relationship. In particular, let $y_{it}$, $t = 1,...,n$ and $t = 1,...,T$ be the $i-$th cross section unit at time $t$. We suppose that $y_{it}$ is generated according to the following semiparametric panel data generating process (see, \cite{su2012sieve})
\begin{align}
y_{it} = \mathsf{g}_i ( x_{it} ) + \gamma_{1i}^{\prime} f_{1t} + e_{it},
\end{align}
where $x_{it} \in \mathcal{X}_i \in \mathbb{R}^d$ is a vector of observed individual-specific regressors on the $i-$th cross-section unit at time $t$. In this section, we focus on the sieve estimation of semiparametric panel data models with multi-factor error structure. We develop the asymptotic theory under fairly general conditions when both the cross-section and time-dimensions are large. For instance, if only homogeneous regression relationships are of interest, the time dimension need not pass to infinity. Moreover one can consider testing the constancy of the nonparametric relationship over individuals in the presence of multi-factor error structure. In addition the individual specific regressors have the following structure
\begin{align}
x_{it} = \Gamma_{1i}^{\prime} f_{1t} + \Gamma_{2i}^{\prime} f_{2t} + v_{it},    
\end{align}
In practice, one may also be interested in estimating a restricted submodel of the following form 
\begin{align}
y_{it} = \mathsf{g}(x_{it}) + \gamma_{1i}^{\prime} f_{1t} + e_{it}.    
\end{align}

\newpage

\subsubsection{Sieve M inference on irregular parameters}

We follow the framework proposed by \cite{chen2014sieve} who consider plug-in sieve M estimators. More precisely, we assume that the data $\left\{ Z_i \right\}_{i=1}^n$ is a random sample from the distribution of $Z$ defined on an underlying complete probability space. Let $\mathcal{L} (., .) : \mathcal{Z} \times \mathcal{A} \mapsto \mathbb{R}$ be a measurable function and $\mathbb{E} \left[ \mathcal{L} \left( Z, \alpha \right) \right]$ be a population criterion. For simplicity we assume that there is a unique $\alpha_0 \in \left( \mathcal{A}, d_{A} \right)$ such that $\mathbb{E} \left[ \mathcal{L} \left( Z, \alpha_0 \right) \right] > \mathbb{E} \left[ \mathcal{L} \left( Z, \alpha \right) \right]$ for all $\alpha \in \left( \mathcal{A}, d_{A} \right)$ with $d_A \left( \alpha, \alpha_0 \right) > 0$. Different models in economics correspond to different choices of the criterion function $\mathbb{E} \left[ \mathcal{L} \left( Z, \alpha \right) \right]$ and the parameter space $\left( \mathcal{A}, d_{A} \right)$. A model does not need to be correctly specified and $\alpha_0$ could be pseudo true parameter. In this paper, we are interested in the estimation of and inference of a functional $f \left( \alpha_0 \right)$ via the method of sieves. 

Let $\mathcal{A}_n$ be a sieve space for the whole parameter space $\mathcal{A}$. Then an approximate sieve M estimator $\widehat{a}_n \in \mathcal{A}_n$ of $\alpha_0$ solves
\begin{align}
\frac{1}{n} \sum_{i=1}^n \mathcal{L} \left( Z_i, \widehat{a}_n \right) \geq \underset{ \widehat{a}_n \in \mathcal{A}_n }{ \mathsf{sup}  } \frac{1}{n} \sum_{i=1}^n \mathcal{L} \left( Z_i, a_n \right) - o_p \left( \frac{1}{n} \right)
\end{align}

\begin{example}[A partially additive quantile regression model, see \cite{chen2014sieve}]

\

Suppose that the i.i.d data $\left\{ Y_i, X_i^{\prime} = \left( X_{0i}^{\prime}, X_{1i}..., X_{qi} \right) \right\}_{i=1}^n$ is generated according to the process below: 
\begin{align}
Y-i = X_{0i}^{\prime} \theta_0 + \sum_{j=1}^q h_{j,0} \left( X_{j,i}  \right) + U_i,
\end{align} 
with $\mathbb{E} \big[ 1\left\{ U_i \leq 0 \right\} | X_i \big] = \tau \in (0,1),$ where $dim(X_0) = d_{\theta}$, $dim(X_j) = 1$ for $j = 1,..., q$, $dim(X) = d_{\theta} + q$ and $dim(Y) = 1$. Let $\alpha_0 = \left( \theta_0, h_0 \right)$, where $\theta_0 \in \Theta$ and $h_0 = \left( h_{1,0},..., h_{q,0} \right) \in \mathcal{H}$.  A functional of interest could be for instance, $f \left( \alpha_0 \right) = \lambda^{\prime} \theta_0$ for any $\lambda \in \mathbb{R}^{d_{\theta} }$ with $\lambda \neq 0$. This is an extension of the parametric quantile regression model of \cite{koenker1978regression} to allow for unknown additive functions $\sum_{j=1}^q h_{j,0}\left( X_{j,i} \right)$. We can estimate $\alpha_0 = \left( \theta_0, h_0 \right)$ by the sieve QR estimator $\widehat{\alpha}_n = \left( \widehat{\theta}_n, \widehat{h}_n  \right)$ that solves 
\begin{align}
\underset{\theta \in \Theta, h \in \mathcal{H}_n }{ \text{max} } \ \sum_{i=1}^n \left( 1 \left\{ Y_i \leq X_{0,i}^{\prime} \theta   + \sum_{j=1}^q h_{j,0}\left( X_{j,i} \right) \right\} - \tau \right) \times \left[ Y_i - X_{0,i}^{\prime} \theta  -  \sum_{j=1}^q h_{j,0}\left( X_{j,i} \right)  \right]
\end{align}
Specifically, if $\mathcal{A}$ is a Holder, Sobolev or Besov space of functions with bounded supports and $\mathcal{A}_n$ is a linear sieve space consisting of spline, wavelet, or cosine bases, then one typically has 
\begin{align}
d_A \left( \pi_n \left( \alpha_0 \right), \alpha_0 \right) = \norm{\pi_n \left( \alpha_0 \right) - \alpha_0  }_{\mathsf{sup}} 
\end{align} 
Given the existing results on the convergence rates for sieve $M$ estimators of semi-nonparametric models, we can restrict our attention to a shrinkage neighborhoud of $\alpha_0$. Let $\delta_{A,n} = \delta_{A,n}^{*} \gamma_n$ and $\delta_{s,n} = \delta_{s,n}^{*} \gamma_n$, where $\gamma_n$ is a positive sequence that diverges to infinity very slowly (say log log $n$) such that $\delta_{A,n} = 0(1)$. 
\end{example}

\newpage

In the example above, we have that 
\begin{align}
\mathcal{L} \left( Z, \alpha \right) = \big[ 1 \left\{ Y \leq \alpha ( Y )  - \tau  \right\} \big] \big[ Y - \alpha(Y) \big]
\end{align}
with $\alpha(Y) = X_{0,i}^{\prime} \theta   + \sum_{j=1}^q h_{j,0}\left( X_{j,i} \right)$. For any $\alpha \in \mathcal{A}$, we define a strong metric $\norm{.}_s$ as below
\begin{align}
\norm{ \alpha - \alpha_0 }_s = \mathbb{E} \left[ \left| X_{0,i}^{\prime} \left( \theta - \theta_0 \right)  + \sum_{j=1}^q \big[ h_{j}\left( X_{j} \right) - h_{j,0}\left( X_{j,i} \right) \big] \right|  \right]
\end{align}
By the definitions of the metrics, we have that for any $\alpha \in \mathcal{A}$, (see, \cite{chen2014sieve}) 
\begin{align}
\norm{ \alpha - \alpha_0 }^2 = \mathbb{E} \bigg[ f(0|X) \left| \alpha(X) - \alpha_0 (X) \right|^2 \bigg]
\end{align}
Let $v^{*}_n = \left( v^{*}_{\theta, n}, v^{*}_{h, n}  \right)$ be the Riesz representer of the functional $\frac{ \partial f \left( \alpha_0 \right) }{  \partial \alpha } [v]$ on $\mathcal{V}_n$. 
Denote with 
\begin{align}
v^{*}_n(X) = X_{0,i}^{\prime} v^{*}_{\theta,n}  + \sum_{j=1}^q v^{*}_{h_j,n}  \left( X_{j} \right)
\end{align}
Then, the variance of the plug-in sieve $M$ estimator $f \left( \widehat{\alpha}_n \right)$ of $f \left( \alpha_0 \right)$ is expressed as below 
\begin{align}
\norm{v_n^{*} }^2_{sd} = \tau ( 1 - \tau) \mathbb{E} \bigg[ \big| v_n^{*}   (X) \big|^2 \bigg] = \tau \left( 1 - \tau \right) \norm{  v_n^{*} }_s^2. 
\end{align} 

\medskip

\begin{proposition}[\cite{chen2014sieve}]
Under Assumptions and Conditions we have that
\begin{align}
\sqrt{n} \frac{ \displaystyle f \left( \widehat{\alpha}_n \right) - f \left( \alpha_0 \right)}{ \displaystyle \left\{ \tau \left( 1 - \tau \right) \mathbb{E} \left[ \left| X_{0}^{\prime} v^{*}_{\theta,n}  + \sum_{j=1}^q v^{*}_{h_j,n}  \left( X_{j} \right) \right|^2 \right] \right\}^{1 / 2 } } \to \mathcal{N} \left( 0, 1 \right).
\end{align}
\end{proposition}

\medskip

An application to specification testing in econometrics using invariance principles on Sobolev spaces is proposed by  \cite{kuersteiner2019invariance}. Furthermore, an application of general sieves for inference purposes in nonparametric time series regression models identified with conditional moment restrictions is studied by \cite{chen2022inference}. In particular, nonlinear sieve learning employs a more general functional class of sieves that can approximate nonlinear functions of high dimensional variables (such as in the case of an infinite-dimensional parameter space).

\newpage

\subsection{Further Statistical Algorithms in Economic Applications}

This section is motivated from the Seminar of Dennis Kristensen, Professor of Economics at the Department of Economics of University College London (UCL), who presented the paper with title: 
"\textit{Iterative estimation of structural models with an application to perturbed utility models}", in May 2023 at the Department of Economics, University of Exeter Business School. Specifically, the particular steam of literature is based on the framework proposed by \cite{hotz1994simulation} and is relevant to estimators for dynamic models of discrete choice, that is, a general class of iterative estimators. 

Take for example, the classical \textit{likelihood estimation}, which is a statistical methodology that iteratively solves the statistical problem until convergence to an estimate that is close to the true unknown parameter of interest. In the same spirit, without sacrificing efficient statistical estimation we assume that such a sequence of estimators converge to the full solution of the optimization problem. These optimization steps are summarized as below:   
\begin{itemize}

\item[1.] Obtain a nonparametric estimate of the solution as the global optimizer from a local estimator $\hat{\theta}$.

\item[2.] Use first order conditions to update the initial conditions estimate.

\item[3.] Use the updated solution path based on the estimated parameters of the previous steps.

\item[4.] Update solution using parameter estimates.

\item[5.] Repeat Steps 2-4 until convergence. 

\end{itemize}

\medskip

\begin{remark}
Notice that the existence of some noise due to the parametric estimation is unavoidable. Therefore, the larger the presence of noise during the first stage of estimation the more iterations will be needed for algorithmic convergence (uncertainty quantification). Furthermore, the error bound of the true parameter vector based on the estimated parameters from the above algorithmic procedure is heavily based on the assumed estimator of the first stage (large sample theory). Thus, the resulting sequence of iterative estimators is assumed to have asymptotic behaviour that is within the vicinity of the domain of attraction; although might not necessarily result in solving the structural model under examination. 
\end{remark}

A different stream of literature considers doubly robust estimation methods which have been developed in various semiparametric problems, including partially linear models, instrumental variable analysis, and dimension reduction among others. Therefore, as a somewhat under-appreciated limit result, we point out that the familiar least-squares estimator for each individual coefficient in linear regression is doubly robust in the context of a partially linear model. This result is also closely related to debiased Lasso estimation in high-dimensional linear regression (e.g., see \cite{guo2022doubly}).

\newpage 

\appendix
\numberwithin{equation}{section}
\makeatletter

\newpage

\section{Elements of Weak Convergence of Empirical Processes}

\begin{example}
(M-dependent sequences) Let $X_n, n \in \mathbb{Z}$ be a stationary sequence with $\mathbb{E} X_n = 0$, $\mathbb{E} X^2_n < \infty$. Assume that $\sigma \big( \left\{ X_j, j \leq 0 \right\} \big)$ and  $\sigma \big( \left\{ X_j, j \geq M \right\} \big)$ are independent. 
\end{example}

\subsection{Sub-Gaussian processes}

Notice that Sub-Gaussian processes satisfy the increment bound $\norm{ X_s - X_t }_{ \psi_2 } \leq \sqrt{6} d(s,t)$. Therefore, the   general maximal inequality leads for sub-Gaussian processes to a bound in terms of an entropy integral. 

\begin{lemma}
Let $\left\{ X_t : t \in T \right\}$ be a separable sub-Gaussian process. Then, for every $\delta > 0$, 
\begin{align}
\mathbb{E} \ \underset{ d(s,t) \leq \delta }{ \mathsf{sup} } | X_s - X_t | \leq K \int_0^{\delta} \sqrt{ \mathsf{log} D (\epsilon, d)} d \epsilon,
\end{align}
for a constant $K$. In particular, for any $t_0$, 
\begin{align}
\mathbb{E} \ \underset{ t }{ \mathsf{sup} } | X_t | \leq \mathbb{E} | X_{t_0} | + K \int_0^{\infty} \sqrt{ \mathsf{log} D (\epsilon, d)} d \epsilon.
\end{align}
\end{lemma}
\begin{proof}
Apply the general maximal inequality with $\psi_2 (x) = e^{x^2} - 1$  and $\eta = \delta$. Since $\psi_2^{-1}(m) = \sqrt{ \log(1 + m)}$, we have that $\psi_2^{-1} \big( D^2 ( \delta, d ) \big) \leq \sqrt{2} \psi_2^{-1} \big( D ( \delta, d ) \big)$. Thus, the second term in the maximal inequality can first be replaced by $\sqrt{2} \delta \psi^{-1} \big( D ( \delta, d ) \big)$, (at the cost of increasing the constant) gives
\begin{align}
\norm{ \underset{ d(s,t) \leq \delta }{ \mathsf{sup} } | X_s - X_t | }_{\psi_2} \leq K \int_0^{\delta} \sqrt{ \mathsf{log} \left( 1 + D(\epsilon, d) \right)} d \epsilon. 
\end{align} 
\end{proof}

\paragraph{Symmetrization}

Let $\epsilon_1,...,\epsilon_n$ be \textit{i.i.d} Rademacher random variables. Replace the empirical process
\begin{align}
f \mapsto ( P_n - P ) f = \frac{1}{n} \sum_{i=1}^n \big( f(X_i) - Pf \big), 
\end{align}
with the corresponding symmetrized empirical process defined by $\displaystyle f \mapsto P_n^{o} f = \frac{1}{n} \sum_{i=1}^n \epsilon_i f(X_i)$,where $\epsilon_1,...,\epsilon_n$ are independent of $( X_1,..., X_n )$. Both processes have mean function zero. 
\begin{lemma}
For every nondecreasing, convex $\Phi: \mathbb{R} \to \mathbb{R}$ and class of measurable function $\mathcal{F}$
\begin{align}
\mathbb{E}^{*} \Phi \big( \norm{ P_n - P }_{ \mathcal{F}} \big) \leq \mathbb{E}^{*} \Phi \big( 2 \norm{ P_n^{0} }_{ \mathcal{F}}   \big)
\end{align}
\end{lemma}

\newpage

\subsection{Clivenko-Cantelli theorems}

In this section, we prove two types of Clivenko-Cantelli theorems. The first theorem is the simplest and is based on entropy with bracketing. The second theorem, uses random $L_1-$entropy numbers and is proved through symmetrization followed by a maximal inequality.  

\begin{definition}
(Covering numbers) The covering numbers $N \left( \epsilon, \mathcal{F}, \norm{.} \right)$ is the minimal number of balls $\left\{ g : \norm{ g - f } < \epsilon \right\}$ of radius $\epsilon$ needed to cover the set $\mathcal{F}$. The entropy (without bracketing) is the logarithm of the covering numbers. 
\end{definition}

\begin{definition}
(bracketing numbers) Given two functions $l$ and $u$, the bracket $[l,u]$ is the set of all functions $f$ with $l \leq f \leq u$. An $\epsilon-$bracket is a bracket $[l,u]$ with $\norm{ u - l } < \epsilon$. Then, the bracketing number $N_{[ \ ]} \left( \epsilon, \mathcal{F}, \norm{.} \right)$ is the minimum number of $\epsilon-$brackets needed to cover $\mathcal{F}$. The entropy with bracketing is the logarithm of the bracketing number. 
\end{definition}

\begin{theorem}
Let $\mathcal{F}$ be a class of measurable functions such that $N_{[ \ ]} \left( \epsilon, \mathcal{F}, L_1 (P) \right) < \infty$ for every $\epsilon > 0$. Then, $\mathcal{F}$ is Glivenko-Cantelli. 
\end{theorem}

\begin{proof}
Fix $\epsilon > 0$. Choose finitely many $\epsilon-$brackets $[ \ell_i , u_i ]$ whose union contains $\mathcal{F}$ and such that $P( u_i - \ell_i ) < \epsilon$ for every $i$. Then, for every $f \in \mathcal{F}$, there is a bracket such that 
\begin{align}
( P_n - P ) f \leq ( P_n - P ) u_i + P( u_i - f  ) \leq ( P_n - P ) u_i + \epsilon
\end{align}
Consequently, 
\begin{align}
\underset{ f \in \mathcal{F} }{ \mathsf{sup} } \left( P_n - P \right) \leq \underset{ i }{ \mathsf{max} } (P_n - P ) u_i + \epsilon. 
\end{align}
The right side converges almost surely to $\epsilon$ by the strong law of large numbers for real variables. 
\end{proof}

\medskip

\begin{theorem}
Let $\mathcal{F}$ be a $P-$measurable class of measurable functions with envelope F such that $P^{*} F < \infty$. Let $\mathcal{F}_M$ be the class of functions $f \mathbf{1} \left\{ F \leq M \right\}$ when $f$ ranges over $\mathcal{F}$. If log$N_{[ \ ]} \left( \epsilon, \mathcal{F}, L_1 (P_n) \right) = o_p^{*}(n)$ for every $\epsilon$ and $M > 0$, then $\norm{ P_n - P }_{ \mathcal{F} }^{*} \to 0$ both almost surely and in mean. In particular, $\mathcal{F}$ is GC.  
\end{theorem}

\begin{proof}
By the symmetrization lemma, measurability of the class $\mathcal{F}$, and Fubini's theorem 
\begin{align*}
\mathbb{E}^{*} \norm{ P_n - P }_{\mathcal{F}} 
&\leq 
2 \mathbb{E}_{X} \mathbb{E}_{\epsilon} \norm{ \frac{1}{n} \sum_{i=1}^n \epsilon_i f(X_i) }_{ \mathcal{F} }
\\
&\leq 
2 \mathbb{E}_{X} \mathbb{E}_{\epsilon} \norm{ \frac{1}{n} \sum_{i=1}^n \epsilon_i f(X_i) }_{ \mathcal{F_M} } + 2 \mathbb{P}^{*} F \left\{ F > M  \right\}
\end{align*}
by the triangle inequality, for every $M > 0$.

\newpage

Thus, for sufficiently large $M$, the last term is arbitrarily small. To prove convergence in mean, it suffices to show that the first term converges to zero for fixed $M$. Fix $X_1,...,X_n$. If $\mathcal{G}$ is an $\epsilon-$net in $L_1( P_n )$ over $\mathcal{F}_M$, then the following inequality holds
\begin{align}
\mathbb{E}_{\epsilon} \norm{ \frac{1}{n} \sum_{i=1}^n \epsilon_i f(X_i)   }_{ \mathcal{F}_M } \leq  \mathbb{E}_{\epsilon} \norm{ \frac{1}{n} \sum_{i=1}^n \epsilon_i f(X_i)   }_{ \mathcal{G} } + \epsilon. 
\end{align}
The cardinality of $\mathcal{G}$ can be chosen equal to $N( \epsilon, \mathcal{F}_M, L_1(P_n) )$. Bound the $L_1-$norm on the right using the   Orlicz-norm for $\psi_2 (x) = \mathsf{exp}( x^2 ) - 1$, and using the maximal inequality to find that the last expression does not exceed a multiple of 
\begin{align}
\sqrt{1 + \mathsf{log} N( \epsilon, \mathcal{F}_M, L_1(P_n) ) } \ \underset{ f \in \mathcal{G} }{ \mathsf{sup} } \norm{ \frac{1}{n} \sum_{i=1}^n \epsilon_i f(X_i) }_{ \psi_2 | X } + \epsilon,
\end{align}
where the Orlicz-norm $\norm{ . }_{ \psi_2 | X }$ are taken over $\epsilon_1,..., \epsilon_n$ with $X_1,...,X_n$ fixed.  By Hoeddding's inequality, then can be bounded by $\sqrt{6 / n} \left( P_n f^2 \right)^{1 / 2}$, which is less than $\sqrt{6 / n} M$. 
\end{proof}

\subsection{Donsker Theorems}

\paragraph{Uniform Entropy:} We establish the weak convergence of the empirical process under the condition that the envelope function $F$ be square integrable, combined with the uniform entropy bound
\begin{align}
\int_{0}^{\infty} \sqrt{ \mathsf{log}  N \left( \epsilon, \mathcal{F}_{Q,2}, L_2(Q) \right) } d \epsilon < \infty.  
\end{align}
\begin{theorem}
Let $\mathcal{F}$ be a class of measurable functions that satisfies the uniform entropy bound. Let the class $\mathcal{F}_{\delta} = \left\{ f - g: f,g, \in \mathcal{F}, \norm{ f - g }_{P,2} < \delta \right\}$ and $\mathcal{F}^2_{\infty}$ be $P-$measurable for every $\delta > 0$. If $P^{*} F^2 < \infty$, then $\mathcal{F}$ is $P-$Donsker.  
\end{theorem}

\begin{proof}
Let $\delta_n \to 0$ be a fixed constant. Using Markov's inequality and the symmetrization lemma: 
\begin{align}
\mathbb{P}^{*} \left( \norm{ G_n }_{\mathcal{F}_{ \delta_n }} > x \right) \leq \frac{2}{x} \mathbb{E}^{*} \norm{ \frac{1}{ \sqrt{n} } \sum_{i=1}^n \epsilon_i f( X_i )  }_{ \mathcal{F}_{\delta_n} }. 
\end{align}
Therefore, we can see that the inner expectation is bounded as below
\begin{align}
\mathbb{E}_{\epsilon} \norm{ \frac{1}{ \sqrt{n} } \sum_{i=1}^n \epsilon_i f( X_i )  }_{ \mathcal{F}_{\delta_n} } 
\leq
\int_0^{\infty} \sqrt{ \mathsf{log}  N \left( \epsilon, \mathcal{F}_{\delta_n }, L_2( P_n ) \right) } d \epsilon. 
\end{align}
Notice that for large values of $\epsilon$, the set $\mathcal{F}_{\delta_n }$ fits in a single ball of radius $\epsilon$ around the origin, in which case the integrand is zero.Furthermore, we have that the covering numbers of the class $\mathcal{F}_{\delta_n }$ are bounded by covering numbers of $\mathcal{F}_{\infty} = \left\{ f - g: f,g \in \mathcal{F} \right\}$ (see, \cite{wellner2013weak}).  
\end{proof}

\newpage

\section{Elements of Stochastic Processes}

\subsection{Asymptotic Equicontinuity}

A class of measurable functions is called pre-Gaussian if the (tight) limit process $\mathbb{G}$ in the uniform central limit theorem exists. We focus on Brownian bridge processes. Firstly, it is desirable to have a more concrete description of the tightness property of a Brownian bridge and hence of the motion of pre-Gaussianity. More specifically, tightness of a random map into $\ell^{\infty} ( \mathcal{F} )$ is closely connected to continuity of its sample paths. A Donsker class $\mathcal{F}$ satisfies a stronger condition that the sequence $\mathcal{G}_n$ is asymptotically tight. Therefore, this entails replacing the condition that the sample paths of the limit process are continuous by the condition that the empirical process is asymptotically continuous. 

Thus, for every $\epsilon > 0$ it holds that 
\begin{align}
\underset{ \delta \to 0 }{ \mathsf{lim} } \ \underset{ n \to \infty }{ \mathsf{lim \ sup} } \ \mathbb{P}^{*} \bigg(  \underset{ \rho_p ( f - g) < \delta }{ \mathsf{sup} } \big| \mathcal{G}_n ( f - g ) \big| > \epsilon \bigg) = 0. 
\end{align}
Related discussion can be found in \cite{newey1991uniform} and \cite{ hagemann2014stochastic} among others.

\paragraph{Maximal Inequalities}

Therefore, it follows that the law of large numbers and the central limit theorem are concerned with showing that the supremum of real-valued variables are converges to zero. Thus, to show these results we need to make use of maximal inequalities that bound probabilities involving suprema of random variables. Notice that bounds on finite suprema can be extended to general maximal inequalities with the help of the chaining method.

\paragraph{CLT in Banach Spaces} 

In this particular section it is shown that any CLT in a Banach space can be stated in terms of empirical processes. Furthermore, a class of maximal inequalities can be used to establish the asymptotic equicontinuity of the empirical process.

\begin{theorem}
Let $\psi$ be convex, nondecreasing, nonzero function with $\psi(0) = 0$ and $\underset{ x, y \to \infty }{ \mathsf{lim \ sup} } \ \psi(x) \psi(y) / \psi(cxy) < \infty$, for some constant $c$. Let $\big\{ X_t : t \in T \big\}$ be a separable stochastic process with 
\begin{align}
\norm{ X_s - X_t }_{\psi} \leq C d(s,t), \  \ \ \text{for every} \ s,t,
\end{align}
for some semi-metric $d$ on $T$ and a constant $C$. Then for any $\eta, \delta > 0$, 
\begin{align}
\norm{  \underset{ d(s,t) \leq \delta }{ \mathsf{sup} } \big| X_s - X_t \big| }_{\psi} \leq K \left[ \int_0^{\eta} \psi^{-1} \big( D ( \epsilon, d ) \big) d \epsilon + \delta \psi^{-1} \big( D^2 ( \epsilon, d ) \big) \right],
\end{align}
for a constant $K$ depending on $\psi$ and $C$ only. 
\end{theorem}

\newpage

\begin{corollary}
The constant $K$ can be chosen such that 
\begin{align}
\norm{ X_s - X_t }_{\psi} \leq K \int_0^{ \mathsf{diam} T} \psi^{-1} \big( D ( \epsilon, d ) \big) d \epsilon,
\end{align}
where $\mathsf{diam} T$ is the diameter of $T$. 
\end{corollary}
In practise, maximal inequality means that no point can be added without destroying the validity of the inequality. Furthermore, a stochastic process is called sub-Gaussian with respect to the semimetric $d$ on its index set if
\begin{align}
\mathbb{P} \big( \left| X_s - X_t \right| \big) \leq 2 e^{ - \frac{1}{2} x^2 / d^2 (s,t) }, \ \ \text{for every} \ \ s,t \in T, x > 0.
\end{align}

Another example, is the Rademacher process given by 
\begin{align}
X_{\alpha} = \sum_{i=1}^n \alpha_i \epsilon_i, \ \ \alpha \in \mathbb{R}^n,
\end{align}
for Rademacher variables $\epsilon_1,...,\epsilon_n$. Therefore, by Hoeffding's inequality, this is sub-Gaussian for the Euclidean distance $d(a,b) = \norm{a - b}$.

\paragraph{Tightness under an Increment Bound}

We focus on deriving a general CLT for empirical processes through the application of maximal inequalities. 

\begin{example}
Let $\big\{ X_n(t): t \in [0,1] \big\}$ be a sequence of separable stochastic processes with bounded sample paths and increments satisfying the following condition 
\begin{align}
\mathbb{E} \big| X_n(s) - X_n(t) \big|^p \leq K \left| s - t \right|^{1 + r}, 
\end{align}
for constants $p, K, r > 0$ independent of $n$. Assume that the sequences of martingales $\big( X_n (t_1),..., X_n (t_k) \big)$ converge weakly to the corresponding marginals of a stochastic process $\big\{ X(t): t \in [0,1] \big\}$. Then, there exists a version of $X$ with continuous sample paths and $X_n \Rightarrow X$ in $\ell^{\infty} [0,1]$. Hence, also in $D[0,1]$ or $C[0,1]$, provided every $X_n$ has all its sample paths in these spaces.  
\end{example}

\newpage

\subsection{LLNs for Hilbert Space-Valued Mixingales}

\begin{example}
\textbf{(Regression via Orthonormal Bases)}
\
 
Let $\left\{ X_t \right\}$ and $\left\{ Y_t \right\}$ be real-valued random sequences (see, also \cite{hu2004complete}):
\begin{align}
\mathbb{E} \left[ Y_t | X_t = x \right] := \theta^0 (x)
\end{align} 
Define $\epsilon_t \equiv Y_t - \theta^0 ( X_t )$, and suppose $\mathbb{E} \left[ \epsilon_t^2 | X_t \right] = \rho^2_t \leq \sigma^2 < \infty$. Suppose that, for all $t$, $X_t$ has the same marginal density $f$ with bounded support and $0 < \mathsf{inf} \ f(x) \leq \mathsf{sup} \ f(x) < \infty$ where $M^0 (x) \equiv \theta^0 (x)$ and $f(x)$ belong to an infinite-dimensional separable Hilbert space $H$ with inner product induced norm $| \ . \ |$. Let $\left\{ g_j  \right\}$ be an orthonormal basis for $H$ and $\left\{ J_n \right\}$ a nondecreasing integer sequence. Estimate $M^0$ as 
\begin{align}
\hat{M}_n (x) = \underset{ 1 \leq j \leq J_n }{ \sum } \left[ n^{-1} \sum_{ t = 1 }^n Y_t g_j( X_t ) \right] g_j (x). 
\end{align}
For known $f$, estimate $\theta^0$ by $\hat{\theta}_n = \hat{M}_n / f$. For unknown $f$, estimate $\theta^0$ by $\hat{\theta}_n = \hat{M}_n / \hat{f}_n$, 
\begin{align}
\hat{f}_n \equiv \underset{ 1 \leq j \leq J_n }{ \sum } \left[ n^{-1} \sum_{ t = 1 }^n g_j( X_t ) \right] g_j (x). 
\end{align}
Furthermore, if $\left\{ X_t \right\}$ and $\left\{ Y_t \right\}$ are $\mathbb{R}-$valued near-epoch dependent (NED) functions of some mixing random sequences, $\left\{ \hat{M}_n - \mathbb{E} \hat{M}_n \right\}$ and $\left\{ \hat{f}_n - \mathbb{E} \hat{f}_n \right\}$, become weighted sums of Hilbert space-valued mixingale arrays. Our results, in the next section can establish the convergence of $\mathbb{E} \left[ | \hat{\theta}_n - \theta^0 | \right]$. For example, when $H = L_2 ( \mathbb{R} )$, we obtain that $\int_{R} \left[ \hat{\theta}_n (x) - \theta^0 (x) \right]^2 dx \to 0$, in probability.  
\end{example}

\begin{theorem}
Let $\left\{ W_{n,i}, \mathcal{F}^{n,i} \right\}$ be an $L_p(H)$ mixingale with $p \geq 1$. 

\begin{enumerate}
\item[(i)] If $p \geq 2$, then there exists a doubly infinite summable sequence of positive constants defined with $a \equiv \left\{ a_m : - \infty < m < \infty \right\}$ where $a_m = a_{ - m}$ such that 
\begin{align}
\mathbb{E} \left[ \underset{ j \leq k }{ \mathsf{max} } \left| \sum_{ i=1}^j W_{n,i} \right|^2 \right] &\leq K( \psi, a ) \sum_{ i=1}^k c^2_{n,i}, 
\\
K( \psi, a ) &= 4 \left( \underset{ - \infty < m < \infty }{ \sum } a_m \right) \left[ a_0^{-1} \left( \psi_0^2 + \psi_1^2 \right) + 2 \sum_{m =1}^{ \infty } \psi_m^2 | a_m^{-1} - a_{m-1}^{-1} | \right]
\end{align}

\item[(ii)] If $1 < p \leq 2$, then there exists a constant $C_p > 0$ depending only on $p$ such that 
\begin{align*}
\norm{ \underset{ j \leq k }{ \mathsf{max} } \left| \sum_{i=1}^j W_{n,i} \right| }_p 
\leq 
C_p \underset{ - \infty < m < \infty  }{ \sum  } \left( \sum_{i=1}^k \mathbb{E} \left[ \left| W_{n,i}^m  \right|^p \right] \right)^{ 1 / p} 
\leq 2 C_p \sum_{ m = 1 }^{ \infty } \psi_m \sum_{ i = 1 }^{ k } 
\left( c_{n,i}^p \right)^{1/ p}, 
\end{align*}
where $W_{n,i}^m \equiv \mathbb{E} \left[ W_{n,i} | \mathcal{F}^{n, i - m} \right] - \mathbb{E} \left[ W_{n,i} | \mathcal{F}^{n, i - m - 1} \right]$. 
\end{enumerate}
\end{theorem}

\newpage 

\section{Elements of Bayesian Statistics}

\begin{quotation}
"...\texttt{the simple idea of splitting a sample in two and then developing the hypothesis on the basis of one part and testing it on the remainder may perhaps be said to be one of the most seriously neglected ideas in statistics}...".
\end{quotation}

The applications of the Bayesian framework in economics, econometrics and statistics are widely spread and it is worth mentioning some key elements here, especially due to the fact that model selection is commonly presented using the Bayesian approach. Generally, statistical problems are concerned for the model $f( x | \theta )$, such that $\theta \in \Theta$ with the associated hypothesis testing of interest is formulated as below:
\begin{align}
H_0: \theta = \theta_0 \ \ \ \text{vs} \ \ \ \theta \neq \theta_0,    
\end{align}
where the data $X = \big( X_1,..., X_n \big) \overset{ \textit{i.i.d} }{ \sim } f ( X | \theta )$. Denote the prior probability function with $\pi ( \theta )$, then one can construct the observed posterior $\pi (  \theta | X )$, assumed to be a proper density function even if the prior is improper (see,   ). Next, one needs to consider what we compare $\pi( \theta | X )$ with. Assume that we can generated data $Y = \big( Y_1,..., Y_n \big) \overset{ \textit{i.i.d} }{ \sim } f( Y | \theta_0 )$ under the null hypothesis. Using the same prior we can construct $\pi ( \theta | Y )$ and hence we can define the expected posterior under the null; $\pi_0 ( \theta ) = \displaystyle \int \pi ( \theta | y ) f ( y | \theta_0 ) dy$, where $y$ is an $n-$vector. Then, the test statistic of interest is the KL divergence between the expected posterior under the null and the observed posterior such that
\begin{align}
T(X) = \int \pi_0( \theta ) \mathsf{log} \left[ \frac{  \pi_0 ( \theta ) }{  \pi( \theta | X ) } \right] d \theta.     
\end{align}
Therefore, it can be shown that $T(X)$ is related to the Bayes factor. 

\medskip

\begin{example}
Consider the probability distributions $\mathbb{P} ( x | H_0 )$ and $\mathbb{P} ( x | H_1 )$. Moreover, assume some prior knowledge on prior probabilities $\mathbb{P} ( H_0 )$ and $\mathbb{P} ( H_1 )$. Then, Bayes theorem combines the prior probabilities and the data to produce posterior probabilities $\mathbb{P} ( H_0 | x )$ and $\mathbb{P} ( H_1  | x )$. Therefore, the transformation of prior to posterior itself represents evidence provided by the data, which takes the following form 
\begin{align}
\mathbb{P} ( H_k | x ) = \frac{ \mathbb{P} ( x | H_k ) \mathbb{P}( H_k ) }{  \mathbb{P} ( x | H_0 ) \mathbb{P}( H_0)  + \mathbb{P} ( x | H_1 ) \mathbb{P}( H_1 ) }    
\end{align}
In other words, the Bayes factor represents the the likelihood ratio. When there is an unknown parameter $\theta_k$, corresponding to hypothesis $H_k$, the Bayes factor is the marginal likelihood ratio, where the marginal likelihood densities $m_k(x)$ are obtained by integrating over the parameter space with respect to the specified prior $\pi ( \theta_k | H_k )$, such that
\begin{align}
m_k(x) = \int f_k ( x | \theta_k ) \pi_k ( \theta_k ) d \theta_k, \ \ \ \text{for} \ k \in \left\{ 0, 1 \right\}.     
\end{align}
However, there are difficulties with the Bayes factor when prior information about the unknown parameters of the models is weak, in particular, with the use of improper priors \cite{chen2021new}. 
\end{example}

\newpage

Below, we present some popular econometric models. Some relevant literature that studies Bayesian asymptotics within the time series econometric context include among others   \cite{kim1994bayesian}. 
\begin{example}
Consider the autoregressive model 
\begin{align}
y_t = \rho y_{t-1} + \epsilon_t
\end{align}
Without imposing the normality assumption on the innovation sequences, we can consider the asymptotic posterior distribution. A posterior density is defined as below (see, e.g., \cite{kim1994bayesian})
\begin{align}
P_T( g, Y_T(\omega) ) = \int_G \pi( \theta | Y_T(\omega) ) d \theta 
\end{align}
The Bayesian posterior is determined by the prior and the likelihood function. Therefore, in order to verify the asymptotic normality of the posterior distribution, then we need to determine the limiting behaviour of each of the components comprising the prior and the posterior distributions.
\end{example} 
Additional useful reading within the context of model selection and empirical bayes include among others \cite{varin2011overview}\footnote{Professor Cristiano Varin gave a seminar with title "An approximate empirical Bayes approach to paired comparisons", at the S3RI Departmental Seminar Series at the University of Southampton on the 28th of April 2022. \\}, \cite{efron2014two} as well as \cite{jewson2022general}\footnote{Professor David Rossell gave a seminar with title "Improper models for data analysis", at the S3RI Departmental Seminar Series at the University of Southampton on the 10th of February 2022.}. Moreover, a framework for Bayesian Inference in Econometric Models using Monte Carlo Integration is proposed by \cite{geweke1989bayesian}.

\begin{example}
Consider the structural regression model given by 
\begin{align}
Y_t &= \alpha + \beta X_{t-1} + U_t
\\
X_t &= \mu + \rho X_{t-1} + V_t
\end{align}
where $( U_t, V_t )^{\top}$ is a sequence of independent and identically distributed random vectors with means zeros and finite variances. As it is well documented in the literature the least squares estimator for $\beta$ based on the predictive regression is biased in finite sample behaviour due to the correlation structure between the innovation sequences of the equations of the system $U_t$ and $V_t$. Hence, several bias-corrected estimators and tests for both stationary $( | \rho | < 1 )$ and nonstationary or nearly nonstationary, $\rho = \left( 1 - \displaystyle \frac{c}{n} \right)$ have been proposed in the literature. In other words, inference on the parameter $\beta$ of the model is challenging due to the fact that the asymptotic distribution of an estimators or a test statistic depends heavily on whether $X_t$ is stationary or nearly integrated or unit root, and whether the model intercept $\phi$, of the nonstationary autoregressive equation is zero or not. Therefore, is of paramount importance to have a unified inference approach to avoid making a mistake in characterizing the predicting variable. The persistence endogeneity of covariates in predictive regressions has been studied by several authors.

\newpage

Therefore an estimation methodology is to employ the empirical likelihood approach which is defined: 
\begin{align}
L_n (\beta) = \mathsf{sup} \left\{ \prod_{t=1}^n (n p_t): p_1 \geq 0,..., p_n \geq 0, \sum_{t=1}^n p_t = 1, \sum_{t=1}^n p_t Z_t ( \beta )     \right\}    
\end{align}
where $Z_t (\beta) = \big( Y_t - \beta X_{t-1} \big) X_{t-1} \big/ \sqrt{ 1 + X^2_{t-1} }$. Then, based on the Lagrange multiplier optimization method it follows that 
\begin{align}
\ell_n (\beta) := - 2 \mathsf{log} L_n (\beta) = 2 \sum_{t=1}^n \mathsf{log} \big\{ 1 + \lambda Z_t (\beta) \big\}    
\end{align}
Then, taking the first-order-condition for some $\lambda = \lambda (\beta)$ satisfies
\begin{align}
\sum_{t=1}^n \frac{ Z_t (\beta) }{ 1 + \lambda Z_t (\beta) } \equiv 0.    
\end{align}
\end{example}

\medskip

\begin{remark}
The first rigorous work to define and construct tests which are asymptotically optimal was \cite{wald1943efficient}. He argued that maximum likelihood estimators may be asymptotically sufficient for detecting local deviations from the null hypothesis and showed that the Wald test - is asymptotically most stringent: the asymptotic power function is closest to the asymptotic envelope power function in the minimax sense in local neighborhoods of the null hypothesis. In particular, Wald tests can be used for more general problems if we can find \textit{asymptotically efficient} estimates for the parameters of interest - an estimator $\hat{\vartheta}_n$ with $\sqrt{n} \big( \hat{\vartheta}_n - \vartheta_n \big)$ asymptotically $\mathcal{N} \big( 0, B^{* - 1} \big)$ under every $\big( \vartheta_n( h_{\vartheta} ), \eta_n ( h_n  ) \big)$. 
\end{remark}

Moreover, the notion of \textit{asymptotically uniformly most powerful} (AUMP), AUMPU (unbiased) and AUMPI (invariant) tests are useful in semiparametric econometrics. In particular, characterization is done by stating the asymptotic local power function. Furthermore, sufficient for \textit{optimality} is that a test be equivalent to a canonical effective score test - an optimal test requiring knowledge of nuisance parameters. In addition, Stein's notion of \textit{adaptation} implies replacing, in tests which are optimal when certain nuisance parameters are known, these parameters by estimates without affecting the asymptotic performance of the test. The particular aspect is the equivalent of large-sample studentization, a variance parameter can be replaced by an estimate without large-sample penalty. Note that a statistical problem is \textit{invariant under locally linear transformation} of the parameter of interest, \textit{iff} standardized effective score tests are rotation invariant (asymptotically). 

Furthermore, another large stream of literature in econometrics focuses on estimation and inference methods for \textit{conditional moment models} which can be extended within a high-dimensional environment using techniques based on large sample approximations (see, \cite{ai2003efficient} and \cite{dominguez2004consistent} among others). A Bayesian perspective is presented by \cite{chib2022bayesian}.

\newpage

\bibliographystyle{apalike}
\bibliography{myreferences1}

\addcontentsline{toc}{section}{References}

\newpage

\end{document}